\documentclass{uiucthesis2021}
\usepackage[utf8]{inputenc}
\usepackage[english]{babel}
\usepackage{csquotes}
\usepackage{microtype}
\usepackage{amsmath,amsthm,amssymb,graphicx,mathrsfs,mathtools}
\usepackage[bookmarksdepth=3,linktoc=all,colorlinks=true,urlcolor=blue,linkcolor=blue,citecolor=blue]{hyperref}
\usepackage[capitalize]{cleveref}
\usepackage[style=ieee,citestyle=numeric-comp]{biblatex}
\usepackage{booktabs}
\usepackage{bm}
\usepackage{appendix}
\usepackage{pifont,blkarray}
\usepackage{titletoc}
\usepackage{tabularx}
\usepackage{ltablex}
\usepackage{tikz-cd}
\usepackage{enumerate}
\usepackage{nccmath}
\usepackage{CJKutf8}

\usepackage{lipsum}  

\newcommand{\nn}{\nonumber}
\newcommand{\bb}{\mathbb}
\newcommand{\scr}{\mathscr}
\newcommand{\p}{\partial}
\newcommand{\tr}{\text{tr}}
\newcommand{\LCGamma}{\mathring{\Gamma}}

\newcommand{\LCnabla}{\mathring{\nabla}}

\newcommand{\un}{\underline}

\newcommand{\ti}{\tilde }
\newcommand{\td}{\textnormal{d}}
\newcommand{\ts}{\textnormal{s}}
\newcommand{\E}{\textnormal{e}}
\def\I{\textnormal{i}}

\newcommand{\be}{\begin{equation}}
\newcommand{\ee}{\end{equation}}
\newcommand{\beqn}{\begin{eqnarray}}
\newcommand{\eeqn}{\end{eqnarray}}
\newcommand{\beq}{\begin{eqnarray}}
\newcommand{\eeq}{\end{eqnarray}}
\newcommand{\mX}{\mathfrak{X}}
\newcommand{\mY}{\mathfrak{Y}}
\newcommand{\mZ}{\mathfrak{Z}}
\newcommand{\RR}{\mathbb{R}}
\newcommand{\mD}{\mathfrak{D}}
\newcommand{\mg}{\mathfrak{g}}
\newcommand{\umX}{\un{\mX}}
\newcommand{\umY}{\un{\mY}}
\newcommand{\umZ}{\un{\mZ}}
\newcommand{\uX}{\un{X}}
\newcommand{\uY}{\un{Y}}

\newcommand{\uE}{\un{E}}
\newcommand{\ut}{\un{t}}
\newcommand{\uA}{{\un{A}}}
\newcommand{\uB}{{\un{B}}}
\newcommand{\uM}{{\un{M}}}
\newcommand{\uN}{{\un{N}}}
\newcommand{\umu}{\un{\mu}}
\newcommand{\unu}{\un{\nu}}
\newcommand{\ualpha}{\un{\alpha}}
\newcommand{\ubeta}{\un{\beta}}
\newcommand{\End}{\text{End}}
\newcommand{\Der}{\text{Der}}

\newcommand{\hatd}{\hat{\td}}
\newcommand{\Aconn}[1]{\phi_{#1}}

\newcommand{\econn}{\hat{A}}
\newcommand{\ecurv}{\hat{F}}

\theoremstyle{definition}
\newtheorem{defn}{Definition}[chapter]
\newtheorem{rem}{Remark}[chapter]
\newtheorem{ex}{Example}[chapter]

\theoremstyle{plain}
\newtheorem{thm}{Theorem}[chapter]

\newtheorem{prop}[thm]{Proposition}
\newtheorem{lemma}[thm]{Lemma}

\newcommand{\bthm}{\begin{thm}}
\newcommand{\ethm}{\end{thm}}
\newcommand{\bpro}{\begin{prop}}
\newcommand{\epro}{\end{prop}}
\newcommand{\bdefi}{\begin{defn}}
\newcommand{\edefi}{\end{defn}}
\newcommand{\bex}{\begin{ex}}
\newcommand{\eex}{\end{ex}}
\newcommand{\brem}{\begin{rem}}
\newcommand{\erem}{\end{rem}}
\newcommand{\bpf}{\begin{proof}}
\newcommand{\epf}{\end{proof}}

\newcounter{counterforappendices}

\allowdisplaybreaks[1]

\begin{document}

\title{Topics in Weyl Geometry and Quantum Anomalies}
\author{Weizhen Jia}
\department{Physics}
\phdthesis
\degreeyear{2024}
\committee{
    Associate Professor Thomas Faulkner, Chair\\
    Professor Robert G. Leigh, Director of Research\\
    Professor Taylor L. Hughes\\
    Associate Professor Jorge Noronha}
\maketitle

\frontmatter

\begin{abstract}
The interplay between geometry, symmetry, and physics reveals fundamental insights of Nature. In this thesis we explore several facets of these topics, including Weyl geometry and its applications in holographic duality, and the geometric structure of gauge theory and quantum anomalies in the language of Lie algebroids.

The first part of this thesis focuses on the Weyl-covariant nature of holography. The conformal boundary of an asymptotically locally AdS (ALAdS) spacetime carries a conformal geometry. The commonly used Fefferman-Graham (FG) gauge explicitly breaks the Weyl symmetry of the boundary theory. This can be resolved by applying the Weyl-Fefferman-Graham (WFG) gauge, in which the boundary carries a Weyl geometry, which is a natural extension of conformal geometry with the Weyl covariance mediated by a Weyl connection. Based on this idea, we generalize the Fefferman-Graham ambient construction for conformal geometry to a corresponding construction for Weyl geometry. We modify the FG ambient metric into a Weyl-ambient metric by implementing the WFG gauge, then we show that the Weyl-ambient space as a pseudo-Riemannian geometry at codimension-2 a Weyl manifold. Conversely, we also show that the Weyl-ambient metric can be uniquely reconstructed from a codimension-2 Weyl manifold provided the initial data of the metric and Weyl connection. Through the Weyl-ambient construction, we investigate Weyl-covariant quantities on the Weyl manifold and define Weyl-obstruction tensors. We show that Weyl-obstruction tensors appear as poles in the Fefferman-Graham expansion of the ALAdS bulk metric for even boundary dimensions. Under holographic renormalization in the WFG gauge, we compute the Weyl anomaly of the boundary theory in multiple dimensions and demonstrate that Weyl-obstruction tensors can be used as the building blocks for the Weyl anomaly of the dual quantum field theory (QFT). Furthermore, the holographic calculation with a background Weyl geometry also suggests an underlying geometric interpretation of the Weyl anomaly, which motivates the second part of this thesis.

The second part of this thesis is devoted to understanding the geometric nature of the Becchi-Rouet-Stora-Tyutin (BRST) formalism and quantum anomalies. Conventionally, the geometric interpretation for anomalies is studied through the Wess-Zumino consistency condition and descent equations, where the anomaly lives in the ghost number one sector of the BRST cohomology. Using the language of Lie algebroids, the BRST complex can be encoded in the exterior algebra of an Atiyah Lie algebroid derived from the principal bundle of the gauge theory. We develop the correspondence of the BRST cohomology and the Lie algebroid cohomology. We showed explicitly that the cohomology of an Atiyah Lie algebroid in a trivialization gives rise to the BRST cohomology. In addition, in the framework of Lie algebroid, the gauge transformations and diffeomorphisms are implemented on an equal footing. We then apply the Lie algebroid cohomology in studying quantum anomalies and demonstrate the computation for chiral and Lorentz-Weyl (LW) anomalies. In particular, we pay close attention to the fact that the geometric intuition afforded by the Lie algebroid (which was absent in the traditional BRST complex) provides hints of a deeper picture that simultaneously geometrizes the consistent and covariant forms of the anomaly. In the algebroid construction, the difference between the consistent and covariant anomalies is simply a different choice of basis. This indicates that the Lie algebroid cohomology is indeed a suitable formulation for geometrizing quantum anomalies. 

The two parts of this thesis are structured to be self-contained and can be read independently. While each part delves into distinct topics, they converge on the subject of the Weyl anomaly. Collectively, they contribute to advancing our understanding of the Weyl anomaly from various perspectives.
\end{abstract}

\begin{dedication}
To my parents, whose love sustained me through my years overseas, enabling the pursuit of my aspirations.
\par
\&
\par
To my motherland, whose profound culture and history have imparted wisdom, shaping me as a human.
\end{dedication}

\begin{acknowledgments}

As I finalize this thesis, my journey of PhD is also coming to its completion. I am grateful to those who supported me along the way. 
\par
First and foremost, I would like to thank my advisor Prof.~Rob Leigh. Spending two years taking his classes on quantum field theory and AdS/CFT laid a solid foundation for conducting the research presented in this thesis. After joining his group, his passion for physics and profound way of thinking have left a lasting impact on me. I am also grateful for his guidance in fostering my skills of independent study, which I consider critical for my career in physics. I would like to sincerely thank the members on my doctoral committee, Profs.~Tom Faulkner, Taylor Hughes and Jorge Noronha. The enlightening conversations with each of them over the years have offered me valuable insights from across different fields in physics, which encouraged me to develop an interdisciplinary mindset. 
\par
This thesis could not be completed without the contributions of my collaborators. I would like to thank Manthos Karydas and Marc Klinger, with whom I had the privilege to work. Manthos's support, especially during my early days in the group, has been particularly meaningful to me. His insightful ``interrogations" have consistently deepened my understanding of the issues addressed within or beyond this thesis. Marc possesses an exceptional breadth and depth of knowledge in mathematics and a remarkable sense of connecting them with physics. The discussions with him never failed to be educational to me. I would like to express special thanks to Dr.~Luca Ciambelli for his constant support throughout my research projects. His illuminating comments have inspired the resolution of many challenges I encountered. His advice and support regarding my academic career have also been truly invaluable.
\par
I would like to thank Prof.~Mike Stone for his lectures on mathematical physics and for always being open to having discussions regarding anomalies, which have been extremely useful. I also thank Profs.~Nigel Goldenfeld, Jessie Shelton, Sheldon Katz, Rui Fernandes, Eduardo Fradkin and Barry Bradlyn, for their excellent lectures in physics and math departments. Taking their course and interacting with them has not only enriched my understanding of relevant subjects covered in this thesis, but also broadened my horizon and made me more versatile for my future studies. I am also thankful to Prof. Lance Cooper for his miscellaneous assistance and financial support, which played a crucial role in the completion my PhD and this thesis. 
\par
\begin{CJK*}{UTF8}{bsmi}
Over the years, I have received countless support from many colleagues and friends. Many thanks to Sammy Goldman and Pin-Chun Pai (白秉錞) for being wonderful research groupmates. They provided numerous valuable feedbacks on my presentations during group meetings, and always created a positive and motivating environment.
I would also like to thank the fellows from the high energy theory group, including Fikret Ceyhan, Min Li (李敏), Xuchen Cao (曹旭晨), Kort Beck, Jia Wang (王佳), Antony Speranza, Animik Ghosh, and Zhencheng Wang (王振丞), as well as Enrico Speranza and Kevin Ingles from the nuclear theory group, for all the helpful conversations.\end{CJK*} \begin{CJK*}{UTF8}{gbsn}I also want to thank former students and postdocs of the high energy theory group who have been supportive at various stages of my PhD, including Simon Lin (林山), Srivatsan Balakrishnan, Souvik Dutta, Yidong Chen (陈一东), Udit Gupta, Jackson Fliss, and Sagar Lokhande. Besides, I am deeply grateful to my friends at ICMT, Akash Vijay, Jacopo Gliozzi and Dmitry Manning-Coe, who have always been warm and encouraging to me. The stimulating clash of ideas between my background and their condensed matter expertise has also taught me a great deal.
\end{CJK*}

\begin{CJK*}{UTF8}{gbsn}
I am greatly honored for the friendships I have formed with Brandon Buncher, Dalton Chaffee, Marcus Rosales, Sai Paladugu, Nick Abboud, Spencer Johnson, Shaun Lahert and Pranjal Ralegankar. I appreciate all the fun times we shared and especially for helping me blend into the culture when I first arrived in this country. To my Chinese friends in Urbana, Cunwei Fan (范存炜), Xueying Wang (王雪莹), Minhui Zhu (朱\end{CJK*}\begin{CJK*}{UTF8}{bsmi}旻\end{CJK*}\begin{CJK*}{UTF8}{gbsn}晖), Jiayu Shen (沈佳宇), Yu Ding (丁遇), Yunkai Wang (王云开), Kai Zheng (郑楷) and Yumu Yang (杨雨木), interacting with each of you has been inspiring and uplifting.

Lastly, I want to express my gratitude to Penghao Zhu (朱鹏浩) and Mengdi Zhao (赵梦迪) for their unwavering friendship. The moments spent hanging out with this lovely couple, chatting about physics and life, and sharing homemade food to each other provided me with the mental strength to get through tough times. These are treasured memories that I will always cherish.
\end{CJK*}

\end{acknowledgments}

{
    \hypersetup{linkcolor=black}  
    \tableofcontents
    \addtocontents{toc}{\vspace{-\baselineskip}}
}

\mainmatter
\part{Weyl-Ambient Metrics, Obstruction Tensors and Holography}
\chapter{Introduction}
\label{chap:Intro1}
\section{Backgrounds on Geometry}
Conformal geometry is a very rich area of mathematics with its history deeply intertwined with that of physics. Historically, the subject was initiated at the beginning of the twentieth century with the work of Hermann Weyl \cite{Weyl:1918pdp}, \'Elie Cartan \cite{CartanLesE} and Tracy Y.~Thomas \cite{thomas1925invariants}. In physics, there have been numerous applications of conformal geometry, from conformal compactification \cite{Penrose:1962ij} and conformal gravity \cite{Mannheim:2011ds} to the anti-de Sitter/conformal field theory (AdS/CFT) correspondence \cite{Maldacena:1997re,Witten:1998qj}. 
 \par
The fundamental structure appearing in conformal geometry is a manifold $M$ endowed with a \emph{conformal class} of metrics $[g]$. Two metrics belong in the same conformal class $[g]$ if one metric is a smooth positive multiple of the other. Local rescalings of the metric tensor by an arbitrary smooth positive function are called \emph{Weyl transformations}. Compared to pseudo-Riemannian manifolds $(M,g)$, conformal manifolds are endowed with an enlarged symmetry group with both diffeomorphisms and Weyl transformations, denoted by $\text{Diff}(M)\ltimes \text{Weyl}$. A tensor $T$ on a conformal manifold $(M,[g])$ is said to be conformally covariant if it transforms covariantly under a Weyl transformation:
\begin{equation}\label{conformal_Weyl_tensor}
T\to {\cal B}(x)^{w_{T}}T\,,\qquad\text{when}\qquad g\to {\cal B}(x)^{-2}g\,,
\end{equation}
where $w_{T} $ is  the Weyl weight of the tensor $T$. On the physics side, conformal-covariant tensors appear as expectation values of operators in conformal field theories coupled to a background metric. As an important example, the expectation value of the trace of the energy-momentum tensor acquires an anomalous term after quantization, namely the celebrated Weyl anomaly \cite{Capper:1974ic}, which will be discussed in detail shortly. By investigating the effective action in dimensional regularization, Deser and Schwimmer \cite{Deser:1993yx} made a conjecture regarding the possible candidates for the Weyl anomaly, which are global conformal invariants. This conjecture was later proven in \cite{alexakis2012decomposition,Boulanger:2007st,Boulanger:2007ab}.\footnote{The analysis in \cite{alexakis2012decomposition} concerns local conformal invariants, corresponding to the type B Weyl anomaly, while \cite{Boulanger:2007ab,Boulanger:2007st} deals with the type A Weyl anomaly.}

Just as diffeomorphism-covariant quantities, i.e.,\ tensors, on pseudo-Riemannian manifolds can easily be constructed out of the metric, Riemann tensor and covariant derivatives, one might expect to find conformal-covariant tensors on conformal manifolds.  
However, unlike the abundance of diffeomorphism-covariant quantities on $(M,g)$, it is significantly harder to construct conformal-covariant tensors on $(M,[g])$. Before the work of Fefferman and Graham, the only known examples of conformal tensors were the Weyl tensor $W_{ijkl}$ (traceless part of the Riemann tensor $R_{ijkl}$) in any dimension, the Cotton tensor $C_{ijk}$ \cite{cotton1899varietes} in $3d$ and the Bach tensor $B_{ij}$ \cite{bach1921weylschen} in $4d$. By means of the Schouten tensor
\be
P_{ij}=\frac{1}{d-2}\bigg(R_{ij}-\frac{1}{2(d-1)}Rg_{ij}\bigg)\,,
\ee
these tensors can be expressed as
\begin{align}
\label{Weyltensor}
W_{ijkl}&=R_{ijkl}-g_{ik}P_{jl}-g_{jl}P_{ik}+g_{jk}P_{il}+g_{il}P_{jk}\,,\\
C_{ijk}&=\nabla_{k}P_{ij}-\nabla_{j} P_{ik}\,,\\
B_{ij}&=\nabla^k\nabla_k P_{ij}-\nabla^k\nabla_{j} P_{ik}-W_{ljik}P^{kl}\,.
\end{align}
In their seminal work \cite{AST_1985__S131__95_0,Fefferman:2007rka}, Fefferman and Graham introduced the ambient metric construction based on previous work by Fefferman \cite{fefferman1979parabolic}, which provided a systematic method of finding conformal-covariant tensors. The basic idea of the construction was to associate  a $(d+2)$-dimensional ``ambient" pseudo-Riemannian manifold to a $d$-dimensional conformal manifold. One can then find a specific class of ambient diffeomorphisms that induces Weyl transformations on the conformal manifold. 
\par
An important outcome of the ambient construction was to define \emph{extended obstruction tensors} from covariant derivatives of the ambient Riemann tensor \cite{graham2009extended}. Obstruction tensors are the generalization to higher (even) dimension of the Bach tensor. For each even dimension $d\geqslant4$, the corresponding obstruction tensor is the only irreducible conformal-covariant tensor in that dimension \cite{graham2005ambient}. Defined through the ambient space, the $k^{th}$ extended obstruction tensor $\Omega^{(k)}_{ij}$ has a simple pole at $d=2k+2$, whose residue is the obstruction tensor in that dimension. For example, the first extended obstruction tensor reads
\be
\Omega^{(1)}_{ij}=-\frac{1}{d-4}B_{ij}\,,
\ee
where $B_{ij}$ is the Bach tensor, namely the obstruction tensor in $4d$. 
\par
A different perspective on conformal geometry was introduced by Weyl \cite{Weyl:1918pdp}, whose idea was to make the physical scale a local quantity. The Weyl connection was introduced so that one can transport the physical scale between two points of the manifold. Although Weyl's initial attempt to identify the Weyl connection with the electromagnetic gauge field failed, the consistent mathematical structure he introduced was developed further in \cite{10.4310/jdg/1214429379,doi:10.1063/1.529582}. In this approach, a Weyl connection $a$ is introduced on the conformal manifold which transforms together with the metric $g$ under a Weyl transformation. One can modify the conformal class $[g]$ to a \emph{Weyl class} $[g,a]$, which is the equivalence class formed by the pairs $(g,a)\sim ({\cal B}(x)^{-2}g, a- \td\ln{\cal B}(x))$. This defines a Weyl manifold $(M,[g,a])$, and the conformal geometry is promoted to \emph{Weyl geometry} \cite{10.4310/jdg/1214429379,doi:10.1063/1.529582,scholz2018unexpected}. Equivalently, a Weyl connection can be thought of as a connection on the \emph{Weyl structure}, which is a principal bundle with the Weyl symmetry group as the structure group \cite{10.4310/jdg/1214429379}. 
\par
Similarly to a conformal-covariant tensor, one can define a Weyl-covariant tensor $T$ on a Weyl manifold $(M,[g,a])$ to be a tensor that transforms covariantly under a Weyl transformation:
\begin{equation}\label{Weyl_tensor}
\begin{split}
T\to {\cal B}^{w_{T}}(x)T\,,\qquad\text{when}\qquad g\to {\cal B}(x)^{-2}g\,,\quad a\to a- \td \ln {\cal B}(x)\,.
\end{split}
\end{equation}
Although conformal-covariant tensors on a conformal manifold $(M,[g])$ are hard to find, Weyl-covariant tensors on a Weyl manifold $(M,[g,a])$ can be constructed quite easily. Recall that on a pseudo-Riemannian manifold $(M,g)$, one can define a Levi-Civita (LC) connection $\nabla$, and it is well-known that diffeomorphism-covariant quantities can be constructed from the metric, Riemann curvature, and covariant derivatives  of the Riemann curvature. On a Weyl manifold $(M,[g,a])$, one can define a Weyl-Levi-Civita connection $\hat\nabla$, and a plethora of Weyl-covariant quantities can similarly be constructed from the metric, Weyl-Riemann curvature, and  Weyl-covariant derivatives $\hat\nabla$ of the Weyl-Riemann curvature. This indicates that the $\text{Diff}(M)\ltimes \text{Weyl}$ symmetry is manifested more naturally on a Weyl manifold, and the representation has a similar structure as that of  $\text{Diff}(M)$ on  pseudo-Riemannian manifolds. There are corresponding notions of Weyl metricity, Weyl torsion and a uniqueness theorem giving a Weyl-LC connection \cite{10.4310/jdg/1214429379,Condeescu:2023izl}. 
\par
From the geometry side, the main goal of Part I of this thesis is to provide an ambient construction for Weyl manifolds. We start by introducing the Weyl-ambient metric as a modification of the FG ambient metric. We will then present two perspectives. The first one is a top-down approach. We will see that one naturally obtains a codimension-2 Weyl manifold $(M,[g,a])$. A more formal approach is the bottom-up perspective, where we start from a $d$-dimensional conformal manifold $(M,[g])$, which is then enhanced into a Weyl manifold $(M,[g,a])$ by introducing a connection on the Weyl structure over $M$. A $(d+2)$-dimensional Weyl-ambient space can then be constructed by taking the Weyl structure as an initial surface, which follows the rigorous ambient space construction in \cite{Fefferman:2007rka}. We also provide a definition of Weyl-obstruction tensors on a Weyl manifold $(M,[g,a])$ through the Weyl-ambient space $(\tilde M,\tilde g)$, in a way analogous to how obstruction tensors were defined in \cite{graham2009extended,Fefferman:2007rka}. Many properties of the extended Weyl-obstruction tensors can also be derived from the Weyl-ambient space. 

\section{Backgrounds on Physics}
\par
To physicists, perhaps a more familiar scenario is lying on a hyperbola in the ambient space, namely a $(d + 1)$-dimensional asymptotically locally AdS (ALAdS) geometry, usually referred to in the physics literature as the ``(AL)AdS bulk." The conformal boundary of an ALAdS spacetime is an important example of conformal geometry, as it carries not a single metric but a conformal class of metrics, given that the asymptotic boundary is formally located at conformal infinity. The AdS/CFT correspondence \cite{Maldacena:1997re,Witten:1998qj} conjectures a duality between quantum gravity theories in the AdS bulk and conformal field theories on the boundary. This duality is an example of gauge/gravity dualities and a realization of the holographic principle of quantum gravity \cite{tHooft:1993dmi,Susskind:1994vu}. The large-$N$ limit of the boundary CFT corresponds to the semiclassical limit of the bulk gravity theory, where the Einstein-Hilbert action dominates the effective theory. Moreover, a strongly coupled boundary theory corresponds to a weakly coupled gravity theory in the bulk. Thus, besides the motivation for quantum gravity, the AdS/CFT duality has provided a versatile toolkit applied in various fields, including condensed matter physics \cite{Herzog:2007ij,Hartnoll:2009sz,Sachdev:2010ch}, nuclear physics \cite{Kim:2011ey,Pahlavani:2014dma,Aoki:2012th}, hydrodynamics \cite{Policastro:2002se,Bhattacharyya:2007vjd,Haack:2008cp,Hubeny:2011hd}, and quantum information theory \cite{Almheiri:2014lwa,Dong:2016eik,Pastawski:2015qua,Hayden:2016cfa,Banks:2020dus}.
\par
In the context of AdS/CFT, diffeomorphisms that induce Weyl transformations of the boundary metric are the Weyl diffeomorphisms in the bulk. Thus, conformal-covariant tensors can descend from ambient Riemannian tensors, and their Weyl transformations can be derived from certain ambient diffeomorphisms. In a suitable coordinate system $\{z,x^\mu\}$ ($\mu=0,\cdots,d-1$), the metric of any $(d+1)$-dimensional ALAdS spacetime can be expanded with respect to the bulk coordinate $z$ into two series, called the Fefferman-Graham expansion \cite{Ciambelli:2019bzz,Leigh}. The Weyl transformations can be represented by a local scaling of the coordinate $z$. 
\par
Usually when discussing AdS/CFT, one picks a specific representative of the conformal class. For example, the most commonly used choice for studying the conformal boundary of an ALAdS spacetime is the Fefferman-Graham (FG) gauge \cite{AST_1985__S131__95_0,Fefferman:2007rka}. However, the FG gauge explicitly breaks the Weyl symmetry by fixing a specific boundary metric. This is also manifested by the fact that the FG ansatz of the bulk metric is not preserved under a Weyl diffeomorphism. More specifically, in this case one can introduce a Penrose-Brown-Henneaux (PBH) transformation \cite{Imbimbo_2000,Bautier:2000mz,Rooman:2000ei} in the bulk to induce a Weyl transformation on the boundary, but the subleading terms in the $z$-expansion will not transform in a Weyl-covariant way if the form of the FG ansatz is to be preserved. 

In order to resolve this issue, one can relax the FG ansatz of the ALAdS bulk metric to the Weyl-Fefferman-Graham (WFG) ansatz \cite{Ciambelli:2019bzz}. In this way, the form of the bulk metric is preserved under a Weyl diffeomorphism, and all the terms in the $z$-expansion transform in a Weyl-covariant way, which brings a powerful reorganization of the holographic dictionary. It was shown \cite{Ciambelli:2019bzz} that in the WFG gauge, the bulk LC connection induces a Weyl connection on the conformal boundary. Thus, the ALAdS bulk geometry in the WFG gauge induces a Weyl geometry instead of only a conformal geometry on the conformal boundary. Following \cite{Ciambelli:2019bzz}, the WFG gauge was further investigated in \cite{Jia:2021hgy,Jia:2023gmk,Ciambelli:2023ott}. We have seen that in the FG ambient construction, the conformal boundary $(M,[g])$ of a $(d+1)$-dimensional ALAdS bulk is associated with a $(d+2)$-dimensional ambient space, and the ALAdS bulk in the FG gauge can be considered as a hypersurface in the ambient space. A natural question to ask is whether such a construction exists for the conformal boundary as a Weyl manifold. In this thesis we will provide such a construction. We introduce the Weyl-ambient space $(\tilde M,\tilde g)$ as a modification of the FG ambient space, in which the ALAdS bulk in the WFG gauge is a hypersurface and its boundary is associated with a codimension-2 Weyl manifold $(M,[g,a])$.
\par
For an even-dimensional boundary, the two series in the FG expansion will mix and the solution to the equations of motion encounters a pole. Formulating the FG expansion is using the technique of dimensional regularization, i.e.\ regarding $d$ as a variable (formally complex), the extended obstruction tensor $\Omega^{(k)}_{ij}$ can be read off from the pole of the FG expansion in $2k$-dimension. Equivalently, the obstruction tensor can also be introduced as a logarithmic term at order $O(z^{d-2})$ for $d=2k$, causing an obstruction to the power series expansion\cite{graham2005ambient}. Using the technique of dimensional regularization, the Weyl-obstruction tensors and extended Weyl-obstruction tensors were introduced in \cite{Jia:2021hgy} as the poles in the on-shell metric expansion. The extended obstruction tensors also play an integral role in the context of holography as the basic building blocks of the holographic Weyl anomaly \cite{Henningson:1998gx,graham2009extended}.

The Weyl anomaly, also known as the conformal anomaly or trace anomaly, reflects the violation of the Weyl symmetry in a quantum theory that is present in a classical theory. (For a general overview of quantum anomalies, see Section \ref{sec:anomalyoverview} in Part II). It is quantified by the nonvanishing trace of the energy-momentum tensor in even dimensions, which has been computed for various conformal field theories \cite{Capper1974,Deser1976,Duff1977,Polyakov1981,Fradkin:1983tg,Bonora:1985cq,Osborn1991,Deser1993,Henningson:1998gx,Boulanger:2007st,Boulanger:2007ab} and exhibits many physical consequences. For example, it has been found that it significantly contributes to the proton mass \cite{Ji:1994av,Ji:1995sv}. In condensed matter systems, experimentally accessible effects have been discussed in \cite{Northe:2022tjr}. In string theory, the cancellation of the Weyl anomaly determines the dimensionality of bosonic string theory to be 26 and superstring theory to be 10 \cite{Polchinski:1998rq,Polchinski:1998rr}. The results of Weyl anomaly in $2d$ and $4d$ are well-known:
\begin{align}
\label{WA2d4d}
2d:\langle T^\mu{}_\mu\rangle= -\frac{c}{24\pi} R\,,\qquad4d:\langle T^\mu{}_\mu\rangle= cW^2-aE^{(4)}\,,
\end{align}
where $W^2$ is the contraction of two Weyl tensors, and $E^{(4)}$ is the Euler density in $4d$. The coefficient $c$ in $2d$ is the central charge of the $2d$ CFT, which has the crucial property that it monotonically decreases along the renormalization group (RG) flow from the the ultraviolet (UV) to the infrared (IR), a result known as the c-theorem \cite{Zamolodchikov:1986gt}. Similarly, in $4d$, the coefficient $a$ follows $a_{UV}>a_{IR}$, known as the $a$-theorem \cite{Komargodski:2011vj}. These results highlight one of the key aspects of the unique nature of the Weyl anomaly compared to other kinds of anomalies. 
\par
In the context of holography, the Weyl anomaly was first suggested in \cite{Witten:1998qj}, and was then calculated from the bulk in \cite{Liu:1998bu} and \cite{Henningson:1998gx}.
For a holographic theory where we have the vacuum Einstein theory in the bulk, one gets $a=c$ in the 4-dimensional boundary theory as a constraint on the central charges. In the FG gauge, after going through the holographic renormalization procedure by adding counterterms to cancel the divergence extracted by the regulator, one finds that the holographic Weyl anomaly in an even dimension corresponds to the logarithmic term in the bulk volume expansion. In mathematical literature this is also referred to as the Q-curvature \cite{branson1991explicit,Branson1995,2001math.....10271F,Fefferman2003} (see \cite{Chang2008} for a short review), which has been studied by means of obstruction tensors and extended obstruction tensors in \cite{graham2005ambient} and \cite{graham2009extended}. Going into the WFG gauge, it was shown in \cite{Ciambelli:2019bzz} using dimensional regularization that the Weyl anomaly in $2k$-dimension can be extracted directly from the variation of the pole term at the $O(z^{2k-d})$-order of the ``bare" on-shell action under the $d\to 2k^-$ limit. Using this method in the WFG gauge, it was found in \cite{Jia:2021hgy} that the holographic Weyl anomaly can be expressed in terms of extended Weyl-obstruction tensors. 
\par
From the physics side, our goal in Part I of this thesis is to find the holographic Weyl anomaly in higher dimensions utilizing the the features of the Weyl geometry and WFG gauge, and organize the results in a form that manifests its general structure.\footnote{For discussions on the Weyl anomaly in non-holographic contexts utilizing Weyl geometry, see \cite{Zanusso:2023vkn,Ghilencea:2023sti}.} It has been shown in \cite{Ciambelli:2019bzz} that, up to total derivatives, the Weyl anomaly in $2d$ and $4d$ in the WFG gauge has the same form of that in the FG gauge, but now become Weyl-covariant. We generalize these results to $6d$ and $8d$ by calculating the Weyl anomaly explicitly, and we find that the same statement still holds. Furthermore, we show that by promoting the obstruction tensors in the FG gauge to the Weyl-obstruction tensors in the WFG gauge, one can use them as natural building blocks for the Weyl anomaly. In this way, we will see clearly how  the WFG gauge Weyl-covariantizes the Weyl anomaly without introducing additional nontrivial cocycles. Our results also reveal some interesting clues about the general form of the holographic Weyl anomaly in any dimension.

\section{Organization of Part I}
The rest of Part I is organized as follows. 

In Chapter \ref{chap:pre}, we provide necessary preliminaries. Section \ref{sec:Weyl} introduces Weyl geometry, including useful quantities and identities. Section \ref{sec:FG} discusses obstruction tensors and extended obstruction tensors in the FG gauge and their properties. Section \ref{sec:WFG} reviews the WFG gauge as a Weyl-covariant modification of the FG gauge and explains how the bulk LC connection induces a Weyl connection on the conformal boundary.

In Chapter \ref{chap:Weylambiant}, we first review the Fefferman-Graham ambient metric before introducing the Weyl-ambient metric $\tilde g$ at the end of Section \ref{sec:ambient}. To build intuition, we start with the flat ambient metric and generalize to Ricci-flat ambient metrics. Different coordinate systems presented in Section \ref{sec:ambient} are described in Appendix \ref{app:coords}. In Section \ref{sec:WAS}, we formulate Weyl-ambient geometry from two perspectives. First, from a top-down perspective, we demonstrate how $(\tilde M,\tilde g)$ induces a codimension-2 Weyl manifold $(M,[g,a])$. Then, we introduce the bottom-up construction of the Weyl-ambient metric. We show that the Weyl-ambient metric has a well-defined perturbative initial value problem, with Ricci-flatness as the equation of motion, following and generalizing \cite{Fefferman:2007rka}. Some major theorems from \cite{Fefferman:2007rka} are extended with suitable modifications. 

Chapter \ref{chap:WOT} is dedicated to Weyl-obstruction tensors. In Section \ref{sec:WOTbulk}, we generalize the obstruction tensors derived from in Section \ref{sec:FG} to Weyl-obstruction tensors by solving the Einstein equations in the WFG gauge. Expansions of the Einstein equations can be found in Appendix \ref{app:B0}. In Section \ref{sec:WOTambient}, we discuss how the Weyl-covariant tensors on $(M,[g,a])$ are derived from the Riemann tensor of $(\tilde M,\tilde g)$, and define the extended Weyl-obstruction tensors. We use a first-order formalism in Section \ref{sec:topdown} with a null frame, with details provided in Appendix \ref{app:Null}. We then discuss Weyl-covariant tensors and extended Weyl-obstruction tensors in the second-order formalism, and prove the extended Weyl-obstruction tensors defined from both approaches. The results of Chapter \ref{chap:Weylambiant} and Chapter \ref{chap:WOT} are summarized in Section \ref{sec:conclu}.

In Chapter \ref{chap:HWA}, we introduce the anomalous Weyl-Ward identity in Weyl geometry and discuss the holographic Weyl anomaly in the WFG gauge in Section \ref{sec:WWI}. Using Weyl-Schouten and extended Weyl-obstruction tensors, we derive the holographic Weyl anomaly in the WFG gauge up to $8d$ in Section \ref{sec:HWA}. More details of the calculation are provided are in Appendix \ref{app:expansion}. In Section \ref{sec:roleWFG}, we explore aspects of Weyl structure in the formulas for Weyl-obstruction tensors and Weyl anomaly. Finally, in Section \ref{sec:disc1}, we summarize our results and point out possible directions for future research.

The results presented in Part I sourced mostly from the joint research works \cite{Jia:2021hgy,Jia:2023gmk} with the author's advisor Robert~G.~Leigh, and collaborator Manthos Karydas.

\section{Notation}
We will label the indices in a $d$-dimensional manifold $M$ by lowercase Latin letters $i,j,\cdots$, in a $(d+1)$-dimensional ALAdS bulk by lowercase Greek letters $\mu,\nu,\cdots$, and in a $(d+2)$-dimensional ambient space $\tilde M$ by uppercase Latin letters $I,J,\cdots$. The vectors on $M$ are denoted by $ \un U,\un V$, on the Weyl structure ${\cal P}_W$ over $M$ are denoted by $\un u,\un v$, and on the ambient manifold $\tilde M$ are denoted by $\un {\cal U},\un{\cal V}$.

In Subsections \ref{sec:topdown} and \ref{sec:firstorder}, we mainly use the dual frame $\{\bm e^I\}$, and the ambient frame indices are $I=+,1,\cdots,d,-$. Unless otherwise indicated, in Subsections \ref{sec:bottomup}, \ref{sec:proofs} and  \ref{sec:secondorder} we mainly use the ambient coordinate system $\{t,x^i,\rho\}$, and the indices are $I=0,1,\cdots,d,\infty$, where $0$ labels the $t$-component and $\infty$ labels the $\rho$-component. The notation $(0,x^i,\infty)$ is also used for the components in a trivialization ${\cal P}_W\times \bb R\simeq \bb R_+\times M\times\bb R$, even without specifying a choice of coordinates on $M$. The above-mentioned notation is summarized in Table \ref{t1}.

In Chapter \ref{chap:pre}, Section \ref{sec:WOTbulk} and Chapter \ref{chap:HWA} we use $\gamma^{(2k)}_{ij}$ and $a^{(2k)}_{i}$ for the bulk expansions in $z$-coordinate, while in Chapter \ref{chap:Weylambiant} and Section \ref{sec:WOTambient} we use $\gamma^{(k)}_{ij}$ and $a^{(k)}_{i}$ for the ambient expansions in $\rho$-coordinate, which correspond to $(-2)^k\gamma^{(2k)}_{ij}/L^{2k}$ and $(-2)^ka^{(2k)}_{i}/L^{2k}$ in the $z$-expansion, respectively.
\begin{table}[!h]
\centering
\caption{Notation for Part I}
\begin{tabularx}{\textwidth}{c|c|c|X}
\toprule
Dimension & Manifold & Vectors & Indices \\
\midrule
$d$ & $M$ & $ \un U,\un V$ & $i,j,\cdots$ $\quad\{x^i\}$ $\quad i=1,\cdots,d$\\
\midrule
$d+1$ & (AL)AdS$_{d+1}$ &  & $\mu,\nu,\cdots$ $\quad\{x^\mu\}=\{z,x^i\}$  $\quad i=1,\cdots,d$ \\
\midrule
$d+1$ & ${\cal P}_W$ & $\un u,\un v$ &   \\
\midrule
$d+2$ & $\tilde M$ & $\un {\cal U},\un{\cal V}$ & $I,J,\cdots$\newline  In the frame $\{\bm e^I\}=\{\bm e^+,\bm e^i,\bm e^-\}$, $I=+,1,\cdots,d,-$. \newline In the coordinates $\{x^I\}=\{t,x^i,\rho\}$, $I=0,1,\cdots,d,\infty$.\\
\bottomrule
\end{tabularx}
\label{t1}
\end{table} 

\chapter{Preliminaries}
\label{chap:pre}
\section{Weyl Geometry}
\label{sec:Weyl}
In this section we provide a brief review of Weyl geometry (see also \cite{10.4310/jdg/1214429379,doi:10.1063/1.529582}). We will mainly introduce the geometric quantities equipped with Weyl connection as well as some useful relations we will use later in this thesis. We use $a,b,\cdots$ to label the internal frame indices and $i,j,\cdots$ to label the spacetime indices. For clarity, we also put $\circ$ on the top of Levi-Civita quantities, e.g.\ $\mathring R^a{}_{bcd}$, $\mathring P_{ab}$, etc.
\par
Given a generalized Riemannian manifold $(M,g)$ with a connection $\nabla$, in an arbitrary basis $\{\un e_a\}$, the connection coefficients $\Gamma^c{}_{ab}$ are defined as
\begin{align}
\label{conncoef}
\nabla_{\un e_a}\un e_b=\Gamma^c{}_{ab}\un e_c\,.
\end{align}
The torsion tensor and Riemann curvature tensor of $\nabla$ in this basis are given by
\begin{align}
\label{Tor}
T^c{}_{ab}\un e_c&\equiv \nabla_{\un e_a}\un e_b-\nabla_{\un e_b}\un e_a-[\un e_a,\un e_b]\,,\\
\label{Rie}
R^a{}_{bcd}\un e_a&\equiv \nabla_{\un e_c}\nabla_{\un e_d}\un e_b-\nabla_{\un e_d}\nabla_{\un e_c}\un e_b-\nabla_{[\un e_c,\un e_d]}\un e_b\,.
\end{align}
When $\nabla$ is associated with $g$ and is torsion-free, it is called a Levi-Civita (LC) connection, denoted by $\LCnabla$. Using $\LCGamma$ to denote the LC connection coefficients, we have $\mathring\nabla_{\un e_a}\un e_b=\LCGamma^c{}_{ab}\un e_c$. 
By definition, the conditions satisfied by the LC connection coefficients $\LCGamma^c{}_{ab}$ are
\begin{align}
\label{NM}
0&=(\LCnabla g)(\un e_a,\un e_b,\un e_c)=\LCnabla_{\un e_c}g(\un e_a,\un e_b)-\LCGamma^d{}_{ca}g(\un e_d,\un e_b)-\LCGamma^d{}_{cb}g(\un e_d,\un e_a)\,,\\
0&=T^a{}_{bc}=\LCGamma^c{}_{ab}-\LCGamma^c{}_{ba}-C_{ab}{}^c\,,
\end{align}
where $C_{ab}{}^c$ are the commutation coefficients defined by 
\begin{align}
\label{comcoeff}
[\un e_a,\un e_b]=C_{ab}{}^c\un e_c\,. 
\end{align}
Denote $g_{ab}\equiv g(\un e_a,\un e_b)$ as the component of the metric in the frame $\{\un e_a\}$. From these conditions $\LCGamma^c{}_{ab}$ can be derived as
\begin{align}
\label{LC}
\LCGamma^c{}_{ab}=&\,\frac{1}{2}g^{cd}\big(\un e_a(g_{db})+\un e_b(g_{ad})-\un e_d(g_
{ab})\big)-\frac{1}{2}g^{cd}(C_{ad}{}^e g_{eb}+C_{bd}{}^e g_{ae}-C_{ab}{}^e g_{ed})\,.
\end{align}
If we choose a local coordinate basis $\{\un\p_i\}$ with $\un e_a=e_a^i\un\p_i$, the dual frame $e^a=e^a_i \td x^i$ satisfies $e^a_ie_a^j=\delta^j_i$. Then noticing that \eqref{comcoeff} in this coordinate basis reads
\begin{align}
e^i_a\p_i e^j_b-e^i_b\p_i e^j_a=C_{ab}{}^c e^i_c\,, 
\end{align}
we can see that the LC connection coefficients in this coordinate basis go back to the familiar Christoffel symbol
\begin{align}
\label{LCGamma}
\LCGamma^k{}_{ij}\equiv\LCGamma^c{}_{ab}e^a_i e^b_j e_c^k=\frac{1}{2}g^{kl}(\p_ig_{jl}+\p_jg_{il}-\p_lg_{ij})\,.
\end{align}
\par
Now we will work in a coordinate basis $\{\un\p_i\}$.\footnote{Note that $\un e_a\equiv e_a^i\un\p_i$ and $\bm e^a\equiv e^a_i\td x^i$ have Weyl weights $+1$ and $-1$ respectively, while $\un\p_i$ and $\td x^i$ have no Weyl weights. This is because the Weyl transformation of the frame only comes from the soldering of the vector bundle associated with the frame bundle to the tangent space of $M$.} Consider a Weyl transformation
\begin{align}
\label{WeylA}
g_{ij}\to{\cal B}^{-2}g_{ij}\,.
\end{align}
The metricity tensor $\nabla g$ any connection $\nabla$ will transform non-covariantly under \eqref{WeylA}:
\begin{align}
\nabla_i g_{jk}\to B^{-2}(\nabla_i g_{jk}-2\nabla_i\ln{\cal B}g_{jk})\,.
\end{align}
To restore the Weyl covariance, one can introduce a Weyl connection $A=A_i\td x^i$ which transforms under a Weyl transformation as
\begin{align}
A_i\to A_i-\nabla_i\ln{\cal B}\,.
\end{align}
Then, we obtain an object that is Weyl-covariant:
\begin{align}
(\nabla_i g_{jk}-2A_i g_{jk})\to{\cal B}^{-2}(\nabla_i g_{jk}-2A_i g_{jk})\,.
\end{align}
More generally, for a tensor $T$ of an arbitrary type (with indices suppressed) that transforms under a Weyl transformation with a specific Weyl weight $\omega_T$, i.e.\ $T\to B^{\omega_T}T$, we can define
\begin{align}
\label{WeylD}
\hat\nabla_i T\equiv\nabla_i T+w_TA_i T\,.
\end{align}
In this way, $\hat\nabla$ acting on $T$ will also transform Weyl-covariantly as 
\begin{align}
\hat\nabla_i T\to B^{\omega_T}\hat\nabla_i T\,.
\end{align}
\par
Now we choose the connection $\nabla$ by setting the metricity as follows
\begin{align}
\nabla_i g_{jk}=2A_i g_{jk}\,.
\end{align}
Equivalently, we say that this connection has vanishing \emph{Weyl metricity}, since
\begin{align}
\hat\nabla_i g_{jk}=0\,.
\end{align}
We will also require $\nabla$ defined in the above equation to be torsion-free. Then, $\hat\nabla$ is called a \emph{Weyl-LC connection}. The connection coefficients of $\nabla$ in the coordinate basis become
\begin{align}
\label{WeylLC}
\Gamma^k{}_{ij}={}&\frac{1}{2}g^{kl}(\p_k g_{lj}+\p_j g_{il}-\p_l g_
{ij})-(A_i\delta^k{}_j+A_j\delta^k{}_i-g^{kl}A_l g_{ij})\,.
\end{align}
We can see that this is different from the Christoffel symbols \eqref{LCGamma} due to the extra terms involving the Weyl connection. When $\nabla$ and $\LCnabla$ act on a vector, their difference can be reflected by
\begin{align}
\label{div_v}
\nabla_i v^j=\LCnabla_i v^j-(A_i\delta^j{}_k+A_k\delta^j{}_i-g^{jl}A_l g_{ik})v^k\,.
\end{align}
It is worthwhile to notice that if $v^i$ has Weyl weight $d=\dim M$, then it follows from \eqref{WeylD} and \eqref{div_v} that $\hat\nabla_i v^i=\LCnabla_i v^i$.
\par
Now one can compute the Riemann tensor of $\nabla$ and its contractions. Denoting the coordinate components of the Riemann tensor of $\LCnabla$ as $\mathring R^i{}_{jkl}$, one finds from \eqref{Rie} that
\begin{align}
\label{WRiem}
R^i{}_{jkl}={}&\mathring{R}^i{}_{jkl}+\LCnabla_l A_j\delta^i{}_k-\LCnabla_k A_j\delta^i{}_l
+(\LCnabla_l A_k-\LCnabla_k A_l)\delta^i{}_j
+\LCnabla_k A^i g_{jl}-\LCnabla_l A^i g_{jk}\nn\\
&+A_j(A_l\delta^i{}_k-A_k\delta^i{}_l)
+A^i(g_{jl}A_k-g_{jk}A_l)
+A^2(g_{jk}\delta^i{}_l-g_{jl}\delta^i{}_k)\,,\\
R_{ij}={}&\mathring{R}_{ij}-\frac{d}{2}F_{ij}+(d-2)(\LCnabla_{(i}A_{j)}+A_i A_j)+(\LCnabla\cdot A-(d-2)A^2)g_{ij}\,,\\
\label{WRic}
R={}&\mathring{R}+2(d-1)\LCnabla\cdot A-(d-1)(d-2)A^2\,,
\end{align}
where $R_{ij}\equiv R^k{}_{ikj}$, $R\equiv R_{ij}g^{ij}$, and we defined the curvature of $A_i$ as $F_{ij}=\LCnabla_{i}A_{j}-\LCnabla_{j}A_{i}$. It is easy to see from \eqref{WRiem} that, unlike $\mathring R^i{}_{jkl}$, the $R^i{}_{jkl}$ of $\nabla$ now is not antisymmetric in the first two indices, and it does not have the interchange symmetry for the two index pairs. Also, the $R_{ij}$ of $\nabla$ is not symmetric due to the appearance of the $F_{ij}$ term.
\par
On the other hand, from \eqref{conncoef} we have the connection coefficients $\hat\Gamma^c{}_{ab}$ for $\hat\nabla$:
\begin{align}
\label{WLCsplit}
\hat\Gamma^c{}_{ab}\un e_c=\hat\nabla_{\un e_a}\un e_b=\nabla_{\un e_a}\un e_b+A(\un e_a)\un e_b=\Gamma^c{}_{ab}\un e_c+A(\un e_a)\un e_b\,,
\end{align}
where we used the fact that the basis vector $\un e_a$ has Weyl weight $+1$. Plugging this into \eqref{Rie}, we find that the Riemann tensor of $\hat\nabla$ and its contractions satisfy
\begin{align}
\label{hatR}
\hat R^i{}_{jkl}=&R^i{}_{jkl}+\delta^i{}_j F_{kl}\,,\qquad\hat R_{ij}=R_{ij}+F_{ij}\,,\qquad\hat R=R\,.
\end{align}
We refer to $\hat R^i{}_{jkl}$, $\hat R_{ij}$ and $\hat R$ as the Weyl-Riemann tensor, Weyl-Ricci tensor, and Weyl-Ricci scalar, respectively.\footnote{Note that in some literature, e.g.~\cite{Ciambelli:2019bzz}, the quantities defined using $\nabla$ instead of $\hat\nabla$ are called Weyl quantities.} Similar to the curvature tensors for $\nabla$, the Weyl-Riemann tensor is not antisymmetric in the first two indices and does not have the interchange symmetry for the two index pairs, and the Weyl-Ricci tensor is not symmetric. Also notice that the Weyl-Weyl tensor, namely the traceless part of the Weyl-Riemann tensor, is equal to the LC Weyl tensor, i.e.
\begin{align}
\hat W^i{}_{jkl}=\mathring W^i{}_{jkl}\,.
\end{align}
\par
Unlike the LC curvature quantities, which transform in a non-covariant way under the Weyl transformation, the Weyl-Riemann tensor, Weyl-Ricci tensor, and Weyl-Ricci scalar transform under the Weyl transformation as
\begin{align}
\label{hatRtrans}
\hat R^i{}_{jkl}\to\hat R^i{}_{jkl}\,,\qquad\hat R_{ij}\to\hat R_{ij}\,,\qquad\hat R\to{\cal B}^2\hat R\,.
\end{align}
Furthermore, we can define the Weyl-Schouten tensor $\hat P_{ij}$ and Weyl-Cotton tensor $\hat C_{ijk}$ as
\begin{align}
\label{WP1}
\hat P_{ij}&=\frac{1}{d-2}\bigg(\hat R_{ij}-\frac{1}{2(d-1)}\hat Rg_{ij}\bigg)\,,\\
\label{WC1}
\hat C_{ijk}&=\hat\nabla_{k}\hat P_{ij}-\hat\nabla_{j}\hat P_{ik}\,.
\end{align}
Although the LC Schouten tensor $\mathring P_{ij}$ defined by substituting $\hat R_{ij}$ and $\hat R$ in \eqref{WP1} with $R_{ij}$ and $R$ is a symmetric tensor, $\hat P_{ij}$ has an antisymmetric part $\hat P_{[ij]}=-F_{ij}/2$. In terms of the LC connection, the Bach tensor is defined by (the indices of the components are raised and lowered by $g$)
\begin{align}
\mathring B_{ij}=\LCnabla^k\LCnabla_k \mathring P_{ij}-\LCnabla^k\LCnabla_{j}\mathring P_{ik}-\mathring W_{ljik}\mathring P^{kl}\,,
\end{align}
which satisfies $\mathring B_{ij}\to{\cal B}^2\mathring B_{ij}$ in $4d$. Now we can define the Weyl-Bach tensor
\begin{align}
\label{WB1}
\hat B_{ij}=\hat\nabla^k\hat\nabla_k \hat P_{ij}-\hat\nabla^k\hat\nabla_{j}\hat P_{ik}-\hat W_{ljik}\hat P^{kl}\,.
\end{align}
Similar to the LC Bach tensor, the Weyl-Bach tensor is also symmetric and traceless; however, it is Weyl-covariant in any dimension. Following \eqref{WRiem}--\eqref{WRic}, here we list the above-mentioned Weyl quantities in terms of their corresponding LC quantities:
\begin{align}
\label{hatP}
\hat P_{ij}&=\mathring P_{ij}+\LCnabla_j A_{i}+A_{i}A_{j}-\frac{1}{2}A^2g_{ij}\,,\\
\hat C_{ijk}&=\mathring C_{ijk}-A_l\mathring W^l{}_{ikj}\,,\\
\hat B_{ij}&=\mathring{B}_{ij}+(d-4)(A^k\mathring{C}_{kji}-2A^k\mathring{C}_{ijk}+A^k A^l \mathring W_{likj})\,.
\end{align}
\par
The Bianchi identity for $\hat\nabla$ reads
\begin{align}
\hat\nabla_i \hat R^m{}_{jkl}+\hat\nabla_k\hat R^m{}_{jli}+\hat\nabla_l\hat R^m{}_{jik}=0\,.
\end{align}
Noticing that $\hat\nabla_i g_{jk}=0$, the contraction of the above equation gives
\begin{align}
\label{BI}
\hat\nabla^i\hat G_{ij}=0\,,
\end{align}
where we defined the Weyl-Einstein tensor $\hat G_{ij}\equiv \hat R_{ij}-\frac{1}{2}\hat Rg_{ij}$. Using \eqref{WP1}, this identity can also be expressed using the Weyl-Schouten tensor as
\begin{align}
\label{BIP}
\hat\nabla^i\hat P_{ij}=\hat\nabla_j\hat P\,.
\end{align}
where $\hat P\equiv\hat P_{ij}g^{ij}$. Starting from \eqref{WB1} and using \eqref{BIP} repeatedly, one obtains
\begin{align}
\label{nablaB}
\hat\nabla^i\hat B_{ij}=(d-4)\hat P^{ik}(\hat C_{kij}+\hat C_{jik})\,.
\end{align} 
Note that since  $\mathring P_{ij}$ is symmetric, while the Cotten tensor is antisymmetric in the last two indices. Thus, the above equation in the LC case becomes
\begin{align}
\LCnabla^i\mathring B_{ij}=(d-4)\mathring P^{ik}\mathring C_{kij}\,.
\end{align} 
It is also useful to notice that in the LC case, the divergence of the Cotton tensor vanishes
\begin{align}
\label{divC}
\LCnabla^i\mathring C_{ijk}=0\,,
\end{align} 
while for the Weyl-Cotton tensor we have instead
\begin{align}
\hat\nabla^i\hat C_{ijk}=\hat W_{lkmj}F^{lm}\,.
\end{align} 
In the end of this section, we list the Weyl weights of the above-mentioned Weyl quantities in Table~\ref{table1}.
\begin{table}[!htbp]
\centering
\caption{Weyl weights of Weyl-covariant quantities}
\begin{tabular}{ccccccccccc}
\toprule
$\un e_a$ & $\bm e^a$ & $g_{ij}$ & $g^{ij}$ &$\hat R^i{}_{jkl}$ & $\hat R_{ij}$ & $\hat R$ & $F_{ij}$& $\hat P_{ij}$ &$\hat C_{ijk}$ &$\hat B_{ij}$\\
\midrule
$+1$ & $-1$ & $-2$ & $+2$ &$0$ & $0$&$+2$ & $0$& $0$ &$0$ &$+2$\\
\bottomrule
\end{tabular}
\label{table1}
\end{table}

\section{Fefferman-Graham Expansion and Obstruction Tensors}
\label{sec:FG}
The obstruction tensor is known as the only irreducible conformal covariant tensor besides the Weyl tensor in an even-dimensional spacetime. The general references for obstruction tensors are \cite{graham2005ambient,Fefferman:2007rka}, where they were defined precisely in terms of the ambient metric. Instead of providing the formal definition immediately, in this section we will demonstrate the obstruction tensors as poles of the Fefferman-Graham expansion. The same method will also be used in Section \ref{sec:WOTbulk} for Weyl-obstruction tensors. In Section \ref{sec:WOTambient} we will introduce the precise definition of Weyl-obstruction tensors using the ambient formalism.
\par
According to the Fefferman-Graham theorem \cite{AST_1985__S131__95_0}, the metric of a $(d+1)$-dimensional asymptotically locally AdS (ALAdS) spacetime can always be expressed in the following form
\begin{equation}
\label{FG}
\td s^2=L^2\frac{\td z^2}{z^2}+h_{ij}(z;x)\td x^i \td x^j\,,\qquad i,j=0,\cdots,d-1\,,
\end{equation}
where the coordinate $z$ can be considered as a ``radial" coordinate, and $z=0$ is the ``location" of the conformal boundary. When $h_{ij}=L^2 \eta_{ij}/z^2 $ with $\eta_{ij}$ the flat metric, this represents the Poincar\'e metric for AdS$_{d+1}$. Near the conformal boundary, $h_{ij}$ can be expanded with respect to $z$ as follows \cite{Ciambelli:2019bzz}:
\begin{equation}
\label{hex0}
h_{ij}(z;x)=\frac{L^2}{z^2}\left[\gamma^{(0)}_{ij}(x)+\frac{z^2}{L^2}\gamma^{(2)}_{ij}(x)+\cdots\right]+\frac{z^{d-2}}{L^{d-2}}\left[\pi^{(0)}_{ij}(x)+\frac{z^2}{L^2}\pi^{(2)}_{ij}(x)+\cdots\right]\,.
\end{equation}
As we mentioned in Chapter \ref{chap:Intro1}, the conformal boundary carries a conformal class of metrics. In the FG expansion $\gamma^{(0)}_{ij}$ serves as the ``canonical" representative of the conformal class sourcing the energy-momentum tensor of the dual field theory on the boundary, while $\pi^{(0)}_{ij}$ corresponds to the expectation value of the energy-momentum tensor \cite{Leigh}. Once $\gamma^{(0)}_{ij}$ is given, each term in the first series can be determined by solving the vacuum Einstein equations with negative cosmological constant in the bulk. Similarly, once $\pi^{(0)}_{ij}$ is given, the second series will be determined. However, $\pi^{(0)}_{ij}$ is not completely arbitrary but is actually constrained by the Einstein equations. To be more specific, the $zz$-component of the Einstein equations tells us that $\pi^{(0)}_{ij}$ is traceless while the $zi$-components indicate that it is also divergence-free.
\par
Nevertheless, subtleties will arise when the boundary dimension $d$ is an even integer, since the two series in \eqref{hex0} mix into one. To resolve this issue for an even $d=2k$, we treat $d$ formally as a variable $d\in \mathbb{C}$ in the expansion \eqref{hex0} and let $d$ approach $2k$ from below. As we will see explicitly, when the Einstein equations are satisfied, $\gamma^{(2k)}_{ij}$ has a first order pole at $d=2k$. For any integer $k\geqslant2$, up to some factor, the coefficient of the pole term (which is actually a meromorphic function of the boundary dimension) is what we define as the \emph{obstruction tensor}, denoted by $\mathcal{O}^{(2k)}_{ij}$:
\begin{equation}\label{gamma2k}
\gamma^{(2k)}_{ij}= \frac{c_{(2k)}}{d-2k}\mathcal{O}^{(2k)}_{ij}+ \tilde{\gamma}_{ij}^{(2k)}\,,\qquad c_{(2k)}=-\frac{L^{2k}}{2^{2k-3}k!}\frac{\Gamma(d/2-k+1)}{ \Gamma(d/2-1)}\,,
\end{equation}
where the normalization factor $c_{(2k)}$ has been chosen so that the obstruction tensor agrees with the convention of \cite{Fefferman:2007rka}, and the tensor $\tilde{\gamma}^{(2k)}_{ij}$ is analytic at $d=2k$. 
\par
 Besides holographic dimensional regularization \cite{Leigh}, another common approach is to introduce a logarithmic term for $d=2k$ \cite{Henningson:1998gx}, which turns out to be proportional to the obstruction tensor. This is also the origin of the name obstruction tensor, as it obstructs the existence of a formal power series expansion. Note that the tensor $\mathcal{O}^{(2k)}_{ij}$ is well-defined in any dimension, but only behaves as an ``obstruction" when $d=2k$. The relation between the two approaches will be cleared up at the end of this section once we show how to correctly take the limit for an even $d$ in holographic dimensional regularization. 
\par
Now we present the obstruction tensors in $d=2,4,6$ explicitly. First, by solving the bulk Einstein equations to the $O(z^2)$-order one finds that
\begin{align}
\label{gamma2}
\frac{\gamma^{(2)}_{ij}}{L^2}=-\frac{1}{d-2}\left(R^{(0)}_{ij}-\frac{R^{(0)}}{2(d-1)}\gamma_{ij}^{(0)}\right)\,,
\end{align}
where $R^{(0)}_{ij}$ and $R^{(0)}$ represent the Ricci tensor and Ricci scalar of $\gamma_{ij}^{(0)}$ on the boundary, respectively. One can recognize $\gamma^{(2)}_{ij}/L^2$ as the Schouten tensor $P_{ij}$ on the boundary (with a minus sign):
\begin{align}
\label{P}
P_{ij}&=\frac{1}{d-2}\left(R^{(0)}_{ij}-\frac{R^{(0)}}{2(d-1)}\gamma_{ij}^{(0)}\right)\,.
\end{align}
Indeed we notice that there is a first order pole when $d=2$ as expected. However, it is easy to see that  the residue of the pole vanishes identically for $d=2$. This is the reason $P_{ij}$ is usually not referred to as the obstruction tensor for $d=2$.
\par
At the $O(z^4)$-order, the Einstein equations give us
\begin{align}
\label{gamma4}
\frac{\gamma^{(4)}_{ij}}{L^4}=-\frac{1}{4(d-4)}B_{ij}+\frac{1}{4}P_{ki}P^{k}{}_{j}\,.
\end{align}
Note that on the boundary, the tensor indices are lowered and raised using $\gamma^{(0)}_{ij}$ and its inverse $\gamma_{(0)}^{ij}$. The tensor $B_{ij}$ is the Bach tensor, which is defined as
\begin{align}
\label{B}
B_{ij}&=\nabla_{(0)}^l\nabla^{(0)}_l P_{ij}-\nabla_{(0)}^l\nabla^{(0)}_{j} P_{il}-W^{(0)}_{kjil}P^{lk}\,,
\end{align}
where $\nabla^{(0)}_i$ is the derivative operator on the boundary associated with $\gamma^{(0)}_{ij}$, and $W^{(0)}_{kijl}$ is the Weyl tensor of $\gamma_{ij}^{(0)}$. We notice that the first term has a pole at $d=4$ and it follows from \eqref{gamma2k} that the obstruction tensor for $d=4$ is just the Bach tensor, i.e.\ $
{\cal O}^{(4)}_{ij}= B_{ij}$.\par
Similarly, if we move on to the $O(z^6)$-order of the Einstein equations, we find that $\gamma^{(6)}_{ij}$ has a pole at $d=6$ and can be written as
\begin{align}
\label{gamma6}
\frac{\gamma^{(6)}_{ij}}{L^6}&=-\frac{1}{24(d-6)(d-4)}{\cal O}^{(6)}_{ij}+\frac{1}{6(d-4)}B_{ki}P^k{}_j\,.
\end{align}
From \eqref{gamma2k} one can see that ${\cal O}^{(6)}_{ij}$ is the obstruction tensor for $d=6$, now given by
\begin{align}
\label{O6}
{\cal O}^{(6)}_{ij}={}&\nabla_{(0)}^l\nabla^{(0)}_l B_{ij}-2W^{(0)}_{kjil}B^{lk}-4B_{ij}P+2(d-4)\big(2P^{kl}\nabla^{(0)}_l C_{(ij)k}+\nabla^{(0)}_l PC_{(ij)}{}^l\nn\\
&\qquad\qquad-C^{k}{}_{i}{}^{l}C_{ljk}+ \nabla_{(0)}^l P^k{}_{(i}C_{j)kl}-W^{(0)}_{kijl}P^{l}{}_m P^{mk}\big)\,,
\end{align}
where $P\equiv P_{ij}\gamma_{(0)}^{ij}$, and $C_{ijk}$ is the Cotton tensor on the boundary defined as
\begin{align}
C_{ijk}=\nabla^{(0)}_k P_{ij}-\nabla^{(0)}_j P_{ik}\,.
\end{align}
\par
Let us make a few remarks on some important properties of the obstruction tensors. First, they are symmetric traceless tensors for any boundary dimension $d$. The traceless condition can be derived from the $zz$-component of the Einstein equations at the $O(z^{2k})$-order. Also, the obstruction tensor ${\cal O}^{(2k)}_{ij}$ is divergence-free when $d=2k$. For instance, divergence of the Bach tensor gives
\begin{align}
\label{divB}
\nabla_{(0)}^j B_{ji}=(d-4)P^{jk}C_{kji}\,.
\end{align}
The divergence of the Bach tensor can be read  from the $O(z^4)$-order of the $zi$-component of  Einstein equations. In general, at any $O(z^{2k})$-order one finds that the divergence of ${\cal O}^{(2k)}_{ij}$ is proportional to $d-2k$ and thus vanishes when $d=2k$. The divergence of ${\cal O}^{(2k)}_{ij}$ can also be obtained by using the following identity
\begin{align}
\label{divP}
\nabla_{(0)}^j P_{ji}=\nabla^{(0)}_i P\,.
\end{align}
This is equivalent to the contracted Bianchi identity at the boundary [similar to \eqref{BIP} for the Weyl-Schouten tensor], which can also be read from the leading order of the $zi$-component of  Einstein equations. Finally, a notable feature of ${\cal O}^{(2k)}_{ij}$ is that it is Weyl-covariant when $d=2k$ with Weyl weight $2k-2$ (which will be proved from the ambient space in Subsection \ref{sec:firstorder}).
\par
For convenience, we can also absorb the $d$-dependent factors in $\gamma^{(2k)}_{ij}$ by introducing Graham's extended obstruction tensor $\Omega^{(k-1)}_{ij}$ ($k\geqslant 2$):
\begin{align}
\label{extO}
\Omega^{(1)}_{ij}=-\frac{1}{d-4}B_{ij}\,,\qquad\Omega^{(2)}_{ij}=\frac{1}{(d-6)(d-4)}\mathcal O^{(6)}_{ij}\,,\qquad\cdots
\end{align}
The extended obstruction tensor $\Omega^{(k)}_{ij}$ was precisely defined in \cite{graham2009extended} in the context of the ambient metric. The general relation between the obstruction tensor and extended obstruction tensor is
\begin{align}
\Omega^{(k)}_{ij}=\frac{(-1)^{k}}{2^k}\frac{\Gamma(d/2-k-1)}{\Gamma(d/2-1)}{\cal O}^{(2k+2)}_{ij}\qquad (k\geqslant1)\,.
\end{align}
\par
We finish this section by describing how to get the $d\to 2k^{-}$ limit of the two series in \eqref{hex0} properly. By taking the limit carefully we will recover a logarithmic term in the expansion whose coefficient is exactly the obstruction tensor for $d=2k$, which also justifies the name ``obstruction" as we mentioned before. There are two issues one has to deal with while taking the  $d\to 2k^{-}$ limit. First, as we already noted, $\gamma^{(2k)}_{ij}$ has a pole at $d-2k$, so it diverges in this limit.  Second, the two series mix since both $\gamma^{(2k)}_{ij}$ and $\pi^{(0)}_{ij}$ appear at the same order $O(z^{2(k-1)})$ in \eqref{hex0}, for $d=2k$. To keep the $O(z^{2k})$-order finite we pose that  $\pi_{ij}^{(0)}$ should also have a pole for $d=2k$ proportional to ${\cal O}^{(2k)}_{ij}$ so that the divergence in $\gamma^{(2k)}_{ij}$ gets canceled, i.e.\ we claim that $\pi^{(0)}_{ij}$ has the following form:
\begin{equation}\label{pi0}
\pi^{(0)}_{ij}= - \frac{c_{(2k)}}{d-2k}{\cal O}^{(2k)}_{ij} + \tilde{\pi}^{(0)}_{ij}\,,
\end{equation}
where $\tilde{\pi}^{(0)}_{ij}$ is finite at $d=2k$. Substituting back \eqref{pi0} and \eqref{gamma2k} to \eqref{hex0} we get
\begin{equation}
h_{ij}(z;x)=\sum_{n=0}^{k-1}\gamma_{ij}^{(2n)}\left(\frac{z}{L}\right)^{2n-2} + \big(\tilde\gamma_{ij}^{(2k)}+ \tilde{\pi}^{(0)}_{ij}\big)\left(\frac{z}{L}\right)^{2k-2}-c_{(2k)}\left(\frac{z}{L}\right)^{2k-2}\text{ln}\left(\frac{z}{L}\right){\cal O}^{(2k)}_{ij} +o\big((z/L)^{d}\big)\,.
\end{equation}
This makes contact with the expansion with a logarithmic term (for an even $d$) presented in the literature, e.g.\ \cite{Henningson:1998gx,deHaro:2000vlm,Skenderis2002}.

\section{Weyl-Fefferman-Graham Formalism}
\label{sec:WFG}
In this section we provide a brief review of the Weyl-Fefferman-Graham (WFG) formalism established in \cite{Ciambelli:2019bzz}. We will see that in the WFG gauge, the conformal boundary of an ALAdS spacetime is endowed with Weyl geometry, and the geometric quantities are naturally upgraded to the ``Weyl quantities" that we introduced in Section \ref{sec:Weyl}.
\par
The Fefferman-Graham ansatz \eqref{FG} is quite convenient for calculations, especially in the context of holographic renormalization. In this setup, one can induce a Weyl transformation of the boundary metric by a bulk diffeomorphism, namely the PBH transformation \cite{Imbimbo_2000},
\begin{equation}
z\to z'= z/{\cal{B}}(x)\,,\qquad  x^{i}\to x'^{i}= x^{i}+ \xi^{i}(z;x)\, , 
\end{equation}
where $\xi^{i}(z;x)$ vanish at the boundary $z=0$. The functions $\xi^{i}(z;x)$ can be found (infinitesimally) in terms of ${\cal{B}}(x)$ by the constraint that the form of the FG ansatz is preserved under the transformation. However, under the PBH transformation, the subleading terms in the FG expansion \eqref{hex0} do not transform in a Weyl-covariant way. The source of this complication is the compensating diffeomorphisms $\xi^{i}(z;x)$ introduced for preserving the FG ansatz. 
\par
This above-mentioned issue motivated the authors of \cite{Ciambelli:2019bzz} to replace the FG ansatz with
\begin{align}
\label{WFG}
\td s^2=L^{2}\left(\frac{\td z}{z}-a_i(z;x) \td x^i\right)^2
+h_{ij}(z;x)\td x^i \td x^j\,,
\end{align}
which was named the Weyl-Fefferman-Graham ansatz. With the additional Weyl structure $a_{i}$ added, the form of the WFG ansatz is now preserved under the Weyl diffeomorphism
\begin{align}
\label{weyl}
z\to z'=z/{\cal B}(x)\,,\qquad x^i\to x'^i=x^i\,.
\end{align}
It is not hard to see that the Weyl diffeomorphism \eqref{weyl} induces the following transformation of the fields $a_{i}$ and $h_{ij}$:
\begin{align}
\label{weylha}
a_i(z;x)\to a'_{i}(z';x)= a_i({\cal B}(x)z';x)-\p_i\ln{\cal B}(x)\,,\quad h_{ij}\to h'_{ij}(z';x)= h_{ij}({\cal B}(x)z';x)\,.
\end{align}
Thus, we can now induce a Weyl transformation on the boundary and preserve the form of the metric without introducing the irritating $\xi^{i}(z;x)$. Note that according to the FG theorem, any ALAdS spacetime can always be expressed in the FG form, and so \eqref{WFG} can be transformed into \eqref{FG} under a suitable diffeomorphism. This indicates that $a_i$ is actually pure gauge in the bulk. Another way of going back to the FG gauge is to simply set  $a_i$ to zero; in this perspective, the FG gauge is nothing but a special case of the WFG gauge with a fixed gauge.
\par
The main utility of the WFG gauge is that all the terms (except one) in the $z$-expansions of $h_{ij}(z;x)$ and $a_{i}(z;x)$ transform as Weyl tensors under Weyl diffeomorphisms. To see this, let us expand $h_{ij}$ and $a_i$ near $z=0$:
\begin{align}
\label{hex}
h_{ij}(z;x)&=\frac{L^2}{z^2}\left[\gamma^{(0)}_{ij}(x)+\frac{z^2}{L^2}\gamma^{(2)}_{ij}(x)+\cdots\right]+\frac{z^{d-2}}{L^{d-2}}\left[\pi^{(0)}_{ij}(x)+\frac{z^2}{L^2}\pi^{(2)}_{ij}(x)+\cdots\right]\,,\\
\label{aex}
a_{i}(z;x)&=\left[a^{(0)}_{i}(x)+\frac{z^2}{L^2}a^{(2)}_{i}(x)+\cdots\right]
+\frac{z^{d-2}}{L^{d-2}}\left[p^{(0)}_{i}(x)+\frac{z^2}{L^2}p^{(2)}_{i}(x)+\cdots\right]\,.
\end{align}
In the FG gauge where $a_i$ is turned off, the FG expansion only includes \eqref{hex}, and the subleading terms $\gamma^{(2k)}_{ij}$ in the first series are determined solely by the boundary induced metric $\gamma^{(0)}_{ij}$ and its derivatives. Now with the extra series \eqref{aex}, $\gamma^{(2k)}_{ij}$ will also depend on $a^{(0)}_i$, $a^{(2)}_i$, $a^{(4)}_i$, etc. 
Moving on, from the transformations \eqref{weylha} under a Weyl diffeomorphism, one finds the transformation of each term in the expansions \eqref{hex} and \eqref{aex} as follows \cite{Ciambelli:2019bzz}:
\begin{align}
\label{GP1}
\gamma^{(2k)}_{ij}(x)&\to\gamma^{(2k)}_{ij}(x){\cal B}(x)^{2k-2}\,,\qquad
\pi^{(k)}_{ij}(x)\to \pi^{(2k)}_{ij}(x){\cal B}(x)^{d-2+2k}\,,\\
\label{AP}
a^{(2k)}_{i}(x)&\to a^{(2k)}_{i}(x){\cal B}(x)^{2k}-\delta_{k,0}\p_i\ln{\cal B}(x)\,,\qquad
p^{(2k)}_{i}(x)\to p^{(2k)}_{i}(x){\cal B}(x)^{d-2+2k}\,.
\end{align}
Indeed, we see that almost all the terms in the expansions transform Weyl-covariantly. The only exception is $a^{(0)}_{i}$, which transforms inhomogeneously under Weyl transformation, and thus does not have a definite Weyl weight. All the other terms in the expansions \eqref{hex} and \eqref{aex} can be viewed as tensor fields on the boundary and we can easily read off their Weyl weights from the power of ${\cal B}(x)$  appearing in \eqref{GP1} and \eqref{AP}. 
\par
Having the expansion of $h_{ij}$, it is also useful to expand its inverse:
\begin{align}
\label{hinv}
h^{ij}(z;x)&=\frac{z^2}{L^2}\left[\gamma_{(0)}^{ij}(x)+\frac{z^2}{L^2}\gamma_{(2)}^{ij}(x)+...\right]+\frac{z^{d+2}}{L^{d+2}}\left[\pi_{(0)}^{ij}(x)+\frac{z^2}{L^2}\pi_{(2)}^{ij}(x)+...\right]\\
&=\frac{z^2}{L^2}\left[\gamma_{(0)}^{ij}(x)-\frac{z^2}{L^2}\tilde{m}_{(2)k}^{i}\gamma_{(0)}^{kj}(x)-\frac{z^4}{L^4}\tilde{m}_{(4)k}^{i}\gamma_{(0)}^{kj}(x)+\cdots\right]\nn+\frac{z^{d+2}}{L^{d+2}}\left[\tilde{n}_{(2)k}^{i}\gamma_{(0)}^{kj}(x)+\cdots\right]\,,
\end{align}
where $\tilde{m}_{(2k)j}^{i}\equiv-\gamma_{(2k)}^{ik}\gamma^{(0)}_{kj}$, $\tilde{n}_{(2k)j}^{i}\equiv-\pi_{(2k)}^{ik}\gamma^{(0)}_{kj}$. Denoting $m_{(k)j}^{i}\equiv\gamma_{(0)}^{ik}\gamma^{(k)}_{kj}$ and $n_{(k)j}^{i}\equiv\gamma_{(0)}^{ik}\pi^{(k)}_{kj}$, we can solve the above expansion order by order and get
\begin{align}
\gamma_{(0)}^{ij}&=(\gamma^{(0)}_{ij})^{-1}\,,\qquad \tilde m_{(2)j}^{i}=m_{(2)j}^{i}\,,\qquad\tilde m_{(4)j}^{i}=m_{(4)j}^{i}-m_{(2)k}^{i}m_{(2)j}^{k}\,,\qquad\cdots\\
\tilde n_{(0)j}^{i}&=n_{(0)j}^{i}\,,\qquad\tilde n_{(2)j}^{i}=n_{(2)j}^{i}-m_{(2)k}^{i}n_{(0)j}^{k}-n_{(0)k}^{i}m_{(2)j}^{k}\,,\qquad\cdots\nn
\end{align}
\par
For a metric in the form of \eqref{WFG} defined on the bulk manifold $M$, one can choose a dual form basis and its corresponding vector basis as follows:
\begin{align}
\label{basis}
\bm e^z&=L\frac{\td z}{z}-La_i(z;x)\td x^i\,,\qquad \bm e^i=\td x^i\,,\\
\un e_z&=\frac{z}{L}\un\p_z\equiv \un D_z\,,\qquad \un e_i=\un\p_i+za_i(z;x)\un\p_z\equiv \un D_i\,.
\end{align}
Then the tangent space at any point $(z,x^{i})\in M$ can be spanned by the basis $\{\un D_z,\un D_i\}$, and the basis vectors $\{\un D_i\}$ form a $d$-dimensional distribution on $M$ which belongs to the kernel of $\bm e^z$. The Lie brackets of these basis vectors are
\begin{align}
\label{DmDn}
[\un D_i,\un D_j]=Lf_{ij}\un D_z\,,\qquad[\un D_z,\un D_i]=L\varphi_i\un D_z\,,
\end{align}
where $\varphi_i\equiv D_za_i$ and $f_{ij}\equiv D_i a_j-D_j a_i$ ($D_z$ and $D_i$ represent taking the derivatives along $\un e_z$ and $\un e_i$). According to the Frobenius theorem, the condition for the distribution spanned by $\{\un D_i\}$ to be integrable is that $[\un D_i,\un D_j]=0$, i.e.\ $f_{ij}=0$. In this case, this distribution defines a hypersurface. For instance, in the FG gauge where  $a_{i}$ is turned off, the distribution $\{\un D_i\}$ becomes $\{\un \p_i\}$, which generates a foliation of constant-$z$ surfaces. However, $\{\un D_i\}$ in the WFG gauge is not necessarily an integrable distribution, and thus one needs to keep in mind that the boundary hypersurface $z=0$ is in general not part of a foliation.
\par 
Suppose $\nabla$ is the Levi-Civita (LC) connection on $M$. One can find the connection coefficients of $\nabla$ in the frame $\{\un D_z,\un D_i\}$ from its definition \eqref{conncoef}:
\begin{align}
\nabla_{\un D_i}\un D_j=\Gamma^k{}_{ij}\un D_k+\Gamma^z{}_{ij}\un D_z\,.
\end{align}
The coefficients $\Gamma^k{}_{ij}$ in the above equation define the induced connection coefficients on the distribution over $M$ spanned by $\{\un D_i\}$. Using the LC condition (torsion-free and metricity-free) of the bulk $\nabla$ we obtain that
\begin{align}
\label{bulkLC}
\Gamma^k{}_{ij}=\frac{1}{2}h^{kl}( D_i h_{lj} +D_j h_{il} - D_l h_{ji})\,,
\end{align}
where we have read from \eqref{DmDn} that the commutation coefficients vanish. Expanding $\Gamma^k{}_{ij}$ with respect to $z$, at the leading order one finds that
\begin{align}
 \label{IndWeyl}
\Gamma^{k}_{(0)}{}_{ij}
 &=\frac{1}{2}\gamma_{(0)}^{kl}\big(
 \p_i \gamma^{(0)}_{jl}
 +\p_j \gamma^{(0)}_{il}-\p_l \gamma^{(0)}_{ij}\big)-\big(a^{(0)}_i\delta^k{}_j+a^{(0)}_j\delta^k{}_i+a^{(0)}_l\gamma_{(0)}^{kl}\gamma^{(0)}_{ij}\big)\,.
\end{align} 
We can see that \eqref{IndWeyl} gives exactly the connection coefficients of a torsion-free connection with Weyl metricity shown in \eqref{WeylLC} (where $A_i$ and $g_{ij}$ correspond to $a^{(0)}_i$ and $\gamma^{(0)}_{ij}$). That is, on the boundary with $z\to0$ we have a connection $\nabla^{(0)}$ satisfying
\begin{align}
\label{nonmetry}
\nabla^{(0)}_i\gamma^{(0)}_{jk}=2a^{(0)}_i\gamma^{(0)}_{jk}\,.
\end{align}
This indicates that although $a_i$ is pure gauge in the bulk, its leading order $a_{i}^{(0)}$ serves as a Weyl connection at the conformal boundary. Together with the induced metric $\gamma^{(0)}_{ij}$, they provide a Weyl geometry at the boundary \cite{10.4310/jdg/1214429379}. Under a boundary Weyl transformation
\begin{align}
\label{WT}
\gamma_{ij}^{(0)} \to {\cal B}(x)^{-2} \gamma^{(0)}_{ij}\,,\qquad a_{i}^{(0)}\to a_{i}^{(0)}- \partial_{i}\text{ln}{\cal}B(x)\,,
\end{align}
for any tensor $T$ (with indices suppressed) with Weyl weight $w_T$ on the boundary, we have
\begin{align}
T\to B^{w_T}T\,,\qquad(\nabla^{(0)}_i T+w_Ta^{(0)}_i T)\to B^{w_T}(\nabla^{(0)}_i T+w_Ta^{(0)}_i T)\,.
\end{align} 
One can also absorb the Weyl connection and define $\hat\nabla^{(0)}$ such that
\begin{align}
\hat\nabla^{(0)}_i T\equiv\nabla^{(0)}_i T+w_Ta^{(0)}_i T\,,
\end{align} 
which renders $\hat\nabla^{(0)}_i T$ Weyl-covariant. Particularly, Eq.\ \eqref{nonmetry} indicates that $\hat\nabla^{(0)}$ is a Weyl-LC connection, which makes it convenient for boundary calculations.
\par
Now that we have the Weyl geometry on the boundary, the geometric quantities there are promoted to the ``Weyl quantities" as we demonstrated in Section \ref{sec:Weyl}. More precisely, for any geometric quantity constructed by the boundary metric $\gamma^{(0)}_{ij}$ and the LC connection in the FG case, we now have a Weyl-covariant counterpart of it constructed by $\gamma^{(0)}_{ij}$, $a^{(0)}_i$ and $\hat\nabla^{(0)}$ in the WFG case. For instance, we have the Weyl-Riemann tensor $\hat R^{i}{}^{(0)}_{jl\sigma}$, Weyl-Ricci tensor $\hat R^{(0)}_{ij}$ and Weyl-Ricci scalar $\hat R^{(0)}$. In addition, $f_{ij}$ induces on the boundary a tensor $f^{(0)}_{ij}=\p_i a^{(0)}_j-\p_j a^{(0)}_{i}$, namely the curvature of the Weyl connection $a^{(0)}$, which is obviously Weyl-invariant. We can also define the Weyl-Schouten tensor $\hat P_{ij}$ and Weyl-Cotton tensor $\hat C_{ijl}$ on the boundary as follows:
\begin{align}
\label{WP}
\hat P_{ij}&=\frac{1}{d-2}\bigg(\hat R^{(0)}_{ij}-\frac{1}{2(d-1)}\hat R^{(0)}\gamma^{(0)}_{ij}\bigg)\,,\\
\label{WC}
\hat C_{ijl}&=\hat\nabla^{(0)}_{l}\hat P_{ij}-\hat\nabla^{(0)}_{j}\hat P_{il}\,.
\end{align}
In Chapter \ref{chap:WOT}, we will also see the Weyl-covariant counterparts of the obstruction tensors.
\par
We emphasis again that the symmetry of the indices of a Weyl quantity is not necessarily the same as the corresponding quantity defined with the LC connection. For instance, the Weyl-Ricci tensor is not symmetric, with its antisymmetric part $\hat R^{(0)}_{[ij]}=-(d-2)f^{(0)}_{ij}/2$, and hence the Weyl-Schouten tensor $\hat P_{ij}$ also contains an antisymmetric part $\hat P_{[ij]}=-f^{(0)}_{ij}/2$. 

\chapter{Weyl-Ambient Geometries}
\label{chap:Weylambiant}
\section{Ambient Metrics}
In this section we will start by reviewing the FG ambient metric and then introduce the Weyl-ambient metric. To build up some intuition, we begin with the flat ambient metric and then generalize to Ricci-flat ambient metrics.
\label{sec:ambient}
\subsection{Flat Ambient Metrics}
\label{sec:flatambiant}
The simplest example of an ambient space is the flat ambient space. Consider the $(d+2)$-dimensional Minkowski spacetime $\mathbb{R}^{1,d+1}$ with the metric
\begin{equation}\label{Flat_Ambient metric}
\eta= - (\td X^{0})^{2}+\sum_{i=1}^{d+1}(\td X^{i})^{2}\,.
\end{equation}
One can describe $(d+1)$-dimensional Euclidean AdS spaces as the following codimension-1 hyperboloids:\footnote{One can also take the signature in \eqref{Flat_Ambient metric} to be $(2,d)$. Then, $g^+$ will be the Lorentzian signature AdS spacetime and the $\delta_{ij}$ in \eqref{Flat_Ambient_3} becomes $\eta_{ij}$. More generally, if one takes the signature in \eqref{Flat_Ambient metric} to be $(p,d+2-p)$, then the signature of $g^+$ will be $(p-1,d+2-p)$.}
\be
(X^{0})^2-R^2=L^2\,,\qquad R^2=\sum_{i=1}^{d+1}(X^{i})^{2}\,,
\ee
where $L$ represents the AdS radius. The hyperboloids with different $L$ form a one-parameter family of hypersurfaces foliating the interior of the future light cone, denoted by ${\cal N}^+$, emanating from the origin of the Lorentzian coordinate system $\{X^{0}, X^{i}\}$. Then, one can also write the Minkowski metric in the following ``cone'' form:
\begin{equation}\label{Flat_Ambient_4}
\eta = -\td \ell^2 + \frac{\ell^2}{L^2} g^{+}\,,\qquad \ell>0\,,
\end{equation}
where the coordinate $\ell=\sqrt{(X^{0})^2-R^2}$, and $g^{+}$ is the $(d+1)$-dimensional Euclidean AdS metric. Now the Euclidean AdS space is represented by the hyperbola $\ell=L$. The metric $g^{+}$ can be expressed in the Fefferman-Graham (FG) form in the following different ways (see Appendix \ref{app:coords} for details):
\begin{align}
\label{eq:adsglob}
\qquad g^{+}_{S}&=  \frac{L^2}{z^2} \Big(\td z^2 + L^2(1- \frac{1}{4}(z/L)^2)^2\td\Omega_{d}^2\Big)\,,\qquad 0<z<2L\,,\\
\label{eq:adsPoin}
\qquad g^{+}_{F}&= \frac{L^2}{z^2}\left(\td z^2 + \delta_{ij}\td x^{i}\td x^{j}\right)\,,\qquad i=1,\cdots,d\,,\qquad z>0\,.
\end{align}
The metric \eqref{Flat_Ambient_4} with $g^+=g^{+}_{S}$ or $g^{+}_{F}$ is defined in the whole interior of the light cone ${\cal N}^+$,\footnote{Note that for Lorentzian signature AdS spacetime, the metric \eqref{Flat_Ambient_4} with $g^{+}_{F}$ only covers half of the interior of the future light cone.} while their AdS boundaries have different topologies. It is easy to see that the AdS boundary at $z\to0^+$ of $g^+_S$ in \eqref{eq:adsglob} is conformally a $d$-sphere while that of $g^+_F$ in \eqref{eq:adsPoin} is conformally flat. 

While the metric \eqref{Flat_Ambient_4} is singular in the limit $z\to 0^{+}$ with $\ell$ fixed, it is well-defined when taking both $z$ and $\ell$ to zero with $z/\ell$ fixed. To make this evident we introduce a new coordinate system $\{t,x^{i},\rho\}$, called the \emph{ambient coordinate system}, with $t=\ell/z$ and $\rho=-z^{2}/2$. First we look at the metric \eqref{Flat_Ambient_4} with $g^+_S$ in \eqref{eq:adsglob}, which in the ambient coordinate system becomes
\begin{align}\label{Flat_Ambient_2}
\eta &= 2\rho \td t^2 + 2t \td t \td\rho + t^2 (1+ \frac{\rho}{2 L^2})^2 L^2 \td\Omega_{d}^2 \,.
\end{align} 
The coordinate patch of $\{\ell,x^i,z\}$ which covers the interior of the light cone surface ${\cal N}^{+}$, corresponds to $t\in (0,\infty) $, $\rho \in (-2L^2,0)$ (see Figure~\ref{fig:cones}). However, it is apparent now that the limit $\rho\to 0^{-}$ of the above metric is well-defined, and thus we can extend the coordinate patch of $\{t,x^i,\rho\}$ to include an open neighborhood of the surface ${\cal N}^{+}$ at $\rho=0$. Hence, ${\cal N}^{+}$ is parametrized by $\{t,x^{i}\}$, where $t\in \mathbb{R}_{+}$ and $x^{i}$ are the coordinates of the $d$-sphere $S_{d}$. In other words, ${\cal N}^{+}$ can be regarded as a line bundle over $S^d$ whose fibers are parametrized by $t$.
\begin{figure}[!htb]
\begin{center}
\includegraphics[width=3.5in]{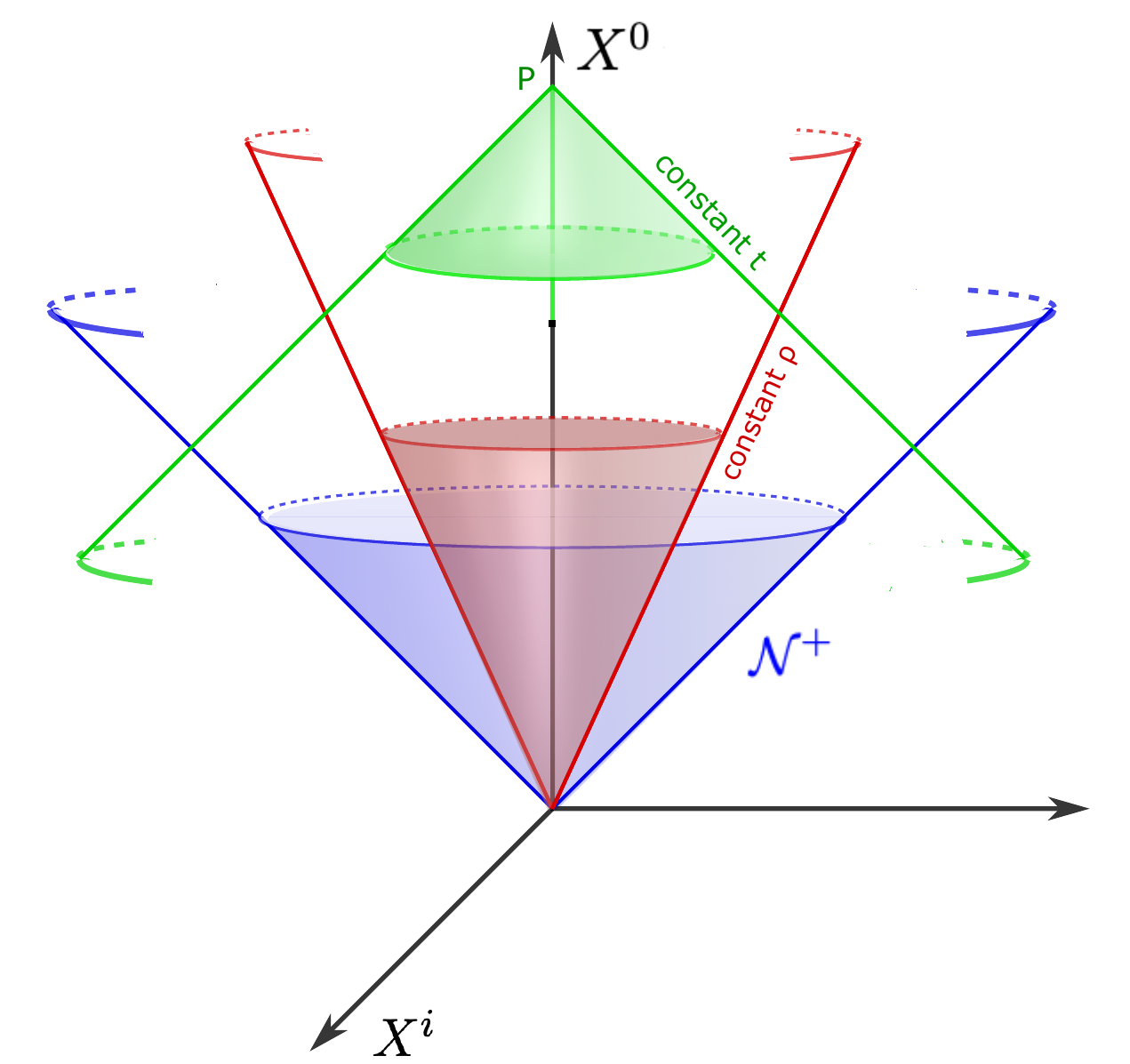}
\caption{Sketch of a constant-$\rho$ surface (red) and a constant-$t$ surface (green) of the flat ambient metric \eqref{Flat_Ambient_2} in the Lorentzian coordinate system $\{X^{0}, X^{i}\}$. Constant-$t$ surfaces are past directed light cones. Changing $t$ moves the apex $P$ of the cone along the $X^{0}$-axes. Constant-$\rho$ surfaces are future directed timelike cones. When $\rho\to 0^{-}$ the constant $\rho$ surface becomes the light cone ${\cal N}^{+}$ (blue) \cite{Jia:2023gmk}.} 
\label{fig:cones}
\end{center}
\end{figure}
\par
Suppose $\phi$ is a function on $\bb R^{1,d+1}$, which defines a hypersurface $\Sigma$ by the locus of points $p\in \mathbb{R}^{1,d+1}$ such that $\phi(t,x^{i},\rho)|_p=0$. In order to find the intersection $\Sigma \cap \mathcal{N}^{+}$, one can set $\rho=0$ and solve for $t$ as a function $t(x^{i})$ of the $d$-sphere coordinates from $\phi(t,x,\rho=0)=0$. The pullback metric on the intersection submanifold is $\eta |_{\Sigma \cap {\cal N}^{+}} = t (x)^2L^2 \td\Omega_{d}^2$. The function $t(x)$ depends on the choice of function $\phi$ (which is arbitrary) that defines $\Sigma$, and thus we see that the pullback metric is conformally equivalent to the metric of $S_{d}$. An example is to take $\phi=\ln t$, and to consider the pull back of the metric at $\rho=0$, $t=1$, namely $\eta|_{\rho=0,t=1}= L^2\td\Omega^{2}_{d}$. If we perform a diffeomorphism $t={\cal B}(x)^{-1}t'$ and pull back the metric at $\rho=0$, $t'=1$, then we find $\eta|_{\rho=0,t'=1}= {\cal B}(x)^{-2}L^2\td\Omega^{2}_{d}$. Therefore, at $\rho=0$ we have a \emph{conformal class} $[g]$ of $d$-dimensional metrics, and the $(d+2)$-dimensional Minkowski metric expressed in \eqref{Flat_Ambient_2} is said to be the ambient metric of $[g]$. This implies that the null surface ${\cal N}^{+}$ at $\rho=0$ is associated with a metric bundle, which will be important for the formal construction later in Subsection \ref{sec:bottomup}.
\par
Similarly, the metric \eqref{Flat_Ambient_4} with $g^+_{F}$ in \eqref{eq:adsPoin} can also be expressed in the ambient coordinates as

\begin{equation}\label{Flat_Ambient_3}
\eta = 2\rho \td t^2 + 2t\td t \td\rho + t^2\delta_{ij} \td x^{i}\td x^{j}\,,\qquad i=1,\cdots,d\,.
\end{equation}
In this case, the original coordinate patch of $\{\ell,x^i,z\}$ corresponds to $t\in(0,\infty)$, $\rho\in(-\infty,0)$, and the null surface ${\cal N}^+$ is again covered by the $\{t,x^i,\rho\}$ system at $\rho=0$. Intersecting the null surface with a hypersurface and taking the pullback metric on the intersection, we now obtain a $d$-dimensional metric $\td s^2= t(x)^2 \delta_{ij}\td x^{i}\td x^{j}$ that is conformally flat. This metric is also in the conformal class $[g]$ but the topology is different from the $d$-dimensional metric obtained from \eqref{Flat_Ambient_2}. Note that the flat ambient metric in either \eqref{Flat_Ambient_2} or \eqref{Flat_Ambient_3} is homogeneous of degree 2 with respect to the $t$-coordinate; that is, under a constant scaling $t\to s t$ the metric transforms as $\eta \to s^{2}\eta$, or in the infinitesimal form,
\begin{equation}
\label{eq:homo}
{\cal L}_{\un T}\eta= 2\eta \,,\qquad\un T=t\un\p_{t}\,.
\end{equation}
We will retain this property also for Ricci-flat ambient metrics and the Weyl-ambient metric. For relaxation of this homogeneity condition, see \cite{graham2008inhomogeneous}.

\subsection{Ricci-Flat Ambient Metrics}
\label{sec:RFambient}
The flat ambient metric combines hyperbolic metrics and their conformal boundaries in a unified  framework. Before we describe its utility, we will review the  generalization of flat ambient metrics to Ricci-flat ambient metrics. This will allow us to consider $(d+1)$-dimensional asymptotically locally Anti-de Sitter (ALAdS) spaces which are especially relevant in holographic theories.
\par
The main observation that allows an extension to Ricci-flat ambient metrics is that \eqref{Flat_Ambient_4} can be generalized in the following form:
\begin{equation}\label{R_Flat_Ambient}
\ti g= -\td \ell^2 + \frac{\ell^2}{L^2} g^{+}_{\mu\nu}(x)\td x^{\mu}\td x^{\nu},\qquad \mu,\nu=1,\cdots d+1\,,\quad \ell>0\,,
\end{equation}
where now $g^{+}(x)$ is an arbitrary $(d+1)$-dimensional metric independent of $\ell$. We will refer to this $(d+1)$-dimensional geometry as the ``bulk''. The ambient Ricci tensor $\tilde Ric(\tilde g)$ can be decomposed in terms of the Ricci tensor of $g^+$ as \cite{Fefferman:2007rka,Graham:1991jqw}
\begin{equation}
\ti Ric(\ti g)= Ric(g^{+}) + \frac{d}{L^2} g^{+} \,.
\end{equation}
The right-hand side of the above equation can also be written as $G_{\mu\nu}(g^{+})+ \Lambda g^{+}_{\mu\nu}$ with $\Lambda= -\frac{d(d-1)}{2L^2}$.  Therefore, when the ambient metric $\ti g$ is Ricci-flat, $g^{+}$ is an Einstein metric and thus satisfies the vacuum Einstein equations.
\par
According to the Fefferman-Graham theorem \cite{AST_1985__S131__95_0,Graham:1991jqw}, any ALAdS Einstein metric $g^+$ can be expressed in the Fefferman-Graham form \eqref{FG}
\begin{equation}
\label{eq:g+FG}
g^{+}= L^2\frac{\td z^2}{z^2}+ \frac{L^2}{z^2}\gamma_{ij}(x,z)\td x^{i}\td x^{j}\,,\qquad i,j=1,\cdots ,d\,,\quad z>0\,,
\end{equation}
where $h_{ij}(x,z)=\gamma_{ij}(x,z)/z^2$ in \eqref{FG}. Then, by a coordinate transformation $t=\ell/z$ and $\rho=-z^{2}/2$, the metric \eqref{R_Flat_Ambient} takes the form
\begin{equation}\label{ambient_metric}
\ti g = 2\rho \td t^2 + 2t\td t\td\rho + t^2 \gamma_{ij}(x,\rho)\td x^{i}\td x^{j} \,,\qquad t>0\,.
\end{equation}
We can see that the flat ambient metrics \eqref{Flat_Ambient_2} and \eqref{Flat_Ambient_3} are nothing but special cases of \eqref{ambient_metric} when $\tilde g=\eta$ and $g^+$ is taken to be \eqref{eq:adsglob} and \eqref{eq:adsPoin}, respectively. The codimension-2 metric is now generalized to an arbitrary $\gamma_{ij}(x,z)$ whose corresponding $g^+$ in \eqref{eq:g+FG} is an Einstein metric.
\par
Note that the advantages of the ambient coordinate system $\{t,x^i,\rho\}$ mentioned before for the flat ambient space are now carried over to the Ricci-flat case. One can see that the surface at $\rho=0$ is still a null hypersurface, denoted by $\cal N$, which is a coordinate singularity in the original $\{\ell,x^i,z\}$ coordinate system. Hence, the ambient coordinate system permits one to extend the spacetime region to include an open neighborhood of the null surface ${\cal N}$. Denoting the extended spacetime manifold as $\tilde M$, then ${\cal N}$ is a hypersurface in $\tilde M$ parametrized by $\{t,x^{i}\}$, which furnishes a conformal class $[\gamma]$ of codimension-2 metrics. Suppose $M$ is a $d$-dimensional manifold equipped with the conformal class $[\gamma]$, then $(\tilde M,\tilde g)$ is called the $(d+2)$-dimensional ambient space of $(M,[\gamma])$. 
\par
Being part of the Ricci-flat ambient space, ${\cal N}$ can be regarded as an initial value surface. Then given the initial data $\gamma_{ij}(x,\rho)|_{\rho=0}$, the Ricci-flatness condition can be used to ``propagate'' the metric beyond the initial surface to a neighborhood around $\rho=0$. That is, the Ricci-flatness condition $\ti Ric(\ti g)=0$ is a set of differential equations for $\tilde g_{ij}(x,\rho)$, which can be solved iteratively in a series around $\rho=0$ given the initial value $\ti g(x,\rho)|_{\rho=0}$. The initial value problem for the Ricci-flat ambient space has been defined and evaluated rigorously in \cite{Fefferman:2007rka}, the results of which will be carried over to the Weyl-ambient space in Subsection \ref{sec:bottomup}.

\subsection{Weyl-Ambient Metrics}
\label{sec:Weylambiant}
Now we are ready to introduce the Weyl-ambient metric. We start from the $(d+2)$-dimensional ambient metric in the form of \eqref{R_Flat_Ambient}. The expression of $g^+$ in \eqref{eq:g+FG} is the FG ansatz for an ALAdS spacetime, which is not  preserved under a Weyl diffeomorphism $z\to z/{\cal B}(x)$, $x^i\to x^i$ as we explained in Section \ref{sec:WFG}. To manifest the Weyl covariance, one should apply the WFG gauge to $g^+$ by adding an additional mode $a_\mu$ to \eqref{eq:g+FG} as follows:
\begin{equation}
\label{eq:g+WFG}
g ^+_{\text{WFG}}=L^2\left(\frac{\td z}{z} - a_{i}(x,z)\td x^{i}\right)^{2}+ \frac{L^2}{z^2}\gamma_{ij}(x,z)\td x^{i}\td x^{j}\,,\qquad z>0\,.
\end{equation}
Now we substitute the $g^+$ in \eqref{R_Flat_Ambient} with the WFG ansatz \eqref{eq:g+WFG}, then transforming back to the ambient coordinates $\{t,x^i,\rho\}$, we obtain the \emph{Weyl-ambient metric}
\begin{equation}\label{Weyl_ambient}
\ti g = 2\rho \td t^2 + 2t^2\td\rho\left(\frac{\td t}{t} + a_{i}(x,\rho)\td x^{i}\right) + t^2 g_{ij}(x,\rho)\td x^{i}\td x^{j}\,,\qquad t>0\,,
\end{equation}
where $g_{ij}(x,\rho):= \gamma_{ij}(x,\rho)- 2\rho a_{i}(x,\rho)a_{j}(x,\rho)$. We call the pseudo-Riemannian space $(\tilde M,\tilde g)$  a \emph{Weyl-ambient space}. Having the form of the Weyl-ambient metric,  the \emph{ambient Weyl diffeomorphism}\footnote{In terms of the coordinates $\ell,z$, the ambient Weyl diffeomorphism acts as $(\ell',x'^i,z')=(\ell, x^i, {\cal B}(x)^{-1}z)$.}
\begin{equation}
\label{eq:Weyldiff}
t'={\cal B}(x)t\,,\qquad x'^i=x^i\,,\qquad \rho'={\cal B}(x)^{-2}\rho
\end{equation}
induces a change in the constituents $a_i$ and $\gamma_{ij}$ of the form
\begin{equation}
\label{eq:Weyldiff2}
a'_i(x',\rho')=a_i(x,\rho)-\p_i\ln{\cal B}(x)\,,\qquad\gamma'_{ij}(x',\rho')={\cal B}(x)^{-2}\gamma_{ij}(x,\rho)\,.
\end{equation}
If we regard the ALAdS bulk as a hypersurface of the Weyl-ambient space, the above transformation gives rise to the Weyl diffeomorphism which preserves the WFG ansatz. In addition, we want to point out that just as the ambient metric \eqref{ambient_metric} is homogeneous with respect to $t$, the homogeneity property \eqref{eq:homo}  also pertains for the Weyl-ambient metric \eqref{Weyl_ambient} since both $a_{i}(x,\rho)$ and $\gamma_{ij}(x,\rho)$ are independent of $t$.   This homogeneity property will be repeatedly used throughout this thesis. In the following we use this property in order to show how an induced Weyl class arises from the Weyl-ambient metric; it is also crucial for the bottom-up construction and for proving Propositions~\ref{prop1} and \ref{Proposition_6.5_Graham}.
\par
The Ricci-flatness condition $\ti Ric(\ti g)=0$ for the Weyl-ambient metric \eqref{Weyl_ambient}, similar to that for the ambient metric \eqref{ambient_metric}, is a set of differential equations for $\tilde g_{ij}(x,\rho)$ which can be solved order by order in a neighborhood of $\rho=0$ given the initial value $\ti g_{ij}(x,\rho)|_{\rho=0}$. To be precise, in a neighborhood of $\rho=0$ we can expand $\gamma_{ij}$ and $a_{i}$ as\footnote{Similar to \eqref{hex0}, there will be a second series starting from the $\rho^{d/2}$ order in the expansion \eqref{eq:gexpan}:
\begin{align*}
\gamma_{ij}(x,\rho)&=( \gamma^{(0)}_{ij}(x)+ \gamma^{(1)}_{ij}(x)\rho+ \cdots)+\rho^{{d/2}}( \pi^{(0)}_{ij}(x)+ \pi^{(1)}_{ij}(x)\rho  + \cdots)\,.
\end{align*}
However, to solve for the second series in $\gamma_{ij}$ order by order one needs the interior data $\pi_{ij}^{(0)}$ of the ambient space.}
\begin{align}
\label{eq:gexpan}
\gamma_{ij}(x,\rho)&= \gamma^{(0)}_{ij}(x)+ \gamma^{(1)}_{ij}(x)\rho +\gamma^{(2)}_{ij}(x)\rho^2 + \cdots\,,\\
\label{eq:aexpan} 
a_{i}(x,\rho)&= a^{(0)}_{i}(x) + a^{(1)}_{i}(x)\rho + a^{(2)}_{i}(x)\rho^2 + \cdots\,.
\end{align}
Notice that the $\gamma^{(k)}_{ij}$ and $a^{(k)}_{i}$ in the $\rho$-expansion here correspond to $(-2)^k\gamma^{(2k)}_{ij}/L^{2k}$ and $(-2)^ka^{(2k)}_{i}/L^{2k}$ in the $z$-expansion in \eqref{hex} and \eqref{aex}, respectively. From the equation $\ti Ric(\ti g)=0$, one can solve for $\gamma^{(n)}_{ij}(x)$ in terms of $\gamma^{(k)}_{ij}(x)$ and $a^{(k)}_{i}(x)$ with $k$ up to $n-1$. However, the modes $a^{(n)}_{i}(x)$ are not determined by the Ricci flatness condition and hence we regard $a^{(k)}_{i}(x,\rho)$ as input data. This initial value problem will be examined in detail in Subsection~\ref{sec:bottomup} after the Weyl-ambient space is defined in terms of the Weyl structure and the ansatz in \eqref{Weyl_ambient} will be shown to be the uniquely determined Weyl-ambient metric for any given $\gamma^{(0)}_{ij}(x)$ and $a_i(x,\rho)$. 
\par
From the transformation \eqref{eq:Weyldiff2} and the expansions \eqref{eq:gexpan} and \eqref{eq:aexpan}, we can see that $\gamma_{ij}^{(k\geqslant0)}$ and $a^{(k\geqslant1)}_{i}(x)$ transform covariantly under the ambient Weyl diffeomorphism \eqref{eq:Weyldiff}, with Weyl weights $2k-2$ and $2k$, respectively:
\be
\gamma^{(k\geqslant0)}_{ij}(x)\to{\cal B}(x)^{2k-2}\gamma^{(k\geqslant0)}_{ij}(x)\,,\qquad a_i^{(k\geqslant1)}(x)\to{\cal B}(x)^{2k}a_i^{(k\geqslant1)}(x)\,.
\ee
On the other hand, $a^{(0)}_i$ transforms as $a^{(0)}_i\to a^{(0)}_i-\p_i\ln{\cal B}$. Therefore, we should anticipate that $a^{(0)}_i$ can be interpreted as a Weyl connection on the codimension-2 geometry.  In Section \ref{sec:WFG} we have shown that the bulk metric of an ALAdS spacetime in the WFG gauge provides a Weyl geometry on the conformal boundary. In the next section we will show that by introducing $a_i(x,\rho)$ in the ambient metric, we indeed obtain a Weyl geometry at codimension-2, where $\gamma^{(0)}_{ij}$ and $a^{(0)}_i$ play the role of a metric and a Weyl connection, respectively.
\par
Closing this section, we remark that the codimension-1 surface ${\cal N}$ at $\rho=0$ is again a null surface parametrized by $(t,x)$ with $t\in \mathbb{R}_{+}$, just like the case of the ambient metric \eqref{ambient_metric}. This surface in fact has the structure of a line bundle with each fiber parametrized by $t$, which turns out to be a principal bundle with the structure group $\bb R_+$. The new ingredient $a_i$ in the Weyl-ambient metric \eqref{Weyl_ambient} induces naturally a connection on this principal bundle, represented by $a^{(0)}_i=a_i|_{\rho=0}$. We will explore this in Section~\ref{sec:bottomup}. 

\section{Weyl-Ambient Space}
\label{sec:WAS}
The goal of this section is formulate the Weyl-ambient geometry from two perspectives. First we analyze the Weyl-ambient metric from a top-down perspective by showing explicitly that the Weyl-ambient metric \eqref{Weyl_ambient} leads to a Weyl geometry at codimension-2. Then we introduce the more formal bottom-up construction of the Weyl-ambient space in Subsection \ref{sec:bottomup} and show that the Weyl ambient metric can be constructed from the codimension-2 Weyl geometry. 

\subsection{Top-Down Perspective}
\label{sec:topdown}
We start from a $(d+2)$-dimensional manifold $\tilde M$. Define a dual frame $\{\bm e^{P}\}$ on the $\tilde M$ as follows:
\begin{align}
\label{e+-def}
\bm e^+&=\td t+ta_i(x,\rho)\td x^i\,,\qquad \bm e^i=\td x^i\,,\qquad\bm e^-=t\td\rho+\rho \td t-t\rho a_i(x,\rho)\td x^i\,,
\end{align}
where now $P=\{+,i,-\}$. In this frame the Weyl-ambient metric \eqref{Weyl_ambient} can be written as
\begin{align}
\label{eq:metricnull}
\tilde g=\bm e^+\otimes\bm e^-+\bm e^-\otimes\bm e^++t^2\gamma_{ij}\bm e^i\otimes\bm e^j\,.
\end{align}
It is easy to check that the 1-forms defined in \eqref{e+-def} are covariant under \eqref{eq:Weyldiff} and \eqref{eq:Weyldiff2}, and thus the form of $\tilde g$ in \eqref{eq:metricnull} is preserved under an ambient Weyl diffeomorphism. The corresponding frame $\{\un D_{P}\}$ of \eqref{e+-def} reads
\begin{align}
\label{eq:D+-i}
\un D_+&=\un\p_t-\frac{\rho}{t}\un\p_\rho\,,\qquad\un D_i=\un\p_i-ta_i(x,\rho)\un\p_t+2\rho a_i(x,\rho)\un\p_\rho\,,\qquad\un D_-=\frac{1}{t}\un\p_\rho\,.
\end{align}
From \eqref{eq:metricnull} it is clear that $\un D_+$ and $\un D_-$ are null vectors. $\{\un D_i\}$ form a basis of a $d$-dimensional \emph{distribution} $C_d\subset T\tilde M$, defined as
\begin{align}
C_d=\big\{\un{\mathcal V}\in T\tilde M\,|\,i_{\un{\mathcal V}}\bm e^\pm=0\big\}\,.
\end{align}
It follows from \eqref{eq:D+-i} that
\begin{align}
[\un D_i,\un D_j]=-tf_{ij}\un D_++t\rho f_{ij}\un D_-\,,
\end{align}
where $f_{ij}=D_ia_j-D_ja_i$ is the curvature of $a_i(x,\rho)$. The Frobenius theorem implies that the distribution $C_d$ is integrable when $f_{ij}=0$, though we will not generally assume this to be the case. One should note that the codimension-1 distribution spanned by $\{\un D_i,\un D_+\}$ is integrable at $\rho=0$, and thus defines a codimension-1 subspace (see Appendix \ref{app:Null} for relevant details).

Suppose $M$ is a $d$-dimensional manifold with a local coordinate system $\{y^i\}$ on $U\subset M$, and a point $\tilde p\in\tilde M$ has coordinates $(t,x^i,\rho)$. One can consider the coordinate patch $\tilde U$ of the ambient coordinate system $\{t,x^i,\rho\}$ as a fiber bundle with the projection $\pi:\tilde U\to U$ such that $\pi(\tilde p)=p\in M$ has coordinates $y^i=x^i$, i.e.\ each fiber in $\tilde U$ is parametrized by $(t,\rho)$. For simplicity, in what follows we will refer to $\tilde U$ as $\tilde M$ and $U$ as $M$, and we will not distinguish $\{x^i\}$ and $\{y^i\}$. Now that we have a bundle structure $\pi:\tilde M\to M$, we can see that $a_i(x,\rho)$ plays the role of an Ehresmann connection that specifies the horizontal subspace $H_{\tilde p}=C_d|_{\tilde p}\subset T_{\tilde p}\tilde M$, which defines the horizontal lift $T_pM\to H_{\tilde p}$ with $\un\p_i\mapsto\un D_i$. In general then, we are describing an isolated surface.

\par
Since we have a bundle structure $\pi:\tilde M\to M$, each section defines an embedding $\phi:M\to\tilde M$ such that a point $p\in M$ with coordinates $x^i$ is mapped to $\phi(p)=(t(x),x^i,\rho(x))$. With the horizontal subspace defined, we have $\pi_*:H_p\to T_pM$ such that $\pi_*(\un D_i)=\un\p_i$. Now consider the embedding $\phi$ with $\phi(p)=(t=1,x^i,\rho=0)$. We can define an induced metric $\gamma_{ij}^{(0)}(x)$ on $M$ by ``pulling back''\footnote{Note that we abuse the term as this is technically not a standard pullback by the embedding $\phi$, because $\un D_i$ is not tangent to $\phi[M]$.} $\tilde g_{ij}(t,x,\rho)= \tilde g(\un D_i,\un D_j)$ from the subspace of $\tilde M$ at $t=1$ and $\rho=0$ similar to what we did for the flat ambient space:
\be
\label{eq:inducedmetr}
\gamma^{(0)}_{ij}= \tilde g_{ij}|_{t=1,\rho=0}\,.
\ee
Under the coordinate transformation \eqref{eq:Weyldiff} in $\tilde M$ induced by an ambient diffeomorphism, we can consider the pullback $\gamma'^{(0)}(x')$ of $\tilde g'(t',x',\rho')$ by $\phi'(p)=(t'=1,x'^i,\rho'=0)$:
\be
\label{eq:inducedmetr'}
\gamma'^{(0)}_{ij}= \tilde g'_{ij}|_{t'=1,\rho'=0}\,,
\ee
where $\tilde g'_{ij}=g'(\un D'_i,\un D'_j)$, with $\un D'_i=\un\p'_i-t'a'_i(x',\rho')\un\p'_t+2\rho' a'_i(x',\rho')\un\p'_\rho$. Since $\tilde g'_{ij}=t'^2\gamma'_{ij}(x',\rho')$, we have
\be
\gamma'^{(0)}_{ij}={\cal B}(x)^{-2} \tilde g'_{ij}|_{t'={\cal B}(x),\rho'=0}={\cal B}(x)^{-2} \tilde g_{ij}|_{t=1,\rho=0}={\cal B}(x)^{-2}\gamma^{(0)}_{ij}\,.
\ee
That is, under the ambient Weyl diffeomorphism in $\tilde M$, we obtain two induced metrics which are related by a Weyl transformation in $M$. Hence, the ambient Weyl diffeomorphisms acting on the surface $\rho=0$, namely the null surface $\cal N$, gives rise to a conformal class of metrics on $M$.\footnote{If one only performs a local scaling in the coordinate $t$, i.e.\ $t'=B(x)t, x'^i=x^i, \rho'=\rho$, then one can also get a conformal class of metrics from other constant-$\rho$ surfaces. However, to obtain the induced Weyl connection and a Weyl class, one needs to perform the ambient Weyl diffeomorphism, and thus needs the restriction of $\rho=0$.}

\par
Having a conformal class of induced metrics on $M$, now let us look at how a connection is induced from $\tilde M$ onto $M$. Suppose $\tilde\nabla$ is the Levi-Civita connection of the ambient space $(\tilde M,\tilde g)$, i.e.\ it is torsion-free and has zero metricity $\tilde\nabla_{\un D_P}\tilde g_{MN}=0$. The ambient connection coefficients $\tilde\Gamma^P{}_{MN}$ of $\tilde\nabla$ are defined with respect to the frame $\un D_M$ of $T\tilde M$ as:
\begin{align}
\tilde\nabla_{\un D_M}\un D_N=\tilde\Gamma^i{}_{MN}\un D_i+\tilde\Gamma^+{}_{MN}\un D_++\tilde\Gamma^-{}_{MN}\un D_-\,.
\end{align}
In the following discussion we will denote the covariant derivative $\tilde\nabla_{\un D_P}$ along $\un D_P$ as $\tilde\nabla_P$ for brevity ($P=+,i,-$); we emphasize that these are not however the coordinate frame components.
The ambient connection 1-form $\tilde{\bm\omega}^{M}{}_{N}=\tilde\Gamma^M{}_{PN}\bm e^P$ in this frame is then found to be (the matrix elements are arranged in the order of $+,i,-$)
\begin{align}
\tilde{\bm\omega}^{M}{}_{N}=&\left(\begin{array}{ccc}a_k & -t\psi_{kj}& 0  \\\frac{1}{t}(\delta_k{}^i-\rho\psi_k{}^i) & \tilde\Gamma^i{}_{kj}& \frac{1}{t}\psi_k{}^i  \\0 & -t(\gamma_{kj}-\rho\psi_{kj})& -a_k \end{array}\right)\bm e^k\nn\\
\label{eq:conn1form}
&+\left(\begin{array}{ccc}0 & \rho\varphi_j& 0  \\\frac{\rho^2}{t^2}\varphi^i &\frac{1}{t}(\delta_j{}^i-\rho\psi_{j}{}^i)& -\frac{\rho}{t^2}\varphi^i  \\0& -\rho^2\varphi_j & 0 \end{array}\right)\bm e^++\left(\begin{array}{ccc}0 & -\varphi_j & 0 \\-\frac{\rho}{t^2}\varphi^i &\frac{1}{t}\psi_j{}^i & \frac{1}{t^2}\varphi^i  \\0  & \rho \varphi_j & 0 \end{array}\right)\bm e^- \,,
\end{align}
where the upper $i,j$ indices are raised by $\gamma^{ij}\equiv(\gamma_{ij})^{-1}$, and
\begin{align}
\label{eq:psiphi}
\psi_{ij}&=\frac{1}{2}(\p_\rho \gamma_{ij}+f_{ij})\,,\qquad\varphi_i=\p_\rho a_i \,,\qquad f_{ij}=D_ia_j-D_ja_i\,,\\
\label{eq:Gammaijk}
\tilde\Gamma^{i}{}_{jk}&=\frac{1}{2}\gamma^{im}(D_j\gamma_{mk}+D_k\gamma_{jm}-D_m\gamma_{jk})-(a_j\delta^i{}_k+a_k\delta^i{}_j-a^i\gamma_{jk})\,.
\end{align}

We note that the Levi-Civita condition $\tilde\nabla_{i}\tilde g_{jk}=0$ evaluates to $\nabla_{i}\gamma_{jk}=2a_i\gamma_{jk}$, where $\nabla$ is the connection on the distribution $C_d$ induced by $\tilde\nabla$, with $\nabla_{i}\gamma_{jk}:=D_i\gamma_{jk}-\tilde\Gamma^m{}_{ij} \gamma_{m k}-\tilde\Gamma^m{}_{ik} \gamma_{jm}$. Hence, if we interpret $\gamma_{ij}$, i.e.\ $\tilde g_{IJ}$ restricted to the $i,j$ indices, as giving rise to a metric on the distribution $C_d$ spanned by $\{\un D_i\}$ in $\tilde M$, then the connection $\nabla$ on $C_d$ has a nonvanishing metricity $2a_i\gamma_{jk}$. Equivalently, this connection has vanishing Weyl metricity, and it is therefore convenient and natural to introduce a connection $\hat\nabla$ on $C_d$, such that \[\hat\nabla_i\gamma_{jk}:=\nabla_i\gamma_{jk}-2a_i\gamma_{jk}=0.\] The vanishing of the Weyl metricity is a Weyl-covariant condition, whereas the vanishing of the usual metricity $\nabla_i\gamma_{jk}$ is not. More generally, for any tensor $T$ defined on $C_d$ (i.e.,\ $ T$ has no $+,-$ components) that transforms covariantly under an ambient Weyl diffeomorphism as $T(t,x^i,\rho)\to {\cal B}(x)^{w_T} T({\cal B}(x)^{-1}t,x^i,{\cal B}^{2}(x)\rho)$, the derivative
\be
\label{eq:hatnabla}
\hat\nabla_i T:=\nabla_iT+w_Ta_iT
\ee
will also transform covariantly with the same weight. For example, it follows from the definitions in \eqref{eq:psiphi} that $\varphi_i(x,\rho)\to{\cal B}(x)^2\varphi_i(x,{\cal B}(x)^2\rho)$ and $\psi_{ij}(x,\rho)\to\psi_{ij}(x,{\cal B}(x)^2\rho)$, and thus we can write their Weyl-covariant derivatives as
\be
\hat\nabla_i\varphi_j=\nabla_i\varphi_j+2a_i\varphi_j\,,\qquad 
\hat\nabla_i\psi_{jk}=\nabla_i\psi_{jk}\,.
\ee
From the above behavior of the induced connection on $C_d$, we can naturally expect that the induced connection on $M$ will give us a codimension-2 Weyl geometry. However, since $\{\un D_i\}$ is not an integrable distribution when $a_i$ is turned on, the connection coefficients \eqref{eq:Gammaijk} cannot be pulled back directly to $M$. As we will see below, this problem does not exist if we focus on the surface at $\rho=0$.

Notice that $\tilde\Gamma^{i}{}_{jk}$ does not depend on $t$, and thus at any value of $t$ at $\rho=0$, the induced connection coefficients can be expressed as
\begin{align}
\label{eq:Weylconn}
\Gamma_{(0)}^{i}{}_{jk}\equiv\tilde\Gamma^{i}{}_{jk}|_{\rho=0}=\frac{1}{2}\gamma_{(0)}^{im}(\p_j\gamma^{(0)}_{mk}+\p_k\gamma^{(0)}_{jm}-\p_m\gamma^{(0)}_{jk})-(a^{(0)}_j\delta^i{}_k+a^{(0)}_k\delta^i{}_j-a_{(0)}^i\gamma^{(0)}_{jk})\,.
\end{align}
To define an induced connection on $M$, let us take $t=1$ as a representative, i.e.\ take $\phi(M)$ to be a $d$-dimensional surface in $\tilde M$ at $\rho=0$ and $t=1$. At  first sight, the connection defined by \eqref{eq:Weylconn} is still an induced connection on the distribution spanned by $\{\un D_i\}$, which does not lie on the codimension-2 surface $\phi[M]$ when $a_i$ is turned on. However, when the dual frame $\{\bm e^P\}$ gets pulled back on $M$, we get $\{\bm e^i=\td x^i\}$, and the corresponding vector basis on $TM$ is $\{\un \p_i\}$. Hence, the ambient LC connection $\tilde\nabla$ defined on $T^*\tilde M$ induces a connection $\nabla^{(0)}$ on $T^*M$ in the following natural manner
\begin{align}
\nabla^{(0)}_{\un \p_j}\bm e^i\equiv\nabla_{\un D_j}\bm e^i|_{\rho=0,t=1}=-\Gamma_{(0)}^i{}_{jk}\bm e^k\,.
\end{align}
Then, $\nabla^{(0)}$ can also be defined on $TM$, which defines the parallel transport of a vector along a curve on $M$:
\begin{align}
\nabla^{(0)}_{\un\p_i}\un\p_j=\Gamma_{(0)}^k{}_{ij}\un\p_k\,.
\end{align}
In this way we get a connection $\nabla^{(0)}$ on $M$ whose connection coefficients are given by \eqref{eq:Weylconn}. This is a connection that satisfies $\nabla^{(0)}_i\gamma^{(0)}_{jk}=2a^{(0)}_i\gamma^{(0)}_{jk}$, i.e.\ it has vanishing Weyl metricity, and $a^{(0)}_i$ plays the role of a Weyl connection on $M$. One can also define a metricity-free connection $\hat\nabla^{(0)}$ on $M$ satisfying $\hat\nabla^{(0)}_{i}\gamma^{(0)}_{jk}=\nabla^{(0)}_i\gamma^{(0)}_{jk}-2a^{(0)}_i\gamma^{(0)}_{jk}=0$, which can be referred to as a Weyl-LC connection. 
\par
An ambient Weyl diffeomorphism in $\tilde M$ induces on $M$ a Weyl transformation $\gamma^{(0)}_{ij}\to{\cal B}^{-2}\gamma^{(0)}_{ij}$, $a^{(0)}_i\to a^{(0)}_i-\p_i\ln{\cal B}$.\footnote{If one considers a more general version of the diffeomorphism \eqref{eq:Weyldiff} where $x'=x'(x)$, then 
\begin{equation*}
\frac{\p x'^j}{\p x^i}a'^{(0)}_j(x')=a_i^{(0)}(x)-\p_i\ln{\cal B}(x)\,,\qquad\frac{\p x'^i}{\p x^k}\frac{\p x'^j}{\p x^l}\gamma_{kl}'^{(0)}(x')={\cal B}(x)^{-2}\gamma_{ij}^{(0)}(x)\,.
\end{equation*}
The transformation $(t,x^i,\rho)\to(t,x'^i(x),\rho)$ realizes the Diff$(M)$ part of the $\text{Diff}(M)\ltimes \text{Weyl}$ symmetry on $M$.} This means that we get a \emph{Weyl class} $[\gamma^{(0)},a^{(0)}]$, which is the equivalence class formed by all the pairs of $\gamma^{(0)}$ and $a^{(0)}$ that are connected by Weyl transformations, i.e.,\
\be
(\gamma^{(0)}_{ij},a^{(0)}_i)\sim({\cal B}(x)^{-2}\gamma^{(0)}_{ij},a^{(0)}_i-\p_i\ln{\cal B}(x))\,.
\ee
With the Weyl class defined on $M$, we obtain a $d$-dimensional Weyl manifold $(M,[\gamma^{(0)},a^{0}])$ induced by the Weyl-ambient space $(M,\tilde g)$, where the geometric quantities defined in terms of the Weyl connection are Weyl covariant. For example, one can define on $M$ the Weyl-Riemann tensor $\hat R_{(0)}^i{}_{jkl}$, Weyl-Ricci tensor $\hat R^{(0)}_{ij}$, Weyl-Ricci scalar $\hat R^{(0)}$, etc.

\subsection{Bottom-Up Perspective}
\label{sec:bottomup}
In this subsection we will present a geometric interpretation of the Weyl-ambient metric \eqref{Weyl_ambient} as well as the Weyl connection therein in terms of a bottom-up construction. By ``bottom-up'' we mean to construct a $(d+2)$-dimensional Weyl-ambient space from a $d$-dimensional manifold $M$. The majority of this subsection will follow a similar construction in Section 2 and Section 3 of \cite{Fefferman:2007rka} where a more detailed exposition of the ambient construction can be found. We will generalize the main definitions and theorems there with the inclusion of a Weyl connection on the principal $\RR_+$-bundle. (See Section \ref{sec:primer} for the basics of principal bundles.) The resulting Weyl structure together with the metric bundle, viewed as an associated bundle, will be then used to define the Weyl-ambient metric. For this subsection to be self-contained we repeat some of the definitions and proofs of \cite{Fefferman:2007rka} when necessary while generalizing them appropriately. 
\par
We start with a $d$-dimensional manifold $M$ and introduce a principal $\bb{R}_{+}$-bundle $\mathcal{P}_W$ over $M$ that we call a {\it Weyl structure}.\footnote{We use this name since ${\cal P}_W$ can be regarded as a $G$-structure of the frame bundle, in which the structure group is reduced from $GL(d,\RR)$ to $\RR_+$.} 

\begin{defn}\label{Weylstructure}
Given a $d$-dimensional manifold $M$, a \emph{Weyl structure} is a $(d+1)$-dimensional manifold ${\cal P}_W$ together with the structure group $\bb R_+$, which is equipped with
\par
\ding{172} a free right action $\delta:\mathcal{P}_W\times\bb R_+\to\mathcal{P}_W$, such that $\delta_s(p)= p\cdot s$, $\forall p\in\mathcal{P}_W$, $s\in\bb R_+$;
\par
\ding{173} a projection map $\pi: \mathcal{P}_W\rightarrow M$, such that $\pi(p)=\pi(p\cdot s)$, $\forall p\in\mathcal{P}_W$, $s\in\bb R_+$;
\par
\ding{174} a local trivialization $T_i:\pi^{-1}(U_i)\to U_i\times \bb{R}_{+}$ for each open set $U_i\subset M$ with $T_i(p)=(\pi(p),t_i(p))$, where $t_i:\pi^{-1}(U_i)\to \bb{R}_{+}$ satisfies $t_i(p\cdot s)=t_i(p)\cdot s$ for all $s\in\bb{R}_{+}$.
\par
For brevity, suppose $U_i\subset M$ has local coordinates $\{x^i\}$, we can express a point $p\in{\cal P}_W$ as $(x,t)$ with $t\in \bb{R}_{+}$. 
\end{defn}
A connection on the Weyl structure can be described as follows. First we note that the push forward $\pi_*:T{\cal P}_W\to TM$ defines the vertical sub-bundle $V\subset T{\cal P}_W$ given at any point $p\in{\cal P}_W$ by
\beq
V_p=\ker(\pi_*)\equiv\{\un v\in T_p{\cal P}_W\,|\,\pi_*(\un v)=\un 0\}.
\eeq
In the present case $V_p$ is a one-dimensional vector space spanned by the fundamental vector field which generates the  group action along the fibers; in the local trivialization, it is expressed as $\un T=t\un\p_t$. From the perspective of ${\cal P}_W$, we can then think of the action of $\bb{R}_+$ as corresponding to a dilatation of the fibers. To assign a connection on ${\cal P}_W$ is to specify a horizontal sub-space $H_p\subset T_p{\cal P}_W$ such that $T_p{\cal P}_W=H_p\oplus V_p$ at any $p$. In the local trivialization given above, the horizontal bundle can be described as the span of vectors of the form $\un D_i=\un\p_i-a_i(x)t\un\p_t$.\footnote{Here we have required that $a(x)$ be independent of $t$ in order to make the Weyl-ambient metric homogeneous of degree 2 with respect to $t$.} Equivalently, it can be described as the kernel of a form $n:=t^{-1}\td t+a_i(x)\td x^i\in T^*{\cal P}_W$, i.e.
\beq
\label{eq:HV}
H_p:= \{\un u\in T_p{\cal P}_W\,|\,i_{\un u}n=0\}\qquad \forall p\in{\cal P}_W.
\eeq
We note that under the Abelian group action $(x,t(x))\mapsto (x,t'(x))=(x,t(x)s(x))$, we have
\beq
\label{eq:horizontaln}
n'= n+\big(a_i'(x)-a_i(x)+\p_i\ln s(x)\big)\td x^i\,,
\eeq
and so we see that the coefficients $a_i(x)$ transform as connection coefficients.
Note also that it is natural to introduce the projector ${\bf a}:T{\cal P}_W\to V$ as
\beq\label{projector_triv_1}
{\bf a}=t\un\p_t\otimes \big(t^{-1}\td t+a_i(x)\td x^i\big)\,,
\eeq
which is an alternative way to express the connection on ${\cal P}_W$. We will refer to both $\mathbf a$ and $a_i(x)$ as the \emph{Weyl connection}.

This line bundle has an important representation given by a conformal class of metrics. Indeed, all the non-trivial representations are one-dimensional, and thus a representation of $\bb{R}_+$ is given by specifying a Weyl weight $w$. We call the corresponding associated bundle ${\cal E}_{w}$ and its sections respond to the group action as
\beq
T_x\mapsto s(x)^w T_x\,.
\eeq
Equivalently, this determines the transition functions on the associated bundle.

Suppose a conformal class $[g]$ of smooth metrics of signature $(p,q)$ is given on $M$, in which any two representatives $g$ and $g'$ are related by a smooth function ${\cal B}(x)$ as $g'_x={\cal B}(x)^{-2}g_x$, where $g_x$ is the value of $g$ at a point $x\in M$. Then, $(M,[g])$ is a conformal manifold. One can define a \emph{metric bundle} ${\cal G}$ as follows \cite{Fefferman:2007rka}:
\begin{defn}\label{Metric Bundle}
A \emph{metric bundle} $\mathcal{G}$ is the collection of pairs $(x,h)$ where $h= s^2 g_{x}$, $\forall s\in \mathbb{R}_{+}$ and $\forall x\in M$, which is equipped with 
\par
\ding{172} a dilatation map $\tilde\delta_{s}: \mathcal{G}\to \mathcal{G}$ such that $\tilde\delta_{s}(x,h)= (x,s^2 h)$, $\forall s\in \mathbb R_{+}$. 
\par
\ding{173} a projection map $\tilde\pi: \mathcal{G}\rightarrow M$ such that $(x,h)\mapsto x$;
\end{defn}
This definition simply identifies a conformal class of metrics with a bundle associated to the Weyl structure given by the weight $w=-2$ representation of $\bb{R}_+$. We note that it is isomorphic to the Weyl structure ${\cal P}_W$, as is any non-trivial associated bundle of ${\cal P}_W$.\footnote{Note that in \cite{Fefferman:2007rka}, the metric bundle $\cal G$ itself is treated as the principal $\bb R_+$-bundle through an isomorphism. Here we introduced the Weyl structure ${\cal P}_W$ and distinguish it from $\cal G$ in order to emphasize that a conformal class of metrics furnishes a representation of the group $\bb{R}_+$ with $w=-2$.} Under a trivialization, assigning an isomorphism between ${\cal P}_W$ and the metric bundle $\cal G$ can be thought of as a choice of representative $g$ of the conformal class $[g]$ if we identify
\begin{equation}
\label{eq:trivial}
(x,t)\in U_i\times\mathbb{R}_{+} \qquad \text{with }\qquad (x,t^2 g_{x})\in {\cal G}\,.
\end{equation}
Given $g\in[g]$, for any $p\in{\cal P}_W$, by means of the corresponding $(x,h)\in{\cal G}$ one can define a symmetric tensor $\mathbf{g}_{0}$ of type $(0,2)$ called the \emph{tautological tensor} that acts on vector fields $\un w_1,\un w_2 \in T_{p}{\cal P}_W$ as follows:
\begin{equation}\label{tautological_tensor}
\mathbf{g}_{0} (\un w_1,\un w_2)\equiv h(\pi_{*}\un w_1,\pi_{*}\un w_2)\,,
\end{equation}
which can be expressed as $\mathbf{g}_0=t^2\pi^*g$ under the identification in \eqref{eq:trivial}. 

If we pick another representative $g'_{x}={\cal B}(x)^{-2}g_{x}$ of the conformal class $[g]$, following the identification in \eqref{eq:trivial}, we obtain another isomorphism between ${\cal P}_W$ and ${\cal G}$ by identifying
\begin{equation}
(x,t')\in U_i\times\mathbb{R}_{+} \qquad \text{with }\qquad (x,t'^2 g'_{x})\in {\cal G}.
\end{equation}
It is easy to see that the two isomorphisms are related by setting $t'= {\cal B}(x)t$. To preserve the horizontal subspace on ${\cal P}_W$, from \eqref{eq:horizontaln} we can see that $a'_{i}(x)$ satisfies
\begin{equation}
a_{i}'(x)= a_{i}(x)- \p_{i}\text{ln}{\cal B}(x)\,.
\end{equation}
In the present circumstances, it is natural to replace the notion of conformal class $[g]$ by the {\it Weyl class} $[g,a]$, with the property
\beq
\forall (g,a), (g',a')\in [g,a],\quad \exists\; {\cal B}(x)\; {\rm such \; that} \; (g'_x,a'_x)=({\cal B}(x)^{-2}g_x,a_x-\td\ln{\cal B}(x))\,,
\eeq
where $\td$ is the exterior derivative on $M$.
\par
Before we proceed to define the Weyl-ambient space based on the Weyl structure ${\cal P}_W$, we would like to make a few remarks. Recall that for the Weyl-ambient metric \eqref{Weyl_ambient}, the coordinates $t$ and $x^{i}$ parametrize a codimension-1 null hypersurface ${\cal N}$ located at $\rho=0$. One can see that this surface is exactly a Weyl structure. In Section \ref{sec:topdown}, the degenerate ``induced metric" of $\tilde g$ on ${\cal N}$ is the tautological tensor, the induced metric $\gamma^{(0)}$ on $M$ is a representative $g$ in the conformal class, and the Weyl connection $a^{(0)}_i(x)$ on $M$ is the $a_i(x)$ in \eqref{projector_triv_1}. Thus, the Weyl class $[\gamma^{(0)},a^{(0)}]$ corresponds to $[g,a]$ in this section, and $(M,[g,a])$ defines a Weyl manifold. We will discuss more details of the role of the Weyl connection and the horizontal subspace it defines in Theorem \ref{diff_Weyl_normal_form} below. It is noteworthy that the projector in \eqref{projector_triv_1}, which defines the Weyl connection on ${\cal P}_W$, is a special case of the construction presented in \cite{Ciambelli:2019lap} with restricted diffeomorphisms.

\par
Now we will define a Weyl-ambient space for a Weyl manifold generalizing the definition of a Fefferman-Graham ambient space for a conformal manifold introduced in \cite{Fefferman:2007rka}. Consider a $(d+2)$-dimensional space $\tilde M$ which looks at least locally like ${\cal P}_W \times\bb{R}$ where each point can be labeled by $(p,\rho)$ with $\rho\in\bb R$. The inclusion map $\iota:{\cal P}_W\to\tilde M$ is defined such that $p\mapsto (p,0)$. By letting the  map $\delta_s$  act only on $p\in{\cal P}_W$,  we can extend $\delta_s$ to a map on $\tilde M$, which commutes with $\iota$. The vector field $\un T$ which generates the Weyl group action is extended to a vector field $\un{\cal T}=\iota_*\un T=t\un\p_t$ on $\tilde M$.
\begin{defn}\label{Ambient}
Suppose $M$ is a $d$-dimensional manifold equipped with a Weyl class $[g,a]$, and ${\cal P}_W$ is a Weyl structure over $M$. A pseudo-Riemannian space $(\tilde M,\tilde g)$ is called the \emph{Weyl-ambient space} for $(M,[g,a])$ if
\par
\ding{172} $\tilde M$ is a dilatation-invariant open neighborhood of ${\cal P}_W\times\{0\}$ in ${\cal P}_W\times\bb R$, and the pullback $\iota^*\tilde g$ is the tautological tensor $\mathbf{g}_0$ defined above;
\par
\ding{173} $\tilde g$ is a smooth metric on $\tilde M$ of signature $(p+1,q+1)$, which is homogeneous of degree 2 on $\tilde M$, i.e.,\ $\delta^*_s\tilde g=s^2\tilde g$, $\forall s\in\bb R_+$;
\par
\ding{174} $Ric(\tilde g)$ vanishes to infinite order at every point of ${\cal P}_W\times\{0\}$.
\end{defn}
\noindent Without condition \ding{174}, $(\tilde M,\tilde g)$ is called a \emph{Weyl pre-ambient space} for $(M,[g,a])$. Note that the condition (3) in \cite{Fefferman:2007rka} is presented differently when $d$ is even and odd, and $Ric(\tilde g)$ has an obstruction in the order $O(\rho^{d/2-1})$ for even $d$. 
Here we take the dimension to be a continuous complex variable, and so the Ricci-flatness condition always holds to infinite order. As explained in Section~\ref{sec:FG}, the obstruction at even dimension will be manifested by the pole of the expansion of $\tilde g$ at even $d$, which is identified as the extended Weyl-obstruction tensor. 
\par
Now we introduce the final ingredient in our Weyl-ambient construction---the \emph{Weyl-normal form}, which is a generalization of the \emph{normal form} defined in \cite{Fefferman:2007rka}.
 \begin{defn}\label{Weyl_normal_form}
 A Weyl pre-ambient space $(\ti M,\ti g)$ for $(M,[g,a])$ is said to be in \emph{Weyl-normal form} with acceleration $\un{\cal A}$ if
 \par
\ding{172} For each fixed $p\in \mathcal{P}_W$, the set of $\rho \in \mathbb{R}$ such that $(p,\rho)\in \ti M$ is an open interval $I_{p}\in\bb R$ containing $0$.
\par
\ding{173} For each $p\in {\cal P}_W$, the parametrized curve $C_p:I_{p} \to \tilde M$, $\rho\mapsto(p,\rho)$ has a tangent vector $\un{\cal U}$, whose acceleration $\un {\cal A}\equiv \tilde\nabla_{\un {\cal U}}\un {\cal U}$ satisfies $\tilde g(\un {\cal T}, \un {\cal A})=0$, where $\tilde\nabla$ is the Levi-Civita connection of $(\tilde M,\tilde g)$.
\par
\ding{174} Let $(t,x,\rho)$ represent a point in $\mathbb{R}_{+}\times M\times \mathbb{R}\simeq {\cal P}_W\times \mathbb{R}$ under the local trivialization induced by $g$. Then, at each point $(t,x,0)\in {\cal P}_W\times \{0\}$, the metric $\ti g$ takes the form
 \begin{equation}\label{form_1}
 \ti g|_{\rho=0}= \mathbf{g}_{0}  + 2t^2(t^{-1}\td t + a_{i}(x)\td x^{i})\td \rho\,,
 \end{equation}
 where $\mathbf{g}_{0}$ is the  tautological symmetric tensor defined in \eqref{tautological_tensor}.
 \end{defn} 
Definition \ref{Weyl_normal_form} is engineered for the purpose of generating the Weyl-ambient metric from the ``initial surface'' at $\rho=0$. At $\rho=0$, the Weyl-ambient metric we have seen in \eqref{Weyl_ambient} has the form \eqref{form_1}, which motivates condition \ding{174}. Since $\un{\cal T}=t\un\p_t$ everywhere in $\tilde M$, condition \ding{173} implies that the covector ${\cal A}$ of the acceleration does not have a $t$-component. Furthermore, one can also parametrize the accelerated curve $C_p$ such that $\tilde g(\un{\cal A},\un{\cal U})=0$, and let $\cal A$ have no $\rho$-component either.\footnote{Suppose $C_p$ has a parameter $\lambda$, then under a reparametrization $\lambda\to f(\lambda)$ we have $\un{\cal U}\to f'\un{\cal U}$, and the acceleration vector transforms  $\un{\cal A}\to f'^2 \un{\cal A} + f'\un{\cal U}(f)\un{\cal U}$, and thus $\tilde g(\un{\cal A}, \un{\cal U})$ can always be set to zero for non-null $\un{\cal U}$ by choosing an appropriate function $f$. For null $\un{\cal U}$ the condition holds automatically.} We will assume that $\rho$ is such a parametrization. Note that in the special case where $\un{\cal A}=0$, the $\rho$-coordinate lines are geodesics, and condition \ding{173} goes back to that of normal form in \cite{Fefferman:2007rka}, while condition \ding{174} will still be different as long as $a_{i}(x)$ are nonvanishing. The acceleration $\un{\cal A}$ encodes all the higher modes $a^{(k\geqslant 1)}_{i}(x)$ in the expansion \eqref{eq:aexpan} of $a_i(x,\rho)$, as we will see in Lemma \ref{Lemma}. In fact, if both $a_{i}(x)$ and $\un{\cal A}$ are zero, the mode $a_{i}(x,\rho)$ in \eqref{Weyl_ambient} vanishes. 
\par
The following Theorem is a generalization of Proposition 2.8 in \cite{Fefferman:2007rka}.
 \begin{thm}\label{diff_Weyl_normal_form}
 Let $(M,[g,a])$ be a Weyl manifold, with $(g,a)$ a representative of the Weyl class. Let ${\cal P}_W$ be the Weyl structure over $M$, and $(\ti M,\ti g)$ be a Weyl pre-ambient space for $(M,[g,a])$. Then, there exists a dilatation-invariant open set $\tilde M'\subset {\cal P}_W\times \mathbb{R}$ containing ${\cal P}_W\times \{0\}$ on which there is a unique diffeomorphism $\phi:\tilde M'\to\ti M$ commuting with dilatations with $\phi |_{{\cal P}_W\times \{0\}}$ being the identity map, such that the Weyl pre-ambient space $(\tilde M', \phi^{*}\ti g)$ is in Weyl-normal form with acceleration $\un{\cal A}'$.
\end{thm}
This theorem indicates that given a representative pair $(g,a)$, any Weyl pre-ambient space can be put into Weyl-normal form by a diffeomorphism $\phi$. $(\tilde M, \ti g)$ and $(\tilde M', \phi^{*}\ti g)$ are also said to be ambient-equivalent (see Definition 2.2 in \cite{Fefferman:2007rka} for the precise definition of ambient equivalence).
The proof of this theorem will be presented in Subsection~\ref{sec:proofs}.

\par
Before we move on to the main result of this section, namely Theorem \ref{Weyl_ambient_existence}, let us introduce some useful notation. Given a local coordinate system $\{x^i\}$ ($i=1,\cdots,d$) on $M$, the fiber coordinate $t$ of ${\cal P}_W$ and the parameter $\rho$ naturally defines an \emph{ambient coordinate system} $\{t,x^i,\rho\}$ on $\tilde M$. Later on, we will follow \cite{Fefferman:2007rka} and use $I,J,\cdots=(0,i,\infty)$ to label the ambient coordinate indices, where $0$ labels the $t$-component and $\infty$ labels the $\rho$-component. It is also convenient to interpret the notations $(0,i,\infty)$ as representing the components in a trivialization ${\cal P}_W\times \bb R\simeq \bb R_+\times M\times\bb R$, even without specifying a choice of coordinates on $M$.

We will now present Theorem \ref{Weyl_ambient_existence}, which is a natural generalization of Theorem 2.9 of \cite{Fefferman:2007rka}, based on our definition of Weyl-normal form. As a corollary of this theorem, we will show that for a Weyl-ambient space in Weyl-normal form, the Weyl-ambient metric \eqref{Weyl_ambient} emerges from the initial surface uniquely under the Ricci-flatness condition. We emphasize again that we consider the dimension $d$ of the manifold $M$ formally as a complex parameter, and do not need to distinguish between even and odd dimensions. 
\begin{thm}\label{Weyl_ambient_existence}
Let $(M,[g,a])$ be a Weyl manifold, and let $(g,a)$ be a representative in the Weyl class.
 \begin{enumerate}[(A)]
 \item There exists a Weyl-ambient space $(\ti M,\ti g)$ for $(M,[g,a])$ which is in Weyl-normal form with acceleration $\un{\cal A}$.
 \item Suppose that $(\ti M_{1},\ti g_{1})$ and $(\ti M_{2},\ti g_{2})$ are two Weyl-ambient spaces for $(M,[g,a])$, both of which are in Weyl-normal form with acceleration $\un{\cal A}$. Then $\ti g_{1}- \ti g_{2}$ vanishes to infinite order at every point of ${\cal P}_W\times \{0\}$.
 \end{enumerate}
 \end{thm}
\par
The proof of Theorem \ref{Weyl_ambient_existence} employs the following lemma. 
\begin{lemma}\label{Lemma}
Let $({\ti M},\ti g)$ be a Weyl pre-ambient space for $(M,[g,a])$. Suppose for each $p\in {\cal P}_W$, the set of all $\rho\in \mathbb{R}$ such that $(p,\rho)\in \ti  M$ is an open interval $I_{p}$ containing $0$. Let $g$ be a metric in the representative $(g,a)$ of the Weyl class, which provides a local trivialization ${\cal P}_W\times \mathbb{R}\simeq \mathbb{R}_{+}\times M\times \mathbb{R}$. Then $(\ti M,\ti g)$ is in Weyl-normal form with acceleration $\un{\cal A}$ if and only if one has on $\ti M$:
\begin{equation}\label{lemma_conditions}
\ti g_{0\infty}=t\,,\qquad \ti g_{i\infty}=t^{2} a_{i}(x,\rho)\,,\qquad \ti g_{\infty\infty}=0\,,
\end{equation}
where $a_{i}(x,\rho)\equiv a_{i}(x)+t^{-2} \int_{0}^{\rho}{\cal A}_{i}(t,x,\rho')\td\rho' $.
\end{lemma}

\begin{proof}
Suppose $\tilde g$ satisfies \eqref{lemma_conditions}, then it follows from the condition $\iota^{*}\ti g= \mathbf{g}_{0}$ for the pre-ambient space that $\ti g|_{\rho=0}$ must have the form \eqref{form_1}. Thus, all we have to prove is that for $\ti g$ satisfying \eqref{form_1} at $\rho=0$, the condition that the $\rho$-coordinate lines have acceleration $\un{\cal  A}$ with $\tilde g(\un{\cal T}, \un{\cal A})=0$ is equivalent to \eqref{lemma_conditions}. The fact that the $\rho $-coordinate lines have an acceleration $\un{\cal A}$ implies
\begin{equation}\label{Gamma_condition_rho_lines}
\ti\Gamma_{\infty \infty I}= \mathcal{A}_{I}\,,
\end{equation}
where $\ti\Gamma_{IJK}\equiv\ti g_{KL}\ti\Gamma^{L}{}_{IJ}$.
The condition $\tilde g(\un{\cal T},\un{\cal A})=0$ leads to ${\cal A}_0=0$. As we have mentioned, one can also parametrize the curve $C_p:I_p\to\tilde M$ such that $\tilde g(\un{\cal U}, \un{\cal A})=0$, then we also have ${\cal A}_{\infty}=0$, and thus $\un{\cal A}_{I}=\left({\cal A}_{0}, {\cal A}_{i},{\cal A}_{\infty}\right)= \left(0,t^{2}\varphi_{i}(x,\rho),0\right)$. The functions $\varphi_{i}(x,\rho)$ are considered as external input and cannot be determined from the initial conditions. The factor $t^2$ is derived from the homogeneity property of $\ti g$ and \eqref{Gamma_condition_rho_lines}. If we set $I=\infty$ in \eqref{Gamma_condition_rho_lines} we get 
\begin{equation}
\ti\Gamma_{\infty\infty \infty}= {\cal A}_{\infty}= 0\,\implies\,\p_{\rho}g_{\infty\infty}=0\,\implies\, g_{\infty\infty}=0\,,
\end{equation}
where in the last step we used the initial condition $g_{\infty\infty}|_{\rho=0}=0$. Similarly, setting $I=0$ in \eqref{Gamma_condition_rho_lines} we find
\begin{equation}
\p_{\infty }g_{\infty 0}=0\, \implies\, g_{\infty 0}= t\,,
\end{equation}
where we used the initial condition $g_{0\infty}|_{\rho=0}=t$. Finally, setting $I=i$ yields
\begin{equation}
\p_{\rho}g_{\infty i}={\cal A}_{i}(t,\rho;x)\implies g_{\infty i}= t^{2}a_{i}(x)+ t^{2}\int_{0}^{\rho}\varphi_{i}(\rho;x)\td\rho \equiv  t^2 a_{i}(\rho;x)\,,
\end{equation}
where we used the initial condition $\ti g_{\infty i}|_{\rho=0}= t^2 a_{i}(x)$.
 \end{proof}
The main logic of the proof of Theorem \ref{Weyl_ambient_existence} will follow part of Section 3 in \cite{Fefferman:2007rka}. To show part (A) of Theorem \ref{Weyl_ambient_existence}, namely the existence of the Weyl-ambient space ${\tilde M}$ in Weyl-normal form, we need to show the following: for a Weyl manifold $(M,[g,a])$, given a representative $(g,a)$ of the Weyl class and $a_i(x,\rho)$ determined by $\un{\cal A}$, there exists a metric $\ti g$ on an open neighborhood $\ti M$ of ${\cal P}_W\times\{0\}$ with the following properties:
 \begin{enumerate}[(1)]
 \item $\delta^{*}_{s}\ti g= s^2 \ti g$, $\forall s>0$ (homogeneity property);
 \item $\ti g = t^2 g(x)+ 2t^2(t^{-1}\td t+ a_{i}(x)\td x^{i})\td\rho$ when $\rho=0$;
 \item $\ti g_{0\infty}=t,\quad \ti g_{i\infty}=t^{2} a_{i}(x,\rho),\quad \ti g_{\infty\infty}=0$;
 \item $\ti Ric(\ti g)=0$ to infinite order at $\rho=0$.
 \end{enumerate}
 The first property above is the homogeneity property which is still taken to be true for the Weyl-ambient metric. Property (3) is equivalent to condition \ding{173} of  Definition \ref{Weyl_normal_form} due to Lemma \ref{Lemma}, which indicates that $\ti g_{I\infty}$ components are known, while the rest are now regarded as unknown functions. Property (2) can be considered as the initial data of these components at the initial surface at $\rho=0$, while the Ricci-flatness property (4) is a system of partial differential equations that one can solve to find the metric components beyond the initial surface. We will show that this is a well defined initial value problem so that the unknown components of the Weyl-ambient metric can be uniquely determined in a series expansion in $\rho$, which will prove part (B) of Theorem~\ref{Weyl_ambient_existence}. The complete proof will be presented in Subsection~\ref{sec:proofs}.

As an important corollary, we now show in Theorem \ref{thm:Weyl_ambient} that the metric $\tilde g$ determined from Theorem~\ref{Weyl_ambient_existence} has exactly the form of the Weyl-ambient metric \eqref{Weyl_ambient}. First we need the following lemma.
\begin{lemma}\label{Lemma_Ricci_Weyl_ambient}
Suppose a metric $\ti g$ has the following form:
\begin{equation}\label{WFG_ambient}
 \ti g_{IJ}=\left(\begin{array}{ccc}
2\rho & 0 &  t\\
0 & t^2 g_{ij}(x,\rho)& t^2 a_{j}(x,\rho)\\
t&t^2 a_{i}(x,\rho)&0
\end{array}\right)\,.
\end{equation}
Then the Ricci curvature of $\ti g$ satisfies $\ti R_{0I}=0$.
\end{lemma}
\begin{proof}
For $\ti g$ of the form \eqref{WFG_ambient}, we can write the inverse metric as
\begin{equation}\label{inverse_Weyl_ambient}
\begin{split}
\ti g^{IJ}&= \frac{1}{1+ 2\rho a^2}
\left(\begin{array}{ccc}
a^2 & -t^{-1}a^{j} &t^{-1}\\
-t^{-1}a^{i} &t^{-2}(1+ 2\rho a^2 )g^{ij}- 2t^{-2}\rho a^{i}a^{j}&2t^{-2}\rho a^{i}\\
t^{-1}&2t^{-2}\rho a^{j}& -t^{-2}2\rho
\end{array}\right)\,,
\end{split}
\end{equation}
and the Christoffel symbols $\ti\Gamma_{IJK}=\ti g_{KL}\ti\Gamma^{L}{}_{IJ}$ are given by
\begin{equation}\label{Christoffel_Weyl_ambient}
\begin{split}
\ti \Gamma_{IJ0}&=
\left(\begin{array}{ccc}
0& 0 &1\\
0 &-tg_{ij}&-ta_{i}\\
1&-ta_{j}& 0
\end{array}\right)\,,\quad
\ti \Gamma_{IJ\infty}=
\left(\begin{array}{ccc}
0& t a_{j} &0\\
 t a_{i} &-t^{2}\left(\frac{1}{2}\p_\rho g_{ij} - \p_{(i}a_{j)}\right)&0\\
0&0& 0
\end{array}\right)\,,\\
\ti \Gamma_{IJk}&=
\left(\begin{array}{ccc}
0& t g_{jk} &t a_{k}\\
t g_{ik} &t^{2}\Gamma_{ijk}&\frac{t^2}{2}\left(\p_\rho g_{ik} + F_{ik}\right)\\
 t a_{k}&\frac{t^2}{2}\left(\p_\rho g_{jk} + F_{jk}\right)& t^{2}\p_{\rho}a_{k}
\end{array}\right)\,,
\end{split}
\end{equation}
where $\Gamma_{ijk}=g_{kl}\Gamma^{l}{}_{ij}$ are the Christoffel symbols of $g_{ij}(x,\rho)$, and $F_{jk}= \p_{j}a_{k}- \p_{k}a_{j}$. Plugging \eqref{inverse_Weyl_ambient} and  \eqref{Christoffel_Weyl_ambient} into the Ricci curvature [see \eqref{Ricci}] we can compute $\ti R_{0I}$ explicitly and find that $\ti R_{0I}=0$.
\end{proof}

\begin{thm}
\label{thm:Weyl_ambient}
Suppose $(M,[g,a])$ is a Weyl manifold. Let $(\ti M,\ti g)$ be the unique ambient space for $(M,[g,a])$ which is in Weyl-normal form with acceleration $\un{\cal A}$. Then, for any representative $(g,a)$, the uniquely determined metric $\tilde g$ has the following form 
\begin{equation}
\ti g = 2\rho \td t^2 + 2td\rho\left(\frac{\td t}{t} +  a_{i}(x,\rho)\td x^{i}\right) + t^2 g_{ij}(x,\rho)\td x^{i}\td x^{j}\,,
\end{equation}
where $a_{i}(x,\rho)\equiv a_{i}(x)+t^{-2} \int_{0}^{\rho}{\cal A}_{i}(t,x,\rho')$. This metric is exactly the Weyl-ambient metric introduced in \eqref{Weyl_ambient}. 
\end{thm}
\begin{proof}
Based on Theorem \ref{Weyl_ambient_existence}, all we have to prove is that $\tilde g_{00}=2\rho$ and $\tilde g_{0i}=0$ to all orders. Let $\ti g^{(m)}$ be the $m^{th}$ order of $\ti g$, and let  $\ti g^{[k]}$ represent $\tilde g$ with all the orders higher than $O(\rho^{k})$ in the $\rho$-expansion excluded, i.e. $\tilde g=\tilde g^{[k]}+O(\rho^{k+1})$. From \eqref{metric_components_order_1} we find to the first order that $\tilde g^{[1]}_{00}=2\rho$ and $\tilde g^{[1]}_{0i}=0$. Assuming that $\tilde g^{[m-1]}_{00}=2\rho$ and $\tilde g^{[m-1]}_{0i}=0$, it follows from Lemma \ref{Lemma_Ricci_Weyl_ambient} that $\tilde R^{[m-1]}_{00}=\tilde R^{[m-1]}_{0i}=0$. Then, from \eqref{Ricci_inductive} we obtain that $\phi_{00}=\phi_{0i}=0$, and hence $\tilde g^{(m)}_{00}=\tilde g^{(m)}_{0i}=0$ [see \eqref{eq:g+phi}], $\forall m>1$. Therefore, by induction we can deduce to infinite order that $\tilde g_{00}=2\rho$ and $\tilde g_{0i}=0$, which completes the proof.
\end{proof}

\subsection{Proofs}
\label{sec:proofs}
\subsubsection{Proof of Theorem \ref{diff_Weyl_normal_form}}
To prove Theorem \ref{diff_Weyl_normal_form}, we first need to introduce a \emph{$(g,a)$-transversal vector} (generalized from the concept of a $g$-transversal vector in \cite{Fefferman:2007rka}), where the horizontal subspace $H_p$ defined by the Weyl connection plays an important role. Once we pick a representative $(g,a)$ in the Weyl class, $g$ induces an isomorphism between ${\cal P}_W$ and $\cal G$ through \eqref{eq:trivial}, which determines the fiber coordinate $t$ of ${\cal P}_W$; $a$ defines for any $p\in {\cal P}_W$ a horizontal subspace $H_p\subset T_p{\cal P}_W$ given in \eqref{eq:HV}, which can also be viewed as a subspace of $T_{(p,0)}({\cal P}_W\times \mathbb{R})$ via the inclusion map $\iota : {\cal P}_W\to {\cal P}_W\times \mathbb{R}$. We define a vector $\un {\cal V}\in T_{(p,0)}({\cal P}_W\times \mathbb{R})$ to be a \emph{$(g,a)$-transversal vector} for $\tilde g$ if it satisfies the following three conditions at $(p,0)$:
\begin{equation}\label{Weyl g-transversal vector}
\text{\ding{172} }\ti g(\un {\cal V}, \un{\cal T})= t^2\,,\qquad \text{\ding{173} }\ti g(\un{\cal V},\un{\cal H})=0\quad\forall\, \un{\cal H}\in H_{p}\,,\qquad\text{\ding{174} }\ti g(\un{\cal V},\un{\cal V})=0\,.
\end{equation}
When $a_i(x)=0$ in \eqref{projector_triv_1}, i.e.,\ $\mathbf{a}=\un\p_{t} \otimes \td t$, the $(g,a)$-transversal vector for $\tilde g$ goes back to the $g$-transversal vector for $\tilde g$ defined in \cite{Fefferman:2007rka}. From \eqref{form_1} one can see that for $(\tilde M,\tilde g)$ in Weyl-normal form, $\un \p_{\rho}$ is $(g,a)$-transversal for $\tilde g$ at $(p,0)$.
Following the proof of Lemma 2.10 in \cite{Fefferman:2007rka}, it is straightforward to show that the $(g,a)$-transversal vector is unique and dilatation-invariant (i.e.\ $\delta_{s*}V_p=V_{\delta_s(p)}$) for $\tilde g$ at $(p,0)$.  

The proof of Theorem \ref{diff_Weyl_normal_form} proceeds similar to the proof of Proposition 2.8 in \cite{Fefferman:2007rka}; one only has to let the $g$-transversal vector $\un{\cal V}$ to be a $(g,a)$-transversal vector. Here we will not repeat all the details but only outline the proof and elaborate on the steps when the Weyl connection $a$ is relevant. 
\begin{proof}[Proof of Theorem \ref{diff_Weyl_normal_form}]
Suppose $p\in{\cal P}_W$ and let $\un {\cal V}_{p}$ be the $(g,a)$-transversal vector for $\tilde g$ at $(p,0)$.
One can parametrize the (non-geodesic) curve $C_p:\lambda\mapsto \phi(p,\lambda)\in \ti M$ with initial conditions
\begin{equation}
\label{phi_initial_conditions}
\phi(p,0)= (p,0)\,,\qquad \p_{\lambda}\phi(p,\lambda)|_{\lambda=0}= \un{\cal V}_{p}\,,
\end{equation} 
with the ``equation of motion'' $\nabla_{\un{\cal U}}\un{\cal U}=\un{\cal A}$, where $\un{\cal U}= \frac{\td}{\td\lambda}$ is the tangent vector to the accelerated curve $C_p$, and the acceleration vector $\un{\cal A}$ satisfies $\tilde g(\un{\cal T}, \un{\cal A})=0$. Suppose the domain of $\phi$ is $\tilde U_0\subset{\cal P}_W\times\bb R$, which is dilatation-invariant. Then $\phi:\tilde U_0\to\tilde M$ is a smooth map commuting with dilatation, and it can be proved that there exists  $\tilde U_1\subset \tilde U_0$ as a dilatation-invariant neighborhood of ${\cal P}_W\times\bb \{0\}$ such that $\phi:\tilde U_1\to \tilde M$ is a diffeomorphism (see \cite{Fefferman:2007rka}).

\par
Furthermore, one can define $\tilde M'=\{(p,\lambda)\in\tilde U_1|\,(p,\mu)\in\tilde U_1, \forall\mu\in\bb R$ satisfying $|\mu|\leqslant|\lambda|\}$. It is easy to verify that $(\tilde M',\phi^*\tilde g)$ satisfies the conditions of Definition \ref{Ambient} and thus is a Weyl pre-ambient space for $(M,[g,a])$. It follows that for each $p\in {\cal P}_W$, the set for $\lambda$ such that $(p,\lambda)\in\tilde M'$ is an open interval $I_p$ containing $0$, and the parametrized curve $C'_p:\lambda\mapsto(p,\lambda)$ with tangent vector $\un{\cal U}'$ and the acceleration $\un {\cal A}'=\nabla'_{\cal U'}{\cal U'}$ satisfies $\phi^*\tilde g(\un {\cal T}', \un {\cal A}')=0$, where $\un{\cal T}'\equiv \phi^*\un{\cal T}$, and $\nabla'$ is the Levi-Civita connection associated with $\phi^*\tilde g$. Hence, conditions \ding{172} and \ding{173} of Definition \ref{Weyl_normal_form} are satisfied by $(\tilde M',\phi^*\tilde g)$.
\par
Finally let us verify condition \ding{174} of Definition \ref{Weyl_normal_form}. Since $\un{\cal V}$ satisfies the conditions in \eqref{Weyl g-transversal vector} and $\phi$ satisfies \eqref{phi_initial_conditions}, under the identification $\mathbb{R}_{+}\times M\times \mathbb{R}\simeq {\cal P}_W\times \mathbb{R}$ induced by $g$ we have at $(\lambda=0,p)$:
 \begin{align}
 (\phi^{*}\ti g)(\un\p_{\lambda},\un{\cal T})&=t^2\,\nn\\
 (\phi^{*}\ti g)(\un\p_{\lambda},\un{\cal H})&=0\qquad\forall\,\un{\cal H}\in {\cal H}_{p}\,,\\
 (\phi^{*}\ti g)(\un\p_{\lambda},\un\p_{\lambda})&= 0\,.\nn
 \end{align}
For a given connection $\mathbf{a}=t\un\p_t\otimes \big(t^{-1}\td t+a_i(x)\td x^i\big)$ on ${\cal P}_W$, the horizontal subspace ${\cal H}_p$ at $(p,0)$ is spanned by $\un D_i=\un\p_i-ta_i\un\p_t$. Since $(\tilde M',\phi^*\tilde g)$ is a Weyl pre-ambient space for $(M,[g,a])$, $\iota^*(\phi^*g)$ is the tautological tensor $\mathbf g_0$ on ${\cal P}_W$. Then, the above equations give that $\phi^{*}\ti g|_{\lambda=0}=  t^2 \mathbf{g}_{0}  + 2t(\td t + ta_{i}(x)\td x^{i})\td\lambda $. Therefore, all the conditions in Definition \ref{Weyl_normal_form} are satisfied by $(M',\phi^*g)$, which completes the existence part of the Proposition. The uniqueness part follows from the fact that the above construction of $\phi$ is forced. Suppose $\phi:M\to M'$ is a diffeomorphism such that $(M',\phi^*g)$ is a pre-ambient space in Weyl-normal form, then $\un{\cal V}_p$ must be $(g,a)$-transversal for $\ti g$ at $(p,0)$, and the curve $C'_p:\lambda\mapsto \phi(z,\lambda)$ must be the unique curve satisfying the initial conditions \eqref{phi_initial_conditions} and having the acceleration $\un{\cal A}$, which determines $\phi:\tilde M\to \tilde M'$ uniquely.
\end{proof}

\subsubsection{Proof of Theorem \ref{Weyl_ambient_existence}}
\begin{proof}[Proof of Theorem \ref{Weyl_ambient_existence}]
 The proof of this theorem has two main parts. First, from $\ti Ric(\ti g)=0$ and the initial value of $\tilde g$ at $\rho=0$ we will determine the first $\rho$-derivative of the metric components at $\rho=0$. Then, using an inductive argument we will show that all higher derivatives (to infinite order) at $\rho=0$ can also be determined from the Ricci-flatness condition. Let us write the unknown components of $\tilde g$ as
 \begin{equation}
 \ti g_{00}= c(x,\rho)\,,\qquad\ti g_{0 i}= t b_{i}(x,\rho)\,,\qquad\ti g_{ij}= t^2 g_{ij}(x,\rho)\,,
 \end{equation}
where $g_{ij}(x,\rho)$ can be considered as a one-parameter family of metrics on $M$. From property (2) above we have the initial values $c(x,0)=0$ and $b_{i}(x,0)=0$. The general metric has the form
\begin{equation}\label{_metric_formal_second_order}
\ti g_{IJ}= \begin{blockarray}{cccc}
 & 0 & j & \infty  \\
\begin{block}{c(ccc)}
  0 & c(x,\rho)& tb_i(x,\rho) & t\\
   i &  tb_i(x,\rho)  & t^2g_{ij}(x,\rho)&  t^2a_i(x,\rho)\\
  \infty & t & t^2a_i(x,\rho) & 0\\
\end{block}
\end{blockarray}
\,\,,
 \end{equation}
 and the inverse metric is
\begin{equation}\label{inverse_metric_formal}
\ti g^{IJ}=\left(\begin{array}{ccc}\frac{a^2}{\chi}& -\frac{(1- a\cdot b)a^{j}+ a^2 b^{j}}{t\chi} & \frac{1- a\cdot b}{t\chi}\\
 -\frac{(1- a\cdot b)a^{i}+ a^2 b^{i}}{t\chi}  & \frac{g^{ij}}{t^2}+ \frac{(1-a\cdot b)(a^{i}b^{j}+ a^{j}b^{i}) + a^2 b^{i}b^{j}- (c- b^2)a^{i}a^{j}}{t^2 \chi}& \frac{(c- b^2)a^{i}- (1- a\cdot b)b^{i}}{t^2 \chi} \\
 \frac{1- a\cdot b}{t\chi}&\frac{(c- b^2)a^{j}- (1- a\cdot b)b^{j}}{t^2 \chi}& \frac{b^2 - c}{t^2 \chi}\\\end{array}\right)\, ,
 \end{equation}
 where $a^{i}\equiv g^{im}a_{m}$, $b^{i}\equiv g^{im}b_{m}$ and $\chi= a^{2}(c-b^2)+ (1- a\cdot b)^2$, with $a^{2}= a_{k}a^{k}$, $b^2= b_{k}b^{k}$ and $a\cdot b= a_{k}b^{k}$. 
The Christoffel symbols $\ti \Gamma_{IJK}\equiv \ti g_{KM}\ti \Gamma^{M}{}_{IJ}$ are
 
 \begin{equation}\label{Christoffel_Initial_value}
 \begin{split}
2\ti \Gamma_{IJ 0}&=
\left(\begin{array}{ccc}
0& \p_{j}c &\p_{\rho}c\\
\p_{i}c  & t(\p_{i}b_{j}+ \p_{j}b_{i}- 2 g_{ij})& t(\p_{\rho}b_{i}- 2a_{i})\\
\p_{\rho}c &t(\p_{\rho}b_{j}- 2a_{j})& 0
\end{array}\right)\,,\\
2\ti \Gamma_{IJ k}&=
\left(\begin{array}{ccc}
2b_{k}- \p_{k}c & t\left(2 g_{jk}+ \p_{j}b_{k}- \p_{k}b_{j}\right)&t(2a_{k} + \p_{\rho}b_{k})\\
t(2 g_{ik}+ \p_{i}b_{k}- \p_{k}b_{i})  &2 t^2 \gamma_{ijk}& t^2 (\p_{\rho}g_{ik}+ F_{ik}) \\
t(2a_{k}+ \p_{\rho}b_{k}) & t^2 \left(\p_{\rho}g_{jk}+ F_{jk}\right) & 2t^2 \p_{\rho}a_{k}
\end{array}\right)\,,\\
2\ti \Gamma_{IJ \infty}&=
\left(\begin{array}{ccc}
2- \p_{\rho}c& t(2a_{j}- \p_{\rho}b_{j}) & 0\\
t(2a_{i}- \p_{\rho}b_{i})  & t^2 (\p_{i}a_{j}+ \p_{j}a_{i}- \p_{\rho}g_{ij})& 0\\
0 &0& 0
\end{array}\right)\,,
\end{split}
\end{equation}
where $\gamma_{ijk}= g_{km}\gamma^{m}{}_{ij}$ with $\gamma^{m}{}_{ij}= \frac{1}{2}g^{mk}\left(\p_{i}g_{jk}+ \p_{j}g_{ik}-\p_{k}g_{ij}\right)$ and  $F_{jk}= \p_{j}a_{k}- \p_{k}a_{j}$. Calculating the components $\ti R_{IJ}$ of $\tilde Ric(\tilde g)$ to the leading order in $\rho$-expansion from
\begin{equation}\label{Ricci}
\ti R_{IJ}= \frac{1}{2}\ti g^{KL}\left(\p^{2}_{IL}\ti g_{JK}+ \p^{2}_{JK}\ti g_{IL}- \p^{2}_{KL}\ti g_{IJ}- \p^{2}_{IJ}\ti g_{KL}\right)+ \ti g^{KL}\ti g^{PQ}\big(\ti\Gamma_{ILP}\ti\Gamma_{JKQ}- \ti\Gamma_{IJP}\ti\Gamma_{KLQ}\big) \,,
\end{equation}
and setting them to zero as the Ricci-flatness condition demands, we obtain
\begin{equation}\label{metric_components_order_1}
\begin{split}
c(x,\rho)&= 2\rho + O(\rho^{2})\,,\qquad b_{i}(x,\rho)= O(\rho^2)\,,\\ g_{ij}(x,\rho)&=g_{ij}(x)+ \rho\big(2\hat P_{(ij)}- 2a_{i}(x)a_{j}(x)\big)+ O(\rho^{2})\,,
\end{split}
\end{equation} 
where $\hat P_{ij}$ is the Weyl-Schouten tensor. One can observe that this agrees with \eqref{Weyl_ambient}, where $g_{ij}(x)$ corresponds to $\gamma^{(0)}_{ij}$ in the expansion \eqref{eq:gexpan}, and the order $O(\rho)$ matches $\gamma^{(1)}_{ij}$ [see \eqref{Pgf}]. Note that the above components of a Weyl-ambient metric reduce to the components of an ambient metric in \cite{Fefferman:2007rka} when the Weyl connection $a_i$ is turned off. 
\par
The next stage of the proof is to carry out an inductive perturbation calculation for higher orders in $\rho$. The purpose of this calculation is to prove (inductively) that the Ricci-flatness condition can be used to determine the unknown components of $\tilde g$ in Weyl-normal form to infinite order in $\rho$.\par
Let $\ti g^{[k]}$ represent a metric that includes the terms of the $\rho$-expansion of $\tilde g$ up to (including) order $O(\rho^{k})$, i.e.,\ $\tilde g=\tilde g^{[k]}+O(\rho^{k+1})$. Then, the Ricci-flatness condition of $\tilde g$ implies that the components $\ti R^{[k]}_{IJ}$ of $Ric(\tilde g^{[k]})$ satisfy
\begin{align}
\label{eq:Ricciflatm}
\begin{split}
\ti R_{IJ}(\ti g^{[k]})= O(\rho^{k})\quad I,J\neq\infty\,,\qquad
\ti R_{I\infty}(\ti g^{[k]})= O(\rho^{k-1})\,.
\end{split}
\end{align}
To carry out the induction, we assume that $\ti g^{[m-1]}$ has been uniquely determined from the condition \eqref{eq:Ricciflatm} with $k=m-1$. We have seen this is true for $m=2$ above by explicit calculation. Now we want to show that $\ti g^{[m]}$ then can be uniquely determined from the condition \eqref{eq:Ricciflatm} with $k=m$. Set $\ti g^{[m]}_{IJ}= \ti g^{[m-1]}_{IJ}+ \Phi_{IJ}$, with 
\begin{equation}
\label{eq:g+phi}
\Phi_{IJ}:=
\left(\begin{array}{ccc}
\Phi_{00}&\Phi_{0j}&0\\
\Phi_{i0}&\Phi_{ij}&\Phi_{i\infty}\\
0&\Phi_{j\infty}&0
\end{array}\right)
= \rho^m 
\left(\begin{array}{ccc}
\phi_{00}(x)&t\phi_{0j}(x)&0\\
t\phi_{0i}(x)&t^2 \phi_{ij}(x)&t^{2}a^{(m)}_{i}(x)\\
0&t^{2}a^{(m)}_{j}(x)&0
\end{array}\right)\,,
\end{equation}
where $a_{i}^{(m)}(x)$ is the $m^{th}$ order term of $a_{i}(x,\rho)$ [see \eqref{eq:aexpan}], and we have considered the fact that $\ti g^{[m]}_{IJ}$ satisfies \eqref{lemma_conditions}. All we have to show is that $\phi_{00}$, $\phi_{0i}$ and $\phi_{ij}$ can all be uniquely determined. From \eqref{Ricci} one finds that 
\begin{equation}\label{Ricci_2}
\begin{split}
\ti R^{[m]}_{IJ}={}&\ti R^{[m-1]}_{IJ}+  \frac{1}{2}\ti g_{[m]}^{KL}\left(\p^{2}_{IL} \Phi_{JK}+ \p^{2}_{JK} \Phi_{IL}- \p^{2}_{KL} \Phi_{IJ}- \p^{2}_{IJ} \Phi_{KL}\right) \\&
+ \ti g_{[m]}^{KL}\ti g_{[m]}^{PQ}\left(\ti \Gamma^{[m]}_{ILP}\Gamma^{\Phi}_{JKQ}+  \Gamma^{\Phi}_{ILP}\ti\Gamma^{[m]}_{JKQ}- \ti\Gamma^{[m]}_{IJP} \Gamma^{\Phi}_{KLQ}- \Gamma^{\Phi}_{IJP}\ti \Gamma^{[m]}_{KLQ}\right)+ O(\rho^{m})\,,
\end{split}
\end{equation}
where $\ti g_{[m]}^{KL}$ and $\ti \Gamma_{IJK}^{[m]}$ are the inverse and Christoffel symbols of $\ti g^{[m]}_{KL}$, respectively, and $\Gamma^{\Phi}_{IJK}\equiv\frac{1}{2}( \p_{J}\Phi_{IK}+ \p_{I}\Phi_{JK}- \p_{K}\Phi_{IJ})$. The components of $\Gamma^{\Phi}_{IJK}$ can be expressed as follows:
\begin{equation}\label{GammaPhi}
\begin{split}
2\Gamma^{\Phi}_{IJ0}&= 
\left(\begin{array}{ccc}
0&0&\p_\rho\Phi_{00}\\
0&0&\p_\rho\Phi_{i0}\\
\p_\rho\Phi_{00}&\p_\rho\Phi_{0j}&0
\end{array}\right)+O(\rho^m)\,,\\
2\Gamma^{\Phi}_{IJk}&= 
\left(\begin{array}{ccc}
0&0&\p_\rho\Phi_{0k}\\
0&0&\p_\rho\Phi_{ik}\\
\p_\rho\Phi_{0k}&\p_\rho\Phi_{jk}&2\p_\rho\Phi_{\infty k}
\end{array}\right)+O(\rho^m)\,,\\
2\Gamma^{\Phi}_{IJ\infty}&= 
\left(\begin{array}{ccc}
-\p_\rho\Phi_{00}&-\p_\rho\Phi_{0j}&0\\
-\p_\rho\Phi_{i0}&-\p_\rho\Phi_{ij}&0\\
0&0&0
\end{array}\right)+O(\rho^m)\,.
\end{split}
\end{equation}

Substituting \eqref{GammaPhi} and the leading order of $\tilde\Gamma^{[m]}_{IJK}$ and $\tilde g_{[m]}^{IJ}$ [i.e.,\ the leading order of $\tilde g^{IJ}$, $\tilde\Gamma_{IJK}$ in \eqref{inverse_metric_formal},\eqref{Christoffel_Initial_value}] into \eqref{Ricci_2}, one finds 
\begin{equation}\label{Ricci_inductive}
\begin{split}
t^2 \ti R^{[m]}_{00}&=t^2 \ti R^{[m-1]}_{00} + m\rho^{m-1}\left(m- 1 - \frac{d}{2}\right)\phi_{00}+O(\rho^{m})\,,\\
t\ti R^{[m]}_{0i} &= t\ti R^{[m-1]}_{0i}+ m\rho^{m-1}\left[\frac{1}{2}\p_{i}\phi_{00}+ \left(m-1-\frac{d}{2}\right)\phi_{0i}\right]+ O(\rho^{m})\,,\\
\ti R_{ij}^{[m]}&= \ti R_{ij}^{[m-1]}+ m\rho^{m-1}\left[(m- \frac{d}{2})\phi_{ij}- \frac{1}{2}g_{ij}g^{km}\phi_{km}+ \mathring\nabla_{(i}\phi_{j)0}+ \mathring P_{ij}\phi_{00}\right] + O(\rho^{m})\,,\\
t\ti R^{[m]}_{0\infty}&=t\ti R^{[m-1]}_{0\infty}+ \frac{1}{2}m(m-1)\rho^{m-2}\phi_{00} + O(\rho^{m-1})\,,\\
\ti R^{[m]}_{i\infty}&=\ti R^{[m-1]}_{i\infty}+ \frac{1}{2}m(m-1)\rho^{m-2}\phi_{i0}+ O(\rho^{m-1})\,,\\
\ti R^{[m]}_{\infty \infty}&=\ti R^{[m-1]}_{\infty\infty} -m(m-1)\rho^{m-2}\left(\frac{1}{2}a^{2}\phi_{00}- a^{k}\phi_{k0} +\frac{1}{2}g^{km}\phi_{km}\right)+ O(\rho^{m-1})\,,
\end{split}
\end{equation}
where $\mathring P_{ij}$, $\mathring\nabla$ are the LC Schouten tensor and LC connection associated with the metric $g_{ij}(x)$. Although the Weyl connection $a^{(0)}_{i}(x)$ appears throughout the calculation, it cancels itself out rather unexpectedly, except for the terms in $\ti R^{(m)}_{\infty\infty}$. The inductive argument then proceeds in the same way as \cite{Fefferman:2007rka}. First we consider the Ricci components with $I,J\neq \infty$. From the first two equations in \eqref{Ricci_inductive} one can uniquely determine $\phi_{00}$ and $\phi_{0i}$ such that $\ti R^{[m]}_{00}$ and $\ti R^{[m]}_{0i}$ both vanish up to order $O(\rho^{m})$. Then, from the third equation in \eqref{Ricci_inductive}  one can uniquely solve for $\phi_{ij}$ such that the order $O(\rho^{m-1})$ of $\ti R^{[m]}_{ij}$ vanishes. Therefore, $\tilde g^{[m]}$ will be uniquely determined by $\ti R^{[m]}_{IJ}=O(\rho^m)$ for $I,J\neq \infty$ once $\tilde g^{[m-1]}$ is determined, and hence the unknown components of $\ti g_{IJ}$ can be determined to infinite order. 
\par
Note that when $d=2m$, the situation becomes subtle because the term $\phi_{ij}$ vanishes in $\tilde R^{[m]}$. In \cite{Fefferman:2007rka}, this is attributed to the obstruction of the Ricci-flatness condition at $O(\rho^{d/2-1})$ when $d$ is an even integer, and one has to carefully consider even and odd $d$ separately. Nevertheless, since we consider the dimension $d$ as a continuous parameter, we can always solve for $\phi_{ij}$ from the Ricci-flatness condition for any $d$, and the information regarding these obstructions is not lost but takes the form of poles in $\phi_{ij}$ at $d=2m$. As is shown in Proposition \ref{prop:obspole}, since $\phi_{ij}$ represents the order $O(\rho^m)$ of $g_{ij}(x,\rho)$ in the Weyl-ambient metric \eqref{Weyl_ambient}, this pole represents exactly the Weyl-obstruction tensor.
\par
So far we have proved that the unknown components of $\ti g$ are determined to infinite order by the Ricci-flatness condition for $I,J\neq \infty$. To finish the analysis we also need to show that the remaining Ricci components $\ti R_{I\infty}$ also vanish to infinite order when we plug in the solution for $\tilde g$ obtained from $\ti R_{IJ}=0$ for $I,J\neq \infty$. Consider the Bianchi identity $\ti g^{JK}\nabla_{I}\ti R_{JK}= 2\ti g^{JK} \nabla_{J}\ti R_{IK}$. Expanding the covariant derivative in terms of the Christoffel symbols we get
\begin{equation}\label{Bianchi}
2\ti g^{JK}\p_{J}\ti R_{IK}- \ti g^{JK}\p_{I}\ti R_{JK}- 2 \ti g^{JK}\ti g^{PQ}\ti\Gamma_{JKP}\ti R_{QI}=0\,.
\end{equation}
Since $\ti R_{I\infty}= {\cal O}(\rho^{m-2})$ is trivially true for $m=2$, now we want to show that $\ti R_{I\infty}= {\cal O}(\rho^{m-2})$ leads to $\ti R_{I\infty}= {\cal O}(\rho^{m-1})$ by means of the Bianchi identity. Expanding \eqref{Bianchi} for $I=0,i,\infty$ and making use of the homogeneity property of the metric we get
 \begin{equation}\label{Bianchi_2}
 \begin{split}
& \left(d -2 -2 \rho\p_{\rho}\right)\ti R_{0\infty} = {\cal O}(\rho^{m-1})\\&
 (d-2 -2 \rho\p_{\rho})\ti R_{i\infty}- t\p_{i}\ti R_{0\infty}={\cal O}(\rho^{m-1})\\&
 a^{2}\left( t^{-1}d \ti R_{\infty 0}+2\p_{0}\ti R_{\infty 0}\right)- 2 t ^{-1}a^{m}\left( \p_{m}\ti R_{\infty 0}-(2-d)t^{-1}\ti R_{\infty m}\right) \\&
+2 t^{-2}\left(d-2 - \rho\p_{\rho}\right)\ti R_{\infty\infty}+2 t^{-2}g^{mk}\mathring \nabla_{m}\ti R_{\infty k}+ 2t^{-1}\mathring{P}\ti R_{\infty 0} ={\cal O}(\rho^{m-1})\,.
 \end{split}
 \end{equation}
We can see that the Weyl connection appears only in the last equation of \eqref{Bianchi_2}. Note that all the Ricci terms $\ti R_{IJ}$ with $I,J \neq \infty$ has been dropped from \eqref{Bianchi_2} since they vanish to infinite order. Suppose $\ti R_{I\infty}= \gamma_{I}\rho^{m-2}$. The first equation in \eqref{Bianchi_2} gives $(d+2 -2m)\gamma_{0}= {\cal O}(\rho)$, and thus $\ti R_{0\infty}={\cal O}(\rho^{m-1})$. The second equation in \eqref{Bianchi_2}  gives  $(d+2 -2 m)\gamma_{i}= {\cal O}(\rho)$, and thus $\ti R_{i\infty}= {\cal O}(\rho^{m-1})$. The last equation then gives $(d-m)\gamma_{\infty}= {\cal O}(\rho)$, so $\ti R_{\infty\infty}={\cal O}(\rho^{m-1})$. This completes the inductive argument and thus $\ti R_{I\infty}$ can also be made to vanish to infinite order.
\par
To summarize, we have shown by an inductive argument that there exists a Weyl-ambient space $(\tilde M,\tilde g)$ for $(M,[g,a])$ in Weyl-normal form with acceleration $\un{\cal A}$. Some components of $\tilde g$ have the form in \eqref{lemma_conditions}, and all the unknown components are determined uniquely to infinite order of $\rho$ at ${\cal P}_W\times \{0\}$ by the Ricci-flatness condition.
\end{proof}

\chapter{Weyl-Obstruction Tensors}
\label{chap:WOT}
In Section \ref{sec:FG} we saw that the poles of asymptotic expansion of the ALAdS bulk in even dimensions give rise to obstruction tensors, which are covariant quantities on conformal manifolds $(M,[g])$. The goal of chapter is to carry over this concept to Weyl manifolds $(M,[g,a])$. First we introduce Weyl-obstruction tensors as the poles of the ALAdS bulk metric in the WFG gauge. Then we provide the precise definitions of Weyl-obstruction via the Weyl-ambient construction in first and second formalisms, respectively, and show that they are equivalent. Notice that in Section~\ref{sec:WOTbulk}, $\gamma^{(2k)}_{ij}$ will stand for terms in the $z$-expansion \eqref{hex} of the ALAdS bulk metric, while in Section~\ref{sec:WOTambient}, $\gamma^{(k)}_{ij}$ will stand for terms in the $\rho$-expansion \eqref{eq:gexpan} of the Weyl-ambient metric.
\section{Poles of the Metric Expansion}
\label{sec:WOTbulk}
In the previous chapters we saw that the WFG gauge in the bulk induces a Weyl geometry on the boundary. Now we would like to determine the higher order terms in the $z$-expansion \eqref{hex} and find the obstruction tensors with the Weyl connection turned on. The method is exactly analogous to that in Section \ref{sec:WFG} for the FG gauge. By solving the bulk Einstein equations order by order in the WFG gauge, we find that $\gamma^{(2k)}_{ij}$ still has the same form as \eqref{gamma2k}, except that the obstruction tensor ${\cal O}^{(2k)}_{ij}$ is now promoted to the \emph{Weyl-obstruction tensor} $\hat {{\cal O}}^{(2k)}_{ij}$. Unlike ${\cal O}^{(2k)}_{ij}$, which is only Weyl-covariant in $2k$-dimension, the Weyl-obstruction tensors $\hat {{\cal O}}^{(2k)}_{ij}$ are Weyl-covariant with a weight $2k-2$ in any dimension; that is, under a Weyl transformation \eqref{WT} it transforms in any $d$ as $\hat {{\cal O}}^{(2k)}_{ij} \to {\cal B}(x)^{2k-2}\hat {{\cal O}}^{(2k)}_{ij}$.
\par
In principle, $\gamma^{(2k)}_{ij}$ at any order can be obtained from the Einstein equations by iteration. In this section, we will show solutions of $\gamma^{(2k)}_{ij}$ obtained from Einstein equations up to $k=3$, and read off the corresponding Weyl-obstruction tensors from them. The details of the expansions of Einstein equations can be found in Appendix \ref{app:B0}. 
\par
First, the leading order of the $ij$-components of the Einstein equations gives
\begin{align}
\label{g2}
\frac{\gamma^{(2)}_{ij}}{L^2}&=-\frac{1}{d-2}\bigg(\hat R^{(0)}_{(ij)}-\frac{1}{2(d-1)}\hat R^{(0)}\gamma^{(0)}_{ij}\bigg)\,.
\end{align}
We notice that this is the symmetric part of the Weyl-Schouten tensor defined in \eqref{WP} with a minus sign, i.e.\
\begin{align}
\label{Pgf0}
\frac{\gamma^{(2)}_{ij}}{L^2}&=-\hat P_{(ij)}=-\hat P_{ij}-\frac{1}{2}f^{(0)}_{ij}\,.
\end{align}
Similar to the FG gauge, one can check that the residue of the pole in \eqref{g2} vanishes identically when $d=2$. Hence, there is no Weyl-obstruction tensor for $d=2$ and so no logarithmic term will appear in the metric expansion in the $d\to 2^{-}$ limit.
\par
Then, solving the $O(z^2)$-order of the $ij$-components of the Einstein equations yields
\begin{align}
\label{g40}
\frac{\gamma^{(4)}_{ij}}{L^4}&=-\frac{1}{4(d-4)}\hat{\cal O}^{(4)}_{ij}+\frac{1}{4}\hat P^{k}{}_{ i}\hat P_{k j}-\frac{1}{2L^2}\hat\nabla^{(0)}_{( i} a_{ j)}^{(2)}\,,
\end{align}
where $\hat{\cal O}^{(4)}_{ij}$ is the Weyl-obstruction tensor for $d=4$, namely the Weyl-Bach tensor $\hat B_{ij}$, given by
\begin{align}
\hat{\cal O}^{(4)}_{ij}=\hat B_{ij}=\hat\nabla^{(0)}_ k\hat\nabla_{(0)}^ k \hat P_{ij}-\hat\nabla^{(0)}_ k\hat\nabla^{(0)}_{ j} \hat P_{ i}{}^{ k}-\hat W^{(0)}_{l j i k}\hat P^{ kl}\,.
\end{align}
If we compare \eqref{g4} with the corresponding result \eqref{gamma4} in the FG case, we see that the form of the expression stays almost the same, with all the LC quantities now being promoted to the corresponding Weyl quantities. Besides, in the WFG gauge $\gamma^{(4)}_{ij}$ also has an additional term involving $a^{(2)}_{ i}$, which does not contribute to the pole at $d=4$.
\par
Moving on to the $O(z^4)$-order of the Einstein equations we get
\begin{equation}
\label{g60}
\begin{split}
\frac{\gamma^{(6)}_{ij}}{L^6}=&-\frac{1}{24(d-6)(d-4)}\hat{\cal O}^{(6)}_{ij}+\frac{1}{6(d-4)}\hat B_{k( i}\hat P^k{}_{ j)} -\frac{1}{3L^4}\hat\nabla^{(0)}_{( i}a^{(4)}_{ j)}\\
&-\frac{1}{L^4}a^{(2)}_{ i}a^{(2)}_{ j}+\frac{1}{6L^2}a^{(2)}\cdot a^{(2)}\gamma^{(0)}_{ij}+\frac{1}{6L^2}\hat\nabla^{(0)}_{( i}(\hat P^{k}{}_{ j)}a^{(2)}_k)+\frac{1}{2L^4}\hat\gamma_{(2)}^{k}{}_{ij}a^{(2)}_k\, ,
\end{split}
\end{equation}
where $\hat\gamma_{(2)}^{k}{}_{ij}\equiv-\frac{L^2}{2}(\hat\nabla^{(0)}_ i\hat P^k{}_{ j}+\hat\nabla^{(0)}_ j\hat P_ i{}^k-\hat\nabla_{(0)}^k\hat P_{ij})$, and $\hat{\cal O}^{(6)}_{ij}$ is the Weyl-obstruction tensor for $d=6$:
\begin{equation}
\label{WO6}
\begin{split}
\hat{\cal O}^{(6)}_{ij}={}&\hat\nabla_{(0)}^ k\hat\nabla^{(0)}_ k\hat  B_{ij}-2\hat W^{(0)}_{ l j i k}\hat B^{ k l}-4\hat P\hat B_{ij}+2\hat P_{ k( j}\hat B^{ k}{}_{ i)}-2\hat B^{ k}{}_{( i}\hat P_{ j) k}\\
&+2(d-4)\bigg(\hat\nabla_{(0)}^ k\hat  C_{ k l( i}\hat P^{ l}{}_{ j)} -\hat P^{ k l}\hat\nabla^{(0)}_{( i}\hat  C_{ j) l k}+2\hat P^{( l k)}\hat\nabla^{(0)}_ k \hat C_{(ij) l}+\hat\nabla^{(0)}_ k\hat P^{ l k}\hat C_{(ij) l}\\
&\qquad\qquad\qquad-\hat C^{ l}{}_{ i}{}^{ k}\hat C_{ k j l}+ \hat\nabla_{(0)}^ k\hat P^ l{}_{( i}\hat C_{ j) l k}-\hat W^{(0)}_{ l( j i) k}\hat P^{ k}{}_m\hat P^{m l}\bigg)\,.
\end{split}
\end{equation}
It is easy to verify that  \eqref{g6} and \eqref{WO6} go back to the FG expressions \eqref{gamma6} and \eqref{O6} when we turn off the Weyl structure $a_ i$. Note that when the Weyl connection is turned off, the first term inside the parentheses of \eqref{WO6} vanishes due to \eqref{divC}, and the second term there vanishes since the LC Schouten tensor $\mathring P_{ij}$ is symmetric. Once again, we observe that all the $a^{(2)}_ i$ and $a^{(4)}_ i$ terms that appear in $\gamma^{(6)}_{ij}$ do not contribute to the pole at $d=6$ and thus are not part of the obstruction tensor $\hat{\cal O}_{ij}^{(6)}$. 
\par
Just as ${\cal O}_{ij}^{(2k)}$ derived in the FG gauge, all the $\hat{\cal O}_{ij}^{(2k)}$ are also symmetric traceless tensors, and they are divergence-free when $d=2k$. These properties can either be verified by using the result from the $ij$-components of the Einstein equations (``evolution equations"), or read off from the $zz$- and $z i$-components of the Einstein equations (``constraint equations"). More specifically, plugging $\gamma^{(2k)}_{ij}$ into the $zz$-component of the Einstein equations we can see that $\mathcal{\hat O}^{(2k)}_{ij}$ is traceless in any dimension, and the same result can also be obtained by taking the trace of the $ij$-components of the Einstein equations. To see that $\hat{\cal O}_{ij}^{(2k)}$ is divergence-free when $d=2k$, we can plug $\gamma^{(2k)}_{ij}$ into the $z i$-components of the Einstein equations. For instance, the $O(z^4)$-order of the $z i$-equations gives
\begin{align}
\label{divWB}
\hat\nabla_{(0)}^j\hat B_{ji}=(d-4)\hat P^{jk}(\hat C_{kji}+\hat C_{ijk})\,,
\end{align}
and so the divergence of $\hat B_{ij}$ vanishes when $d=4$. In the FG gauge where the Schouten tensor is symmetric, the second term in the bracket vanishes and so \eqref{divWB} goes back to \eqref{divB}. On the other hand, the divergence of $\hat{\cal O}_{ij}^{(2k)}$ can also be derived from a direct calculation by using repeatedly the Weyl-Bianchi identity
\begin{align}
\label{divWP}
\hat\nabla_{(0)}^i\hat P_{ij}=\hat\nabla^{(0)}_j\hat P\,,
\end{align}
which can be read off from the $O(z^2)$-order of the $z i$-equation. The above discussion indicates that the $zz$- and $z i$-components of the Einstein equations do not contain more information about $\gamma^{(2k)}_{ij}$ than the $ij$-components of  Einstein equations. Note that here we only talk about the equations of motion for $\gamma^{(2k)}_{ij}$. At $O(z^d)$-order the $zz$- and $z i$-equations do provide new constraints on $\pi^{(0)}_{ij}$, while the $ij$-equations on $\pi^{(0)}_{ij}$ become trivial.
\par
It is also convenient to define the \emph{extended Weyl-obstruction tensor} $\hat{\Omega}^{(k)}_{ij}$ as the Weyl-covariant version of the extended obstruction tensor defined in \eqref{extO}. For example, for $k=1$ and $k=2$ we have
\begin{align}
\label{extWO}
\hat\Omega^{(1)}_{ij}=-\frac{1}{d-4}\hat B_{ij}\,,\qquad\hat\Omega^{(2)}_{ij}=\frac{1}{(d-6)(d-4)}\hat{\mathcal O}^{(6)}_{ij}\,.
\end{align}
\par
Similar to the FG case, the Weyl-obstruction tensor $\hat{\cal O}_{ij}^{(2k+2)}$ is also proportional to the residue of the extended Weyl-obstruction tensor $\hat{\Omega}^{(k)}_{ij}$. Both the Weyl-obstruction tensors and the extended Weyl-obstruction tensors can be defined following \cite{graham2005ambient,graham2009extended} by promoting the ambient metric to the ``Weyl-ambient metric". We will discuss this in detail in the next section.

\section{Weyl-Obstruction Tensors from the Ambient Construction}
\label{sec:WOTambient}
A very useful property of the ambient metric introduced in \cite{2001math.....10271F} in the context of conformal geometry is the ability to construct conformal-covariant tensors from the ambient Riemann tensor, including the (extended) obstruction tensors. In the last section we saw that these tensors can be generalized to (extended) Weyl-obstruction tensors on Weyl manifolds $(M,[\gamma^{(0)},a^{(0)}])$ by evaluating the poles of the metric expansion of $\gamma_{ij}$ in the ALAdS bulk. However, defining them as poles lead to an ambiguity since a pole has the freedom of being shifted by finite terms. In this section we will see that the (extended) Weyl-obstruction tensors can be defined in a more explicit way from the Weyl-ambient space $(\tilde M,\tilde g)$. 

\subsection{First-Order Formalism}
\label{sec:firstorder}
First, we would like to demonstrate how the Weyl-obstruction tensors on $M$ can be derived from $(\tilde M,\tilde g)$ in the first order formalism using the frame introduced in \eqref{e+-def}. 
\par
Starting from the metric \eqref{eq:metricnull}, one can solve $\ti Ric(\ti g)=0$ order by order to find the $\gamma_{ij}^{(k)}$ in the $\rho$-expansion \eqref{eq:gexpan}, which is equivalent to solving the Einstein equations in the ALAdS bulk shown in Section \ref{sec:WOTbulk}.\footnote{Note again that the $\gamma^{(k)}_{ij}$ and $a^{(k)}_{i}$ defined here correspond to $(-2)^k\gamma^{(2k)}_{ij}/L^{2k}$ and $(-2)^ka^{(2k)}_{ i}/L^{2k}$ in the $z$-expansion \eqref{hex}, respectively.} The results are 
\begin{align}
\label{Pgf}
\gamma^{(1)}_{ij}={}&2\hat P_{(ij)}=2\hat P_{ij}-f^{(0)}_{ij}\,.\\
\label{g4}
\gamma^{(2)}_{ij}={}&\hat{\Omega}^{(1)}_{ij}+\hat P^{k}{}_{i}\hat P_{kj}+\hat\nabla^{(0)}_{(i} a_{j)}^{(1)}\,,\\
\gamma^{(3)}_{ij}={}&\tfrac{1}{3}\hat{\Omega}^{(2)}_{ij}+\tfrac{4}{3}\hat\Omega^{(1)}_{k(i}\hat P^k{}_{j)} +\tfrac{2}{3}\hat\nabla^{(0)}_{(i}a^{(2)}_{j)}\nn+2a^{(1)}_{i}a^{(1)}_{j}-\tfrac{1}{3}a^{(1)}\cdot a^{(1)}\gamma^{(0)}_{ij}\\
\label{g6}
&+\tfrac{1}{3}P^{k}{}_{(i}\hat\nabla^{(0)}_{j)}a^{(1)}_k-\tfrac{1}{3}a_{(1)}^k(\hat\nabla^{(0)}_i\hat P_{kj}+\hat\nabla^{(0)}_i\hat P_{jk}-\hat\nabla^{(0)}_k\hat P_{ji}+2\hat\nabla^{(0)}_j\hat P_{ik}-2\hat\nabla^{(0)}_k\hat P_{ij})\,,\\
\cdots\nn
\end{align}
where $f^{(0)}_{ij}=\p_ia^{(0)}_j-\p_ja^{(0)}_i$, and $\hat P_{ij}$ is the Weyl-Schouten tensor on $(M,[\gamma^{(0)},a^{(0)}])$. Treating $d$ as an continuous complex variable, the solution for each $\gamma^{(k\geqslant2)}_{ij}$ has a pole at $d=2k$ (see Proposition \ref{prop1}) represented by $\hat{\Omega}^{(k-1)}_{ij}$. For now one should simply regard $\hat{\Omega}^{(k-1)}_{ij}$ in the above equations as denoting the pole terms of $\gamma^{(k)}_{ij}$ at $d=2k$ ($\hat P_{ij}$ also represents the ``pole'' of $\gamma^{(1)}_{ij}$ at $d=2$, which identically vanishes in $2d$). Later in this subsection we will recognize them as extended Weyl-obstruction tensors through a precise definition. In terms of $\gamma^{(0)}_{ij}$, these quantities can be written as
\begin{align}
\label{eq:Wsch}
\hat P_{ij}={}&\frac{1}{d-2}\bigg(\hat R^{(0)}_{ij}-\frac{\hat R^{(0)}}{2(d-1)}\gamma^{(0)}_{ij}\bigg)\,,\\
\label{eq:Omega1}
\hat\Omega^{(1)}_{ij}={}&\frac{1}{d-4}\Big(-\hat\nabla^{(0)}_k\hat\nabla_{(0)}^k \hat P_{ij}+\hat\nabla^{(0)}_k\hat\nabla^{(0)}_{j} \hat P_{i}{}^{k}+\hat W^{(0)}_{kjil}\hat P^{lk}\Big)\,,\\
\label{eq:Omega2}
\hat\Omega^{(2)}_{ij}={}&\frac{1}{d-6}\Big(-\hat\nabla_{(0)}^k\hat\nabla^{(0)}_k\hat\Omega^{(1)}_{ij}+2\hat W^{(0)}_{kjil}\hat\Omega_{(1)}^{lk}+4\hat P\hat\Omega^{(1)}_{ij}-2\hat P_{k(j}\hat\Omega_{(1)}^{k}{}_{i)}+2\hat\Omega_{(1)}^{k}{}_{(i}\hat P_{j)k}\nn\\
&\qquad\quad+2\hat\nabla_{(0)}^k\hat  C_{kl(i}\hat P^{l}{}_{j)} -2\hat P^{kl}\hat\nabla^{(0)}_{(i}\hat  C_{j)lk}+4\hat P^{(kl)}\hat\nabla^{(0)}_l\hat C_{(ij)k}+2\hat\nabla^{(0)}_l\hat P^{kl}\hat C_{(ij)k}\nn\\
&\qquad\quad-2\hat C^{k}{}_{i}{}^{l}\hat C_{ljk}+2 \hat\nabla_{(0)}^l\hat P^k{}_{(i}\hat C_{j)kl}-2\hat W^{(0)}_{k(ji)l}\hat P^{l}{}_m\hat P^{mk}\Big)\,,
\end{align}
where $\hat W_{(0)}^i{}_{jkl}$ is the Weyl curvature tensor and {$\hat C_{ijk}\equiv\hat\nabla_{k}^{(0)}\hat P_{ij}-\hat\nabla_{j}^{(0)}\hat P_{ik}$ is the Weyl-Cotton tensor}. Note that indices are lowered with $\gamma^{(0)}_{ij}$ as necessary.
\par
We first look at how the Weyl-Schouten tensor $\hat P_{ij}$ is derived from the Weyl-ambient geometry. Consider the expansion of $\gamma_{ij}$. At $\rho=0$ and $t=1$, the ambient connection 1-form \eqref{eq:conn1form} becomes
\begin{align}
\tilde{\bm\omega}_{(0)}^{M}{}_{N}=&\left(\begin{array}{ccc}a^{(0)}_k & -\hat P_{jk} & 0  \\ \delta^i{}_k & \Gamma_{(0)}^i{}_{kj}& \hat P^i{}_k  \\0 & -\gamma^{(0)}_{jk} & -a^{(0)}_k \end{array}\right)\bm e^k
+\left(\begin{array}{ccc}0 & 0& 0  \\0 &\delta_j{}^i& 0  \\0& 0 & 0 \end{array}\right)\bm e^++\left(\begin{array}{ccc}0 & 0 & 0 \\ 0&\psi_j{}^i & 0  \\0  &0& 0 \end{array}\right)\bm e^- \,.
\end{align}
Notice that the first term, which is the pullback of $\tilde{\bm\omega}_{(0)}^{M}{}_{N}$ from $T^*\tilde M$ to $T^*M$, can be recognized as the Cartan normal conformal connection \cite{10.2996/kmj/1138845392,kobayashi2012transformation}. From here we can see that the Weyl-Schouten tensor of the boundary appears in the leading order ($\rho=0$) of the ambient connection. 
\par
From the connection 1-form \eqref{eq:conn1form}, we can also find the ambient curvature 2-form in the frame $\{\bm e^+,\bm e^i,\bm e^-\}$ using Cartan's second structure equation \cite{Wald,Liang1} (see Appendix \ref{app:Null} for details):
\begin{align}
\label{eq:curv2form}
\tilde{\bm R}^{M}{}_{N}=&\left(\begin{array}{ccc}0 & -t\bm{\mathcal C}_j & 0  \\-\frac{\rho}{t}\bm{\mathcal C}^i & {\bm{\mathcal W}}^i{}_{j} & \frac{1}{t}\bm{\mathcal C}^i \\0 & \rho t\bm{\mathcal C}_j & 0 \end{array}\right)+\left(\begin{array}{ccc}0 &\bm{\mathcal B}_j & 0 \\\frac{\rho}{t^2}\bm{\mathcal B}^i &\frac{1}{t}{\cal C}_{kj}{}^i\bm e^k & -\frac{1}{t^2}\bm{\mathcal B}^i\\0 & -\rho\bm{\mathcal B}_j & 0 \end{array}\right)\wedge(\bm e^--\rho\bm e^+)\,.
\end{align}
Here we defined $\bm{\mathcal B}_i={\cal B}_{ij}\bm e^j$, $\bm{\mathcal C}_i=\frac{1}{2}{\cal C}_{ikj}\bm e^j\wedge \bm e^k$, $\bm{\mathcal W}^i{}_j={\cal W}^i{}_{jkl}\bm e^k\wedge \bm e^l$, with
\begin{align}
\label{eq:Bij}
{\cal B}_{ij}&\equiv\p_\rho\psi_{ij}-\psi_{ik}\psi_j{}^k-\hat\nabla_i\varphi_j-2\rho\varphi_i\varphi_j\,,\\
{\cal C}_{ikj}&\equiv\hat\nabla_j\psi_{ki}-\hat\nabla_k\psi_{ji}-2\rho\varphi_if_{jk}\,,\\
{\mathcal W}^i{}_{jkl}&\equiv\bar R^i{}_{jkl}+\delta_j{}^if_{kl}-\delta_k{}^i\psi_{lj}-\psi_k{}^i\gamma_{lj}+\delta_l{}^i\psi_{kj}+\psi_l{}^i\gamma_{kj}+2\rho(\psi_k{}^i\psi_{lj}-\psi_l{}^i\psi_{kj}-\psi_j{}^if_{kl})\,,
\end{align}
where $\hat\nabla$ is the metricity free connection on the distribution $\{\un D_i\}$ introduced in \eqref{eq:hatnabla}, and
\begin{align}
\bar R^{i}{}_{jkl}=&D_{k}\ti \Gamma^{i}{}_{lj}-D_{l}\ti \Gamma^{i}{}_{kj}+\ti \Gamma^{i}{}_{km}\ti \Gamma^{m}{}_{lj}-\ti \Gamma^{i}{}_{lm}\ti \Gamma^{m}{}_{kj}\,.
\end{align}
Plugging in \eqref{Pgf} and \eqref{g4} from the $\rho$-expansion of $\gamma_{ij}$, one obtains at the leading order
\begin{align}
\label{eq:BCW}
{\cal B}_{ij}^{(0)}&=\hat\Omega_{ij}^{(1)}\,,\qquad{\cal C}_{ikj}^{(0)}=\hat C_{ijk}\,,\qquad
{\mathcal W}_{(0)}^i{}_{jkl}=\hat W_{(0)}^i{}_{jkl}\,.
\end{align}
Therefore, when pulled back from $\tilde M$ to $M$ the Riemann curvature of the Weyl-ambient space gives us on $M$ the Weyl tensor $\hat W_{(0)}^i{}_{jkl}$, Weyl-Cotton tensor $\hat C_{ijk}$ and the tensor $\hat\Omega_{ij}^{(1)}$ we obtained in \eqref{eq:Omega1} as follows:
\begin{align}
\label{eq:ambRiem}
\tilde R_{-ij-}|_{\rho=0,t=1}=\hat{\Omega}^{(1)}_{ij}\,,\qquad \tilde R_{-ijk}|_{\rho=0,t=1}=\hat C_{ijk}\,,\qquad\tilde R_{ijkl}|_{\rho=0,t=1}=\hat{W}^{(0)}_{ijkl}\,.
\end{align}
The corresponding curvature 2-form at $\rho=0, t=1$ can be expressed as
\begin{align}
\label{ea:ambRiem0}
\tilde{\bm R}_{(0)}^{M}{}_{N}=&\left(\begin{array}{ccc}0 & -\hat{\bm{C}}_j & 0  \\0 & \hat{\bm{W}}_{(0)}^i{}_{j} & \hat{\bm{C}}^i \\0 & 0 & 0 \end{array}\right)+\left(\begin{array}{ccc}0 &\hat{\bm\Omega}^{(1)}_j & 0 \\0 &\hat{C}_{kj}{}^i\bm e^k & -\hat{\bm\Omega}_{(1)}^i\\0 & 0 & 0 \end{array}\right)\wedge\bm e^-\,,
\end{align}
where $\hat{\bm\Omega}^{(1)}_i ={\hat\Omega}^{(1)}_{ij}\bm e^j$, $\hat{\bm C}_i=\frac{1}{2}\hat{C}_{ikj}\bm e^j\wedge \bm e^k$, $\hat{\bm W}_{(0)}^i{}_j={\hat W}^i{}_{jkl}\bm e^k\wedge \bm e^l$. As expected, the first matrix in \eqref{ea:ambRiem0}, which represents the components of $\tilde{\bm R}_{(0)}^{M}{}_{N}$ in the $\bm e^i\wedge\bm e^j$ directions, is the curvature 2-form of the Cartan normal connection. The $\bm e^i\wedge\bm e^-$ components, on the other hand, give rise to the tensor $\hat \Omega^{(1)}_{ij}$ on $M$, which is expected to be the first extended Weyl-obstruction tensor. This implies that we can define the extended Weyl-obstruction tensors on the $d$-dimensional manifold $M$ by means of the $(d+2)$-dimensional Weyl-ambient space. Before getting to that, we first provide the following proposition, which shows that diffeomorphism-covariant tensors in the Weyl-ambient space are Weyl-covariant tensors when pulled back to $M$.

\begin{prop}
\label{prop1}
Let $IJKLM_{1}\dots M_{r}$ be a list of indices, $s_{+}$ of which are $+$, $s_{M}$ of which correspond to $x^i$, and $s_{-}$ of which are $-$, then under the ambient Weyl diffeomorphism \eqref{eq:Weyldiff}, we have
\begin{align}
\label{eq:prop1}
\tilde\nabla_{M_1}\cdots\tilde\nabla_{M_r} \tilde R'_{IJKL}|_{\rho'=0,t'=1}={\cal B}(x)^{2s_{-}-2} \tilde\nabla_{M_1}\cdots\tilde\nabla_{M_r} \tilde R_{IJKL}|_{\rho=0,t=1}\,.
\end{align}
\end{prop}
\begin{proof}
Under the ambient Weyl diffeomorphism \eqref{eq:Weyldiff}, the vector basis $\{\un D_P\}$ transforms as
\be
\un D'_+={\cal B}(x)^{-1}\un D_+\,,\qquad\un D'_i=\un D_i\,,\qquad\un D'_-={\cal B}(x)\un D_-\,,
\ee
where
\begin{align}
\label{eq:D+-i'}
\un D'_+&=\un\p'_t-\frac{\rho'}{t'}\un\p'_\rho\,,\qquad\un D'_i=\un\p'_i-t'a'_i(x',\rho')\un\p'_t+2\rho' a'_i(x',\rho')\un\p'_\rho\,,\qquad\un D'_-=\frac{1}{t'}\un\p'_\rho\,.
\end{align}
Hence, 
\begin{align}
\label{eq:R'R}
\tilde\nabla_{M_1}\cdots\tilde\nabla_{M_r} \tilde R'_{IJKL}|_{\rho'=0,t'={\cal B}(x)}={\cal B}(x)^{s_{-}-s_+} \tilde\nabla_{M_1}\cdots\tilde\nabla_{M_r} \tilde R_{IJKL}|_{\rho=0,t=1}\,.
\end{align}
Noticing the fact that $\tilde g$ is homogeneous in $t$ with degree 2, and considering the $t$-dependence of $\un D_+$ and $\un D_-$ in \eqref{eq:D+-i}, we have
\begin{align}
\label{eq:R'R'}
\tilde\nabla_{M_1}\cdots\tilde\nabla_{M_r} \tilde R'_{IJKL}|_{\rho'=0,t'=1}={\cal B}(x)^{s_{-}+s_+-2} \tilde\nabla_{M_1}\cdots\tilde\nabla_{M_r} \tilde R'_{IJKL}|_{\rho'=0,t'={\cal B}(x)}\,.
\end{align}
Combining \eqref{eq:R'R} and \eqref{eq:R'R'} we obtain \eqref{eq:prop1}.
\end{proof}
Since diffeomorphism-covariant tensors can be constructed out of the Riemann tensor and its covariant derivatives \cite{Graham2007jet}, this proposition implies that the pullback of an ambient tensor $\tilde T_{M_1\cdots M_k}$ to $M$:
\be
T_{i_1\cdots i_{s_M}}\equiv \tilde T_{M_1\cdots M_k}|_{\rho=0,t=1},
\ee
is Weyl covariant with Weyl weight $2s_{-}-2$, where among the indices $M_1\cdots M_k$, $s_-$ of which are $-$, and $s_M$ of which correspond to $x^i$. For instance, from Proposition \ref{prop1} we can see that the tensors obtained in \eqref{eq:ambRiem} are all Weyl-covariant tensors on $M$, and the Weyl weights of $\hat{\Omega}^{(1)}_{ij}$, $\hat C_{ijk}$ and $\hat{W}^{(0)}_{ijkl}$ can be read off to be 2, 0, and $-2$, respectively, which are indeed the correct Weyl weights (see Table~\ref{table1} in Section~\ref{sec:Weyl}).
\par
As a special kind of Weyl-covariant tensor, we introduce the extended Weyl-obstruction tensors as follows.
\begin{defn}\label{def1}
Suppose $k$ is a positive integer. The $k^{th}$ extended Weyl-obstruction tensor $\hat \Omega^{(k)}_{ij}$ is defined as
\begin{align}
\label{eq:def2}
\hat \Omega^{(k)}_{ij}= \underbrace{\tilde\nabla_-\cdots\tilde\nabla_-}_{k-1} \tilde R_{- ij- }|_{\rho=0,t=1}.
\end{align}
\end{defn}
Some properties of Weyl-obstruction tensors can be readily seen from the above definition. From the symmetry of the Riemann tensor we can see that $\hat \Omega^{(k)}_{ij}$ is a symmetric tensor. It follows from Proposition~\ref{prop1} that $\hat \Omega^{(k)}_{ij}$ is Weyl covariant with Weyl weight $2k$. Also, from the Ricci-flatness condition we obtain that $\tilde g^{IJ}\tilde\nabla_{M_1}\cdots\tilde\nabla_{M_r}\tilde R_{IKJL}=0$, which gives rise to $\gamma_{(0)}^{ij}\hat \Omega^{(k)}_{ij}=0$, i.e.\ $\hat \Omega^{(k)}_{ij}$ is traceless.
\par
We have seen in \eqref{eq:BCW} that when $k=0$, this definition gives the $\hat\Omega^{(1)}_{ij}$ in \eqref{eq:Omega1}. By computing $\tilde\nabla_-\tilde R_{-ij-}$, one also finds that $\hat \Omega^{(2)}_{ij}$ defined in this way gives exactly the expression in \eqref{eq:Omega2} (see Appendix \ref{app:Null}). Notice again that before introducing Definition \ref{def1}, although we referred to $\hat\Omega^{(k)}_{ij}$ as the $k^{th}$ extended Weyl-obstruction tensor (especially in Section~\ref{sec:WOTbulk}), we should simply regard it as denoting the pole of $\gamma^{(k+1)}_{ij}$ at $d=2k+2$. Since there is an ambiguity when the pole is shifted by a finite term, that should not be treated as a precise definition for extended Weyl-obstruction tensors. Now the $\hat\Omega^{(k)}_{ij}$ defined through the Weyl-ambient space is uniquely determined. The proposition below will show that each $\hat\Omega^{(k)}_{ij}$ defined through the Weyl-ambient space indeed has a pole at $d=2k+2$, whose residue is the same as the pole in $\gamma^{(k+1)}_{ij}$. Therefore, the ambiguity of the pole in $\gamma^{(k+1)}_{ij}$ can be fixed by taking it to be the extended Weyl-obstruction tensor in Definition \ref{def1}. See the following proposition:

\begin{prop}
\label{prop:obspole}
Let $k\geqslant 2$ be an integer. Both the extended Weyl-obstruction tensor $\hat\Omega^{(k-1)}_{ij}$ and $\gamma^{(k)}_{ij}$ in the expansion \eqref{eq:gexpan} have a simple pole at $d=2k$. The residues satisfy
\be
\textnormal{Res}_{d=2k}\hat\Omega^{(k-1)}_{ij}=\frac{k!}{2}\textnormal{Res}_{d=2k}\gamma^{(k)}_{ij}\,.
\ee
More specifically, $\hat\Omega^{(k-1)}_{ij}$ has the following form:
\be
\hat\Omega^{(k-1)}_{ij}=\frac{(-1)^{k-1}\Gamma(d/2-k)}{2^{k-1}\Gamma(d/2-1)}(\Delta_{(0)}^{k-1}\hat P_{ij}-\Delta_{(0)}^{k-2}\hat\nabla^{(0)}_i\hat\nabla^{(0)}_kP_{j}{}^k+\cdots)\,,
\ee
where $\Delta_{(0)}\equiv\hat\nabla^{(0)}_k\hat\nabla_{(0)}^k$ and the ellipsis represents the terms with fewer number of $\hat\nabla^{(0)}$. The terms inside the brackets represent the Weyl-obstruction tensor.
\end{prop}
\begin{proof}
First, let us show that $\gamma^{(k\geqslant2)}_{ij}$ has a pole at $d=2k$, which has the form
\be
\label{eq:gammapole}
\gamma^{(k)}_{ij}=\frac{(-1)^{k-1}\Gamma(d/2-k)}{2^{k-2}k!\Gamma(d/2-1)}(\Delta_{(0)}^{k-1}\hat P_{ij}-\Delta^{k-2}_{(0)}\hat\nabla^{(0)}_i\hat\nabla^{(0)}_kP_{j}{}^k+\cdots)\,.
\ee
We have seen this previously for $k=2$ and $3$. Using mathematical induction, now we will prove the following equation for $k\geqslant2$:
\be
\label{eq:psipole}
(d-2k)\p^{k-1}_{\rho}\psi_{ji}=\frac{(-1)^{k-1}\Gamma(d/2-k+1)}{2^{k-2}\Gamma(d/2-1)}(\Delta^{k-1}\psi_{ji}-\Delta^{k-2}\hat\nabla_i\hat\nabla_k\psi^k{}_{j}+{\cdots})+2\rho\p^{k}_{\rho}\psi_{ij}+O(\rho)\,,
\ee
where $\Delta\equiv\hat\nabla_k\hat\nabla^k$. This relation leads to \eqref{eq:gammapole} when $\rho=0$ since $\psi_{ij}=\frac{1}{2}(\p_\rho\gamma_{ij}+f_{ij})$ (the $f_{ij}$ in the left-hand side are combined in the ellipsis). Differentiating the Ricci-flatness condition of the form \eqref{eq:Rij0} with respect to $\rho$ and use the expression \eqref{eq:Rijprho} we can see that
\be
\label{eq:psipole2}
(d-4)\p_{\rho}\psi_{ji}=-(\Delta\psi_{ji}-\hat\nabla_i\hat\nabla_k\psi^k{}_{j}+{\cdots})+2\rho\p^2_{\rho}\psi_{ij}+O(\rho)\,,
\ee
which is \eqref{eq:psipole} in the case $k=2$. Now we assume \eqref{eq:psipole} holds for $k=n$. Differentiating both sides of \eqref{eq:psipole2}  for $n-1$ times with respect to $\rho$ yields
\be
(d-2n-2)\p^{n}_{\rho}\psi_{ji}=-\p^{n-1}_\rho(\Delta\psi_{ji}-\hat\nabla_i\hat\nabla_k\psi^k{}_{j}+\cdots)+2\rho\p^{n+1}_{\rho}\psi_{ij}+O(\rho)\,.
\ee
Note that $\p_\rho$ produces two $\hat\nabla$ when acting on $\psi$, while it only produces one $\hat\nabla$ when acting on $\tilde\Gamma^i{}_{jk}$, and thus when we commute $\p_\rho$ with $\hat\nabla$, the new terms only contribute to the ellipsis. Hence,
\begin{align}
&(d-2n-2)\p^{n}_{\rho}\psi_{ji}=-(\Delta\p^{n-1}_\rho\psi_{ji}-\hat\nabla_i\hat\nabla^k\p^{n-1}_\rho\psi_{kj}+\cdots)+2\rho\p^{n+1}_{\rho}\psi_{ij}+O(\rho)\nn\\
={}&\frac{(-1)^{n}\Gamma(d/2-n)}{2^{n-1}\Gamma(d/2-1)}(\Delta^{n}\psi_{ji}-\Delta^{n-1}\hat\nabla_i\hat\nabla_k\p_\rho\psi^k{}_{j}+\cdots)+2\rho\p^{n+1}_{\rho}\psi_{ij}+O(\rho)\,,
\end{align}
where {we used \eqref{eq:R+i0} and} the assumption that \eqref{eq:psipole} holds for $k=n$. This is exactly \eqref{eq:psipole} for $k=n+1$, and thus \eqref{eq:psipole} is proved for any $k\geqslant2$. Therefore, at $\rho=0$ we have
\be
\p^k_{\rho}\psi_{ji}|_{\rho=0}=\frac{(-1)^{k-1}\Gamma(d/2-k-1)}{2^{k}\Gamma(d/2-1)}(\Delta_{(0)}^{k}\hat P_{ij}-\Delta_{(0)}^{k-1}\hat\nabla^{(0)}_i\hat\nabla^{(0)}_k\hat P_{j}{}^k+\cdots)\,.
\ee
From \eqref{eq:curv2form} we can read off that
\be
\tilde R_{-ij-}={\cal B}_{ij}=\p_\rho\psi_{ij}-\psi_{ik}\psi_j{}^k-\hat\nabla_i\varphi_j-2\rho\varphi_i\varphi_j\,.
\ee
Hence, the Weyl-obstruction tensor $\hat\Omega^{(k)}_{ij}$ has the form
\begin{align}
\hat\Omega^{(k-1)}_{ij}&=\underbrace{\tilde\nabla_-\cdots\tilde\nabla_-}_{k-2}\tilde R_{-ij-}|_{\rho=0,t=1}=\p^{k-1}_\rho\psi_{ij}|_{\rho=0}+\cdots\nn\\
\label{eq:poleOmega}
&=\frac{(-1)^{k-1}\Gamma(d/2-k)}{2^{k-1}\Gamma(d/2-1)}(\Delta_{(0)}^{k-1}\hat P_{ij}-\Delta_{(0)}^{k-2}\hat\nabla^{(0)}_i\hat\nabla^{(0)}_kP_{j}{}^k+{\cdots})\,,
\end{align}
where finite terms at $d=2k$ are shifted into the pole. On the other hand, from \eqref{eq:poleOmega} we also have
\begin{align}
\text{Res}_{d=2k}\hat\Omega^{(k-1)}_{ij}&=\text{Res}_{d=2k}\p^k_\rho\psi_{ij}|_{\rho=0}=\frac{k!}{2}\text{Res}_{d=2k}\gamma^{(k)}_{ij}\,,
\end{align}
where in the second equality we considered that $f_{ij}$ does not contribute to the pole.
\end{proof}
This proposition indicates that both the extended Weyl-obstruction tensor $\hat\Omega^{(k-1)}_{ij}$ and $\gamma^{(k)}_{ij}$ are meromorphic functions, which are holomorphic in the whole complex plane except at even integers $d=4,6,\cdots,2k$. We have seen that the pole at $d=2k$ is a simple pole, while the pole at a lower even dimension could be of higher order. These two tensors only differ by terms that are finite at $d=2k$. Therefore,  we can express $\gamma^{(k)}_{ij}$ in terms of $\hat\Omega^{(k-1)}_{ij}$ plus finite terms as we have seen for $k=1,2$ in \eqref{g4} and \eqref{g6}.
\par
In the next subsection, we will introduce the extended Weyl-obstruction tensors in the second order formalism \`a la \cite{Fefferman:2007rka} and show that the two definitions are equivalent.

\subsection{Second-Order Formalism}
\label{sec:secondorder}
In Subsection \ref{sec:topdown} we have seen that Weyl-obstruction tensors can be defined as the derivatives of the ambient Riemann tensor in the first order formalism. In this subsection we will follow the setup of the present section in the second order formalism and show that appropriate ambient tensors constructed from the Weyl-ambient Riemann tensor on $\tilde M$ behave as Weyl-covariant tensors on $M$, through which Weyl-obstruction tensors can again be defined as a special case.   Then we will show that the Weyl-obstruction tensors defined in this way agree with the Weyl-obstruction tensors we defined previously in Definition \ref{def1}. 
\par

We have proven in Subsection \ref{sec:bottomup} that for any pair of $(g,a)$ on $M$, there exists a unique Weyl-ambient space $(\tilde M,\tilde g)$ for the Weyl manifold $(M,[g,a])$ where $\tilde g$ has the form of \eqref{Weyl_ambient}. In Subsection \ref{sec:topdown} we saw that the ambient Weyl diffeomorphism
\be
\label{eq:ambientdiff}
(t',x'^i,\rho')=({\cal B}(x)t,x^i,{\cal B}^{-2}(x)\rho)
\ee
induces a Weyl transformation on $M$. Therefore, to find a Weyl-covariant tensor on $(M,[g,a])$, we can find an ambient tensor which is covariant under an ambient Weyl diffeomorphism, and its pullback on $M$ will be Weyl covariant.
\par
The first main result of this subsection is the following proposition. This  provides  the Weyl transformations of tensors constructed from covariant derivatives of the Riemann tensor of a Weyl-ambient metric, from which we can see which tensors are Weyl covariant when pulled back to $M$.

\begin{prop}\label{Proposition_6.5_Graham}
Suppose $(\tilde M,\tilde g)$ is the Weyl-ambient space for $(M,[g,a])$, and let $(g,a)$ and $(g',a')$ be two representatives of $[g,a]$, with $g'_{ij}={\cal B}^{-2}g_{ij}$ and $a'_i=a_i-\p_i\ln{\cal B}$.
Let $IJKLM_{1}\dots M_{r}$ be a list of indices, $s_{0}$ of which are $0$, $s_{M}$ of which are $x^i$ on $M$, and $s_{\infty}$ of which are $\infty$. Then, the following components of the covariant derivatives of the Riemann tensor $\tilde R_{ABCD}$ of $\ti g$ in the trivialization defined by $g$ satisfy the transformation law
\begin{equation}\label{Proposition_6.5_Graham_main}
\ti R'_{IJKL;M_{1}\cdots M_{r}}|_{\rho'=0,t'=1}={\cal B}(x)^{2(s_{\infty }-1)} \ti R_{ABCD;F_{1}\cdots F_{r}}|_{\rho=0,t=1}p^{A}{}_{I}\cdots p^{F_{r}}{}_{M_{r}}
\end{equation}
under an ambient Weyl diffeomorphism \eqref{eq:ambientdiff}, where $p^{A}{}_{I}$ is the matrix 
\begin{equation}\label{p_matrix}
p^{I}{}_{J}= \begin{blockarray}{cccc}
 & 0 & j & \infty  \\
\begin{block}{c(ccc)}
  0 & 1& \Upsilon_{j} &0\\
   i & 0  & \delta^{i}{}_{j}& 0 \\
  \infty & 0&0& 1\\
\end{block}
\end{blockarray}\,\,,
\end{equation}
and $\Upsilon(x)\equiv - \ln {\cal B}(x)$, $\Upsilon_{j}\equiv \p_{j}\Upsilon(x)$.
$\ti R'_{IJKL;M_{1}\dots M_{r}}$ denotes covariant derivatives of the Riemann tensor of $\tilde g$ in the coordinates $X'^{I}=(t', x'^{i},\rho')$ given by the trivialization provided by $g'$.
\end{prop}
\begin{proof}
The logic for the proof of this Proposition follows the proof of Proposition 6.5 in \cite{Fefferman:2007rka} closely. We start by observing that the ambient Weyl diffeomorphism $\psi:(t',x'^{i},\rho')\mapsto(t,x^{i},\rho)$ has the following properties:
\begin{equation}\label{psi_requirements}
\psi(t',x'^{i},0)= \big(t' \E^{\Upsilon (x)},x'^{i},0 \big)\,,\qquad\psi^{*}\ti g|_{\rho'=0}= 2t'\td\rho'\td t'+ t'^2 g'_{ij}\td x'^{i}\td x'^{j}+ 2t'^2 a'_{i}\td x'^{i}\td\rho'\,,
\end{equation}
where the Weyl-ambient metric $\ti g$ has the form of \eqref{Weyl_ambient}, and $g'_{ij}= {\cal B}(x)^{-2}g_{ij}$, $a'_{i}= a_{i}+\Upsilon_i$. The Jacobian $(\psi)^{A}{}_{I}= \left(\frac{\p X^{I}}{\p X'^{J}}\right)$ of this diffeomorphism is
\begin{equation}\label{Jacobian_Weyl_diff}
(\psi)^{I}{}_{J}\equiv 
\left(\begin{array}{ccc}
\psi^{t}{}_{t'} & \psi^{t}{}_{j'} & \psi^{t}{}_{\rho'}\\
\psi^{i}{}_{t'} & \psi^{i}{}_{j'} & \psi^{i}{}_{\rho'}\\
\psi^{\rho}{}_{t'}& \psi^{\rho}{}_{j'}& \psi^{\rho}{}_{\rho'}
\end{array}\right)
=
\left(\begin{array}{ccc}
\E^{\Upsilon(x) } &t' \E^{\Upsilon (x)}  \Upsilon_{j} & 0\\
0 & \delta^{i}{}_{j}& 0\\
0&-2\rho' \E^{-2\Upsilon(x)}\Upsilon_{j}& \E^{-2\Upsilon(x)}
\end{array}\right)\,,
\end{equation}
where $\Upsilon(x)\equiv -\text{ln}{\cal B}(x)$ and $\Upsilon_{i}\equiv \p_{i}\Upsilon(x)$.
At $\rho'=0$ the Jacobian matrix \eqref{Jacobian_Weyl_diff} reads
\begin{equation}\label{Jacobian_Weyl_diff_rho=0}
(\psi)^{A}{}_{I}|_{\rho' =0}= 
\left(\begin{array}{ccc}
\E^{\Upsilon(x) } & t' \E^{\Upsilon (x)}  \Upsilon_{j} & 0\\
0 & \delta^{i}{}_{j}& 0\\
0&0& \E^{-2 \Upsilon(x)}
\end{array}\right)\,.
\end{equation}
The above matrix can be written as the following matrix product:
\begin{equation}
(\psi)^{A}{}_{I}|_{\rho' =0} =d_{1}p d_{2}\,,
\end{equation}
with 
\begin{equation}\label{pd_matrix}
p^{I}{}_{J}=
\left(\begin{array}{ccc}
1& \Upsilon_{j} &  0\\
0 & \delta^{i}{}_{j}& 0\\
0&0& 1
\end{array}\right),\quad
d_{1}=
\left(\begin{array}{ccc}
 t' \E^{\Upsilon(x)}& 0 &  0\\
0 & \delta^{i}{}_{j}& 0\\
0&0& 1
\end{array}\right),\quad
d_{2}=
\left(\begin{array}{ccc}
 t'^{-1} & 0 &  0\\
0 & \delta^{i}{}_{j}& 0\\
0&0& \E^{-2\Upsilon(x)}
\end{array}\right).
\end{equation}

Since the Weyl-ambient metric is homogeneous of degree 2 under dilatations $\delta_{s}^{*}\ti g= s^{2}\ti g$, it follows that the left-hand side of \eqref{Proposition_6.5_Graham_main} satisfies
\begin{equation}\label{homog_Riemann}
\ti R'_{IJKL;M_{1}\cdots M_{r}}|_{\rho'=0,t'=1}= {\cal B}(x)^{s_{0}-2}\ti R'_{IJKL;M_{1}\cdots M_{r}}|_{\rho'=0,t'={\cal B}(x)}\,.
\end{equation}
Under the ambient Weyl diffeomorphism \eqref{eq:ambientdiff} the covariant derivatives of the ambient Riemann curvature components transform tensorially as
\begin{equation}\label{WFG_tensorially}
\ti R'_{IJKL;M_{1}\cdots M_{r}}|_{\rho',t'}= \ti R_{ABCD;F_{1}\cdots F_{r}}|_{\rho,t} (\psi)^{A}{}_{I}\cdots (\psi)^{F_{r}}{}_{M_{r}}\,.
\end{equation}
Evaluating both sides of \eqref{WFG_tensorially} at $\rho'=0$, $t'=e^{-\Upsilon(x)}$ and using \eqref{pd_matrix} we have
\begin{equation}\label{Riemann_tensor_tensorially} 
\ti R'_{IJKL;M_{1}\cdots M_{r}}|_{\rho'=0,t'=e^{-\Upsilon(x)}}={\cal B}(x)^{2s_{\infty}- s_{0}} \ti R_{ABCD;F_{1}\cdots F_{r}}|_{\rho=0,t=1}p^{A}{}_{I}\cdots p^{F_{r}}{}_{M_{r}}\,.
\end{equation}
Plugging \eqref{homog_Riemann} into \eqref{Riemann_tensor_tensorially}, we  obtain \eqref{Proposition_6.5_Graham_main}.
\end{proof}

Theorem \ref{Proposition_6.5_Graham} helps us to find  Weyl-covariant tensors on $(M,[g,a])$. First let us look at the case without derivatives. In the coordinate basis, the nonvanishing components of the Weyl-ambient Riemann tensor $\tilde R_{IJKL}$ are $\ti R_{\infty jk\infty}$, $\ti R_{\infty jkl}$ and $\ti R_{ijkl}$.  Evaluating at $\rho=0$ and $t=1$, they are

\begin{align}
\label{eq:Rijkl}
\ti R_{\infty jk\infty}|_{\rho=0,t=1}= \hat\Omega^{(1)}_{jk}\,,\qquad\ti R_{\infty jkl}|_{\rho=0,t=1}= \hat C_{jkl}\,,\qquad\ti R_{ijkl}|_{\rho=0,t=1}= \hat W_{ijkl}\,.
\end{align}
Here $\hat C_{jkl}$ and $\hat W_{ijkl}$ are the Weyl-Cotton tensor and the Weyl curvature tensor on $M$, respectively, and $\hat\Omega^{(1)}_{jk}$ for now simply denotes the tensor defined in \eqref{eq:Omega1}. Then, applying \eqref{Proposition_6.5_Graham_main} we get $\hat C'_{jkl}= \hat C_{jkl}$ and $\hat W'_{ijkl}= {\cal B}^{-2}(x) \hat W_{ijkl}$  under Weyl transformation as expected, we can also read off from \eqref{Proposition_6.5_Graham_main} that the Weyl weight of $\hat\Omega^{(1)}_{jk}$ is $+2$.

\par
Now we will define Weyl-obstruction tensors as the derivatives of $\ti R_{\infty jk\infty}$. 
 \begin{defn}\label{Weyl_ext_Obst}
Suppose $k$ is a positive integer. The $k^{th}$ extended Weyl-obstruction tensor $\hat \Omega^{(k)}_{ij}$ is defined as
\begin{equation}
\label{eq:def1}
\hat \Omega^{(k)}_{ij}= \ti R_{\infty ij\infty ;\underbrace{\scriptstyle \infty\cdots\infty}_{k-1}}|_{\rho=0,t=1}.
\end{equation}
For $k=1$ we can see from \eqref{eq:Rijkl} that $\ti R_{\infty jk\infty}|_{\rho=0,t=1}= \hat\Omega^{(1)}_{jk}$ is indeed the first extended Weyl-obstruction tensor.
\end{defn}

From the symmetry of the Weyl-ambient Riemann tensor we can immediately see that $\hat \Omega^{(k)}_{ij}$ given by Definition \ref{Weyl_ext_Obst} is symmetric. From the Ricci-flatness condition $\ti Ric(\ti g)=0$ and the fact that $\ti R_{0 IJK}=0$, we can see that $\hat \Omega^{(k)}_{ij}$ is traceless. Now we will show another important property of the extended Weyl-obstruction tensors defined in this way, namely that they are Weyl covariant. 
\begin{lemma}
\label{lem:lemmaR0}
The components of the Riemann tensor of the Weyl-ambient metric $\tilde g$ satisfy
\be
\label{eq:lemmaR0}
\tilde R_{IJK0;M_1\cdots M_r}=-\frac{1}{t}\sum^r_{s=1}\tilde R_{IJKM_s;M_1\cdots\hat{M}_s\cdots M_r}\,,
\ee
where $\hat{M}_s$ means to remove $M_s$ from the indices.
\end{lemma}
\begin{proof}
Computing the Christoffel symbols of the Weyl-ambient metric $\tilde g$ in \eqref{Weyl_ambient}, one finds $\tilde\Gamma^i{}_{j0}=\frac{1}{t}\delta^i{}_j$ and $\tilde\Gamma^\infty{}_{\infty0}=\frac{1}{t}$. Differentiating $\un{\cal T}=t\un\p_t$ we have ${\cal T}^I{}_{;J}=\delta^I{}_{J}$ and ${\cal T}^I{}_{;JK}=0$, then
\begin{align*}
({\cal T}^L\tilde R_{IJKL})_{;M_1\cdots M_r}&=\tilde R_{IJKM_1;M_2\cdots M_r}+({\cal T}^L\tilde R_{IJKL,M_1})_{;M_2\cdots M_r}\\
&=\tilde R_{IJKM_1;M_2\cdots M_r}+\tilde R_{IJKM_2;M_2\cdots M_r}+({\cal T}^L\tilde R_{IJKL,M_1M_2})_{;M_3\cdots M_r}\\
&=\cdots\\
&=\tilde R_{IJKM_1;M_2\cdots M_r}+\cdots+\tilde R_{IJKM_r;M_1\cdots M_{r-1}}+{\cal T}^L\tilde R_{IJKL;M_1\cdots M_r}\,.
\end{align*}
The left-hand side of this equation vanishes since $R_{IJK0}=0$, and thus the above equation leads to \eqref{eq:lemmaR0}.
\end{proof}

\begin{prop}
The extended Weyl-obstruction tensor $\hat \Omega^{(k)}_{ij}$ defined in \eqref{eq:def1} is a Weyl-covariant tensor with Weyl weight $2k$.
\end{prop}
\begin{proof}
According to Proposition \ref{Proposition_6.5_Graham}, if we choose $(IJKL;M_{1}\dots M_{r})= (\infty, i,j, \infty;\underbrace{\infty\dots \infty}_{(k-1)})$, then $s_{\infty}= k+1$ and under a Weyl transformation we have
\begin{align}
\label{eq:obsWeylcov}
\ti R'_{\infty ij\infty ;\underbrace{\scriptstyle \infty\cdots\infty}_{k-1}}|_{\rho'=0,t'=1}={\cal B}(x)^{2k}\big(\ti R_{\infty ij\infty ;\underbrace{\scriptstyle \infty\cdots\infty}_{k-1}}+\Upsilon_{i}\ti R_{\infty 0j\infty ;\underbrace{\scriptstyle \infty\cdots\infty}_{k-1}}+\Upsilon_{j}\ti R_{\infty i0\infty ;\underbrace{\scriptstyle \infty\cdots\infty}_{k-1}}\big)|_{\rho'=0,t'=1}\,.
\end{align}
It follows from Lemma \ref{lem:lemmaR0} that
\begin{align}
R_{\infty i 0\infty;\underbrace{\scriptstyle \infty\cdots\infty}_{k-1}}=\frac{k-1}{t}R_{\infty i \infty\infty;\underbrace{\scriptstyle \infty\cdots\infty}_{k-2}}=0\,.
\end{align}
Therefore, we obtain from \eqref{eq:obsWeylcov} that $\hat \Omega'^{(k)}_{ij}={\cal B}(x)^{2k}\hat \Omega^{(k)}_{ij}$ under a Weyl transformation, i.e.\ $\hat\Omega^{(k)}_{ij}$ is a Weyl-covariant tensor with Weyl weight $2k$.
\end{proof}

Finally, we would like to show that Definition \ref{Weyl_ext_Obst} and Definition \ref{def1} are equivalent; that is, the Weyl-obstruction tensors defined by the derivatives of the ambient Riemann tensor in the frame $\{\bm e^+,\bm e^i,\bm e^-\}$ and the coordinate basis $\{\td t,\td x^i,\td\rho\}$ are equivalent. To start, let us look at the transformation between $\{\bm e^+,\bm e^i,\bm e^-\}$ and the coordinate basis $\{\td t,\td x^i,\td\rho\}$:
\begin{align}
\label{eq:basistrans}
\left(\begin{array}{c}\bm e^+ \\\bm e^j \\\bm e^-\end{array}\right)=\left(\begin{array}{ccc}1 &ta_i & 0 \\0 & \delta^j{}_i & 0 \\ \rho & -\rho t a_i & t\end{array}\right)\left(\begin{array}{c}\td t \\\td x^i \\ \td\rho\end{array}\right)\,.
\end{align}
Denote the transformation matrix as $\Lambda$, i.e.\ $\bm e^{J}=\Lambda^{J}{}_{I'}\td x^{I'}$ ($J=\{+,i,-\}$, $I'=\{0,i,\infty\}$), then the inverse matrix reads
\begin{align}
\Lambda^{-1}=\left(\begin{array}{ccc}1 &-ta_j & 0 \\0  & \delta^i{}_j & 0\\ -\frac{\rho}{t} &2\rho a_j & \frac{1}{t}\end{array}\right)\,.
\end{align}

Comparing \eqref{eq:ambRiem} and \eqref{eq:Rijkl}, we can see that the components $R_{ijkl}$, $R_{-ijk}$ and $R_{-ij-}$ in the null frame match the corresponding components $R_{ijkl}$, $R_{\infty ijk}$ and $R_{\infty ij\infty}$ in the coordinate basis when $\rho=0$ and $t=1$. Now let us show that any Weyl-obstruction tensor defined in \eqref{eq:def2} is equivalent to that in \eqref{eq:def1}. First, notice that although the components $\tilde R_{-+MN}$ of $\tilde R_{IJKL}$ in the frame $\{+,i,-\}$ vanish, the components $\ti\nabla_{P} \ti R_{-+MN}$ are not necessarily zero. (Using the notation in Subsection \ref{sec:topdown}, here we denote $\tilde\nabla_{\un D_P}$ as $\tilde\nabla_P$ for $P=+,i,-$.) The following lemma will be used in the proof of Proposition \ref{prop:obsequiv}.
\begin{lemma}
\label{lem:Riem}
$\tilde\nabla_{P}\underbrace{\tilde\nabla_-\cdots\tilde\nabla_-}_{n}\tilde R_{-+MN}=-\frac{1}{t}\delta^{i}{}_{P}\underbrace{\tilde\nabla_-\cdots\tilde\nabla_-}_{n}\tilde R_{-iMN}$ for any integer $n\geqslant0$.
\end{lemma}
\begin{proof}
See Appendix \ref{app:Lemproof}.
\end{proof}

\begin{prop}
\label{prop:obsequiv}
$\tilde R_{\infty ij\infty;\underbrace{\scriptstyle\infty\cdots\infty}_{n}}=t^{n+2}\underbrace{\tilde\nabla_-\cdots\tilde\nabla_-}_{n}\tilde R_{-ij-}$ for any integer $n\geqslant0$.
\end{prop}

\begin{proof}
For $n=0$ one can see this readily from \eqref{eq:Rijkl}. Since $\un\p_{\tilde N'}=\Lambda^{M}{}_{N'}\un D_{M}$, for $n\geqslant1$ the left-hand side of the above equation can be written as (primes are dropped for simplicity)
\begin{align}
\label{eq:Riem1}
\tilde R_{\infty ij\infty;\underbrace{\scriptstyle\infty\cdots\infty}_{n}}&=\Lambda^{M_1}{}_{\infty}\cdots\Lambda^{M_n}{}_{\infty}\Lambda^{K}{}_{\infty}\Lambda^{I}{}_{i}\Lambda^{J}{}_{j}\Lambda^{L}{}_{\infty}\tilde\nabla_{M_1}\cdots\tilde\nabla_{M_n}\tilde R_{KIJL}\nn\\
&=t^{n+2}\Lambda^{I}{}_{i}\Lambda^{J}{}_{j}\underbrace{\tilde\nabla_{-}\cdots\tilde\nabla_{-}}_{n}\tilde R_{-IJ-}\,,
\end{align}
where $\Lambda^{M}{}_\infty=t\delta^{M}{}_-$ [see \eqref{eq:basistrans}] is used in the second equality. Using the symmetries of the Riemann tensor, we have
\begin{align}
\label{eq:Riem2}
\Lambda^{I}{}_{i}\Lambda^{J}{}_{j}\underbrace{\tilde\nabla_{-}\cdots\tilde\nabla_{-}}_{n}\tilde R_{-IJ-}={}&\underbrace{\tilde\nabla_{-}\cdots\tilde\nabla_{-}}_{n}\tilde R_{-ij-}+\Lambda^{+}{}_{i}\underbrace{\tilde\nabla_{-}\cdots\tilde\nabla_{-}}_{n}\tilde R_{-+j-}\nn\\
&+\Lambda^{+}{}_{j}\underbrace{\tilde\nabla_{-}\cdots\tilde\nabla_{-}}_{n}\tilde R_{-i+-}+\Lambda^{+}{}_{i}\Lambda^{+}{}_{j}\underbrace{\tilde\nabla_{-}\cdots\tilde\nabla_{-}}_{n}\tilde R_{-++-}\nn\\
={}&\underbrace{\tilde\nabla_{-}\cdots\tilde\nabla_{-}}_{n}\tilde R_{-ij-}\,,
\end{align}
where $\Lambda^{i}{}_j=\delta^{i}{}_j$ is used in the first equality and Lemma \ref{lem:Riem} is used in the second equality. Plugging \eqref{eq:Riem2} into \eqref{eq:Riem1} completes the proof.
\end{proof}
From Proposition \ref{prop:obsequiv} we can directly see that the $\hat\Omega^{(k)}_{ij}$ defined in \eqref{eq:def2} is equivalent to \eqref{eq:def1}. Therefore, the descriptions of the Weyl-obstruction tensors in the first order and second order formalisms are equivalent. Each of these two formalisms have their own advantages. The first order formalism is suited for the top-down approach as the metric $\tilde g$ has a simple form in the dual frame $\{\bm e^I\}$. It is also more convenient to construct Weyl-covariant tensors in the first order formalism since \eqref{eq:prop1} gives a covariant transformation while \eqref{Proposition_6.5_Graham_main} has the matrix $p$ with an off-diagonal element. On the other hand, the second order formalism is designed for the bottom-up approach, as one can evaluate the initial value problem more naturally in the coordinate basis.

\section{Discussion}
\label{sec:conclu}
So far in this thesis we have generalized the ambient construction for conformal manifolds to that for Weyl manifolds. Inspired by the WFG gauge for ALAdS \cite{Ciambelli:2019bzz}, we introduced the Weyl-ambient metric $\tilde g$ in \eqref{Weyl_ambient}. From a top-down perspective we showed how the Weyl-ambient space $(\tilde M,\tilde g)$ induces a Weyl geometry on a codimension-2 manifold $M$. The metric $\tilde g$ and the LC connection on $\tilde M$ give rise to a Weyl class $[\gamma^{(0)},a^{(0)}]$ on $M$, in which a representative includes an induced metric $\gamma^{(0)}_{ij}$ together with a Weyl connection $a^{(0)}_i$. The ambient Weyl diffeomorphisms on $\tilde M$ act as Weyl transformations on the $M$. This enhances the codimension-2 conformal geometry in the usual ambient construction to a Weyl geometry $(M,[\gamma^{(0)},a^{(0)}])$. 
\par
From a bottom-up perspective, we formulated the $(d+2)$-dimensional Weyl-ambient space from a $d$-dimensional Weyl manifold $(M,[g,a])$. We first introduced a Weyl structure ${\cal P}_W$ on $M$ together with a Weyl connection. We then generalized the definition of ambient spaces to Weyl-ambient spaces, and proved that any Weyl-ambient space can be put in Weyl-normal form by a diffeomorphism. Besides assigning the Weyl connection $a_{i}$ on ${\cal P}_W$, the $\rho$-coordinate lines of a Weyl-ambient space in Weyl-normal form are not required to be geodesics but can acquire an acceleration $\un{\cal A}$. By taking the Weyl structure as an initial surface, we have shown that there exists a unique Weyl-ambient space in Weyl-normal form for any given Weyl manifold provided the data $(g_{ij},a_i,\un{\cal A})$ is given. The metric generated order by order from the initial value problem is exactly the $\tilde g$ we introduced in \eqref{Weyl_ambient} from the top-down approach, where $g_{ij}$ corresponds to $\gamma^{(0)}_{ij}$, and $(a_i,\un{\cal A})$ corresponds to $a_i(x,\rho)$. 
\par
We provided a detailed analysis of Weyl-obstruction tensors, the counterparts of obstruction tensors in Weyl geometry. By solving the bulk Einstein equations, we explicitly demonstrated how the Weyl-obstruction tensors in $4d$ (i.e., the Weyl-Bach tensor) and $6d$ are derived from the poles of the on-shell metric expansion in the WFG gauge.
Then, building on the Weyl-ambient construction, we investigated Weyl-covariant quantities induced by the ambient tensors in both first and second order formalisms. As an important example, the extended Weyl-obstruction tensor $\hat\Omega^{(k)}_{ij}$ is defined through covariant derivatives of the ambient Riemann tensor, and its definition in the first and second order formalisms are shown to be equivalent. We also proved that $\hat\Omega^{(k-1)}_{ij}$ corresponds to the pole of $\gamma^{(k)}_{ij}$ at $d=2k$ in the ambient metric expansion, which justifies the description of Weyl-obstruction tensors in \cite{Jia:2021hgy}. Compared with the extended obstruction tensor $\Omega^{(k-1)}_{ij}$, whose residue is only conformally covariant in $d=2k$, the extended Weyl-obstruction tensor $\hat\Omega^{(k-1)}_{ij}$ is Weyl covariant in any dimension. 
\par
Before moving on to the investigation of the holographic Weyl anomaly, we now remark on possible extensions and applications of our construction. The Weyl-ambient space induces the $\text{Diff}(M)\ltimes\text{Weyl}$ symmetry on the codimension-2 manifold $M$, which can be regarded as an asymptotic corner symmetry \cite{Ciambelli:2022vot,Ciambelli:2022cfr}. The algebra of corner symmetries and their Noether charges have been studied in \cite{Ciambelli:2021vnn,Ciambelli:2022cfr} (see also \cite{Freidel:2021cjp}), it is possible to apply the results therein to the Weyl-ambient space and study the asymptotic corner symmetries of the Weyl-ambient space. Moreover, since the surface $\cal N$ at $\rho=0$ of the Weyl-ambient space is null, there is an induced Carroll structure  \cite{Ciambelli:2018xat, Ciambelli:2019lap}. This is evident from the fact that the ambient Weyl diffeomorphism acts on the null surface as (a special case of) a Carrollian diffeomorphism. 
\par
One also expects intriguing holographic applications of the Weyl-ambient construction, for example in the context of celestial holography \cite{Pasterski:2016qvg,Raclariu:2021zjz,Pasterski:2021raf} and codimension-2 holography \cite{Akal:2020wfl,Ogawa:2022fhy}. In particular, the $\text{Diff}(M)\ltimes\text{Weyl}$ symmetry on $M$ corresponds to the Weyl-BMS symmetry on $\tilde M$ \cite{Freidel:2021fxf} (with supertranslations turned off). Therefore, we expect that the Weyl-ambient construction will provide a new arena for realizing the holographic principle.
\par
The symmetry correspondence between $M$ and the ambient space $\tilde M$ can also be applied to construct solutions of conformal hydrodynamics on $M$. For example, the Gubser flow \cite{Gubser:2010ze,Gubser:2010ui}, which is relevant for heavy-ion collisions, can be generalized by considering different symmetry constraints of the conformal group, which can be conveniently organized in the ambient space \cite{JLN}. By imposing different possible constraints coming from different subgroups of the conformal group, solutions of conformal hydrodynamics are generated systematically.
\par
The Weyl-ambient metric construction is part of a bigger program of introducing the Weyl connection back into physics. Viewed as an ordinary gauge symmetry, the Weyl symmetry can provide an organizing principle for constructing effective field theories (e.g.,\ for conformal hydrodynamics). Weyl manifolds would be the proper geometric setup for such future explorations. More recently, the ambient construction was used to study correlators of CFTs on general curved backgrounds \cite{Parisini:2022wkb,Parisini:2023nbd}. We hope the Weyl-ambient geometries can be utilized in similar contexts. 

\chapter{Holographic Weyl Anomaly}
\label{chap:HWA}
Utilizing the WFG formalism, in this chapter we will evaluate the Weyl anomaly for a holographic theory and demonstrate how Weyl-obstruction tensors play an important role in the expression of the Weyl anomaly in higher dimensions. We first discuss the anomalous Weyl-Ward identity for a general field theory on a background Weyl geometry, and then we focus on holographic theories in the WFG gauge. Then, we will compute the holographic Weyl anomaly explicitly in the WFG gauge up to $d=8$ and lay out the pattern for the results in general dimensions. In this Chapter, we will work in the Euclidean signature. We also adopt natural units where $c=\hbar=1$.
\section{Weyl-Ward Identity}
\label{sec:WWI}
Essentially, for a $d$-dimensional field theory coupled to a background metric $\gamma^{(0)}_{ij}$ and a Weyl connection $a^{(0)}_{i}$, the Weyl anomaly comes from an additional exponential factor arising in the path integral after applying a Weyl transformation:
\begin{align}
\label{Z}
Z[\gamma^{(0)},a^{(0)}]= \E^{-{\cal A}[{\cal B}(x);\gamma^{(0)},a^{(0)}]} Z[\gamma^{(0)}/{\cal B}(x)^{2},a^{(0)}-\text d\ln {\cal B}(x)]\,.
\end{align}
The anomaly ${\cal A}[{\cal B}(x);g,a]$ should satisfy the 1-cocycle condition \cite{Manvelyan:2001pv,BONORA1983305}
\begin{align}
{\cal A} [{\cal B''} {\cal B'};\gamma^{(0)},a^{(0)}]= {\cal A} [{\cal B'};\gamma^{(0)},a^{(0)}] + {\cal A} [{\cal B''}; \gamma^{(0)}/({\cal B'})^{2},a^{(0)}- \td\ln {\cal B'}]\,.
\end{align}
For any non-exact Weyl-invariant $d$-form $\bm A[\gamma_{(0)},a_{(0)}]$, one can check that ${\cal A}[{\cal B}(x);\gamma^{(0)},a^{(0)}]= \int (\ln {\cal B})\bm{A}$ satisfies the cocycle condition, and thus it is a possible candidate for the Weyl anomaly. However, if $\bm A$ is exact, ${\cal A}$ would be cohomologically trivial since it can be written as the difference of a Weyl-transformed local functional. The linearly independent choices of $\bm A$ in non-trivial cocycles correspond to different central charges.
\par
It follows from \eqref{Z} that the quantum effective action $S\equiv-\ln Z$ of a theory with Weyl anomaly satisfies
\begin{align}
\label{QFTWeylAnom}
-{\cal A}[{\cal B};\gamma^{(0)},a^{(0)}]=S[\gamma^{(0)}/{\cal B}(x)^{2},a^{(0)}-\text d\ln {\cal B}(x)]-S[\gamma^{(0)},a^{(0)}] \,
\end{align}
under the Weyl transformation. For infinitesimal $\ln{\cal B}$, the above equation gives to the first order
\begin{align}
\label{dS}
-\int \td^d x \frac{\delta \cal A}{\delta \ln{\cal B}(x)}\ln{\cal B}(x)=\int \td^d x \frac{\delta S}{\delta a^{(0)}_i(x)}\p_i\ln{\cal B}(x)
+\int \td^d x \frac{\delta S}{\delta \gamma^{(0)}_{ij}(x)}\Big(-2\ln{\cal B}(x)\gamma^{(0)}_{ij}(x)\Big)\,.
\end{align}
In general, the background fields $\gamma^{(0)}_{ij}$ and $a^{(0)}_i$ are the sources of the energy-momentum tensor operator $T^{ij}$ and the Weyl current operator $J^i$, respectively. The variations of the action with respect to them gives the following 1-point functions:
\begin{align}
\langle T^{ij}(x)\rangle=\frac{2}{\sqrt{-\det\gamma^{(0)}}}\frac{\delta S}{\delta \gamma^{(0)}_{ij}(x)}\,,\qquad \langle J^i(x)\rangle=-\frac{1}{\sqrt{-\det\gamma^{(0)}}}\frac{\delta S}{\delta a^{(0)}_i(x)}\,.
\end{align}
Integrating \eqref{dS} by parts and noticing that the ${\cal B}(x)$ is arbitrary, we obtain the anomalous Weyl-Ward identity
\begin{align}
\label{QFTWeylWard}
\frac{1}{\sqrt{-\det\gamma^{(0)}}}\frac{\delta \cal A}{\delta \ln{\cal B}(x)}=\big\langle T^{ij}(x)\gamma^{(0)}_{ij}(x)+\hat{\nabla}^{(0)}_i  J^i(x)\big\rangle\,.
\end{align}
As we can see, besides the trace of the energy-momentum tensor that appears in the usual case, the divergence of the Weyl current also contributes to the Ward identity when the Weyl connection is turned on.
\par
Let us now focus on a holographic field theory dual to the vacuum Einstein theory in the $(d+1)$-dimensional bulk. The holographic dictionary provides the relation between the on-shell classical bulk action $S_{bulk}$ and quantum effective action $S_{bdr}$ of the field theory on the boundary \cite{Witten:1998qj}:
\begin{align}
\label{dict}
\exp\left(-S_{bulk}[g;\gamma_{(0)},a_{(0)}]\right)=\exp\left(-S_{bdr}[\gamma_{(0)},a_{(0)}]\right)\,,
\end{align}
where $\gamma_{(0)}$ and $a_{(0)}$ are the boundary values of $h$ and $a$ as shown in \eqref{hex} and \eqref{aex}. When the bulk action transforms under a Weyl diffeomorphism, the corresponding boundary theory undergoes a Weyl transformation. However, the diffeomorphism invariance of the bulk Einstein theory does not imply the Weyl invariance on the boundary when there is anomaly, since it follows from \eqref{QFTWeylAnom} that
\begin{align}
\label{Weyltrans}
0=S_{bulk}[g|z',x']-S_{bulk}[g|z,x]=S_{bdr}[\gamma'_{(0)},a'_{(0)}|x]-S_{bdr}[\gamma_{(0)},a_{(0)}|x]+ {\cal A}[{\cal B}]\,,
\end{align}
where $(z',x')=(z/{\cal B},x)$ for the bulk and $\gamma'_{(0)}=\gamma_{(0)}/{\cal B}^2$, $a'_{(0)}=a_{(0)}-\td\ln{\cal B}$ for the boundary.
\par
Since $a_{i}$ is pure gauge in the bulk, $a^{(0)}_i$ could be gauged away and hence it is not expected to source any current on the boundary. The role of the $a^{(0)}_{i}$, however, is important since it makes the energy-momentum tensor along with all the geometric quantities on the boundary Weyl-covariant. On the other hand, the $p^{(0)}_i$ also plays a role in the Weyl-Ward identity. In the FG gauge, $\pi^{(0)}_{ij}$ corresponds to the expectation value of $T_{ij}$; the Ward identity for the Weyl symmetry shows that the trace of $\pi^{(0)}_{ij}$ vanishes, which can be read off from the $O(z^d)$-order of the $zz$-component of the Einstein equations \cite{Leigh}. In the WFG gauge, this equation now gives
\begin{align}
\label{boundaryWI}
0=\frac{d}{2L^2}\gamma_{(0)}^{ij}\pi^{(0)}_{ij}+\hat\nabla^{(0)}\cdot p_{(0)}\,.
\end{align}
Besides $\pi^{(0)}_{ij}$, there is an additional term $\hat\nabla^{(0)}\cdot p_{(0)}$ which represents a gauge ambiguity of $a_i$. This suggests that the energy-momentum tensor in the WFG gauge acquires an extra piece, which now can be considered as an ``improved" energy-momentum tensor $\tilde T_{ij}$ (\`a la \cite{BELINFANTE1940449,Callan1970}):
\begin{equation}
\label{improvedT}
\langle \kappa^2\tilde   T_{ij}\rangle=\frac{d}{2 L^2}\pi^{(0)}_{ij}+ \hat\nabla^{(0)}_{(i} p^{(0)}_{j)}\,,
\end{equation}
where $\kappa^2 =8\pi G$.\footnote{The energy-momentum tensor \eqref{improvedT} in the WFG gauge can be verified using the prescription introduced in \cite{deHaro:2000vlm}.}
It is easy to see that the trace of this energy-momentum tensor gives the right-hand side of \eqref{boundaryWI}. One can also find that the $zi$-components of the Einstein equations at the $O(z^d)$-order give exactly the conservation law $\langle \hat\nabla_{(0)}^i\tilde T_{ij}\rangle =0$ [see the last line of \eqref{emz}], which is the Ward identity corresponding to the boundary diffeomorphisms. Therefore, in the holographic case we can write the anomalous Weyl-Ward identity \eqref{QFTWeylWard} as
\begin{equation}
\label{holoWeylWard}
\frac{1}{\sqrt{-\det\gamma^{(0)}}}\frac{\delta \cal A}{\delta \ln{\cal B}(x)}=\big\langle\tilde T^{ij}(x)\gamma^{(0)}_{ij}(x)\big\rangle\,.
\end{equation}
Notice that one should distinguish $p^{(0)}_i$ and the Weyl current $J_i$. Unlike $\pi_{ij}^{(0)}$ which is sourced by $\gamma^{(0)}_{ij}$, $p^{(0)}_i$ is not sourced by $a^{(0)}_i$ since $a_i$ is pure gauge in the bulk. In the boundary field theory, the Weyl current  $J_{i}$ vanishes identically, while $p^{(0)}_{i}$ contributes to the expectation value of $\tilde T_{ij}$ as an ``improvement". In a generic non-holographic field theory defined on the background with Weyl geometry, there may exist a nonvanishing $J_{i}$ sourced by the Weyl connection $a^{(0)}_i$.
\par
Using the basis $\{\bm e^z,\bm e^i=\td x^i\}$ in \eqref{basis}, the bulk on-shell Einstein-Hilbert action with negative cosmological constant can be written as
\begin{align}
\label{SEH}
S_{bulk}&=\frac{1}{2\kappa^2}\int_M\sqrt{-\det g}\,(R-2\Lambda)\bm{e}^z\wedge \td x^1\wedge\cdots\wedge \td x^d\,.
\end{align}
To evaluate this, we first notice that the trace of the vacuum Einstein equation in the bulk gives
\begin{align}
\label{RL}
R=\frac{2(d+1)}{d-1}\Lambda=-\frac{d(d+1)}{L^2}\,,
\end{align}
where we have considered $\Lambda=-\frac{d(d-1)}{2L^2}$. Also, noticing that $\sqrt{-\det g}=\sqrt{-\det h}$, we can expand $\sqrt{-\det h}$ as
\begin{align}
\label{sqrth}
\sqrt{-\det h}&=\left(\frac{L}{z}\right)^d\sqrt{-\det\gamma^{(0)}}\left(1+\frac{1}{2}\left(\frac{z}{L}\right)^2 X^{(1)}+\frac{1}{2}\left(\frac{z}{L}\right)^4X^{(2)}+\cdots+\frac{1}{2}\left(\frac{z}{L}\right)^dY^{(1)}+\cdots\right)\,.
\end{align}
Plugging \eqref{RL} and \eqref{sqrth} into \eqref{SEH} yields
\begin{align}
\label{Sbulk}
S_{bulk}&=-\frac{L^{-2}}{\kappa^2}\int_M \left(\frac{L}{z}\right)^d\left(d+\frac{d}{2}\left(\frac{z}{L}\right)^2 X^{(1)}+\frac{d}{2}\left(\frac{z}{L}\right)^4X^{(2)}+\cdots+\frac{d}{2}\left(\frac{z}{L}\right)^dY^{(1)}+\cdots\right)\bm{e}^z\wedge vol_\Sigma\,,
\end{align}
where we defined $vol_\Sigma\equiv \sqrt{-\det\gamma^{(0)}}\td x^1\wedge\cdots\wedge\td x^d$. 
\par
The above integral is not well-defined since it has divergences. To handle these divergences one should regularize the bulk on-shell action. In the FG gauge, it is common to introducing a cutoff surface at some small value of $z=\epsilon$, and then add counterterms to cancel the divergent parts when $\epsilon\to 0$. This is essentially how the Weyl anomaly arises since the regulator breaks the Weyl symmetry and causes the appearance of a logarithmically divergent term. However, in the WFG gauge since we do not assume that we have an integrable distribution when $a_i$ is turned on, we cannot naively introduce a cutoff surface and go through this procedure. Nevertheless, one can still extract the divergences using dimensional regularization. Suppose $d$ is not an even integer ($2k-2<d<2k$), then the divergent terms in \eqref{Sbulk} are those from the $O(z^{-d})$-order to the $O(z^{2k-2-d})$-order; once they get canceled by the counterterms, the renormalized bulk action, denoted by $S_{bulk}^{re(k-1)}$, will be analytic and thus no anomaly arises. Now if we let $d$ approach an even integer $2k$ from below, the $O(z^{2k-d})$-order of $S_{bulk}^{re(k-1)}$ will encounter a pole at $d=2k$, which corresponds to the logarithmic divergence that appears in the cutoff procedure. This is similar to the discussion at the end of Section \ref{sec:FG} for the bulk metric expansion. After this pole term is removed by a counterterm, one gets the renormalized action $S_{bulk}^{re(k)}$ for $2k\leqslant d<2k+2$, i.e.
\begin{align}
S^{re(k-1)}_{bulk}[z,x]&=S_{bulk}^{re(k)}[z,x]+S^{(k)}_{pole}[z,x]\,,
\end{align}
where $S_{pole}^{(k)}$ is the $O(z^{2k-d})$-order in the expansion of $S_{bulk}$. $S^{re(k-1)}_{bulk}$ being invariant under a Weyl diffeomorphism gives, 
\begin{align}
\label{Weylk-1}
0&=S^{re(k-1)}_{bulk}[z',x]-S^{re(k-1)}_{bulk}[z,x]=S^{(k)}_{pole}[z',x]-S^{(k)}_{pole}[z,x]+S^{re(k)}_{bulk}[z',x]-S^{re(k)}_{bulk}[z,x]\,.
\end{align}
When we take the limit $d\to2k$ from below, the difference of the divergent $S^{(k)}_{pole}$ will have a finite result, and $S^{re(k)}_{bulk}$ corresponds to the renormalized boundary action $S_{bdy}$ by holographic dictionary, which will not be Weyl invariant at $d=2k$. Comparing \eqref{Weylk-1} with \eqref{Weyltrans}, we can see that the Weyl anomaly can be extracted from the difference of $S^{(k)}_{pole}$ under a Weyl diffeomorphism \cite{Mazur2001}:\footnote{Although the previous counterterms make finite contributions to the $O(z^{2k-d})$-order, they do not affect the pole. So the difference of the $O(z^{2k-d})$-order of the $S^{reg}_{bulk}$ is the same as that of the bare on-shell action \eqref{Sbulk} in the limit $d\to2k^-$.} 
\begin{align}
&\lim_{d\to2k^-}S^{(k)}_{pole}[z',x]-S^{(k)}_{pole}[z,x]\nn\\
={}&\frac{d}{2\kappa^2L}\int \td\left(\frac{1}{d-2k}\left(\frac{L}{z{\cal B}}\right)^{d-2k}\right)\wedge X^{(k)}vol_\Sigma-\frac{d}{2\kappa^2L}\int \td\left(\frac{1}{d-2k}\left(\frac{L}{z}\right)^{d-2k}\right)\wedge X^{(k)}vol_\Sigma\nn\\
={}&\frac{k}{\kappa^2L}\int\ln {\cal B} X^{(k)}_{d=2k}vol_\Sigma\,.
\end{align}
This result gives rise to the Weyl anomaly ${\cal A}_k$ of the $2k$-dimensional boundary theory, i.e.
\begin{align}
\label{Ak}
{\cal A}_k=\frac{k}{\kappa^2L}\int\ln {\cal B} X^{(k)}_{d=2k}vol_\Sigma\,.
\end{align}
Therefore, to find the Weyl anomaly in $2k$-dimension, we only have to compute $X^{(k)}$ coming from the expansion of $\sqrt{-\det h}$. 

\section{Holographic Weyl Anomaly}
\label{sec:HWA}
\subsection{Weyl Anomaly in $2d$ and $4d$}
Now let us apply \eqref{Ak} to $2d$ and $4d$. To find the holographic Weyl anomaly in $2d$ and $4d$ all we have to do is plug in the expressions of $X^{(1)}$ and $X^{(2)}$ obtained from the $zz$-components of the Einstein equations (see Appendix \ref{app:B0}); that is,
\begin{align}
\label{2d4d}
X^{(1)} =-\frac{L^2}{2(d-1)}\hat R\,,\qquad X^{(2)}=-\frac{L^4}{4(d-2)^2}\bigg(\hat R_{ij}\hat R^{ji}-\frac{d}{4(d-1)}\hat R^2\bigg)-\frac{L^2}{2}\hat\nabla\cdot a^{(2)}\,.
\end{align}
[From now on we will drop the label ``(0)" for the boundary curvature quantities and derivative operator when there is no confusion.] First we look at the Weyl anomaly in $d=2$:
\begin{align}
\label{2dA}
{\cal A}_1&=\frac{1}{\kappa^2L}\int\ln {\cal B} X^{(1)}_{d=2}vol_\Sigma =-\frac{L}{16\pi G}\int \ln {\cal B}\hat R\sqrt{-\det\gamma^{(0)}}\td^2x\,,
\end{align}
where in the second equality we used \eqref{2d4d}. Then, it follows from \eqref{holoWeylWard} that the Weyl-Ward identity now reads
\begin{align}
\langle{\tilde T}^i{}_i\rangle=-\frac{L}{16\pi G}\hat R\,.
\end{align}
We can see that the right-hand side of this result has exactly the same form as what we get from the standard calculation in the FG gauge, except that the curvature scalar now is Weyl-covariant. Similarly, plugging \eqref{2d4d} into \eqref{Ak}, we find that the Weyl anomaly in $d=4$ can be written as
\begin{align}
\label{4dA}
{\cal A}_2&=\frac{2}{\kappa^2L}\int\ln {\cal B} X^{(2)}_{d=4}vol_\Sigma=-\frac{L}{8\pi G}\int\bigg[\frac{L^2}{8}\Big(\hat R_{ij}\hat R^{ji}-\frac{1}{3}\hat R^2\Big)+\hat\nabla\cdot a^{(2)}\bigg]\ln {\cal B}\sqrt{-\det\gamma^{(0)}}\td^4x\,.
\end{align}
Again, one can immediately tell that the right-hand side of this result matches the standard FG result (e.g. \cite{Henningson:1998gx}) if we turn off the Weyl structure. 
\par
There are a few things worth paying attention to: first, in the $2d$ Weyl anomaly \eqref{2dA}, the Weyl-Ricci scalar is also the Weyl-Euler density $E^{(2)}$ in $2d$, i.e.\ the Euler density Weyl-covariantized by the Weyl connection. Furthermore, we can rewrite the $4d$ Weyl anomaly \eqref{4dA} as
\begin{align}
\label{A2}
{\cal A}_2&=-\frac{L}{8\pi G}\int\bigg[\frac{L^2}{16}\Big(\hat W_{ijkl}\hat W^{klij}-\hat E^{(4)}\Big)+\hat\nabla\cdot a^{(2)}\bigg]\ln {\cal B}\sqrt{-\det\gamma^{(0)}}\td^4x\,,
\end{align}
where $\hat E^{(4)}$ is the Weyl-Euler density in $4d$:
\begin{align}
\hat E^{(4)}=\hat R_{ijkl}\hat R^{klij}-4\hat R_{ij}\hat R^{ji}+\hat R^2\,.
\end{align}
Traditionally, the Euler density $E^{(2k)}$ without the Weyl connection is called the type A Weyl anomaly, which is topological in $2k$-dimension and not Weyl-invariant, while the type B Weyl anomaly is the Weyl-invariant part of the anomaly \cite{Deser1993}. Here we find that in the WFG gauge, this classification of the Weyl anomaly is still available, with the Weyl-Euler density now Weyl-invariant since the curvature quantities in this setup are endowed with Weyl covariance.
\par
Also, notice that the subleading term $a^{(2)}_i$ of $a_i$ only makes an appearance in the anomaly through a cohomologically trivial term, i.e.\ we can express it as a Weyl-transformed local functional as follows:
\begin{align}
\int\td^4x\sqrt{-\det\gamma_{(0)}}\ln {\cal B}\,\hat\nabla_i a_{(2)}^i=\int\td^4x\sqrt{-\det\gamma'_{(0)}}\,a'^{(0)}_i a'^i_{(2)}-\int\td^4x\sqrt{-\det\gamma_{(0)}}\,a^{(0)}_i a_{(2)}^i\,,
\end{align}
where $a'^i_{(2)}={\cal B}^4a_{(2)}^i$, and the boundary term due to integrating by parts is ignored. We will see that this is a generic feature of the Weyl anomaly in the WFG gauge for any dimension.
\par
Although in \eqref{2dA} and \eqref{4dA} we expressed the holographic Weyl anomaly in $2d$ and $4d$ in terms of curvature to match the corresponding familiar results in the FG gauge, we can also express them alternatively in terms of the Weyl-Schouten tensor:
\begin{align}
\label{X1X2}
\frac{X^{(1)}}{L^2}=-\hat P\,,\qquad \frac{X^{(2)}}{L^4}=-\frac{1}{4}\tr(\hat P^2)+\frac{1}{4}\hat P^2-\frac{1}{2L^2}\hat\nabla\cdot a^{(2)}\,.
\end{align}
Then \eqref{2dA} and \eqref{4dA} can be written as
\begin{align}
\label{2dAP}
{\cal A}_1&=-\frac{L}{\kappa^2}\int\td^2x\sqrt{-\det\gamma^{(0)}}\ln {\cal B} \hat P\,,\\
\label{4dAP}
{\cal A}_2&=-\frac{L^3}{\kappa^2}\int\td^4x\sqrt{-\det\gamma^{(0)}}\ln {\cal B} \bigg(\frac{1}{2}\tr(\hat P^2)-\frac{1}{2}\hat P^2+\frac{1}{L^2}\hat\nabla\cdot a^{(2)}\bigg)\,.
\end{align}
In higher dimensions, $X^{(k)}$ can be expressed in terms of $\gamma_{ij}^{(0\leqslant j\leqslant2k)}$ (see Appendix \ref{app:expansion}). By solving the Einstein equations we have seen that these terms can all be expressed in terms of $\hat P_{ij}$ and $\hat{\cal O}^{(2<j<2k)}_{ij}$. Therefore, we will use the Weyl-Schouten tensor and Weyl-obstruction tensors as the building blocks for the Weyl anomaly in even dimensions.

\subsection{Weyl Anomaly in $6d$}\label{Weyl6d}
After revisiting the results in $2d$ and $4d$, we will now present our computations for $6d$ and $8d$. In principle, $X^{(k)}$ can be obtained by solving Einstein equations as we have done for $2d$ and $4d$. However, as the dimension goes higher, computing the curvature will become extremely tedious. To facilitate the computation in higher dimensions, we can use a more efficient way of organizing the Einstein equations which helps us avoid the curvature tensors, namely to use the Raychaudhuri equation of the congruence generated by $\un D_z$. The details of the Raychaudhuri equation and its expansions are given in Appendix \ref{app:expansion}.
\par
To solve for $X^{(3)}$, we need to expand $\sqrt{-\det h}$ to the order $O(z^{6-d})$.  Using \eqref{Ray6d} and plugging the results we have got for $\gamma^{(2)}_{ij},\gamma^{(4)}_{ij}$ and $X^{(1)},X^{(2)}$ into \eqref{X3}, we obtain
\begin{align}
\frac{X^{(3)}}{L^6}=&-\frac{1}{12}\tr(\hat P^3)+\frac{1}{8}\tr(\hat P^2)\hat P-\frac{1}{24}\hat P^3+\frac{1}{12}\tr(\hat\Omega^{(1)}\hat P)\nn\\
\label{X3P}
&+\frac{1}{6L^4}(d-6)a^2_{(2)}-\frac{1}{3L^4}\hat\nabla\cdot a^{(4)}-\frac{1}{12L^2}\hat\nabla_i\big[a^{(2)}_j(3\hat P^{ij}+\hat P^{ji}
-3\hat P \gamma_{(0)}^{ij})\big]\,,
\end{align}
where we used the extended Weyl-obstruction tensor $\hat\Omega^{(1)}_{ij}$ defined in \eqref{extWO}. Notice first that the $a^{(2)}_i$ quadratic term in $X^{(3)}$ vanishes in $6d$, and thus does not contribute to the Weyl anomaly. Then, it follows from \eqref{Ak} that the Weyl anomaly in $6d$ is
\begin{align}
\label{6dA}
{\cal A}_3={}&\frac{3}{\kappa^2L}\int\ln {\cal B} X^{(3)}_{d=6}vol_\Sigma\nn\\
={}&-\frac{L^5}{\kappa^2}\int \td^6x\sqrt{-\det\gamma^{(0)}}\ln {\cal B} \bigg(\frac{1}{4}\tr(\hat P^3)-\frac{3}{8}\tr(\hat P^2)\hat P+\frac{1}{8}\hat P^3-\frac{1}{4}\tr(\hat\Omega^{(1)}\hat P)\nn\\
&+\frac{1}{L^4}\hat\nabla\cdot a^{(4)}+\frac{1}{4L^2}\hat\nabla_i\big[a^{(2)}_j(3\hat P^{ij}+\hat P^{ji}
-3\hat P \gamma_{(0)}^{ij})\big]\bigg)\,.
\end{align}
Just as what we have shown for the $4d$ case, the subleading terms in the expansion of $a_i$ appear only in total derivatives and thus only contribute to cohomologically trivial terms in the $6d$ Weyl anomaly. When we turn off $a_i^{(0)}$ and $a_i^{(2)}$, this result agrees with the holographic Weyl anomaly in the FG gauge computed in \cite{Henningson:1998gx}. \par
Usually, the Weyl anomaly in $6d$ is written as a linear combination of the $6d$ Euler density and three conformal invariants in $6d$ (see \cite{Henningson:1998gx,Deser1993,Henningson:1998gx}), which represents the four central charges in $6d$. The result we obtained can also be written in this way, which means the classification of type A and type B anomalies still holds for the WFG gauge in $6d$. However, as we will discuss shortly, the expression we have in \eqref{X3P} in terms of $\hat P_{ij}$ and $\hat\Omega^{(1)}_{ij}$ reveals some interesting aspects of the Weyl anomaly.

\subsection{Weyl Anomaly in $8d$}\label{Weyl8d}
Expanding $\sqrt{-\det h}$ to the order $O(z^{8-d})$, we have $X^{(4)}$ in \eqref{X4}. Using \eqref{Ray8d} and plugging the results up to $\gamma^{(6)}_{ij}$ and $X^{(3)}$ into \eqref{X4}, we have
\begin{align}
\label{X4P}
\frac{X^{(4)}}{L^8}=&-\frac{1}{32}\tr(\hat P^4)+\frac{1}{24}\tr(\hat P^3)\hat P+\frac{1}{64}(\tr (\hat P^2))^2-\frac{1}{32}\tr(\hat P^2)\hat P^2+\frac{1}{192}\hat P^4\nn\\
&-\frac{1}{24}\tr(\hat\Omega^{(1)}\hat P)\hat P+\frac{1}{24}\tr(\hat\Omega^{(1)}\hat P^2)-\frac{1}{96}\tr(\hat\Omega^{(1)}\hat\Omega^{(1)})-\frac{1}{96}\tr(\hat\Omega^{(2)}\hat P)\nn\\
&+\frac{d-8}{4L^6}a^{(4)}\cdot a^{(2)}+\frac{d-8}{12L^4}a^{(2)}_i a^{(2)}_j(\hat P^{ij}-\hat P\gamma_{(0)}^{ij})+\text{total derivatives}\,.
\end{align}
As expected, all the terms in \eqref{X4P} that involve $a^{(2)}_i$, $a^{(4)}_i$, $a^{(6)}_i$ either vanish when $d=8$ or contribute only to the total derivatives. The details of the total derivatives are given in \eqref{X8t}. Plugging \eqref{X4P} into \eqref{Ak}, we obtain the holographic Weyl anomaly in $8d$:

\begin{align}
\label{8dA}
{\cal A}_4={}&\frac{4}{\kappa^2L}\int\ln {\cal B} X^{(4)}_{d=8}vol_\Sigma\nn\\
={}&-\frac{L^7}{\kappa^2}\int\td^8x\sqrt{-\det\gamma^{(0)}}\ln {\cal B} \bigg(\frac{1}{8}\tr(\hat P^4)-\frac{1}{6}\tr(\hat P^3)\hat P-\frac{1}{16}(\tr (\hat P^2))^2+\frac{1}{8}\tr(\hat P^2)\hat P^2-\frac{1}{48}\hat P^4\nn\\
&+\frac{1}{6}\tr(\hat\Omega^{(1)}\hat P)\hat P-\frac{1}{6}\tr(\hat\Omega^{(1)}\hat P^2)+\frac{1}{24}\tr(\hat\Omega^{(1)}\hat\Omega^{(1)})+\frac{1}{24}\tr(\hat\Omega^{(2)}\hat P)+\text{total derivatives}\bigg)\,.
\end{align}
Once again, we can see that the subleading terms in $a_i$ only have cohomologically trivial contributions. If we go back to the FG gauge, then this result agrees with the renormalized volume coefficient for $k=4$ shown in \cite{graham2009extended}. One can also write the FG version of the above result in the traditional way as a linear combination of the type A and type B anomalies, i.e.\ the Euler density and Weyl invariants (the list of Weyl invariants in $8d$ can be found in \cite{Boulanger:2004zf}). We naturally expect that this classification can also be applied to the holographic Weyl anomaly in the WFG gauge for higher dimensions.

\subsection{Building Blocks of the Weyl Anomaly}
As we have seen, if we ignore the total derivatives that depend on the subleading terms of the $a_i$ expansion, $X^{(1)}$ corresponds to the Weyl-Ricci scalar (i.e.\ the $2d$ Weyl-Euler density) and $X^{(2)}$ corresponds to the classic ``$a=c$" result. For the Weyl anomaly in $6d$ and $8d$ both $X^{(3)}$ and $X^{(4)}$ can also be written as linear combinations of the Weyl-Euler density and type B anomalies. This is true for both the FG and WFG cases, just the quantities in the latter are Weyl-covariant. One just needs to substitute the Weyl quantities with their LC counterparts (i.e.\ set $a_i$ to zero) to get the Weyl anomaly in the FG case. However, when expressing them in terms of the Weyl-Schouten tensor and extended Weyl-obstruction tensors (or Schouten tensor and extended obstruction tensors in the FG case), we observe that the polynomial terms of $X^{(k)}/L^{2k}$ (without the total derivative terms) in $2k$-dimensions, denoted by $\bar X^{(k)}$, have the following structures:
\begin{align}
\label{X1D}
\bar X^{(1)}&=-\delta^i_j\hat P^j{}_i\,,\\
2\bar X^{(2)}&=\frac{1}{2}\delta^{i_1i_2}_{j_1j_2}\hat P^{j_1}{}_{i_1}\hat P^{j_2}{}_{i_2}\,,\\
6\bar X^{(3)}&=-\frac{1}{4}\delta^{i_1i_2i_3}_{j_1j_2j_3}\hat P^{j_1}{}_{i_1} \hat P^{j_2}{}_{i_2}\hat P^{j_3}{}_{i_3}-\frac{1}{2}\delta^{i_1i_2}_{j_1j_2}\hat\Omega_{(1)}^{j_1}{}_{i_1}\hat P^{j_2}{}_{i_2}\,,\\
\label{X4D}
24\bar X^{(4)}&=\frac{1}{8}\delta^{i_1i_2i_3i_4}_{j_1j_2j_3j_4}\hat P^{j_1}{}_{i_1} \hat P^{j_2}{}_{i_2}\hat P^{j_3}{}_{i_3}\hat P^{j_4}{}_{i_4}+\frac{1}{2}\delta^{i_1i_2i_3}_{j_1j_2j_3}\hat\Omega_{(1)}^{j_1}{}_{i_1}\hat P^{j_2}{}_{i_2}\hat P^{j_3}{}_{i_3}\nn\\
&\quad+\frac{1}{4}\delta^{i_1i_2}_{j_1j_2}\hat\Omega_{(1)}^{j_1}{}_{i_1}\hat\Omega_{(1)}^{j_2}{}_{i_2}+\frac{1}{4}\delta^{i_1i_2}_{j_1j_2}\hat\Omega_{(2)}^{j_1}{}_{i_1}\hat P^{j_2}{}_{i_2}\,,
\end{align}
where the Kronecker $\delta$ symbol is defined as 
\begin{align}
\delta^{i_1\cdots i_s}_{j_1\cdots j_s}=s!\delta^{i_1}{}_{[j_1}\cdots\delta^{i_s}{}_{j_s]}\,.
\end{align}
From \eqref{X1D}--\eqref{X4D} we can see that $\bar X^{(k)}$ contains all kinds of possible combinations of $\hat P_{ij}$ and $\hat\Omega^{(2<j<2k)}_{ij}$ whose Weyl weights add up to be $2k$, i.e.\ the Weyl weight of $X^{(k)}$. Using this pattern, one can directly write down the terms in the holographic Weyl anomaly in any dimension. For instance, we can easily predict without explicit calculation that $\bar X^{(5)}$ is the linear combination of the following terms:
\begin{align*}
&\delta^{i_1i_2i_3i_4i_5}_{j_1j_2j_3j_4j_5}\hat P^{j_1}{}_{i_1} \hat P^{j_2}{}_{i_2}\hat P^{j_3}{}_{i_3}\hat P^{j_4}{}_{i_4}\hat P^{j_5}{}_{i_5}\,,\qquad\delta^{i_1i_2i_3i_4}_{j_1j_2j_3j_4}\hat\Omega_{(1)}^{j_1}{}_{i_1} \hat P^{j_2}{}_{i_2}\hat P^{j_3}{}_{i_3}\hat P^{j_4}{}_{i_4}\,,\\
&\delta^{i_1i_2i_3}_{j_1j_2j_3}\hat\Omega_{(2)}^{j_1}{}_{i_1}\hat P^{j_2}{}_{i_2}\hat P^{j_3}{}_{i_3}\,,\qquad\delta^{i_1i_2i_3}_{j_1j_2j_3}\hat\Omega_{(1)}^{j_1}{}_{i_1}\hat\Omega_{(1)}^{j_2}{}_{i_2}\hat P^{j_3}{}_{i_3}\,,\qquad\delta^{i_1i_2}_{j_1j_2}\hat\Omega_{(2)}^{j_1}{}_{i_1}\hat\Omega_{(1)}^{j_2}{}_{i_2}\,,\qquad\delta^{i_1i_2}_{j_1j_2}\hat\Omega_{(3)}^{j_1}{}_{i_1}\hat P^{j_2}{}_{i_2}\,.
\end{align*}
These terms represent the independent central charges that appear in the holographic Weyl anomaly in $d=10$. 
\par
Based on the above pattern, it is natural to expect a general expression that can generate the holographic Weyl anomaly in any dimension, which is an analog of the exponential structure given by the Chern class that generates the chiral anomaly in any dimension (see, e.g.\ \cite{Frampton:1983nr,Zumino:1983rz,Bertlmann:1996xk}). It has been shown that the type A Weyl anomaly can be generated by a mechanism similar to that for the chiral anomaly \cite{Deser1993,Boulanger:2007st,Boulanger:2007ab,Jia:2023tki}. The expressions for the Weyl anomaly in terms of the (Weyl-) Schouten tensor and the extended (Weyl-) obstruction tensors suggest a similar mechanism for the holographic Weyl anomaly.

\section{Role of the WFG Gauge}
\label{sec:roleWFG}
Now that we have discussed the Weyl-obstruction tensors and Weyl anomaly, let us provide some observations on how the $a_{i}$ mode \eqref{aex} is involved. We have already mentioned that according to the FG theorem, this mode is pure gauge in the bulk. Now we have a few clear manifestations of this from our calculations. 
\par
The first one is that the subleading terms $a^{(2k)}_{i}$ with $k>0$ in the expansion of $a_i$ cannot be determined from the Einstein equations when $a^{(0)}_{i}$ is given. This is different from the expansion of $h_{ij}$ where the subleading terms $\gamma_{ij}^{(2k)}$ can be solved (on-shell) in terms of $\gamma_{ij}^{(0)}$.
\par
The second one is that $a_{i}$ appears only inside total derivatives in $X^{(k)}$, and thus represents cohomologically trivial modifications of the boundary Weyl anomaly. For $a_{i}^{(2k)}$ with $k\geqslant 2$, this can easily be seen from the expressions \eqref{4dAP}, \eqref{6dA} and \eqref{8dA}. What is not explicit in these formulas is that $a^{(0)}_{i}$ also appears inside a total derivative. This can be verified by separating the LC quantities out of the Weyl quantities in $X^{(k)}$. For instance, denote the LC Schouten tensor as $\mathring P_{ij}$ and the LC connection as $\LCnabla$, and then $X^{(1)}$ in $2d$ and $X^{(2)}$ in $4d$ can be written as
\begin{align}
L^{-2}X^{(1)}_{d=2}={}& L^{-2}\mathring X^{(1)}_{d=2}+\LCnabla\cdot a^{(0)}\,,\\
L^{-4}X^{(2)}_{d=4}={}&L^{-4}\mathring X^{(2)}_{d=4}-\frac{1}{2}\LCnabla_i(\mathring P^{ij}a^{(0)}_j-\mathring Pa_{(0)}^i)\nn\\
&-\frac{1}{4}\LCnabla_i(a^{(0)}_j\LCnabla^j a_{(0)}^i-a_{(0)}^i\LCnabla\cdot a_{(0)})-\frac{1}{4}\LCnabla_i(a_{(0)}^i a_{(0)}^2)-\frac{1}{2L^2}\LCnabla\cdot a^{(2)}\,,
\end{align}
where $L^{-2}\mathring X^{(1)}=-\mathring{P}$ and $L^{-4}\mathring X^{(2)}=\frac{1}{4}\mathring{P}^2-\frac{1}{4}\tr(\mathring{P}^2)$.\footnote{Note that $\LCnabla\cdot a^{(2)}$ is equivalent to $\hat\nabla\cdot a^{(2)}$ in $4d$, since in $2k$-dimension $\hat\nabla$ and $\LCnabla$ give the same result when acting on a vector with Weyl weight $+2k$ [which follows directly from \eqref{div_v}].}
Notice that although the terms involving $a^{(0)}_{i}$ are total derivatives, they are not Weyl-covariant and so one cannot naively assume that they are trivial cocycles. However, by finding suitable local counterterms, we have checked that all the terms involving $a_{i}^{(0)}$ are indeed part of a trivial cocycle  for $2d$ and $4d$. As $a_{i}$ is pure gauge, we expect this to be generally true.
\par
In principle, the Weyl connection $a^{(0)}_i$ on the boundary brings new Weyl-invariant objects, such as $\tr(f_{(0)}^2)$, which could lead to new central charges in the Weyl anomaly. However, up to $d=8$ we find the classification of type A and type B anomalies is still available, and in such a basis the nonvanishing central charges are still the same as those in the FG case. Once this can be carried over to higher dimensions, then $a^{(0)}_i$ appearing in total derivatives in $X^{(k)}$ can also be deduced by considering the Weyl anomaly as the sum of the type A and type B anomalies. In the FG gauge, under a Weyl transformation the type B anomaly is invariant while the type A anomaly, i.e.\ the Euler density, gets an extra total derivative involving $\ln{\cal B}$. Since the Weyl connection makes the Weyl anomaly in the WFG gauge Weyl-invariant, the terms with $a^{(0)}_i$ in the Weyl-Euler density should exactly compensate the extra total derivative, and hence they must form a total derivative. 
\par
Another observation we have mentioned is that although the subleading terms in the expansion of $a_{i}$ make an appearance in $\gamma^{(2k)}_{ij}$, they do not appear in the Weyl-obstruction tensors. Up to $k=3$, we have seen explicitly in \eqref{g2}, \eqref{g40} and \eqref{g60} that the terms with $a_{i}^{(2)}$ and $a_{i}^{(4)}$ do not contribute to the pole at $d=2k$ in $\gamma^{(2k)}_{ij}$. What is also true but not as obvious, is that the terms with $a^{(0)}_i$ do not contribute to the pole at $d-2$ in the Weyl-Schouten tensor and are proportional to $d-2k$ in Weyl-obstruction tensors. For instance, one can separate the $a^{(0)}_i$ from $\hat P_{ij}$ and get
\begin{align}
\hat P_{ij}=\mathring P_{ij}+\LCnabla_j a^{(0)}_{i}+a^{(0)}_{i}a^{(0)}_{j}-\frac{1}{2}a_{(0)}^2\gamma^{(0)}_{ij}\,,
\end{align}
while the only pole on the right-hand side is in the LC Schouten tensor $\mathring P_{ij}$. Similarly, expressing the Weyl-Bach tensor in terms of LC quantities we have
\begin{align}
\hat B_{ij}=\mathring{B}_{ij}+(d-4)(a_{(0)}^k\mathring{C}_{kji}-2a_{(0)}^k\mathring{C}_{ijk}+a_{(0)}^k a_{(0)}^l\mathring W_{likj})\,.
\end{align}
Thus, when $d=4$, $a^{(0)}_i$ does not contribute to the pole in $\gamma^{(4)}_{ij}$, and the Weyl-Bach tensor $\hat{B}_{ij}$ is equivalent to the LC Bach tensor $\mathring{B}_{ij}$. One should naturally expect that this is also true for any Weyl-obstruction tensors, i.e.\ $\hat{\cal O}^{(2k)}_{ij}$ is equivalent to the LC obstruction tensor $\mathring{\cal O}^{(2k)}_{ij}$ when $d=2k$. Note that when $d>2k$, the $a^{(0)}_i$ terms are included in the Weyl-obstruction tensor so that $\hat{\cal O}^{(2k)}_{ij}$ is always Weyl-covariant.
\par
The statement that any term in the expansion of $a_i$ does not appear in the pole of $\hat\gamma^{(2k)}_{ij}$ is consistent with the following claim: when $d=2k$, the Weyl-obstruction tensor $\hat{\cal O}_{(2k)}^{ij}$ satisfies
\begin{align}
\label{varX}
\hat{\cal O}_{(2k)}^{ij}=\frac{1}{\sqrt{-\det\gamma^{(0)}}}\frac{\delta}{\delta\gamma^{(0)}_{ij}}\int\td^d x\sqrt{-\det\gamma^{(0)}}X^{(k)}\,.
\end{align}
The FG version of this relation for $\mathring{\cal O}_{(2k)}^{ij}$ was proved in \cite{graham2005ambient} (see also \cite{deHaro:2000vlm}). If the claim above can be proved for the WFG gauge, then the reason that none of the terms in the expansion of $a_i$ contributes to $\hat{\cal O}^{(2k)}_{ij}$ at $d=2k$ will be straightforward: as they only appear in  total derivative terms in $X^{(k)}$, they will be dropped in the variation above. Hence, this can be viewed as another manifestation of $a_i$ being pure gauge in the bulk. We have verified by brute force that for $k=2$ the variation in \eqref{varX} indeed gives the Weyl-Bach tensor when $d=4$, and a rigorous proof for any $k$ is worth further study. 
\par
Based on the FG version of relation \eqref{varX}, there is another approach of finding the (LC) obstruction tensors and Weyl anomaly in even dimensions called the dilatation operator method \cite{Anastasiou:2020zwc}. This method is briefly introduced in Appendix~D of \cite{Jia:2021hgy}, where the $8d$ Weyl anomaly was computed in the FG gauge. As a consistency check, the $8d$ FG result in \cite{Jia:2021hgy} agrees with \eqref{8dA} when the $a_i$ is turned off.

\section{Discussion}
\label{sec:disc1}
As the main result of Part I from the physics side, we computed the Weyl anomaly up to $8d$ in the WFG gauge and showed that they can be expressed using Weyl-Schouten tensor and extended Weyl-obstruction tensors as the building blocks. These results indeed go back to the corresponding FG results when the Weyl structure $a_\mu$ is turned off, but now they become Weyl-covariant. By observing the pattern of the Weyl anomaly in different dimensions, we suspect there exists a general formulation that can generate the holographic Weyl anomaly in any dimension, which will be explored in future work.
\par
In the boundary field theory, both the induced metric $\gamma^{(0)}_{\mu\nu}$ and the Weyl connection $a^{(0)}_\mu$ are non-dynamical background fields. However, only $\gamma^{(0)}_{\mu\nu}$ is sourcing a current operator, namely the energy-momentum tensor, while $a^{(0)}_\mu$ does not source any current since $a_\mu$ is pure gauge in the bulk. From the Weyl-Ward identity \eqref{holoWeylWard}, we can see that the trace of the energy-momentum tensor obtains a contribution from $p^{(0)}_\mu$ due to the gauge freedom of WFG. Together we can regard it as an improved energy-momentum tensor $\tilde T_{\mu\nu}$. For non-holographic field theories with background Weyl geometry the corresponding Weyl current $J^\mu$ of the Weyl connection does not need to vanish. The Weyl current in the general case deserves further investigation.
\par
An important corollary in our analysis is that the Weyl structure $a_{\mu}$ only appears as a trivial cocycle in the Weyl anomaly, and thus only contributes cohomologically trivial modifications. From the Weyl anomaly up to $8d$ we can directly see this for the subleading terms of $a_{\mu}$ as they appear only in total derivative terms in $X^{(k)}$. For the leading term $a_{\mu}^{(0)}$ this is less obvious since it plays the role of the boundary Weyl connection, but one can verify that by writing the anomaly in terms of the boundary LC connection, the terms involving $a_{\mu}^{(0)}$ also represent trivial cocycles. This indicates a striking feature of the WFG gauge, namely $a^{(0)}_\mu$ manages to make the expressions Weyl-covariant without introducing new central charges, which, once again, is consistent with the fact that $a_\mu$ is pure gauge in the bulk. Nonetheless, these cohomologically trivial terms might have significant effects in the presence of corners, i.e.\ spacelike codimension-2 surfaces. This may be analyzed using the construction proposed in \cite{Ciambelli:2021vnn,Freidel:2021cjp,Ciambelli:2022cfr}.
\par
Finally, although this part of the thesis focuses on the holographic Weyl anomaly, we believe that the (Weyl-) Schouten tensor and extended (Weyl-) obstruction tensors can also be used as the building blocks for the Weyl anomaly of other theories in general. How can these building blocks arise in a non-holographic context requires a deep understanding of the Lorentz-Weyl structure of a frame bundle, which encodes all the local Lorentz and Weyl transformations. 
Furthermore, the pattern we have observed for the holographic Weyl anomaly in different dimensions is reminiscent of the structure of the chiral anomaly across various dimensions, with the latter being understood as derived from the Chern class in different dimensions. This similarity suggests the potential for a cohomological interpretation of the Weyl anomaly. These observations motivate Part II of this thesis. In Subsection \ref{sec:LW}, we will revisit these issues and formulate the Weyl and Lorentz anomalies in a geometric fashion.

\part{Lie Algebroid Cohomology and Quantum Anomalies}

\chapter{Introduction}
\label{chap:intro2}
\section{An Overview on Anomalies}
\label{sec:anomalyoverview}
Symmetry has always been central to modern physics. Two monumental moments of symmetry in physics are when Emmy Noether \cite{Noether:1918zz} established the connection between symmetry and conservation laws in classical physics and when Eugene Wigner \cite{Wigner1931} and Hermann Weyl \cite{Weyl1928} introduced group theory to quantum physics. Since then, research on symmetry has played a prominent role in all areas of physics and continues to thrive today. For example, spacetime symmetries, including the Weyl symmetry discussed in Part I, are significant in relativity and gravity; internal symmetries, such as isospin, color, and flavor symmetries, are crucial in particle and nuclear physics; crystal symmetries are essential in solid state physics, particularly in the study of band structures, etc. Over the past decade, the concept of symmetry has further expanded in various directions, leading to the development of generalized symmetries, including higher form symmetries \cite{Gaiotto:2014kfa}, subsystem symmetries \cite{Lawler_2004,Vijay:2016phm,Seiberg:2020bhn}, and non-invertible symmetries \cite{Bhardwaj:2017xup,Chang:2018iay,Shao:2023gho}.
\par
While the fundamental laws of nature exhibit a high degree of symmetry, the observable world is remarkably asymmetric and diverse. Thus, it is crucial to study both the symmetries inherent in nature and the various mechanisms by which these symmetries are broken. There are three major types of symmetry breaking, each providing rich physics to explore: \ding{172} explicit symmetry breaking (and approximate symmetries), \ding{173} spontaneous symmetry breaking, and \ding{174} quantum anomalies. In this thesis we will focus on the study of quantum anomalies, which is the phenomenon when the symmetry of a classical theory fail to be hold for the corresponding quantum theory.
\par
Quantum anomalies were first discovered through the violation of chiral symmetry in quantum electrodynamics (QED), manifested by the non-conservation of the axial current \cite{Schwinger:1951nm,Johnson:1963vz}. This phenomenon, known as the \emph{chiral anomaly} or \emph{Adler–Bell–Jackiw (ABJ) anomaly}, resolved the discrepancy between the theoretical calculations and experimental results of the decay rate of the neutral pion ($\pi^0\to\gamma\gamma$) \cite{adler1969axial,bell1969pcac}. This indicates that symmetry violations in quantum theory are not flaws but essential features that reveal the fundamental quantum nature of the theory.
\par
The chiral anomaly was computed perturbatively from the 1-loop Feynman diagram (in $4d$ it is the famous triangle diagram) of the fermionic theory, where the symmetry breaking is caused by the regularization process. Equivalently, it can also be derived from the transformation of the path integral measure \cite{fujikawa1979path}. Despite that the classical action is invariant under the symmetry transformation, the path integral measure of the fermion fields will acquire a nontrivial Jacobian, which under regularization gives rise to a phase factor to the transformed path integral.\footnote{Although anomalies were originally understood in the context of fermionic theories under regularization, it was later realized that anomalies also occur in bosonic theories and in cases even without the introduction of a regulator (see, e.g., \cite{Cheng:2022sgb}).} 
\par
To be precise, consider a fermionic theory defined on a $d$-dimensional manifold $M$ with a continuous symmetry described by a Lie group $G$ (we will also refer to the symmetry as $G$ for convenience)
\be
Z[A]=\E^{\I W[A]}=\int D\psi D\bar\psi \E^{\I S[\psi,\bar\psi,A]}\,,
\ee 
where we introduced a background field $A$ for the symmetry, and $W[A]$ represents the quantum effective action. Under an infinitesimal transformation parametrized by $\epsilon$, the path integral measure is not invariant, which leads to
\be
Z[A+\delta_\epsilon A]=\E^{\I\mathfrak{a_{\rm con}}}Z[A]=\E^{\I\int_M \epsilon(x) a_{\rm con}(x)}Z[A]\,,
\ee 
where the anomaly density $a_{\rm con}(x)$ is a $d$-form. In terms of the quantum effective action $W[A]$, this can be written as
\be
\label{deltaW}
\delta_\epsilon W[A]=\int D\epsilon\wedge^* \frac{\delta W[A]}{\delta A}=\mathfrak{a}_{\text{con}}=\int_M \epsilon(x) a_{\rm con}(x)\,,
\ee 
Recognizing the current $\langle J^\mu\rangle=\delta W[A]/\delta A_\mu$ (the index $\mu$ denotes the coordinate components), we have the anomalous Ward identity
\be
\langle D^*J\rangle=-a_{\rm con}(x)\,,
\ee
which can be viewed as the quantum version of the Noether's theorem, where now the right-hand side can be non-vanishing due to the quantum effect. For chiral anomaly in $2d$ we have $a_{\rm con}(x)=-\td A$. 
\par
It is important to note that $ Z[A] $ can always be modified by local counterterms defined on $M$, reflecting different choices of regularization schemes. Therefore, we only consider the anomalous phases of $ Z[A] $ that cannot be removed by local counterterms. This statement can be encapsulated by the Wess-Zumino consistency condition \cite{wess1971consequences}, and hence $ a_{\text{con}} $ represents the so-called \emph{consistent anomaly}. As we will see shortly, this signifies the cohomological nature of anomalies. However, for a non-Abelian symmetry, $ a_{\text{con}} $ is not covariant under gauge transformations. One can covariantize the consistent anomaly by adding the Bardeen-Zumino polynomials to the anomalous current and obtain the \emph{covariant anomaly} \cite{bardeen1984consistent}, as we will review in Section \ref{sec:BRST}. For chiral anomaly in $2d$, the covariant anomaly reads $ a_{\text{cov}}(x) = -2 F $, where $ F = \td A +\frac{1}{2} [A, A] $ is the curvature of $ A $. 
\par
The physical meaning of the anomaly derived from the above algebra can have different interpretations. If we treat the symmetry $ G $ as a global symmetry and turn on a non-dynamical background field $ A $ to probe the anomaly, the resulting anomaly is called a \emph{'t Hooft anomaly} \cite{tHooft:1979rat,}. The presence of 't Hooft anomalies does not cause any inconsistency in the quantum theory, and the symmetry is still preserved as long as we do not turn on the background field and make it local. On the other hand, if $ G $ is a gauge symmetry, the same algebra still applies, but $ A $ becomes a dynamical gauge field that gets integrated in the path integral, resulting in a \emph{gauge anomaly}. Since gauge symmetries represent redundancies in the theory, breaking gauge symmetry leads to inconsistencies in the path integral. Thus, gauge anomalies must not be present in a consistent quantum theory. There is also a third case, namely the mixed anomaly between global and gauge symmetries.\footnote{A mixed anomaly arises when two subgroups of $ G $ cannot be non-anomalous simultaneously. This concept also applies when both subsymmetries are either global or gauge symmetries.} In this case, while the global symmetry is broken, the theory remains well-defined. The ABJ anomaly is an example of this, where $ G = U(1)_A \times U(1)_V $, and the current for the axial symmetry $ U(1)_A $ is anomalous due to the gauge field of the vector symmetry $ U(1)_V $.
\par
From a theoretical perspective, quantum anomalies have two key utilities. First, an important property of the 't Hooft anomaly is that it is preserved under an RG flow as long as the symmetry is maintained \cite{tHooft:1979rat,}. That is, the anomaly we find for the same symmetry in the UV theory must also be present in the IR theory, and vice versa. This concept, known as \emph{'t Hooft anomaly matching}, provides an important handle for studying the IR dynamics of quantum field theory, which is typically inaccessible through analytical methods. The existence of an anomaly prevents the IR theory from being trivially gapped, constraining it to one of three possibilities \cite{Cordova:2019bsd}: \ding{172} spontaneous symmetry breaking, \ding{173} a gapless theory (CFT), or \ding{174} topologically ordered (TQFT). This approach has proved to be powerful for understanding the phase diagram of the Yang-Mills theory and quantum chromodynamics (QCD) \cite{Gaiotto:2017yup,Cordova:2019bsd,Gomis:2017ixy,}, as well as the Lieb-Schultz-Mattis (LSM) theorem in condensed matter systems \cite{Lieb:1961fr,Cheng:2015kce,Cho:2017fgz}.
\par
The second utility of anomalies is that in any physical theory, gauge anomalies must be completely canceled. This imposes crucial constraints for model building. For example, in the Standard Model, the hypercharges of leptons and quarks are constrained by the anomaly cancellation condition, and the numbers of quark and lepton generations are restricted to be equal \cite{Weinberg:1996kr}. Another famous example is the Green–Schwarz mechanism in superstring theory, where anomaly cancellation restricts type I string theory to have specific gauge groups such as $SO(32)$ or $E_8 \times E_8$ \cite{Green:1984sg}.
\par
Although anomalies cannot be removed by local counterterms on the $ d $-dimensional manifold $ M $, they can generally be canceled by local counterterms in one higher dimension (which are nonlocal on $ M $) through \emph{anomaly inflow}. This mechanism was first observed by Callan and Harvey for the chiral anomaly of domain wall fermions and bulk Chern-Simons theory \cite{callan1985anomalies}, and was soon recognized as essential for understanding the quantum Hall effect and topological phases \cite{Fradkin:1986pq,Stone:1990iw}. Based on this bulk-boundary correspondence picture, in the modern description, anomalies on $M$ are characterized as an invertible topological quantum field theory (TQFT) on a $(d+1)$-dimensional manifold $ \tilde{M} $ with boundary $ \partial \tilde{M} = M $ \cite{Freed:2014iua,Freed:2016rqq,Monnier:2019ytc}. Invertible field theories are the low-energy effective theories of \emph{symmetry protected topological (SPT) phases} \cite{Gu:2009dr,Chen:2011pg,Senthil:2014ooa}. This understanding of anomalies highlights a profound interplay between quantum field theory, condensed matter physics, and mathematical physics.

In Part II of this thesis, one of our main goals is to explore the topological aspects of anomalies. The appropriate mathematical framework for studying anomalies is cohomology. The intersection of gauge theory and cohomology arises through Chern-Weil theory, which establishes a correspondence between characteristic classes, symmetric invariant polynomials in curvature, and cohomology classes \cite{WeilLetter,zbMATH03070474}. Chern demonstrated in \cite{chern1946characteristic} that characteristic classes quantify obstructions to the existence of global sections on a principal bundle $P(M,G)$, providing access to topological data about the base manifold $M$. Then, the topological nature of an anomaly can be captured by a characteristic class in $(d+2)$-dimension, known as the anomaly polynomial, whose integral is an integer known as the \emph{Atiyah-Singer index} \cite{Atiyah:1968mp,atiyah1984dirac}. Historically, this was considered the mathematical description of anomalies, as the geometric and topological structure of anomalies stems from those of the gauge fields \cite{alvarez1984topological,witten1983global}, which are connections on principal bundles (see the next subsection).
\par
However, the formulation of anomalies as characteristic classes of principal bundles is not quite appropriate. A key observation is that the exterior algebra of the principal bundle can be organized into a bi-complex combining the de Rham cohomology of the base manifold and the Chevalley-Eilenberg cohomology of the Lie algebra \cite{chevalley1948cohomology}. A main issue of this is that the Lie algebra alone does not capture the local nature of gauge symmetry. The resolution of this issue is achieved through the BRST cohomology. As will be outlined in Section \ref{sec:BRST}, the possible algebraic forms of anomalies are successfully derived from the Wess-Zumino consistency condition as part of the descent equations \cite{Stora1977,Stora:1983ct,zumino1985chiral}.
\par
In this thesis, we emphasize that the BRST cohomology is not yet the complete picture of characterizing anomalies, as this approach only determines the consistent anomaly. Additional manipulations are necessary to obtain the covariant anomaly. Therefore, we would like to develop a suitable framework that generalizes the BRST cohomology and incorporates the cohomology of the covariant anomaly. In the next subsection, we will elucidate that the BRST formalism can be naturally geometrized by a mathematical structure called the Lie algebroid, and the cohomology within this framework precisely serves our purpose. Motivated by Part I, we will also investigate the cohomological interpretation of the Weyl anomaly in this framework.
\par
Finally, we would like to emphasize that the anomalies we consider in this part all correspond to violations of the conservation law of a symmetry current, which are referred to as \emph{perturbative anomalies} as they can be derived from a given QFT using perturbative methods. However, this is not the end of the story of anomalies. There are two kinds of anomalies that do not correspond to any symmetry current and are intrinsically non-perturbative. One type is known as global anomalies,\footnote{By ``global" it refers to the global structure of the gauge group, rather than saying that the symmetry is a global symmetry.} which are anomalies of large gauge transformations (e.g., the $SU(2)$ anomaly \cite{Witten:1982fp,Wang:2018qoy}), and the other type is anomalies for discrete symmetries (e.g., the parity anomaly \cite{Niemi:1983rq,Redlich:1983kn,Redlich:1983dv,Alvarez-Gaume:1984zst}), these anomalies are also relevant in subjects such as particle physics, string theory and topological insulators. In the modern description of anomalies, it has been proposed that non-perturbative anomalies are also characterized by SPT phases in one higher dimension and can be unified with perturbative anomalies \cite{Witten:2015aba,Witten:2019bou}. Since non-perturbative anomalies do not correspond to characteristic classes in $ d+2 $ dimensions, in the unified picture the Atiyah-Singer index is upgraded to the \emph{Atiyah-Patodi-Singer $\eta$-invariant} \cite{Atiyah:1975jf,Dai:1994kq}. The mathematical framework for classifying anomalies in this general context is called \emph{cobordism} \cite{Kapustin:2014tfa,Kapustin:2014dxa,Yonekura:2018ufj}. There are still many open questions in the general study of anomalies, and we will leave them as future directions, building on insights from our construction.

\section{Geometric Formulation of Gauge Theories}
Yang-Mills theory \cite{Yang:1954ek} is the cornerstone of modern theoretical physics, providing a profound framework for understanding the fundamental interactions in Nature. At the core of the Yang-Mills theory lies the concept of gauge fields, which transform nonlinearly under gauge transformations, ensuring the gauge invariance of the theory. The background field $A$ we introduced in the last subsection for a symmetry $G$ plays precisely this role. In the Yang-Mills theory, one also includes the kinetic term of the field $A$, and in the quantized theory $A$ is integrated over in the path integral (with further gauge-fixing procedures to be discussed later). Physically, these quantized gauge fields mediate the interactions between elementary particles.
\par
For the classical Yang-Mills theory, principal bundles and their associated bundles offer an elegant geometric formulation \cite{Atiyah:1957,Wu:1975es,Atiyah:1978,RevModPhys.52.175,Jordan:2014uza}. Given a principal $G$-bundle $P(M,G)$, the base manifold $M$ represents the physical spacetime and structure group $G$ describes the gauge symmetry. Then a gauge field $A$ on $M$ corresponds to a connection $\bb A$ on $P$, a gauge choice corresponds to a local section on $P$, the gauge strength $F$ of $A$ corresponds to the curvature $\bb F$ of $\bb A$, a local section $\un\psi$ on an associate bundle $E$ corresponds to a matter field, and the induced connection $\nabla$ on the associate bundle corresponds to the covariant derivative of the matter field, etc. This beautiful correspondence, first published by Wu and Yang in \cite{Wu:1975es} and dubbed the Wu-Yang dictionary, is one of the most striking examples of how physical theories and mathematical structures, despite being developed independently, can be seamlessly interwoven into each other.
\par
The situation for quantum gauge theory, however, is more subtle. The path integral over the gauge field includes an infinite amount of gauge redundancy, and we should only count the physically distinct configurations of the gauge field. This is achieved through the Faddeev-Popov procedure \cite{Faddeev:1967fc}, which fixes the gauge at the cost of introducing unphysical degrees of freedom called \emph{ghosts}. These ghost fields have the ``wrong" statistics: they are scalars on the spacetime $M$ but anticommute. Naturally, one might ask if there is a geometric interpretation for ghosts in the language of principal bundles.

The historical approach to the geometric analysis of quantum gauge theories involves the \emph{Becchi-Rouet-Stora-Tyutin (BRST) formalism}, which was originally introduced to formalize the Faddeev-Popov approach of gauge quantization \cite{becchi1974abelian, becchi1976renormalization, Tyutin:1975qk}. In this formalism, the gauge transformation is extended to the BRST transformation, which acts not only on matter fields and gauge fields but also on ghost fields. Under the BRST transformation, the theory remains invariant even after fixing the gauge and introducing ghosts. The action of the BRST transformation is realized by a nilpotent BRST operator. The physical Hilbert space is then constructed from the cohomology of this BRST operator.
\par
It was subsequently realized that the BRST formalism gives rise to an exterior bi-algebra, later dubbed the BRST complex \cite{Bandelloni:1986wz, Henneaux:1989rq, Henneaux:1995ex, Dragon:1996md, DuboisViolette:1991is, Brandt:1996mh, Barnich:2000zw}, which can be used to calculate the cohomology classes relevant to quantum anomalies \cite{adler1969axial,bell1969pcac,t1976symmetry,fujikawa1979path,bonora1983some,alvarez1984topological,witten1983global,zumino1985chiral,Nelson:1984gu}. Starting from a principal bundle $P(M,G)$, the basic objective of the BRST complex is to design an exterior algebra that combines the de Rham cohomology of the base manifold $M$ with the cohomology of the local gauge algebra associated with the structure group $G$. The BRST complex accomplishes this task in a series of steps. First, it takes a local section of $P(M,G)$ to define the gauge field $A$, which descends from a bona-fide principal connection. In this way, it forgets about the vertical sub-bundle of $TP$, and restricts its attention only to the de Rham cohomology of the base manifold. Next, the vacuum left behind by the vertical sub-bundle is filled by introducing a graded algebra generated by a set of Grassmann valued fields $c^A(x)$ representing the ghosts (encoding its anticommuting nature). In this way, one obtains the BRST complex as an exterior bi-algebra consisting of $p$-forms on $M$ contracted with $q$ factors of the ghost field, where the number $q$ is referred to as the ghost number. 

Now we return to the geometric interpretation of ghosts. A priori, ghost fields have no geometric interpretation, rather being introduced as a computational device in the gauge quantization. However, it has been argued that a geometric interpretation for the ghost fields exists as the ``vertical components" of an extended gauge field \cite{thierry1980explicit,thierry1980geometrical,quiros1981geometrical,Thierry-Mieg:1985zmg,Neeman:1979cvl, Baulieu:1981sb, ThierryMieg:1982un, ThierryMieg:1987um,Bonora:1981rw,Hirshfeld:1980te, Hoyos:1981pb, Baulieu:1985md,Hull:1990qg,Bonora:2021cxn}. The basic idea behind this interpretation is to contract the ghost fields with the set of Lie algebra generators $c = c^A \otimes \un{t}_A$ and define the extended ``connection" form  $\econn = A + c$ by appending the ghost field to the gauge field. Viewing $\econn$ as a connection, it is natural to define an associated curvature $\ecurv = \td_{\rm BRST}\econn + \frac{1}{2}[\econn,\econn]$, where the coboundary operator of the BRST complex is identified as $\td_{\rm BRST} = \td + \ts$, which is simply the combination of the de Rham differential $\td$ and the BRST operator $\ts$. Enforcing the extra condition that the curvature should have extent only in the de Rham part of the BRST complex, one arrives at a pair of equations defining the action of the BRST operator which can be identified with the Chevalley-Eilenberg differential appearing in Lie algebra cohomology \cite{chevalley1948cohomology,Zumino:1984ws,Brandt:1989gv}. 

With the ``connection" $\econn$, ``curvature" $\ecurv$, and coboundary operator $\td_{\rm BRST}$ in hand, one can construct ``characteristic classes" in the BRST complex by naively following the Chern-Weil theorem \cite{WeilLetter,zbMATH03070474}. Due to the fact that $\ecurv$ was manufactured to have zero ghost number, the Chern-Simons form associated with a given characteristic class in the BRST complex can be shown to satisfy a series of equations known as the descent equations \cite{Zumino:1984ws,Stora1977,zumino1984chiral,Sorella:1992dr}. One of the resulting equations is the Wess-Zumino consistency condition \cite{wess1971consequences}, which ultimately determines the algebraic form of candidates for quantum anomalies. 

The success of the BRST approach is undeniable. However, it motivates a series of questions. Why should the Grassmann valued fields $c^A(x)$, which started their life in the BRST quantization procedure have an interpretation as the generators of a local gauge transformation? Why is it reasonable to combine the de Rham complex and the ghost algebra into a single exterior bi-algebra? On a related note, why is it reasonable to consider the combination $\econn = A + c$ as a ``connection", and moreover what horizontal distribution does it define? Why should the ``curvature" $\ecurv$ be taken to have ghost number zero, and why does enforcing this constraint turn the BRST operator $\ts$ into the Chevalley-Eilenberg operator for the Lie algebra of the structure group? These are the questions that we will answer in Part II of this thesis. In fact, we will show that there is not an answer to each of these questions individually, but rather each of these individual questions are resolved by the answer to a single question: what is the appropriate geometric interpretation for the BRST complex? Indeed, our main objective will be to demystify the BRST complex once and for all, and in doing so provide a unified geometric picture of quantum anomalies. The mathematical language capable of this task extends beyond that of principal bundles and is found in the framework of Lie algebroids. \cite{MR0214103, MR0216409, zbMATH01186367, kosmannschwarzbach2002differential, fern2007lie, mackenzie_1987,mackenzie2005general}.
\par
Lie algebroids is a generalization of the more familiar Lie algebras to the setting of smooth manifolds, which also captures the algebraic structure of tangent bundles. They were first formally introduced in \cite{MR0216409} as the infinitesimal generating algebras for Lie groupoids, which are a categorical generalization of Lie groups. Although they may not be well-known to the majority of physicists, Lie algebroids have already found a variety of applications in mathematical physics \cite{Roytenberg:2002nu,Kosmann-Schwarzbach:2003skv,Blohmann:2010jd,Blumenhagen:2012nt,Ramandi:2014qoa}. In particular, in the context of formulating gauge theories, discussions can be found in, e.g., \cite{Ramandi:2014qoa,marle2008differential,lazzarini2012connections, Fournel:2012uv, Carow-Watamura:2016lob, Kotov:2016lpx, Attard:2019pvw,Ciambelli:2021ujl,Klinger:2023auu} and the citations therein. In \cite{Ciambelli:2021ujl}, it was argued that the exterior algebra of an Atiyah Lie algebroid derived from a principal $G$-bundle is a geometrization of the physicist's BRST complex. In this thesis, we will provide a novel perspective on this correspondence by elaborating on the concept of the \emph{Lie algebroid trivialization}, which extends the discussion in \cite{Ciambelli:2021ujl} further, and base on the this framework have a geometric understanding of quantum anomalies. Building on this framework, we seek to achieve a geometric understanding of the BRST complex and quantum anomalies.

\par

At the end of this introduction, we supply two important concepts which we will encounter frequently throughout Part II of this thesis, namely exact sequences and the curvature of a map.
\bdefi
Suppose $A_i$ ($i=0,1,2\cdots$) is a series of sets and $\phi_i:A_i\to A_{i+1}$ is a series of maps, together they can be expressed as a sequence
\be
\label{exactseq}
\begin{tikzcd}
A_0 \arrow[r, "\phi_0"] & A_1 \arrow[r, "\phi_1"] & A_2 \arrow[r, "\phi_2"] & \cdots \arrow[r, "\phi_{i-1}"] & A_i \arrow[r, "\phi_{i+1}"] & \cdots 
\end{tikzcd}
\ee
This sequence is called an \emph{exact sequence} if $\text{im}(\phi_i)=\ker(\phi_{i+1})$ $\forall i=0,1,2,\cdots$. An exact sequence of the form
\be
\begin{tikzcd}
0 \arrow[r, ] & A_1 \arrow[r, "\phi_1"] & A_2 \arrow[r, "\phi_2"] & A_3 \arrow[r, "\phi_3"] & 0
\end{tikzcd}
\ee
is called a \emph{short exact sequence}. In this case we have $A_3=A_2/A_1$.
\edefi

\bdefi
Suppose $A$ and $B$ are spaces with algebra structures defined by brackets $[\cdot,\cdot]_{A}$ and $[\cdot,\cdot]_{B}$, respectively. Then, given any $a_1,a_2\in A$, the \emph{curvature} of a map $f: A \rightarrow B$ is defined as
\begin{equation}
	R^{f}(a_1,a_2) = [f(a_1),f(a_2)]_{A} - f([a_1,a_2]_{B})\,.
\end{equation} 
The map $f$ is called a \emph{morphism} if $R^{f}$ vanishes $\forall a_1,a_2\in A$. In other words, the curvature of a map measures the failure of the map to be a morphism. The morphism that is a bijection is called an \emph{isomorphism}.\footnote{This is not generally true in category theory, where a bijective morphism is called a \emph{bimorphism}, which is weaker than an isomorphism. However, in this thesis we will not need to worry about this subtlety. }
\par 
For clarity, later on we will always denote the bracket structure of a space $A$ by $[\cdot,\cdot]_{A}$. 
\edefi

\section{Organization of Part II}
The rest of Part II is organized as follows. 

In Chapter \ref{chap:topological}, we introduce the traditional cohomological approach to quantum anomalies using the language of principal bundles. To make this thesis self-contained, in Section \ref{sec:primer} we provide a primer on principal bundles, formulated to facilitate the discussion of Lie algebroids in later chapters. Following that, Section \ref{sec:cohomology} offers a crash course on the basics of algebraic topology, aiming to explain necessary notions relevant to the later analysis. We then construct cohomology classes by utilizing the Chern-Weil theorem, which relates cohomology classes with characteristic classes as invariant polynomials of curvature. In Section \ref{sec:BRST}, we review the description of anomalies from the BRST complex and demonstrate its deficiencies. We also briefly review the consistent and covariant anomalies and their anomaly inflow pictures.

In Chapter \ref{chap:algebroid}, we provide a general pedagogical introduction to Lie algebroids, paving the way for our discussions on gauge theory and anomalies. We discuss in detail various equivalent descriptions of connections and curvatures on Lie algebroids. Through the representation of Lie algebroids, we introduce a coboundary operator $\hatd$, which defines the Lie algebroid cohomology. Given the unfamiliarity of physicists with Lie algebroids, we aim to provide step-by-step derivations of the formulas in elucidating the core properties of relevant notions. Some lengthy calculations in explaining the properties of Lie algebroids are presented in Appendices \ref{app:dhat}, \ref{app:ROmega}, and \ref{app:commutation}.

After the abstract discussion of Lie algebroids, Chapter \ref{chap:trivialized} focuses on the Lie algebroids derived from principal bundles, namely Atiyah Lie algebroids. In Section \ref{sec:ALA} we begin by reviewing the construction of Atiyah Lie algebroids derived from a principal bundles, and then introduce their local trivializations. In Section \ref{sec:morphisms} we discuss the role of Lie algebroid isomorphisms between Atiyah Lie algebroids and demonstrate how they can be interpreted as implementing both gauge transformations and diffeomorphisms in physical contexts. In Subsection \ref{sec:local} we apply Lie algebroid isomorphisms as a tool for studying Lie algebroid trivializations in a global context. In Subsection \ref{sec:BRSTcohomology} we study trivializations of the exterior algebra associated with an Atiyah Lie algebroid, and demonstrate that the resulting cohomology is equivalent to that of the BRST complex. Appendix \ref{app:LAT} includes some calculation details.

Finally, in Chapter \ref{chap:anomaly} we apply the lessons from the previous Chapters to study quantum anomalies. Section \ref{sec:Chernalg} carries over the Chern-Weil construction of characteristic classes in Section \ref{sec:cohomology} to the Lie algebroid context. In this framework, the Atiyah Lie algebroid cohomology can directly quantify both the consistent and covariant anomaly polynomials, which will be demonstrated in Sections \ref{sec:consistent} and \ref{sec:covariant}, respectively. Then, as concrete examples, we apply this machinery to computing chiral anomaly and the Lorentz-Weyl anomaly explicitly in Section \ref{sec:examples}. We conclude in Section \ref{sec:disc2} in which we provide answers to the questions posed in this introduction, address directions for follow up work, and comment on the overall lessons regarding Weyl anomaly from both parts of this thesis.

The results presented in Part II sourced mostly from the joint research work \cite{Jia:2023tki} with the author's advisor Robert~G.~Leigh, and collaborator Marc~S.~Klinger. The review sections on Lie algebroids in Chapter \ref{chap:algebroid} and Section \ref{sec:ALA} are mainly inspired by \cite{Ciambelli:2021ujl}.

\section{Notation}
We use lowercase Greek letters $\mu,\nu,\cdots$ for the indices on a base manifold $M$, underscored Latin letters $\un M,\un N,\cdots$ for the indices on the Lie algebroid $A$, uppercase Latin letters $A,B,\cdots$ for the indices of Lie algebra $\mg$ and the isotropy bundle $L$ of $A$, and lowercase letters $a,b,\cdots$ for the indices on a vector bundle $E$. In a split basis of $A$, the indices for the horizontal sub-bundle $H$ is denote by underscored Greek letters $\un\alpha,\un\beta,\cdots$, and the indices for vertical sub-bundle $V$ is denote by underscored Latin letters $\un A,\un B,\cdots$.
\par
On a principal bundle $P$, we denote the connection and curvature forms as $\bb A$ and $\bb F$. On a Lie algebroid $A$, we will denote the connection and curvature reforms as $\omega$ and $\Omega$. The local gauge field in a open set $U\in M$ defined in a local trivialization $T_U$ of principal bundle is denoted by $A_U$ and that defined in a local trivialization $\tau_U$ of principal bundle is denoted by $b_U$. The curvature for local gauge field in both cases is denoted by $F_U$. The label $U$ will be omitted in some sections for brevity.
\par 
In Chapter \ref{chap:topological}, we denote the exterior algebra on $M$ using the standard notation $\Omega^p(M)\equiv\Gamma(\wedge^p T^*M)$. Starting from Chapter \ref{chap:algebroid}, as we will mainly focus on vector bundles, we will adopt the notation $\Omega^p(A)\equiv\Gamma(\wedge^p {}^*A)$; for example, $\Omega^p(M)$ will then be denoted by $\Omega^p(TM)$.
\par
The notation for various bundles including their sections, basis and dual basis is summarized in Table \ref{t1}.

\begin{table}[!h]
\centering
\caption{Notation for Part II}
\begin{tabularx}{\textwidth}{c|c|c|c|X}
\toprule
Bundle & Sections & Basis & Dual basis & Indices \\
\midrule
$TM$ & $\un X,\un Y$ & $\{\un\p_\mu\}$ & $\{\td x^\mu\}$ & $ i=1,\cdots,\dim M$\\
\midrule
$TP$ & $\un u,\un v$ &  &  & \\
\midrule
$A$ & $\umX,\umY$ & $\{\un E_{\un M}\}$ or $\{\un E_{\ualpha},\un E_{\un A}\}$ & $\{E^{\un M}\}$ or $\{E^{\ualpha},E^{\un A}\}$ &  $\un M=1,\cdots,\dim M+\dim G$,\newline$\ualpha=1,\cdots, \dim M$, $\un A=1,\cdots,\dim G$  \\
\midrule
$L$ & $\un\mu,\un\nu$ & $\{\un t_A\}$ & $\{t^A\}$ & $A=1,\cdots,\dim G $\\
\midrule
$E$ & $\un\psi$ & $\{\un e_a\}$ &$\{f^a\}$ & $a=1,\cdots,\text{rank}\,E$\\
\bottomrule
\end{tabularx}
\label{t2}
\end{table}

\chapter{Topological Obstructions and Anomalies}
\label{chap:topological}

Characteristic classes on principal bundles quantify topological obstructions to defining a global section, and anomalies arising in quantum field theories with auxiliary background fields quantify the obstructions to global gauge fixing. Since a section on a principal bundle corresponds to a gauge choice in gauge theory, this implies that characteristic classes capture the deep topological foundation of anomalies. After reviewing the relevant geometric and topological setup, in this chapter we will introduce how the topological nature of anomaly is formulated in terms of the BRST cohomology in the principal bundle picture and discuss its limitations. For a more detailed discussion on the theory of principal bundles and their applications in physics, see \cite{nakahara2018geometry} or \cite{Liang3}. For an in-depth introduction to algebraic topology, see \cite{Hatcher}.

\section{Geometry of Principal Bundles}
\label{sec:primer}
\subsection{Principal Bundles and Connections}
A \emph{principal $G$-bundle} consists of a bundle manifold $P$ called the \emph{total space}, a group $G$ called the \emph{structure group} and a \emph{base manifold} $M$. It is equipped with the following pair of maps:
\begin{equation} \label{P(M,G)}
R: P\times G \rightarrow P\,,  \qquad  \pi: P \rightarrow M\,,
\end{equation}
where $R$ is a free right action, and $\pi$ is a projection map. The map $R$ being a right action means that given $g\in G$, $R_g:P\to P$ is a diffeomorphism such that $R_{gh}=R_g\circ R_h$, $\forall g,h\in G$. Denote the image  $R_g(p)$ as $pg$ for short. $R$ being a free action means that $pg\neq p$ $\forall p\in P$ if $g\neq e$, where $e$ is the identity of $G$. $\pi$ being a projection map satisfies $\pi^{-1}[\pi(p)]=\{pg|g\in G\}$, $\forall p\in P$. Given $p\in P$, $R$ also gives rise to $R_p:G\to P$, which is an embedding of $G$ in $P$. We will refer to such a principal bundle as $P(M,G)$, or by the sequence of maps $G \rightarrow P \rightarrow M$. 

Locally, i.e.\ in a subregion $P\rvert_U = \pi^{-1}[U]$ over an open subset $U \subset M$, we require that $P\rvert_U \simeq U \times G$. More precisely, for any open subset $U\subset M$ there exists a diffeomorphism $T_U:P\rvert_U \to U \times G$, called a \emph{local trivialization}, such that $T_U(p)=(\pi(p),g_U(p))$, where $g_U:P\to G$ satisfies $g_U(ph)=g_U(p)h$, $\forall h\in G$. Suppose $\dim M=d$ and $\dim G=r$, it is natural to assign coordinates on the principal bundle through a pair of atlases consisting of charts, $\phi: U \rightarrow \mathbb{R}^d$ defining coordinates on $U$, and $\psi: G \rightarrow \mathbb{R}^r$ specifying coordinates in a connected open subset of $G$. For simplicity, we will refer to these coordinates on $P|_U$ as $(x,g)$, with $x = (x^1, \cdots, x^d)$ coordinates for $U$, and $g = (g^1, \cdots, g^r)$ fiber coordinates for $G$. Given two local trivializations $T_U:P\rvert_U \to U \times G$ and $T_V:P\rvert_V \to V \times G$ with $U\cap V\neq\varnothing$, one needs to define a map $t_{UV}:U\cap V\to G$ called a \emph{transition function} as $t_{UV}(x)=g_U(p)g_V^{-1}(p)$, $\forall x=\pi(p)\in U\cap V$, so that any point in $\pi^{-1}[U\cap V]$ will be map to the same point on $U\times G$ by $T_U$ and $T_V$. In this sense, the local trivialization is globally well-defined on $P$.
\par
Given an open subset $U\subset M$, a map $s_U:U\to P$ satisfying $\pi(s_U(x))=x$ $\forall x\in U$ is called a \emph{local section}. Once a local trivialization $T_U:P_U\to U\times G$ is given, each fiber has a special point $q$ such that $g_U(q)=e$. This naturally gives rise to a local section $s_U$. On the other hand, once a local section $s_U:U\to P$ is given, for any point $p$ on a fiber $\pi^{-1}[x]$ over $x\in U$ there exists a unique $g\in G$ such that $p=s_U(x)g$, which gives rise to a local trivialization $T_U(p)=(x,g)$. Therefore, this establishes a canonical correspondence between a local section and a local trivialization.
\par
The tangent space $T_pP$ at any $p\in P$ has a vertical subspace $V_p$ satisfying
\begin{equation}
V_p=\{\un v_p\in T_pP\,|\,\pi_*(\un v_p)=0\}\,.
\end{equation}
Since the group $G$ can be considered as generated from its Lie algebra $\mathfrak{g}$ by the exponential map: $\text{exp}: \mathfrak{g} \rightarrow G$, by means of the right action $R$, we can define a map $j_p: \mathfrak{g} \rightarrow V_p$ for any $p\in P$ as follows:
\begin{equation} \label{Isotropy Map}
    j_p(\underline{\mu}) := (R_{p})_*\un\mu=\frac{\td}{\td t}\bigg|_{t=0}[R_p\text{exp}(t\underline{\mu})] \,,\qquad\forall\un\mu\in\mathfrak{g}\,,
\end{equation}
which provides a canonical isomorphism between the Lie algebra $\mathfrak{g}$ and $V_p$. If we let $p$ run all over $P$, the resulting objects will become sections of $TP$, which defines a vector bundle over $P$, namely the \emph{vertical sub-bundle} $V_P$ of $TP$:
\begin{equation}
V_P=\{\un v\in \Gamma(TP)\,|\,\pi_*(\un v)=0\}\,.
\end{equation}
The map $j_p$ can subsequently be extended to a map $j:P\times\mathfrak{g} \rightarrow VP$. In the case we have the same $\un\mu\in\mathfrak{g}$ at each point of $P$, the resulting section of $V_P$ under $j$ is called the \emph{fundamental vector field} induced by $\un\mu$. It is important to notice that $\un\mu$ does not have the information of $M$, and hence this isomorphism identifies the Lie algebra of the structure group globally with the fundamental vector fields as sections of $VP$. 

A horizontal subspace is defined at each $p\in P$ as a distribution of vector fields such that: $T_pP = V_p \oplus H_p$, and $H_{pg} = R_{g*}[ H_p]$, $\forall g\in G$. Unlike the vertical subspace, there is no canonical definition of the horizontal subspace. Rather, by defining $H_p$ smoothly for all $p\in P$ we obtain a \emph{horizontal sub-bundle} $H_
P$ of $TP$, which is also referred to as a choice of \emph{connection} on $P$. There are several seemingly different but equivalent ways of defining a connection on $P$, i.e., specifying a choice of horizontal sub-bundle of $P$. First, a connection can be defined as a $\mathfrak{g}$-valued 1-form field on $P$ denoted by $\mathbb A\in\Omega^1(P;\mg)$, which is also a map $\mathbb{A}: TP\times P \rightarrow \mathfrak{g}$, satisfying
\begin{enumerate}
\item[(A1)] $\mathbb{A}|_{p}( j_p(\underline{\mu}) )= -\underline{\mu}$, $\forall\underline{\mu} \in \mathfrak{g}$; 
\par
\item[(A2)] $\mathbb{A}\rvert_{pg}((R_{g})_* \underline{v}\rvert_p) = \text{Ad}_{g^{-1}}(\mathbb{A}\rvert_p(\underline{v}_p))$, $\forall p\in P$, $\underline{v}_p \in T_p P$, $g \in G$.
\end{enumerate}
The horizontal subspace $H_p$ at $p$ associated with such a principal connection is then simply defined by its kernel,
\begin{equation}
	H_p P := \{ \un{v}_p \in T_p P \; | \; \mathbb{A}_p(\un{v}_p) = 0\}\,.
\end{equation}
As $p$ runs all over $P$, we obtain the horizontal sub-bundle $H_P$ of $P$:
    \begin{equation}
        H_P\equiv\{\un{v}\in\Gamma(TP)\,|\,\mathbb A(\un{v})=0\}\,.
    \end{equation}
On the other hand, specifying a horizontal sub-bundle also uniquely corresponds to defining a map $\sigma: TM \rightarrow TP$ such that
\begin{enumerate}
\item[(B1)] $\pi_* \circ \sigma(\underline{X}\rvert_x) = \underline{X}\rvert_x$, $\forall x\in M$, $\underline{X}\rvert_x \in T_xM$;
        \par
\item[(B2)] $\sigma(\underline{X}\rvert_{\pi(p)})\in H_p$, $\forall p\in P$.
\end{enumerate}
The map $\mathbb A$ is called a \emph{vertical projection} or $\emph{Ehresmann connection}$, and $\sigma$ can be referred to as a \emph{horizontal lift} or \emph{covariant derivative}. The Ehresmann connection and the horizontal lift are related in the sense that define the same horizontal distribution. One can easily deduce that image of $\sigma$ coincides with the kernel of $\mathbb A$, i.e.\
\begin{equation}
\mathbb{A} \circ \sigma(\underline{X}) = 0\,, \qquad\forall \un X\in \Gamma(TM)\,.
\end{equation}
Finally, there is a third equivalent way to characterize a connection on $P$. Suppose $s_U: U \rightarrow P$ is a local section of $P$, we can define a local connection as  a $\mathfrak g$-valued 1-form on $U$ by pulling back the Ehresmann connection:
\begin{equation}
    A_U = s_U^* \mathbb{A} \in \Omega^1(U;g)\,.
\end{equation}
In physical contexts, this object is the familiar local gauge field on $M$. Suppose $U$ and $V$ are two open subsets with $U\cup V\neq\varnothing$, and $s_U: U \rightarrow P$ and $s_V: V \rightarrow P$ are two local sections, whose corresponding local trivializations are $T_U$ and $T_V$ with the transition function $t_{UV}$. Then, the local gauge fields $A_U$ and $A_V$ are related by the following equation:
\begin{fleqn}
\begin{align}
\label{condC}
\text{(C)}\,\,\, A_V|_x (\un X|_x) = \text{Ad}_{t^{-1}_{UV}(x)}(A_U|_x) (\un X|_x)  + t^{-1}_{UV }\td_M t_{UV}|_x(\un X|_x) \,,\quad \forall x\in U\cap V\,, \,\,\un X|_x\in T_xM\,,
\end{align}
\end{fleqn}
where $\td_M$ is the exterior derivative on $M$. Conversely, given such a local gauge field on $M$, one can construct the Ehresmann connection $\mathbb{A}_U$ on $P_U$ over the subset $U \subset M$ by means of the trivialization $T_U(p)=(x,g)$ as follows:
\begin{align} \label{Global Connection}
    \mathbb{A}_U|_p(\un v|_p)& = \text{Ad}_{g^{-1}} \big(A_U|_x(\pi_*(\un v|_p)) + g^{-1}\td_G g\big)\nn\\
    &= \text{Ad}_{g^{-1}} A_U|_x(\pi_*(\un v|_p)) +w\,,\qquad \forall p\in P_U\,,\quad \un v|_p\in T_pP\,.
\end{align}
where $\td_G$ is the exterior derivative on $G$, and $w = g^{-1} \td_G g$ is called the Maurer-Cartan form of $G$. It can be shown that $\mathbb{A}_U$ indeed satisfy conditions (A1) and (A2) above, and condition (C) guarantees that $\mathbb{A}_U$ and $\mathbb{A}_V$ obtained via two local trivializations satisfy $\mathbb{A}_U = \mathbb{A}_V$ on $U \cap V$ for any $U, V \subset M$. That is, despite its local appearance in \eqref{Global Connection}, $\mathbb{A}$ is a globally well-defined Ehresmann connection on $P$ which ensures that the vertical-horizontal splitting is well-defined everywhere on $TP$. For a detailed proof of the equivalence of the above three descriptions of principal connections, see \cite{Liang3}. 

To summarize, the geometry structure of a principle $G$-bundle $P(M,G)$ described by the sequence $G\xrightarrow{R_g}P\xrightarrow{\pi}M$ can be illustrated by the following exact sequence:
\begin{equation}
\label{Psequence}
\begin{tikzcd}
0
\arrow{r} 
& 
V_P
\arrow{r}{j} 
& 
TP
\arrow[bend left]{l}{\mathbb A}
\arrow{r}{\pi_*}
& 
TM
\arrow[bend left]{l}{\sigma}
\arrow{r} 
&
0\,.
\end{tikzcd}
\end{equation}
The principle connection can be defined by $\bb A$ satisfying conditions (A1) and (A2), $\sigma$ satisfying conditions (B1) and (B2), or, in each trivialization $T_U$, a gauge field $A_U$ on the base manifold satisfying condition (C).

\subsection{Exterior Algebra and Curvature}
\label{sec:principal}
Given a principal $G$-bundle $P(M,G)$. The local statement that $P|_U\simeq U\times G$ for an open subset $U\subset M$ is sufficient to identify the exterior algebra of $P$ with the exterior bi-algebra consisting of both the exterior algebras of the manifolds $M$ and $G$. In particular we can express the exterior algebra on $P$, denoted by $\Omega(P)$, as
\begin{equation}
\label{bialgebra}
	\Omega(P) = \bigoplus_{p = 1}^{\dim P} \Omega^{p}(P)\,,\qquad\Omega^p(P) = \bigoplus_{r + s = p} \Omega^{(r,s)}(M,G)\,.
\end{equation}
Now we explain how $\Omega^p(P)$, the collection of the $p$-forms on $P$, is decomposed into $\Omega^{(r,s)}(M,G)$. Since the total space of the principal bundle is locally given by the product $M \times G$, the exterior derivative on $\Omega(P)$, denoted by $\td_P$, locally splits as $\td_P = \td_M + \td_G$ given a suitable choice of local frame, where $\td_M$ and $\td_G$ are the exterior derivatives on $M$ and $G$, respectively. When interpreting this splitting one must be careful in specifying the appropriate generators for the exterior bi-algebra. In a coordinate basis, the exterior algebra of $P$ is generated by a dual basis $\{\td_M x^{\mu}, \td_G g^A\}$, where $x^{\mu}$ are coordinates on $M$ and $g^A$ are coordinates on $G$, and hence we should concede that $\td_G x^{\mu} = \td_M g^A = 0$. Then, any $\bb M\in\Omega^p(P)$ can be expanded in this basis as
\be
\bb M=\sum_{r+s=p}\bb M^{(r,s)}_{\mu_1\cdots\mu_r A_1\cdots A_s}\td_Mx^{\mu_1}\wedge\cdots\wedge\td_Mx^{\mu_r}\wedge\td_Gg^{A_1}\cdots\wedge\td_Gg^{A_s}\,,
\ee
where each $\bb M^{(r,s)}$ can be regarded as a form of degree $r$ on $M$ and degree $s$ on $G$, and the collection of such forms is denoted as $\Omega^{(r,s)}(M,G)$, which defines the exterior bi-algebra in \eqref{bialgebra}.

Now let us introduce the \emph{curvature} of a connection on a principal bundle. Recall that a connection specifies a horizontal distribution $HP \subset TP$. The role of curvature is that it measures the failure of this horizontal distribution to be integrable. Similar to the connection, it can be quantified in three ways. Firstly, the curvature form as a $\mathfrak g$-valued 2-form on $P$ is defined as\footnote{We have introduced the graded Lie bracket of $\mg$-valued differential forms. On a manifold, for any forms $\alpha \in \Omega^m(M; \mg)$ and $\beta \in \Omega^n(M; \mg)$, $[\alpha,\beta]_{\mg}$ is defined as
\begin{equation*}
	[\alpha,\beta]_{\mg}(\un X_1, \ldots, \un X_{m+n}) = \sum_{\sigma }\text{sgn}(\sigma) [\alpha(\un X_{\sigma(1)}, \ldots, \un X_{\sigma(m)}), \beta(\un X_{\sigma(m+1)}, \ldots, \un X_{\sigma(m+n)})]_{\mg}\,,
\end{equation*}
where $\un X_1,\ldots ,\un X_{m+n}$ are arbitrary sections on $TM$, and $\sigma$ denotes the permutations of $(1,\ldots,m+n)$, with $\text{sgn}(\sigma)=1$ for even permutations and $\text{sgn}(\sigma)=-1$ for odd permutations.} 
\begin{equation} \label{Curvature on P(M,G)}
    \mathbb{F} = \td_P \mathbb{A} + \frac{1}{2}[\mathbb{A}, \mathbb{A}]_{\mathfrak g} \in \Omega^2(P;g)\,.
\end{equation}
This equation is referred to as the Cartan's second equation of structure. As a geometric object, the curvature 2-form $\mathbb{F}$ transforms in the adjoint representation of the group $G$, namely $R_g^*\, \mathbb{F} = \text{Ad}_{g^{-1}} \mathbb{F}$. Alternatively, the curvature can be quantified as the failure of the horizontal lift $\sigma$ to be a morphism:
\begin{equation}
\label{RsigmaP}
    R^{\sigma}(\underline{X}, \underline{Y}) = [\sigma(\underline{X}), \sigma(\underline{Y})]_{TP} - \sigma([\underline{X}, \underline{Y}]_{TM}) \in TP\,.
\end{equation}
The relationship between these two notions of curvature is given algebraically as
\begin{equation}
\label{jFRsigma}
    j(\mathbb{F}(\underline{u}, \underline{v}) )= R^{\sigma}(\pi_*\underline{u}, \pi_*\underline{v})\qquad\forall \un u, \un v\in TP\,.
\end{equation}
Lastly, in terms of the local gauge field $A_U$ in each local trivialization $T_U$, we can define the local curvature $F_U$ as the following 2-form on each open subset $U$:
\begin{equation}
\label{Ftransform}
    F_U = \td_M A_U + \frac{1}{2}[A_U, A_U]_\mg \in \Omega^2(U;\mathfrak{g})\,,
\end{equation}
Physically, this is recognize as the gauge field strength. As was the case with the connection form, we can define the curvature globally on $M$ by patching together local gauge field strengths. It follows from condition (C) of the local gauge field that on the overlap $U\cap V$ we have
\begin{equation} \label{Transformation of Gauge Field Strength}
    F_V = \text{Ad}_{t_{UV}^{-1}} F_U\,.
\end{equation}
Similar to the relation between $A_U$ and $\bb A$, the curvature 2-form $F_U$ defined on the base manifold is related to the previously defined $\bb F$ is 
\be
\label{Fbase}
F_U = s_U^* \mathbb{F}\,,
\ee
where $s_U$ is the local section associated with the local trivialization $T_U$.

It is straightforward to show from the definition \eqref{Curvature on P(M,G)} that the exterior derivative of the curvature satisfies the Bianchi identity
\begin{equation}
	\td_P \mathbb{F} = -[\mathbb{A}, \mathbb{F}]_{\mg}\,.
\end{equation}
which follows from the nilpotency of $\td_P$. We can observe that the connection and curvature generate a closed exterior subalgebra of $\Omega(P)$ on account of the algebraic relations:
\begin{equation}
\label{AFrelation}
	\td_P \mathbb{A} = \mathbb{F} - \frac{1}{2}[\mathbb{A},\mathbb{A}]_{\mg}\,,\qquad \td_P \mathbb{F} = -[\mathbb{A},\mathbb{F}]_{\mg}. 
\end{equation}
In the next section, we will demonstrate how the curvature on the principal bundle can be utilized to define the cohomology classes on $P$ explicitly. 

\section{Cohomology and Topological Obstructions}
\label{sec:cohomology}
Topological invariants are a key element in studying the global structure of a differentiable manifold. Two effective tools for constructing these invariants are homotopy and homology. Homotopy concerns the continuous deformation between topological objects, while homology studies the equivalence classes of these objects. These two approaches are closely related. Although homotopy may be more intuitive, its mathematical computation is often quite complex. Therefore, the seemingly more abstract homology is in fact more practical, and homotopy analysis is frequently conducted by means of homology. For physicists, usually an even more convenient approach is to study the dual of homology, namely cohomology, since it directly relates to the familiar differential forms.

\subsection{Homology and Cohomology}
\label{sec:homology}
A basic idea of analyzing the global property of a manifold is to divide it into cells and study how they are pieced together. 
\par
Suppose $n,k\in\bb{Z}$ and $n\geqslant k>0$. Points $v_0,v_1,\cdots,v_k\in\bb{R}^n$ are said to be \emph{affinely independent} if a set of vectors $\{v_1-v_0,\cdots,v_k-v_0\}$ is linearly independent. This assures that these points do not lie on a $(k-1)$-plane. Any single point $v_0\in\bb{R}^n$ is affinely independent. Suppose points $v_0,v_1,\cdots,v_k\in\bb{R}^n$ are affinely independent, then they define a \emph{$k$-simplex} as
\be
\langle v_0,\cdots,v_k\rangle=\left\{\sum^k_{i=0}x^iv_i\bigg|\sum^k_{i=0}x^i=1,x^i\geqslant0\right\}\,.
\ee
$v_0,\cdots,v_k$ are called the \emph{vertices} of the simplex. A simplex formed by some of these vertices is called a \emph{face} of the simplex. Suppose $K$ is a set formed by a finite number of simplices, then $K$ is called a \emph{simplicial complex}, or \emph{complex} for short, if
\par
(a) $\forall\sigma\in K$, each face of $\sigma$ belongs to $K$;
\par
(b) $\forall\sigma_1,\sigma_2\in K$, we have $\sigma_1\cap\sigma_2=\varnothing$ or $\sigma_1\cap\sigma_2$ is a face of both $\sigma_1$ and $\sigma_2$.
\par
Suppose $K$ is a simplicial complex in $\bb{R}^n$, then $|K|\equiv\bigcup_{\sigma\in K}\sigma$ as a subspace of $\bb{R}^n$ is called a \emph{polyhedron}. $K$ is called a \emph{simplicial subdivision} or \emph{triangulation} of $|K|$.
For a $k$-simplex $\sigma=\langle v_0,\cdots,v_k\rangle$, any even permutation $j:(0,\cdots,k)\mapsto(j_0,\cdots,j_k)$ of the vertices is said to be equivalent, i.e., $\langle v_{0},\cdots,v_{k}\rangle\sim\langle v_{j_0},\cdots,v_{j_k}\rangle$. It can be proved that there are two equivalent classes, each one is called an \emph{orientation} of $\sigma$. A simplex $\langle v_0,\cdots,v_k\rangle$ together with an orientation is called a \emph{oriented simplex}, denoted by $[v_0,\cdots,v_k]$. Given any permutation $i:(0,\cdots,k)\mapsto(i_0,\cdots,i_k)$, we have $[v_{i_0},\cdots,v_{i_k}]=\text{sgn}(i)[v_0,\cdots,v_k]$.
\par
Since a smooth manifold $M$ is locally diffeomorphic to an open subset of $\bb{R}^n$, we can use the triangulation of $\bb{R}^n$ as the triangulation of $M$. A linear combination of $k$-simplices of $M$, $c_k=\sum_ia_i\sigma_i^k$, with $a_i\in\bb{Z}$ is called a \emph{$k$-chain} on $M$. The collection of all $k$-chains on $M$, $C_k(M)=\{c_k\}$, is a free Abelian group generated by all oriented $k$-simplices, called the \emph{$k$-chain group}. Since the number of generators can be infinite, the practical way to study them is construct the equivalent classes by means of the homomorphisms between the groups of chains. Now we introduce an operator $\p_k$ that maps each $k$-simplex to a $(k-1)$-simplex on its boundary:
\be
\p_k\sigma^k=\p[v_0,\cdots,v_k]=\sum_{i=0}^k(-1)^i[v_0,\cdots,v_{i-1},v_{i+1},\cdots,v_k]=\sum_{i=0}^k(-1)^i\sigma_i^{k-1}\in C_{k-1}(M)\,.
\ee
When $\p_k$ acts on a $k$-chain, we have
\be
\p_kc_k=\p_k(\sum_ia_i\sigma_i^k)=\sum_ia_i(\p_k\sigma_i^k)\in C_{k-1}(M)\,,
\ee
which preserves the addition of the chain group. Thus, $\p_k:C_k(M)\to C_{k-1}(M)$ is indeed a homomorphism, called the $k^{th}$ \emph{boundary operator}. We also stipulate that the boundary of a $0$-chain is zero.
\par
Given an $d$-dimensional manifold $M$, the $k$-chain groups $C_k(M)$ with $k=0,\cdots,d$ and the boundary operators $\p_k$ give rise to the following sequence:
\be
\label{chaincomplex}
\begin{tikzcd}
0 \arrow[r] & C_d(M) \arrow[r, "\partial_d"] & C_{d-1}(M) \arrow[r, "\partial_{d-1}"] & \cdots \arrow[r, "\partial_2"] & C_1(M) \arrow[r, "\partial_1"] & C_0(M) \arrow[r] & {0\,.}
\end{tikzcd}
\ee
This sequence of chain groups is called a \emph{chain complex}, denoted by $(C_{\bullet}(M),\p_{\bullet})$.
\par
From a boundary operator $\p_k$, we can obtain two important subgroups of $C_k(M)$. One is the kernel of a boundary operator:
\be
Z_k=\{c_k\in C_k(M)|\p_k c_k=0\}\,,
\ee
called a \emph{${k}$-cycle group}, where each $c_k$ is called a $k$-cycle. The other is the image of a boundary operator:
\be
B_k=\{b_k=\p_{k+1} c_{k+1}|c_{k+1}\in C_{k+1}(M)\}\,,
\ee
called a \emph{${k}$-boundary group}, where each $b_k$ is called a $k$-boundary. It can be proved that the boundary of a boundary chain is zero, i.e. $\p_k\cdot\p_{k+1}=0$, and hence $B_k\subset Z_k$. 
\par
Since $B_k$ and $Z_k$ are Abelian groups, $B_k$ must be a normal subgroup of $Z_k$. Then, we can define the quotient group
\be
H_k(M)=Z_k(M)/B_k(M)
\ee
as the $k^{th}$ \emph{homology group} of $M$, which is the set of equivalent classes of $k$-cycles. Two $k$-cycles $c_k$ and $d_k$ are said to be \emph{homologous} if their difference is a $k$-boundary chain, i.e., $c_k-d_k\in B_k(M)$. A non-trivial $k$-cycle in $H_k(M)$ can be thought of as a $k$-dimensional manifold with a $(k+1)$-dimensional hole, and any $k$-cycle without a hole is homologous to a $0$-chain. 
\par
Having the homology group, now we consider the collection of homomorphisms from the chain group $C_k(M)$ to $\bb Z$, denoted by $C^k(M)$. This can be regarded as the dual of the chain group, called the \emph{k-cochain group}. The boundary operator $\p_k:C_k(M)\to C_{k+1}(M)$ also induces a homomorphism $\td^k:C^{k-1}(M)\to C^{k}(M)$, called the $k^{th}$ \emph{coboundary operator} defined as follows:
\be
(\td^kc^{k-1})(c_k)\equiv c^{k-1}(\p_kc_k)\,,\qquad\forall c_k\in C_k(M)\,,\quad c^{k-1}\in C^{k-1}(M)\,.
\ee
The cochain groups $C^k(M)$ with $k=0,\cdots,d$ together with the coboundary operators $\td_k$ give rise to the \emph{cochain complex} $(C^{\bullet}(M),\td^{\bullet})$, which is represented by the following sequence:
\be
\label{cochaincomplex}
\begin{tikzcd}
0 \arrow[r] & C^0(M) \arrow[r, "\text{d}^1"] & C^{1}(M) \arrow[r, "\text{d}^{2}"] & \cdots \arrow[r, "\text{d}^{d-1}"] & C^{d-1}(M) \arrow[r, "\text{d}^{d}"] & C^d(M) \arrow[r] & {0\,.}
\end{tikzcd}
\ee
Similar to the case of a chain group, we can define the kernel of the coboundary operator $\td^k$ as the \emph{${k}$-cocycle group}
\be
Z^k=\{c^k\in C^k(M)|\td^k c^k=0\}\,,
\ee
where each $c_k$ is called a $k$-cocycle. And we define the image of $\td^k$ as the \emph{${k}$-coboundary group}
\be
B^k=\{b^k=\td^{k+1} c^{k+1}|c^{k+1}\in C^{k+1}(M)\}\,,
\ee
where each $b_k$ is called a $k$-coboundary. It can also be proved that $\td^k\cdot\td^{k+1}=0$, and $B^k\subset Z^k$ is a normal subgroup. Then, we can define the $k^{th}$ \emph{cohomology group} as
\be
H^k(M)=Z^k(M)/B^k(M)\,,
\ee
which is the set of equivalent classes of $k$-cocycles. Two $k$-cocycles $c^k$ and $d^k$ are said to be \emph{cohomologous} if their difference is a $k$-coboundary, i.e., $c^k-d^k\in B^k(M)$. Note that the sequences \eqref{chaincomplex} and \eqref{cochaincomplex} are not exact sequences, and $H_k(M)$ and $H^k(M)$ can be considered as a measurement of their ``non-exactness''.
\par
So far we only considered the chain and cochain groups with integer coefficients, and thus the homology and cohomology groups may be denoted as $H_k(M,\bb Z)$ and $H^k(M,\bb Z)$, respectively. In general, $\bb Z$ can be replaced by any group $G$, and $H_k(M,G)$ are vector spaces on $G$ while $H^k(M,G)$ are their dual vector spaces. If we take $G=\bb R$, the resulting cohomology group $H^k(M,\bb R)$ is isomorphic to the \emph{de Rham cohomology group} $H^k_{\text{dR}}(M)$, where the $k$-cocycles are the closed $k$-forms on $M$, the $k$-coboundaries are the exact $k$-forms on $M$, and the coboundary operator is the exterior derivative operator $\td$ on $M$ (the label $k$ is omitted). Furthermore, the wedge product $\wedge:H^p_{\text{dR}}(M)\times H^q_{\text{dR}}(M)\to H^{p+q}_{\text{dR}}(M)$ also defines the de Rham cohomology ring:
\be
H_{\text{dR}}(M)\equiv\bigoplus_{k=1}^{d}H^k_{\text{dR}}(M)\,.
\ee 
Such a ring structure can also be defined for any cohomology class, where for a general cohomology ring $H(M,G)=\oplus_{k=1}^{d}H^k(M,G)$ the wedge product is replaced by the cup product $\cup$ (see, e.g., \cite{Hatcher}). This is a property that homology classes do not generally enjoy. Together with many other properties, this provides advantages in the analysis of cohomology over homology. Since the operations of differential forms are much more familiar to physicists, de Rham cohomology is a convenient implement for studying the global topology of a manifold in physics contexts.

\subsection{Characteristic Classes and the Chern-Weil Theorem}
\label{sec:Chern}
A principle bundle $P(M,G)$ in general cannot be globally trivialized as $P\simeq M\times G$ due to its nontrivial topology. This deviation from the trivial bundle can also be manifested as the obstructions of constructing a global section on $P$ or lifting certain structures or fields globally from $M$ to $P$. Characteristic classes are cohomology classes that measure these topological obstructions. After assigning a connection 1-form $\bb A$ on $P$, the Chern-Weil theorem allows us to express a characteristic class as a polynomial of the corresponding curvature 2-form $\bb F$ on $P$, which we will now introduce \cite{WeilLetter,zbMATH03070474,chern1946characteristic,
chern1966geometry}.
\par
Suppose $\mg$ is the algebra of the structure group $G$ of $P$. Let $Q^{(l)}: \mg^{\otimes l} \rightarrow \mathbb{R}$ correspond to a symmetric order-$l$ polynomial function on $\mg$ which is invariant under the adjoint action of the group $G$. Such an object can be represented by a symmetric $l$-linear map in the tensor algebra of $\mg$. That is, given a basis $\{ t^A\}$ of the dual space $\mg^*$ with $A=1,\cdots,\dim G$, we can write 
\be
\label{Ql}
Q^{(l)} = Q_{A_1 ... A_l} \bigotimes_{j = 1}^l t^{A_j}\,. 
\ee
Then, the $l^{th}$ \emph{characteristic class} $\lambda_{Q}$ defined by $Q^{(l)}$ is
\begin{equation}
	\lambda_{Q}(\mathbb{F}) = Q^{(l)}(\underbrace{\mathbb{F}, ..., \mathbb{F}}_l) = Q_{A_1 ... A_l} \bigwedge_{j = 1}^{l} \mathbb{F}^{A_j} \in \Omega^{2l}(P)\,.
\end{equation}
Note that later we will use the $\lambda_Q(\cdot)$ to define the characteristic classes in different exterior algebras. The exterior algebra in which the particular characteristic class takes values should then be made clear by the argument of $\lambda_Q(\cdot)$.

The essence of the Chern-Weil theorem is the existence of a homomorphism from the invariant polynomial ring on $\mg$ to the cohomology ring of $P$.\footnote{Technically, the Chern-Weil homomorphism maps the invariant polynomial ring on $\mg$ to the cohomology ring of $M$. Here we consider the characteristic classes as living in the equivariant cohomology of $P$, which can be identified with the cohomology of $M$.} Specifically, it establishes that each $\lambda_Q(\mathbb{F})$ gives an element of the cohomology class of degree $2l$ in the exterior algebra of $P$. Here we make no attempt to prove the Chern-Weil theorem in any generality, but will only introduce the following two statements it consists of (see also \cite{nakahara2018geometry}):
\begin{enumerate}
	\item Characteristic classes are closed $2l$-forms in $\Omega(P)$:
\begin{equation}
	\td_P \lambda_Q(\mathbb{F}) = l! Q^{(l)}(\td_P \mathbb{F}, \underbrace{\mathbb{F}, \cdots, \mathbb{F}}_{l-1}) = l! Q^{(l)}(\td_P \mathbb{F} + [\mathbb{A}, \mathbb{F}]_{\mg}, \underbrace{\mathbb{F}, \cdots, \mathbb{F}}_{l-1}) = 0\,,
\end{equation}
which follows from the symmetry of $Q^{(l)}$ and the Bianchi identity. 
	\item Given two different principal connections $\mathbb{A}_1$ and $\mathbb{A}_2$, with respective curvatures $\mathbb{F}_1$ and $\mathbb{F}_2$, we have that $\lambda_Q(\mathbb{F}_2) - \lambda_Q(\mathbb{F}_1) \in \Omega^{2l}(P)$ is exact. The relevant $(2l-1)$-form potential is defined by introducing a one-parameter family of connections $\mathbb{A}_t = \mathbb{A}_1 + t(\mathbb{A}_2 - \mathbb{A}_1)$ which interpolates between $\mathbb{A}_1$ and $\mathbb{A}_2$ as $t$ goes from $0$ to $1$. Then,
\begin{equation} \label{Transgression Formula0}
	\lambda_Q(\mathbb{F}_2) - \lambda_Q(\mathbb{F}_1) = \td_P \left[ Q_{A_1 ... A_l} \int_{0}^{1} \td t (\mathbb{A}_2 - \mathbb{A}_1)^{A_1} \bigwedge_{j = 2}^{l} \left(\td_P \mathbb{A}_t + \frac{1}{2}[\mathbb{A}_t , \mathbb{A}_t]_{\mg}  \right)^{A_j} \right]\,.
\end{equation}
\end{enumerate}

An immediate corollary of the Chern-Weil theorem is that the characteristic class $\lambda_Q(\mathbb{F})$ will be globally exact if there exists a one-parameter family of connections for which $\mathbb{A}_2 = \mathbb{A}$ and $\mathbb{A}_1$ is any connection that has zero curvature. This inspires the topological interpretation of the characteristic class which will be cohomologically trivial if and only if any connection $\mathbb{A}$ can be homotopically connected to the trivial connection. Nonetheless, it will always be true locally that any characteristic class can be written as $\td_P$ acting on a $(2l-1)$-form defined using \eqref{Transgression Formula0}. That is,
\begin{equation}
	\lambda_Q(\mathbb{F}) = \td_P \mathscr{C}_Q(\mathbb{A})\,,
\end{equation} 
where
\begin{equation}
\label{CStransg}
	\mathscr{C}_Q(\mathbb{A}):= Q_{A_1 ... A_l} \int_{0}^{1} \td t\, \mathbb{A}^{A_1} \bigwedge_{j = 2}^{l} \left(t\td_P \mathbb{A} + \frac{1}{2}t^2[\mathbb{A} , \mathbb{A}]_{\mg}  \right)^{A_j}
\end{equation} 
is the \emph{Chern-Simons form} associated with the symmetric invariant polynomial $Q^{(l)}$, which plays a very central role in the cohomological approach to anomalies as will will review shortly. Eq.~\eqref{CStransg} is called the \emph{transgression formula} for the Chern-Simons form.
\par
Finally, characteristic classes satisfy an important property called \emph{naturality}. Suppose $M$ and $N$ are manifolds, $f:N\to M$ is a differentiable map. Let $P(G,M)$ and $P'(G,N)$ be principle bundles over $M$ and $N$ with the same structure group $G$, then a characteristic classes $\lambda_Q(\bb F)$ on $P$ can be pulled back to a characteristic classes on $P'$ as
\be
\label{natural}
f^*\lambda_Q(\bb F)=\lambda_Q(f^*\bb F)\,.
\ee
In other words, characteristic classes are natural as they commute with the pullback of $f$.  

\section{The Cohomology of the BRST Complex and Anomalies}
\label{sec:BRST}
\subsection{BRST Complex}
The topological interpretation of characteristic classes on $P(M,G)$ has led many to expect that the same tools can be used to describe the gauge anomaly which is also a topological effect. Ultimately however, this is misguided for reasons we have mentioned: the cohomology of the principal bundle encodes data associated with the global algebra of the structure group, not the local gauge algebra. In order to let it acquire some explicit relationship with gauge transformations, one needs to require some refinement of the principal bundle language. The historical resolution to this problem is the BRST complex. Before introducing the BRST complex, let us briefly recall how infinitesimal gauge transformations are implemented. 

A local gauge transformation is represented by a map $g: M \rightarrow G$. Under a local gauge transformation, the gauge field and its field strength transforms as
\begin{equation}
\label{localgaugetrans}
	A \rightarrow A^g = \text{Ad}_{g^{-1}}(A) + g^{-1} \td g\,,\qquad F \rightarrow F^g = \text{Ad}_{g^{-1}}(F)\,.
\end{equation}
This is what we have seen in \eqref{condC} and \eqref{Ftransform} for the connection and its curvature defined in a local trivialization, and $g$ now plays the role of the transition function. One should notice that here $g$ is not just a group element, but a pointwisely defined field of group element $g(x)$ on $M$. Each local gauge transformation given by $g$ is generated by $\un{\mu}: M \rightarrow \mg$, which is no longer an element of $\mg$, but a field of element of $\mg$ on $M$. The generator $\un\mu$ acting on the gauge field gives rise to an infinitesimal gauge transformation
\begin{equation} \label{Infinitesimal Gauge Transform of A}
	A \rightarrow A^{\un{\mu}} = A + D\un{\mu} \equiv A + \td\un{\mu} + [A, \un{\mu}]_{\mg},
\end{equation}
where $D$ represents the covariant derivative associated with the gauge field $A$. Besides, we can also introduce a matter field $\un\psi$ in a representation $R$, which is a section on a vector bundle $E$. Then, under the infinitesimal gauge transformation generated by $\mu$, we have
\be
\label{inftranspsi}
\un\psi\to\un\psi^{\un\mu}=\un\psi-R(\un\mu)\un\psi\,,
\ee
where the representation $R$ maps $\un\mu$ to an endomorphism $R(\un\mu)$ on $E$. This is the infinitesimal version of the transformation $\psi\to\psi^g=R(g^{-1})\psi$.
\par
The geometric construction of the BRST formalism considers a principle bundle $\scr  P (M,\scr G)$, whose structure group is $\scr{G}=\{g:M\to G\}$ with the group multiplication $g_1g_2(x)=g_1(x)g_2(x)$ inherited from that of $G$ pointwisely. Unlike $P(M,G)$, $\scr P(M,\scr G)$ has an infinite dimensional structure group $\scr G$, where each element is a choice of $g(x)$ that gives rise to a gauge transformation in \eqref{localgaugetrans}. Then, the exterior algebra of $\scr P$ can be decomposed similar to \eqref{bialgebra} as
\begin{equation}
\label{bialgebra2}
	\Omega(\scr P) = \bigoplus_{k = 1} \Omega^{k}(\scr P)\,,\qquad\Omega^k(\scr P) = \bigoplus_{p + q = k} \Omega^{(p,q)}(M,\scr G)\,.
\end{equation}
Note that the form degree $p$ on $M$ is bounded by the $d=\dim M$, while the degree $q$ on $\scr G$ is unbounded since $\dim\scr G$ is infinite. Denote the exterior derivative on $M$ and $\scr G$ as $\td$ and $\ts$, respectively. Then, $\Omega^{(p,q)}(M,\scr G)$ with $\td$ and $\ts$ form cochain complexes in two directions, namely they form a cochain bi-complex, called the \emph{BRST complex}, which can be represented by the following diagram:
\begin{equation}
\begin{tikzcd}
&[-1em]\cdots                                                    & \cdots                                                     & \cdots                                                     &   \cdots                           & \cdots   & [-1em]                       \\
0 \arrow[r]&[-1em]{\Omega^{(0,1)}(M,\scr G)} \arrow[r, "\text{d}"] \arrow[u, "\text{s}"] & {\Omega^{(1,1)}(M,\scr G)} \arrow[r, "\text{d}"] \arrow[u, "\text{s}"] & {\Omega^{(2,1)}(M,\scr G)} \arrow[r, "\text{d}"] \arrow[u, "\text{s}"] & \cdots \arrow[r, "\text{d}"] & {\Omega^{(d,1)}(M,\scr G)} \arrow[u, "\text{s}"]  \arrow[r]&[-1em] 0\\
0 \arrow[r]&[-1em]{\Omega^{(0,0)}(M,\scr G)} \arrow[r, "\text{d}"] \arrow[u, "\text{s}"] & {\Omega^{(1,0)}(M,\scr G)} \arrow[r, "\text{d}"] \arrow[u, "\text{s}"] & {\Omega^{(2,0)}(M,\scr G)} \arrow[r, "\text{d}"] \arrow[u, "\text{s}"] & \cdots \arrow[r, "\text{d}"] & {\Omega^{(d,0)}(M,\scr G)} \arrow[u, "\text{s}"]  \arrow[r]&[-1em] 0\\
&[-1em]0 \arrow[u]                                                   & 0  \arrow[u]                                                    & 0 \arrow[u]                                                      &        \cdots                      & 0 \arrow[u]    &[-1em]
\end{tikzcd}
\end{equation}
The coboundary operator $\td: \Omega^{(p,q)}(M,\mathscr{G}) \rightarrow \Omega^{(p+1,q)}(M,\mathscr{G})$ is de Rham differentiation on $M$ and the coboundary operator $\ts: \Omega^{(p,q)}(M,\mathscr{G}) \rightarrow \Omega^{(p+1,q)}(M,\mathscr{G})$ in the vertical direction is called the \emph{BRST operator}. Then the exterior derivative on $\scr P$ can be recognized by the coboundary operator $\td_{\rm BRST} = \td + \ts$ on the BRST complex. The nilpotency of these operators means $\td^2=\ts^2=\td\ts+\ts\td=0$.

The next step in the BRST construction is to introduce a graded algebra generated by Grassmann valued fields $c^A(x)$ with $A=1,\cdots,\dim G$, which form a $\mg$-valued 1-form $c=c^A\otimes\un t_A \in \Omega^{(0,1)}(M,\mathscr{G};\mg)$. The fields $c^A$ are referred to as ``ghosts", and play a significant role in the quantization of gauge theories. Thus, later on we will refer to the degrees $p$ and $q$ for any $\alpha^{(p,d)}\in\Omega^{(p,q)}(M,\scr G)$ as the form degree and ghost number, respectively. Then, $c$ is added to the gauge field $A$ on $M$ to define an ``extended form'':
\begin{equation} \label{BRST Connection}
	\hat{A} \equiv A + c\,.
\end{equation}
$\hat A$ has form degree $1$ and ghost number 1, which is regarded as a ``connection" in the context of the BRST analysis. Subsequently, we can define its ``curvature" $\hat{F}$ by
\begin{equation}
	\hat{F} \equiv \td_{\rm BRST} \hat{A} + \frac{1}{2}[\hat{A}, \hat{A}]_{\mg}\,.
\end{equation}
Notice that it is not immediately clear that the $\hat{A}$ and $\hat{F}$ on the BRST complex should be interpreted geometrically as a connection and curvature, although they share the same algebraic relations as that of the connection and curvature on a principal bundle given in \eqref{AFrelation}.
\par
In the BRST analysis, one makes a particular choice which makes it an effective device for the quantization of gauge theory. That choice goes by the name of the Russian Formula, which stipulates that the extended curvature $\hat{F}$ should be completely horizontal, i.e.,\ have zero ghost number. Computing $\hat{F}$ explicitly, we find
\begin{equation} \label{Russian Formula 1}
	\hat{F} = \td A + \frac{1}{2}[A , A]_{\mg} + \Big(\ts A + \td c + [A, c]_{\mg} \Big) + \Big(\ts c + \frac{1}{2}[c, c]_{\mg}  \Big)=F\,.
\end{equation}
To uphold the Russian formula, the terms in the last two parentheses must both vanish identically. This in turn defines the action of the operator $\ts$ through the equations\footnote{Note that there is a relative sign difference between the equation for $\ts A$ and \eqref{Infinitesimal Gauge Transform of A}, this is because the conversion between $\delta_{\un\mu}A=A^{\un\mu}-A$ and $\ts A$ is $\delta_{\un\mu}A=i_{V_{\un\mu}}\ts A$, where $V_{\un\mu}$ is an infinitesimal vector field on $\scr G$. Since $\ts A$ is a bi-form in $\Omega^{(1,1)}(M,\scr G)$, the contraction of $V_{\un\mu}$ with the dual basis on $\scr G$ will pick up a minus sign when crossing the dual basis on $M$. The result in \eqref{Russian Formula 2} does not have this issue since $\ts\psi\in\Omega^{(0,1)}(M,\scr G)$ and so the contraction does not need to cross any dual basis on $M$.}
\begin{equation} \label{Russian Formula 2}
	\ts A = -(\td c + [A, c]_{\mg}) = -Dc\,,\qquad \ts c = -\frac{1}{2}[c ,c]_{\mg}\,.
\end{equation}
Comparing the first equation with \eqref{Infinitesimal Gauge Transform of A}, we can interpret $\ts$ as performing an infinitesimal gauge transformation; the second equation can be interpreted as the action of the Chevalley-Eilenberg operator on the generators of an exterior algebra associated with the Lie group $G$. Furthermore, we can also require that the ``extended covariant derivative'' on a matter field $\un\psi$ in a representation $R$ is horizontal, i.e.,
\be
\hat D\un\psi\equiv \td_{\text{BRST}}\un\psi+R(\hat A)\un\psi=D\un\psi\,.
\ee
This requirement gives
\be
\label{Dhorizontal}
\ts\un\psi=-R(c)\un\psi\,,
\ee
which can be recognized as the infinitesimal gauge transformation of a matter field in \eqref{inftranspsi}. In light of \eqref{Russian Formula 2} and \eqref{Dhorizontal}, one obtains the interpretation that the ghost fields $c^A$ should be regarded as the generators of the local gauge algebra, and the complex $\Omega(\mathscr{G})$ should be interpreted as the Chevalley-Eilenberg algebra of the infinite dimensional gauge group whose elements are $g(x)$. We emphasize, however, that these interpretations follow from the Russian formula, rather than precede it.
\par
Before moving on to the analysis of anomalies, we now introduce the cohomology of the BRST complex. On this cochain bi-complex with coboundary operators $\td$ and $\ts$, define the $(p,q)$-cocycle group
\be
Z^{p,q}(\td|\ts)\equiv\{\alpha^{(p,q)}\in\Omega^{(p,q)}(M,\scr G)|\,\ts\alpha^{(p,q)}+\td\alpha^{(p-1,q+1)}=0\}\,.
\ee
and the $(p,q)$-coboundary group
\be
B^{p,q}(\td|\ts)\equiv\{\alpha^{(p,q)}\in\Omega^{(p,q)}(M,\scr G)|\alpha^{(p,q)}=\ts\alpha^{(p,q-1)}+\td\alpha^{(p-1,q)}\}\,.
\ee
One can easily see that any element $\alpha^{(p,q)}\in B^{p,q}(\td|\ts)$ trivially satisfies the condition for $Z^{p,q}(\td|\ts)$, where the corresponding $\alpha^{(p-1,q+1)}$ is $\ts\alpha^{(p-1,q)}$. Then, the quotient group
\be
H^{p,q}(\td|\ts)=Z^{p,q}(\td|\ts)/B^{p,q}(\td|\ts)\,
\ee
defines the \emph{BRST cohomology group}. The BRST cohomology represents the cohomology of $\Omega^{k}(\scr P)$ defined by $\td_{\text{BRST}}$, as one can show that $H^{p,q}(\td|\ts)\simeq H^{p+q}(\td_{\text{BRST}})$ \cite{Brandt:1996mh}. In fact, substituting $P$ by the infinite dimensional bundle $\scr P$ is in some sense a prototype of the Atiyah Lie algebroid construction. In later chapters, we will see that the Atiyah Lie algebroid provides a natural geometric formulation for the BRST complex and BRST cohomology.

\subsection{Anomalies from Characteristic Classes}
\label{sec:anomalyBRST}
In Subsection \ref{sec:Chern} we introduced a characteristic class $\lambda_Q(\bb F) $ on $P$ as a polynomial of the curvature $\mathbb{F}$, which locally can be expressed as the exterior derivative of the Chern-Simons form $\mathscr{C}_Q(\bb A)$ on $P$. Since the triple $(\td, A, F)$ is characterized by the same algebraic data as the triple $(\td_P, \mathbb{A}, \mathbb{F})$, the same construction can be carried over to $M$, and it remains true that the characteristic classes in the gauge field strength $F$ are always closed and locally we have
\begin{equation}
	\lambda_Q(F) = \td\mathscr{C}_Q(A)\,,
\end{equation} 
where $\mathscr{C}_Q(A)$ is the Chern-Simons form on $M$, which can be expressed in terms of the transgressive formula as
\begin{equation}
	\mathscr{C}_Q(A):= Q_{A_1 ... A_l} \int_{0}^{1} \td t\, A^{A_1} \bigwedge_{j = 2}^{l} \left(t\td A + \frac{1}{2}t^2[A, A]_{\mg}  \right)^{A_j}\,.
\end{equation} 
A consequence of the Russian formula is that it ensures that the triple $(\td_{\rm BRST}, \hat{A}, \hat{F})$ are also characterized by the same algebraic relations as the triple $(\td_P, \mathbb{A}, \mathbb{F})$. Notice that the BRST complex now explicitly containing the cohomology of $\scr G$ representing the local gauge transformations, whereas the exterior algebra of the principal bundle only has access to the cohomology of the structure group $G$, which does not have the information of the local gauge algebra. In this way, one can make use of the Chern-Weil homomorphism and the Chern-Weil theorem to construct characteristic classes on the BRST complex, which leads to the topological interpretation of quantum anomalies. 

To introduce the BRST interpretation of anomalies, we start from the characteristic class $\lambda_Q(\hat{F})$ in the BRST complex. From the Chern-Weil theorem, we have
\begin{equation} \label{Descent Derivation 1}
	\lambda_Q(\hat{F}) = \td_{\rm BRST} \mathscr{C}_Q(\hat{A}) = (\td+\ts)\mathscr{C}_Q(A + c)\,.
\end{equation}
On the other hand, the Russian formula tells us that this should be equivalent to the characteristic class $\lambda_Q(F)$ on the base manifold
\begin{equation} \label{Descent Derivation 2}
	\lambda_Q(\hat{F}) = \lambda_Q(F) = \td\mathscr{C}_Q(A)\,.
\end{equation}
Thus, comparing \eqref{Descent Derivation 1} and \eqref{Descent Derivation 2} yields
\begin{equation} \label{Descent Derivation 3}
	(\td+\ts) \mathscr{C}_Q(A + c) = \td \mathscr{C}_Q(A)\,.
\end{equation}
Next, we can expand $\mathscr{C}_Q(A + c)$ in the bi-complex $\Omega^{(p,q)}(M,\mathscr{G})$ and write
\begin{equation}
\label{CSexpansion}
	\mathscr{C}_Q(A + c) = \sum_{p+q = 2l-1} \alpha^{(p,q)}(A,c)\,,
\end{equation}
where $\alpha^{(p,q)}(A,c) \in \Omega^{(p,q)}(M,\mathscr{G})$. It is easy to see that $\alpha^{(2l-1,0)}(A,c) = \mathscr{C}_Q(A)$. Hence, it follows from \eqref{Descent Derivation 3} that 
\begin{equation} \label{Descent Derivation 4}
	(\td + \ts) \sum_{p+q = 2l-1, p \neq 2l-1} \alpha^{(p,q)}(A,c) = 0\,.
\end{equation}
Enforcing (\ref{Descent Derivation 4}) order by order in $(p,q)$, we arrive at a series of equations called the \emph{descent equations}:
\begin{equation} 
	\td\alpha^{(p,q)}(A,c) + \ts\alpha^{(p+1,q-1)}(A,c) = 0\,,\qquad p + q = 2l-1, p \neq 2l-1\,.
\end{equation}
In particular, the equation for $p=2l-2$ is the well-known \emph{Wess-Zumino consistency condition} \cite{wess1971consequences}:
\begin{equation} \label{WZ}
	\td\alpha^{(2l-3,2)}(A,c) + \ts\alpha^{(2l-2,1)}(A,c) = 0\,.
\end{equation}
Physically, a nontrivial solution $\alpha^{(2l-2,1)}(A,c)$ of \eqref{WZ} will be a candidate for the anomaly density of a $(2l-2)$-dimensional theory provided that it is also not exact in the exterior bi-algebra associated with the BRST complex, i.e.,
\be
\alpha^{(2l-2,1)} \neq \td \gamma^{(2l-3,1)} + \ts \gamma^{(2l-2,0)}\,,
\ee 
for any $\gamma^{(2l-3,1)} \in \Omega^{(2l-3,1)}(M,\mathscr{G})$ and $\gamma^{(2l-2,0)} \in \Omega^{(2l-2,0)}(M,\mathscr{G})$. In other words, the anomaly lives in $H^{2l-2,1}(\td|\ts)$, the ghost number 1 sector of the BRST cohomology. To be precise, for a theory defined on a closed $(2l-2)$-dimensional manifold $M$, the anomaly can be obtained by integrating the BRST variation of $\scr C_Q$ over a $(2l-1)$-dimensional manifold $\tilde M$ with boundary $\p\tilde M=M$:
\begin{equation}
 \label{Aconsistent0}
	\mathfrak{a}_{\rm con} = \int_{\tilde M} \ts\scr C_Q(A+c)= \int_{\tilde M} \ts\alpha^{(2l-1,0)}(A)=-\int_{\tilde M} \td \alpha^{(2l-2,1)}(A,c)= -\int_M \alpha^{(2l-2,1)}(A,c)\,.
\end{equation}
where the terms with higher ghost numbers are dropped since they do not have supports on $M$. The anomaly $\mathfrak{a}_{\rm con}$ is called a \emph{consistent anomaly} since it satisfies the consistency condition \eqref{WZ}. 

To explain the reason why anomalies live in $H^{2l-2,1}(\td|\ts)$, now we give a physical interpretation of the BRST cohomology (see \cite{alvarez1984topological,zumino1985chiral,zumino1984chiral,Gockeler:1987an}). Recall that the quantum effective action $W(A)=-\I\ln Z(A)$ can be written as the integral
\be
W(A)=\int_M\scr L(A)\,, 
\ee
where the effective Lagrangian $\scr L(A)$ is a form in $\Omega^{(2l-2,0)}(M,\scr G)$. Noticing that $W(A)$ only determines $\scr L(A)$ up to a total derivative, i.e., $\ts\scr L(A)=\td\gamma^{(2l-1,0)}$ with $\gamma^{(2l-1,0)}\in\Omega^{(2l-1,0)}(M,\scr G)$. This indicates that $\scr L(A)$ is an element in $H^{2l-2,0}(\td|\ts)$. As we have seen in the last subsection, the action of $\ts$ can be viewed as an infinitesimal gauge transformation, then the corresponding anomaly can be read off from the nonvanishing result of $\ts W(A)$. More precisely, the anomaly defined in \eqref{deltaW} can be recasted in the BRST language as
\be
\ts W(A)=\mathfrak{a}_{\rm con}=\int_M a_{\rm con}(A,c)\,,
\ee
where $a_{\rm con}\in\Omega^{(2l-2,1)}(M,\scr G)$ represents the anomaly density. The nilpotency of $\ts$ gives $s^2W(A)=0$, which means
\be
\ts a_{\rm con}=\td m(A,c)\,,
\ee
where $m(A,c)\in\Omega^{(2l-3,2)}(M,\scr G)$. Therefore, the anomaly density satisfies the Wess-Zumino consistency condition \eqref{WZ}, and hence a solution $\alpha^{(2l-2,1)}(A,c)$ to \eqref{WZ} is a candidate of $a_{\rm con}$. On the other hand, if $a_{\rm con}=  \ts \gamma^{(2l-2,0)}+\td \gamma^{(2l-3,1)}$, then one can shift $W(A)$ by a local counterterm $-\gamma^{(2l-2,0)}$ and remove the anomaly. This is synonymous with the fact that $a_{\rm con}\in H^{2l-2,1}(\td|\ts)$ with non-exactness ensuring that it cannot be canceled by a local counterterm. The consistent anomaly being the gauge variation of the Chern-Simons term on the one higher dimension as shown in \eqref{Aconsistent0} is interpreted as the anomaly inflow for the consistent anomaly.

The BRST analysis from constructing the characteristic classes provides a systematic way of deriving anomaly from the topological perspective. Once a characteristic class $\lambda_Q(\hat F)$ is given, the ghost number 1 term in the expansion \eqref{CSexpansion} will be a possible anomaly for some quantum field theory. For example, when the polynomial $Q$ is taken to be the symmetrized trace of $F=F^{A}\otimes\un t_A$:
\begin{equation}
	\text{str}(F,\cdots,F)= \bigwedge_{j = 1}^{l} F^{A_j}\otimes\frac{1}{l!}\sum_{\pi}\tr(\un t_{A_1}\cdots\un t_{A_l})\,,
\end{equation}
then the corresponding characteristic class $\lambda_Q(\hat F)=\text{ch}(F)$ is the Chern class. In this case, $\alpha^{(2l-2,1)}(A,c)$ gives rise to the chiral anomaly of the $(2l-2)$-dimensional Yang-Mills theory. For $l=2$, we have
\begin{equation}
	\text{ch}(F)=\text{tr}(\hat F\wedge\hat F)= \td_{\rm BRST}\scr C_Q(\hat A)\,,
\end{equation}
with
\begin{equation}
\scr C_Q(\hat A)=\tr\Big(\hat A\wedge \hat F-\frac{1}{6}\hat A\wedge[\hat A,\hat A]_\mg\Big)=\delta_{AB}\Big(\hat A^A\wedge \hat F^B-\frac{1}{6}\hat A^A\wedge[\hat A,\hat A]_\mg^B\Big)\,.
\end{equation}
Then from the decomposition in \eqref{CSexpansion} we have
\begin{align}
\label{CS3d}
\alpha^{(3,0)}(A,c)&=\scr C_Q(A)=\delta_{AB}\Big(A^A\wedge F^B-\frac{1}{6}A^A\wedge[A,A]^B_\mg\Big)\,,\\
\label{YMchiral}
\alpha^{(2,1)}(A,c)&=\delta_{AB}\Big(c^A\wedge F^B-\frac{1}{2}c^A\wedge[A,A]^B_\mg\Big)=\delta_{AB}c^A\wedge\td A^B\,,\\
\alpha^{(1,2)}(A,c)&=-\frac{1}{2}\delta_{AB}A^A\wedge[c,c]^B_\mg\,,\qquad
\alpha^{(0,3)}(A,c)=-\frac{1}{6}\delta_{AB}c^A\wedge[c,c]^B_\mg\,.
\end{align}
We can see that $\alpha^{(3,0)}(A)$ is the standard Chern-Simons form in $3d$, and from $\alpha^{(2,1)}(A,c)$ we obtain the anomaly density. Dropping the ghost, we can read off from \eqref{YMchiral} the familiar expression $\delta_{AB}\td A^B$ for the chiral anomaly of a $2d$ Yang-Mills theory.

However, for a non-Abelian gauge group $\mathfrak{a}_{\rm con}$ is not covariant under a gauge transformation. The notion of anomaly that preserves the gauge covariance is the \emph{covariant anomaly}, which cannot be derived directly from the BRST complex as the BRST operator only behaves as the variation along the gauge orbits. Rather, one needs to perform a free variation of the Chern-Simons form $\scr C_Q(A)$ on the $(2l-1)$-dimensional manifold $\tilde M$ with respect to the gauge field $A$, and the result is \cite{bardeen1984consistent,stone2012gravitational}:
\begin{equation}
\label{CSfreevar}
	\delta\scr{C}_Q(A)=lQ^{(l)}(\underbrace{F,\cdots,F}_{l-1},\delta A)+\td\Theta(A,\delta A)\,.
\end{equation}
The first term on the right-hand side of the above equation represents the covariant anomaly:
\begin{equation} \label{Aconvariant}
\mathfrak{a}_{\rm cov} \equiv -\int_{M} \frac{\delta}{\delta A}\scr{C}_Q(A)=- l \int_M G_Q(F)\,,
\end{equation}
where $G_Q(F)$ is a polynomial which can be read off directly from \eqref{CSfreevar}, and the $\Theta$ in the second term on the right-hand side of \eqref{CSfreevar} is a symplectic potential which provides the Bardeen-Zumino polynomial as a current that covariantizes the consistent anomaly. Again, take $Q^{(l)}$ to be the symmetrized trace for example. The $\mathfrak{a}_{\rm cov}$ in \eqref{Aconvariant} gives the covariant chiral anomaly of the $(2l-2)$-dimensional Yang-Mills theory. Let us demonstrate for the $l=2$ case, where $\scr C_Q(A)$ has the form in \eqref{CS3d}, the free variation of which reads
\begin{align}
\label{freevarCS}
\delta \scr C_Q(A)&=\tr\Big(2F\delta A-\td(A\delta A)\Big)=\delta_{AB}\Big(2F^A\delta A^B-\td(A^A\delta A^B)\Big)\,,
\end{align}
and we find that
\begin{align}
Q^{(2)}(F,\delta A)=\delta_{AB}F^A\delta A^B\,,\qquad\Theta(A,\delta A)=-\delta_{AB} A^A\delta A^B\,.
\end{align}
Then, the covariant anomaly $\mathfrak{a}_{\rm cov}$ in \eqref{Aconvariant} can be read off as
\begin{equation}
\mathfrak{a}_{\rm cov}=-\int_M2F\,.
\end{equation}
Now we explain the physical picture of covariant anomaly. Integrating \eqref{freevarCS} over $\tilde M$ and applying the Stokes theorem for the exact term, we have
\begin{align}
\delta\int_{\tilde M} \scr C_Q(A)&=\delta_{AB}\int_{\tilde M} \delta A^A\wedge{}^*J_{\rm bulk}^B+\delta_{AB}\int_{M} \delta A^A\wedge{}^*X^B\,,
\end{align}
where $^*J_{\rm bulk}=2F$ represents the bulk current sourced by $A$, and $^*X=A$ is the Bardeen-Zumino current on the boundary. Recall that the consistent anomaly of theory on $M$ derived above is the covariant divergence of the consistent anomalous current $J_{\rm con}$:
\begin{align}
\label{DJcon}
D^*J_{\rm con}^A=\td A^A\,.
\end{align}
Define the covariant anomalous current on $M$ as $J_{\rm cov}=J_{\rm con}+X$, then its covariant divergence becomes
\begin{align}
\label{DJcov}
D^*J_{\rm cov}^A=\td A^A+\td A^A+[A,A]^A=2F^A\,.
\end{align}
This is the covariant chiral anomaly of the $2d$ Yang-Mills theory. Comparing \eqref{DJcon} and \eqref{DJcov}, we can see that adding the Bardeen-Zumino current covariantizes the consistent anomaly.\footnote{We will present the general proof of this in Appendix \ref{app:covanomaly}, where the connection and curvature are defined in the Lie algebroid context but the algebra follows in the same way.} On the other hand, the charge injected by the bulk current $J_{\rm bulk}$ into $M$ is
\be
Q=\int_{M}{}^*J_{\rm bulk}=\int_M 2F\,,
\ee
which is again the covariant anomaly. Therefore, besides covariantizing the consistent anomaly, the free variation of the Chern-Simons term also provides a physical interpretation for the covariant anomaly: the conservation of the boundary covariant anomalous current $J_{\rm cov}$ is broken because there are bulk charges flowing into the boundary. Thus, the system of bulk plus boundary is anomaly free. This is the anomaly inflow picture for the covariant anomaly. See \cite{Hughes:2012vg,Parrikar:2014usa} for a discussion on its relation to the Hall viscosity of a Chern insulator.

So far we have seen that consistent anomalies can be derived from the BRST cohomology, while to obtain covariant anomalies one needs some additional manipulations. In Chapter~\ref{chap:trivialized} we will see that, after formulating the BRST complex in terms of an Atiyah Lie algebroid, the covariant anomaly and the consistent anomaly can actually be integrated into a unified framework. 
\chapter{Backgrounds on Lie Algebroids}
\label{chap:algebroid}
Our interest in Lie algebroids arises from the fact that the Atiyah Lie algebroid associated with a principal $G$-bundle precisely encodes the algebra of infinitesimal gauge transformations in a manifestly geometric fashion. This allows us to achieve the objective of the infinite-dimensional principal bundle $\mathscr{P}$ introduced for the BRST analysis without having to engage in any of the ad-hoc procedures therein. Before delving into the Atiyah Lie algebroid and gauge theory in the next chapter, we provide in this chapter a general introduction to transitive Lie algebroids following \cite{Ciambelli:2021ujl}. For a more comprehensive discussion on Lie groupoids and Lie algebroids, see, for example, \cite{mackenzie2005general}.

\section{Basics of Lie Algebroids}
\subsection{Transitive Lie Algebroids and Connections}
\label{sec:transitive}
\bdefi
\label{LA}
A vector bundle $A$ over a manifold $M$ together with a map $\rho:A\to TM$ is called a \emph{Lie algebriod} if
\par
(a) $\rho[\umX,\umY]_A=[\rho(\umX),\rho(\umY)]_{TM}\,.\qquad\forall\umX,\umY\in\Gamma(A)$;
\par
(b) $[f\umX,g\umY]_A=fg[\umX,\umY]_{TM}+f(\rho(\umX)g)\umY-g(\rho(\umY)f)\umX\,.\qquad\forall\umX,\umY\in\Gamma(A),\quad f,g\in C^{\infty}(M)$.\\
where $[\umX,\umY]_A$ is the Lie bracket defined on $A$. The map $\rho$ is called the \emph{anchor map}. For vector fields $\uX,\uY$ on $M$, $[\uX,\uY]_{TM}$ is the usual Lie bracket defined on $TM$. $\rho(\umX)g$ is the ordinary derivative of $g$ along $\rho(\umX)\in TM$.
\edefi
Condition (a) above states that $\rho$ is a morphism. Equivalently, the curvature of the map $\rho$ defined as follows vanishes:
\begin{equation}
R^\rho(\umu,\unu)\equiv[\rho(\umX),\rho(\umY)]_{TM}-\rho([\umX,\umY]_A)=0\,.
\end{equation}
If $\rho$ is surjective, then the Lie algebroid is said to be \emph{transitive}. In this case, we have the following short exact sequence 
\begin{equation}
\begin{tikzcd}
0
\arrow{r} 
& 
L
\arrow{r}{j} 
& 
A
\arrow{r}{\rho} 
& 
TM
\arrow{r} 
&
0\,.
\end{tikzcd}
\end{equation}
where $j$ is an inclusion map of a vector bundle $L$ called the \emph{isotropy bundle}, whose image is the kernel of $\rho$, i.e., $\rho\circ j(\umu)=\rho(j(\umu))=0$, $\forall\umu\in\Gamma(L)$. The kernel of $\rho$ is referred to as the vertical sub-bundle $V\subset A$. For sections $\umu$ and $\unu$ on $L$, it is natural to require that $j$ is a morphism, i.e.\
\begin{equation}
\label{jmorphism}
R^j(\umu,\unu)\equiv[j(\umu),j(\unu)]_A-j([\umu,\unu]_L)=0\,.
\end{equation}
\par
Now that we have the vertical sub-bundle $V\subset A$, we would like to define a horizontal sub-bundle $H\subset A$ such that $A=H\oplus V$ globally. From the exact sequence above we can see that the tangent bundle of $M$ can be considered as the quotient $TM=A/V$. However, there is no canonically defined horizontal sub-bundle on the Lie algebroid. Similar to the concept of connections on a principle bundle, choosing a horizontal sub-bundle $H$ of $A$ introduces a connection on $A$.
\bdefi
A map $\sigma:TM\to A$ is called a \emph{connection} (or a \emph{split}) if $\rho\circ\sigma:TM\to TM$ is the identity on $TM$, i.e.\
\begin{equation}
\rho\circ\sigma(\uX)=\rho(\sigma(\uX))=\uX\,,\qquad\forall\uX\in TM\,.
\end{equation}
The map $\sigma\circ\rho:A\to A$ is a projection on $A$, whose image space is the \emph{horizontal sub-bundle} $H\subset A$. 
\edefi
Unlike $\rho$, $\sigma$ is not necessarily a morphism, the curvature of $\sigma$ can be expressed as
\begin{equation}
R^\sigma(\uX,\uY)=[\sigma(\uX),\sigma(\uY)]_A-\sigma([\uX,\uY]_{TM})\,,\qquad\forall\uX,\uY\in\Gamma(TM)\,.
\end{equation}
One can easily verify that $R^\sigma$ is vertical, i.e., lives in the kernel of $\rho$:
\begin{equation}
\rho(R^\sigma(\uX,\uY))=\rho([\sigma(\uX),\sigma(\uY)]_A)-\rho\circ\sigma([\uX,\uY]_{TM})=[\rho\circ\sigma(\uX),\rho\circ\sigma(\uY)]_A-[\uX,\uY]_{TM}=0\,.
\end{equation}
where we used the fact that $\rho$ is a morphism and $\rho\circ\sigma$ is the identity on $TM$. Thus, $R^\sigma(\uX,\uY)\in\Gamma(V)$.
\par
\bdefi
A map $\omega:A\to L$ is called a \emph{connection reform} if it satisfies
\begin{equation}
\text{ker}(\omega)=\text{im}(\sigma)=H\subset A\,.
\end{equation}
For future convenience, we take $\omega\circ j:L\to L$ to be the minus of the identity on $L$, i.e., $\omega(j(\umu))=-\umu$. This will make the definition of curvature align with the familiar form. The map $-j\circ\omega:A\to A$ is a projection on $A$ whose image space is $V$
\edefi
Having a connection on the Lie algebroid characterized by the map $\omega$ and $\sigma$ defines a second short exact sequence in the direction opposite to the first one:
\begin{equation}
\begin{tikzcd}
0
\arrow{r} 
& 
L
\arrow{r}{j} 
& 
A
\arrow[bend left]{l}{\omega}
\arrow{r}{\rho}
& 
TM
\arrow[bend left]{l}{\sigma}
\arrow{r} 
&
0\,.
\end{tikzcd}
\end{equation}
Note that $\omega$ and $\sigma$ are two equivalent ways of defining the Lie algebroid connection, as one will be determined once the other is specified. Later will we also see that there is a third way of characterizing the connection by means of the trivialization. This is exactly what we have seen for connections on a principal bundle in Section~\ref{sec:primer}. The short exact sequence above is also reminiscent of that of a principal bundle in \eqref{Psequence}; however, in the Lie algebroid case each term in the sequence is now a vector bundle over $M$, which brings a lot convenience in implementing maps between vector bundles.
\par
With the two projection maps on $A$ we defined above, Any section $\umX$ of $A$ can be decomposed into its horizontal and vertical parts:
\begin{equation}
\label{XHXV}
\umX=\sigma\circ\rho(\umX)-j\circ\omega(\umX)\equiv \umX_H+\umX_V\,,
\end{equation}
where $\umX_H\equiv\sigma\circ\rho(\umX)$ and $\umX_V\equiv-j\circ\omega(\umX)$. It is useful to keep in mind that 
\begin{equation}
\label{omegaXH}
\omega(\umX_H)=\rho(\umX_V)=0\,.
\end{equation}
The Lie brackets of the horizontal and vertical components of $\umX$ satisfy
\begin{align*}
\rho([\umX_H,\umY_H]_A)&=[\rho(\umX_H),\rho(\umY_H)]_{TM}\,,\\
\rho([\umX_H,\umY_V]_A)&=[\rho(\umX_H),\rho(\umY_V)]_{TM}=0\,,\\
\rho([\umX_V,\umY_V]_A)&=[\rho(\umX_V),\rho(\umY_V)]_{TM}=0\,,
\end{align*}
and hence $[\umX_H,\umY_V]_A$ and $[\umX_V,\umY_V]_A$ are purely vertical, while $[\umX_H,\umY_H]_A$ may have both horizontal and vertical components. This indicates that $V$ is an ideal of $A$ with respect to the Lie bracket of $A$. According to Frobenius's theorem, $[\umX_H,\umY_H]_A$ being purely horizontal means that $H$ is an integrable distribution in $A$. 

\subsection{Exterior Algebra and Coboundary Operators}
Before we introduce the exterior algebra of a Lie algebroid $A$, we first introduce the representation of $A$, namely the action of $A$ on a vector bundle. Suppose $E$ is an arbitrary bundle over $M$, we can introduce a series of bundles representing different operations on $E$. First, the collection of all the endomorphisms on $E$ is denoted by $\End(E)$. An endomorphism is a linear transformation of the section of $E$, whose linearity can be expresses as
\begin{equation}
\label{End}
\varphi(f\un\psi)=f\varphi(\un\psi)\,,\qquad\forall\varphi\in\End(E),\quad f\in C^\infty(M),\quad \psi\in\Gamma(E)\,.
\end{equation}
The bundle of first-order differentiation on $E$ is denoted by $\text{Diff}(E)$, in which $D\in\Gamma(\text{Diff}(E))$ is a first order differential operator on $E$ satisfying the following Leibniz rule\footnote{Technically, one can introduce the bundle of $n^{th}$-order differentiation on $E$, denoted by $\text{Diff}^n(E)$. The bundle $\text{Diff}(E)$  of first-order differentiation is $\text{Diff}^1(E)$, and the bundle $\End(E)$ of endomorphisms on $E$ is $\text{Diff}^0(E)$.}
\begin{equation}
\label{DL}
D(f\un\psi)=fD(\un\psi)+\varphi_f(\un\psi)\,.\qquad f\in C^\infty(M),\quad\varphi_f\in\End(E)\,,
\end{equation}
To introduce the representation of the Lie algebroid, we focus on the following sub-bundle of $\text{Diff}(E)$:
\bdefi
Consider a sub-bundle $\Der (E)$ of $\text{Diff}(E)$ such that $\forall\mD\in\Gamma(\Der(E))$, $\rho_E(\mD)$ is an ordinate derivative on functions, where $\rho_E:\Der(E)\to TM$ is a morphism. In this case, the $\varphi_f$ in \eqref{DL} is a derivative on $f$ associated to $\mD$, i.e.\
\begin{equation}
\label{Der}
\mD(f\un\psi)=f\mD\un\psi+(\rho_E(\mD)f)\un\psi\,.\qquad f\in C^\infty(M),\quad\mD\in\Der(E)\,.
\end{equation}
Each $\mD$ is called a \emph{derivation} on $E$
\edefi
Now we will see that $\Der(E)$ as a vector bundle over $M$ is itself a Lie algebroid. Consider the Lie bracket on $\Der(E)$ given by
\begin{equation}
[\mD,\mD']_{\Der(E)}\un\psi=\mD(\mD'\un\psi)-\mD'(\mD\un\psi)\,.
\end{equation}
Since $\rho_E$ is a morphism, it can be taken as the $\rho$ in the condition (a) of Definition \ref{LA}, and it is straightforward to verify that condition (b) is satisfied. One can also check that
\begin{align*}
[\mD,\mD']_{\Der(E)}(f\un\psi)&=\mD(\mD'(f\un\psi))-\mD'(\mD(f\un\psi))=\mD(f\mD'\un\psi)+\mD((\rho_E(\mD')f)\un\psi)-\mD'(f\mD\un\psi)-\mD'((\rho_E(\mD)f)\un\psi)\\
&=f\mD(\mD'\un\psi)+(\rho_E(\mD)f)\mD'\un\psi+(\rho_E(\mD')f)\mD\un\psi+(\rho_E(\mD)f)(\rho_E(\mD')f)\un\psi\\
&\quad-f\mD'(\mD\un\psi)-(\rho_E(\mD')f)\mD\un\psi-(\rho_E(\mD)f)\mD'\un\psi-(\rho_E(\mD')f)(\rho_E(\mD)f)\un\psi\\
&=f[\mD,\mD']_{\Der(E)}\un\psi+([\rho_E(\mD),\rho_E(\mD')]_{TM}f)\psi=f[\mD,\mD']_{\Der(E)}\un\psi+(\rho_E([\mD,\mD']_{\Der(E)})f)\psi\,,
\end{align*}
which means that $[\mD,\mD']_{\Der(E)}$ is indeed a derivation. Therefore, as a vector bundle over $M$, $\Der(E)$ possesses an anchor map $\rho_E$ and a Lie bracket, which is a well-defined Lie algebroid. Note that when $\rho_E(\mD)=0$, the second term in \eqref{Der} vanishes, and $\mD$ becomes an endomorphism. Hence, the kernel of $\rho_E$ is $\End (E)$ which can be identified as a sub-bundle of $\Der(E)$ by an inclusion map $j_E$. Then, $\Der(E)$ as a Lie algebroid has the following exact sequence:
\begin{equation}
\begin{tikzcd}
0
\arrow{r} 
& 
\End(E)
\arrow{r}{j_E} 
& 
\Der(E)
\arrow{r}{\rho_E}
& 
TM
\arrow{r} 
&
0\,.
\end{tikzcd}
\end{equation}
\par
Now we can introduce a morphism $\phi_E$ between $A$ and $\Der(E)$ that is compatible with the anchor, i.e., $\rho_E\circ\phi_E=\rho$. The morphism condition simply means that $\phi_E$ has a vanishing curvature:
\begin{equation} 
\label{phiEmor}
	R^{\phi_E}(\un\mX,\un\mY) = [\phi_E(\un\mX),\phi_E(\un\mY)]_{\text{Der}(E)} - \phi_E([\un\mX,\un\mY]_A) = 0\,,\qquad\forall  \un\mX,\un\mY \in A\,,
\end{equation}
and the compatibility condition ensures that $\phi_E$ maps a section $\un\mX$ on $A$ into a derivation $\phi(\un\mX)$ satisfying the Leibniz-like identity enforced by \eqref{Der}:
\begin{equation}
\label{phiEf}
	\phi_E(\un\mX)(f \un\psi) = f \phi_E(\un\mX)(\un\psi) + \rho(\un\mX)(f) \un\psi\,,\qquad \forall  \un\mX \in A\,,\quad f \in C^{\infty}(M)\,,\quad \un\psi \in \Gamma(E)\,.
\end{equation}
Then, $\phi_E$ provides a representation of $A$; that is, each section of $A$ corresponds to an action on $E$. Also, we can introduce a morphism $v_E:L\to\End(E)$ satisfying $\phi_E\circ j=j_E\circ v_E$, making $\End(E)$ the representation of $L$. The diagram of the two Lie algebroids $A$ and $\Der(E)$ can be illustrated as follows:
\begin{equation}
\label{graph}
\begin{tikzcd}
0
\arrow{r} 
& 
L
\arrow{dd}{v_E}
\arrow{r}{j} 
& 
A
\arrow{dr}{\rho}
\arrow{dd}{\phi_E}
& 
&
\\
&&&TM\arrow{r} &0\,.
\\
0
\arrow{r} 
& 
\End(E)
\arrow{r}{j_E} 
& 
\Der(E)
\arrow{ur}{\rho_E}
& 
&
\end{tikzcd}
\end{equation}
The above diagram is a commutative diagram in the sense that the square part satisfies $\phi_E\circ j=j_E\circ v_E$ and the triangle part satisfies $\rho_E\circ\phi_E=\rho$.
\par
Suppose $\{\un e_a\}$ is a basis of $\Gamma(E)$, and $\{f^b\}$ is a dual basis, namely a basis of $\Gamma(E^*)$, then $\{\un e_a\otimes f^b\}$ will be a basis of $\Gamma(\End(E))$. For any $\un\psi\in\Gamma(E)$ and $\varphi\in\End(E)$, we have
\begin{equation}
\varphi(\un\psi)=\varphi^a{}_b\un e_a\otimes f^b(\psi^c\un e_c)=(\varphi^a{}_b\psi^b)\un e_a\,.
\end{equation}
Let $\{\un t_A\}$ be a basis of $\Gamma(L)$. For any $\un\mu=\mu^A\un t_A\in\Gamma(L)$, the representation of $L$ offered by $v_E$ gives
\begin{equation}
v_E(\un\mu)=\mu^Av_E(\un t_A)=\mu^A(t_A)^a{}_b\un e_a\otimes f^b\equiv\mu^a{}_b\un e_a\otimes f^b\,.
\end{equation}
In this matrix representation, we can also have the following commutators:
\begin{align}
[v_E(\un t_A),v_E(\un t_B)]_{\End(E)}&=((t_A)^a{}_c(t_B)^c{}_b-(t_B)^a{}_c(t_A)^c{}_b)\un e_a\otimes f^b\,,\\
v_E([\un t_A,\un t_B]_L)&=v_E(f_{AB}{}^C\un t_C)=f_{AB}{}^C(t_C)^a{}_b\un e_a\otimes f^b\,,
\end{align}
where $f_{AB}{}^C$ can be interpreted the structure constants of $L$. Since $v_E$ is a morphism, comparing the above commutators yields
\begin{align}
\label{bracketL}
[t_A,t_B]^a{}_b=f_{AB}{}^C(t_C)^a{}_b\,.
\end{align}
Note that since $t_A$ are sections on $L$, the ``structure constants'' $f_{AB}{}^C$ are actually functions on the base manifold $M$. From \eqref{jmorphism} and the condition (b) of Definition \ref{LA} we have see that the Lie bracket on $L$ is linear, i.e.
\begin{equation}
[f\umu,g\unu]_L=fg[\umu,\unu]_L\,.
\end{equation}
Thus, evaluating at each point $x\in M$, the Lie bracket on the isotropy bundle $L$ defines the fiber over $x$ as a Lie algebra, called the \emph{isotropy Lie algebra} at $x$, then $f_{AB}{}^C(x)$ will be the structure constants of this Lie algebra.
\par
Now we come to the main focus of this subsection, the exterior algebra of the Lie algebroid $A$, which will be crucial in later chapters. The exterior algebra (cochain complex) $\Omega(A) $ of $A$ is defined as\footnote{Previously, we used the standard notation for the exterior algebra by writing $\Omega^p(M)\equiv\Gamma(\wedge^p T^*M)$. Since in the algebroid context we will be mainly dealing with vector bundles, from now on we will switch the notation to $\Omega^p(A)\equiv\Gamma(\wedge^p {}^*A)$; for example, $\Omega^p(M)$ will be denoted by $\Omega^p(TM)$.}
\be 
\Omega(A) = \bigoplus_{p=0}^{\text{rank}\,A} \Omega^p(A)\,,
\ee
where each $\Omega^p(A) \equiv \Gamma(\wedge^p A^*)$ consists of totally antisymmetric $p$-linear maps from $\Gamma(A^{\otimes p})$ to $C^\infty(M)$. The exterior algebra $\Omega(A)$ has a well-defined coboundary operator $\hatd: \Omega^p(A) \rightarrow \Omega^{p+1}(A)$ determined by the anchor map $\rho$ and the bracket on $A$, which acts as the exterior derivative on the forms on $A$.
\bdefi
The map $\hat\td:\Omega^p(A)\to\Omega^{p+1}(A)$ is called a \emph{coboundary operator} or \emph{exterior derivative operator} on $A$ if $\forall\un\eta\in\Omega^{p}(A)$,
\begin{align} \label{dhat for trivial bundle}
    \hatd\eta(\un{\mX}_1, \ldots, \un{\mX}_{p+1}) ={}& \sum_{i} (-1)^{i+1} \rho(\un{\mX}_i) \eta(\un{\mX}_1, \ldots, \widehat{\un{\mX}_i}, \ldots, \un{\mX}_{p+1})\nn \\
   & + \sum_{i < j} (-1)^{i + j} \eta([\un{\mX}_i, \un{\mX}_j]_A, \un{\mX}_1, \ldots, \widehat{\un{\mX}_i}, \ldots, \widehat{\un{\mX}_j}, \ldots, \un{\mX}_{p+1}) \,,
\end{align}
where $\un{\mX}_1,\ldots,\un{\mX}_{p+1}$ are arbitrary sections on $A$, and the hats on $\un{\mX}_i$ stands for omission. This equation is called the \emph{Koszul formula}.
\edefi
\par
By means of a Lie algebroid representation $\phi_E$, the exterior algebra $\Omega(A)$ can be extended to $\Omega(A;E)$, namely the exterior algebra on $A$ with values in the vector bundle $E$. Denote the collection of $E$-valued $p$-forms on $A$ as $\Omega^{p}(A;E)\equiv\Gamma(\wedge^nA^*\times E)$. Then, we define 
\begin{equation}
\Omega(A;E)=\bigoplus_{p=0}^{\text{rank}\,A}\Omega^p(A;E)\,.
\end{equation}
The corresponding coboundary operator can be defined via a generalized Koszul formula follows:
\bdefi
\label{def:hatd}
The map $\hat\td_E:\Omega^p(A;E)\to\Omega^{p+1}(A;E)$ is called a \emph{coboundary operator} or \emph{exterior derivative operator} if $\forall\un\psi_p\in\Omega^{p}(A;E)$,
\begin{align}
(\hat\td_E\un\psi_p)(\umX_1,\cdots,\umX_{p+1})&\equiv\sum_{i=1}^{p+1}(-1)^{i+1}\phi_E(\umX_i)(\un\psi_p(\umX_1,\cdots,\widehat{\umX_i},\cdots,\umX_{p+1}))\nn\\
\label{dhat}
&\quad+\sum_{i<j}^{p+1}(-1)^{i+j}\un\psi_n([\umX_i,\umX_j]_A,\umX_1,\cdots,\widehat{\umX_i},\cdots,\widehat{\umX_j},\cdots,\umX_{p+1})\,.
\end{align}
\edefi
For simplicity, we will later refer to the coboundary operator as simply $\hatd$, leaving the particular representation $E$ implicit.
\par
The operator $\hatd$ can be verified to be nilpotent as a result of \eqref{phiEmor} and the fact that the Lie bracket on $A$ satisfies the Jacobi identity. It can also be verify that the $\hat\td$ defined from the formula above is linear in the $\umX_i$ in each slot, i.e.,
\begin{align}
\label{dhatlinear}
\quad(\hatd\un\psi_p)(\umX_1,\cdots,f\umX_i,\cdots,\umX_{p+1})=f(\hatd\un\psi_p)(\umX_1,\cdots,\umX_i,\cdots,\umX_{p+1})\,,\qquad \forall i=1,\cdots,p+1\,,\quad f\in C^\infty(M)\,.
\end{align}
The proofs of these properties of $\hatd$ can be found in Appendix \ref{app:dhat}.
\par
For the $p=0$ case, the Koszul formula \eqref{dhat} reduces to
\begin{equation}
\label{phiE}
(\hat\td\un\psi)(\umX)=\phi_E(\umX)(\un\psi)\,,\qquad\un\psi\in\Gamma(E)\,.
\end{equation}
That is, the 1-from $\hatd\un\psi$ on $A$ acting on $\umX$ can be seen as the derivation $\phi_E(\umX)$ acting on $\psi$.

For the $p=1$ and $p=2$ cases, \eqref{dhat} reads
\begin{align}
\label{d2}
(\hatd\un\psi_1)(\umX_1,\umX_2)&=\phi_E(\umX_1)\un\psi_1(\umX_2)-\phi_E(\umX_2)\un\psi_1(\umX_1)-\un\psi_1([\umX_1,\umX_2]_A)\,,\\
(\hatd\un\psi_2)(\umX_1,\umX_2,\umX_3)&=\phi_E(\umX_1)\un\psi_2(\umX_2,\umX_3)-\phi_E(\umX_2)\un\psi_2(\umX_1,\umX_3)+\phi_E(\umX_3)\un\psi_2(\umX_1,\umX_2)\nn\\
&\quad-\un\psi_2([\umX_1,\umX_2]_A,\umX_3)+\un\psi_2([\umX_1,\umX_3]_A,\umX_2)-\un\psi_2([\umX_2,\umX_3]_A,\umX_1)\,.
\end{align}

\subsection{Curvature}
\label{sec:curvature}
In this subsection, we will introduce several notions of the curvature on a Lie algebroid $A$, and show that how they eventually are in fact different ways of quantifying the same curvature on $A$.

First, since the connection reform $\omega:A\to L$ can be regarded as an $L$-valued 1-form on $A$, it is natural to define the curvature as an $L$-valued 2-form on $A$ via the Cartan's second equation of structure similar to the curvature 2-form \eqref{Curvature on P(M,G)} on a principal bundle:
\begin{align}
\label{Omega2}
\Omega\equiv\hatd\omega+\frac{1}{2}[\omega,\omega]_L\in\Omega^2(A)\otimes L\,.
\end{align}
The curvature 2-form defined in this way is called the \emph{connection reform} on $A$. On the other hand, using the map $\sigma:TM\to A$, we can define the curvature following \eqref{RsigmaP} on the principal bundle:
\begin{equation}
\label{Rsigma}
    R^{\sigma}(\underline{X}, \underline{Y}) = [\sigma(\underline{X}), \sigma(\underline{Y})]_{A} - \sigma([\underline{X}, \underline{Y}]_{TM}) \in A\,.
\end{equation}
\par
Now we demonstrate how these two notions of curvature are related. Since $L$ is a vector bundle over $M$, we can take $L$ to be the vector bundle $E$ in the last subsection and construct the Lie algebroid $\Der(L)$ in the manner we introduced $\Der(E)$, which provides a representation for a Lie algebroid $A$. This representation is referred to as the \emph{adjoint representation} of $A$. Denote the morphism between $A$ and $\Der(L)$ by $\phi_L$. Given $\umX\in\Gamma(A)$ and $\un\mu\in\Gamma(L)$, we can define $\phi_L$ using the Lie bracket on $A$ as follows:
\begin{equation}
\label{phiL}
j(\phi_L(\umX)(\mu))=[\umX,j(\mu)]_A\,.
\end{equation}
Note that $\phi_L$ being a morphism give that
\begin{align*}
j(\phi_L([\umX,\umY]_A)(\mu))&=j([\phi_L(\umX),\phi_L(\umY)]_{\Der(L)}(\mu))=j(\phi_L(\umX)\phi_L(\umY)(\mu))-j(\phi_L(\umY)\phi_L(\umX)(\mu))\\
&=[\umX,j(\phi_L(\umY)(\mu))]_A-[\umY,j(\phi_L(\umX)(\mu))]_A=[\umX,[\umY,j(\mu)]_A]_A-[\umY,[\umX,j(\mu)]_A]_A\,.
\end{align*}
Then it follows from \eqref{phiL} that
\begin{align*}
[[\umX,\umY]_A,j(\mu)]_A&=[\umX,[\umY,j(\mu)]_A]_A-[\umY,[\umX,j(\mu)]_A]_A\,,
\end{align*}
which is exactly the Jacobi identity for the Lie bracket on $A$. Thus, $\phi_L$ defined in \eqref{phiL} is automatically a morphism as the Lie bracket on $A$ satisfies the Jacobi identity.
\par
Now we evaluate the curvature 2-form $\Omega$ defined in \eqref{Omega2}. Since $\hatd\omega$ is an $L$-valued $2$-form. Using \eqref{d2} and \eqref{phiL}, we have
\begin{align}
j((\hatd\omega)(\umX,\umY))&=j(\phi_L(\umX)\omega(\umY))-j(\phi_L(\umY)\omega(\umX))-j(\omega([\umX,\umY]_A))\nn\\
&=[\umX,j(\omega(\umY))]_A-[\umY,j(\omega(\umX))]_A-j(\omega([\umX,\umY]_A))\,.
\end{align}
Let $\umX_H,\umY_H$ represent the horizontal part of $\umX,\umY$, and $\umX_V,\umY_V$ represent the vertical part of $\umX,\umY$ as we defined in \eqref{XHXV}. Then, the equation above becomes
\begin{align}
j((\hatd\omega)(\umX,\umY))&=-[\umX,\umY_V]_A+[\umY,\umX_V]_A-j(\omega([\umX,\umY]_A))\nn\\
&=-[\umX_H,\umY_V]_A-[\umX_V,\umY_V]_A+[\umY_H,\umX_V]_A+[\umY_V,\umX_V]_A\nn\\
&\quad-j(\omega([\umX_H,\umY_H]_A))+[\umX_H,\umY_V]_A+[\umX_V,\umY_H]_A+[\umX_V,\umY_V]_A\nn\\
&=-[\umX_V,\umY_V]_A-j(\omega([\umX_H,\umY_H]_A))\nn\\
\label{idomega}
&=-j([\omega(\umX),\omega(\umY)]_L)-j(\omega([\umX_H,\umY_H]_A))\,,
\end{align}
where in the second equality we used the fact that $[\umX_H,\umY_V]_A$ and $[\umX_V,\umY_V]_A$ are purely vertical, and in the last equality we used the fact that $j$ is a morphism. Noticing that
\begin{align}
[\omega,\omega]_L(\umX,\umY)=[\omega(\umX),\omega(\umY)]_L-[\omega(\umY),\omega(\umX)]_L=2[\omega(\umX),\omega(\umY)]_L\,,
\end{align}
we can see from the definition of $\Omega$ that \eqref{idomega} gives
\begin{align}
\label{Omega1}
j(\Omega(\umX_H,\umY_H))&=-j(\omega([\umX_H,\umY_H]_A))=[\umX_H,\umY_H]_V\,,
\end{align}
where $[\umX_H,\umY_H]_V$ stands for the vertical part of $[\umX_H,\umY_H]_A$. Applying $\omega$ to both sides of \eqref{Omega1} yields
\begin{align}
\label{Omega2d}
\Omega(\umX_H,\umY_H)&=-\omega([\umX_H,\umY_H]_A)\,.
\end{align}
The right-hand side of \eqref{Omega1} can be further evaluated as
\begin{align*}
[\umX_H,\umY_H]_V&=[\umX_H,\umY_H]_A-\sigma(\rho[\umX_H,\umY_H]_A)=[\sigma(\rho(\umX)),\sigma(\rho(\umY))]_A-\sigma([\rho(\umX_H),\rho(\umY_H)]_{TM})\\
&=[\sigma(\uX),\sigma(\uY)]_A-\sigma([\uX,\uY]_{TM})=R^\sigma(\uX,\uY)\,,
\end{align*}
where $\uX\equiv\rho(\umX)$, $\uY\equiv\rho(\umY)$. Therefore, we have the following correspondence between the two notions of curvature introduced in \eqref{Omega2} and \eqref{Rsigma}:
\begin{align}
\label{Rsigma2}
j(\Omega(\umX,\umY))&=R^\sigma(\uX,\uY)\,,
\end{align}
which is analogous to the relation \eqref{jFRsigma} for the curvature on a principal bundle. 
\par
Beside Cartan's second equation of structure, another way to characterize the curvature through the map $\omega$ is to introduce the curvature of the map itself:\footnote{More precisely, this should be regarded as the curvature of $-\omega$ due to the plus sign of the second term.}
\begin{align}
R^\omega(\umX,\umY)\equiv[\omega(\umX),\omega(\umY)]_L+\omega([\umX,\umY]_A)\,.
\end{align}
Applying $j$ to both sides, we can verify that
\begin{align}
j(R^\omega(\umX_V,\umY_V))&=[j(\omega(\umX_V)),j(\omega(\umY_V))]_A+j(\omega([\umX_V,\umY_V]_A))\nn\\
&=[\umX_V,\umY_V]_A-[\umX_V,\umY_V]_A=0\,.
\end{align}
Since $j$ is an inclusion, this indicates that $R^\omega(\umX_V,\umY_V)=0$. Also, it follows from $\omega(\umX_H)=0$ that
\begin{align}
\label{Romega}
R^\omega(\umX_H,\umY_H)=\omega([\umX_H,\umY_H]_A)\,,\\
\label{RomegaV}
R^\omega(\umX_H,\umY_V)=\omega([\umX_H,\umY_V]_A)\,.
\end{align}
Form \eqref{Omega2d} and \eqref{Romega} we can see that 
\begin{align}
\label{OmegaR}
\Omega(\umX,\umY)=-R^\omega(\umX_H,\umY_H)\,.
\end{align} 
Together with \eqref{Rsigma2}, the curvatures we defined above are related in the following way:
\begin{align}
\label{curvature}
R^\sigma(\uX,\uY)=j(\Omega(\umX,\umY))=-j(R^\omega(\umX_H,\umY_H))\,.
\end{align}
Thus, these notions of curvature actually represent the same thing, namely the curvature of the Lie algebroid. The curvature defined in each way shown in \eqref{curvature} being nonvanishing is then the manifestation of the failure of $\sigma$ and $-\omega$ being morphisms. 
\par One can also easily see from \eqref{OmegaR} that $\Omega(\umX_V,\umY)=0$, i.e.\ the curvature reform of a transitive Lie algebroid is automatically horizontal. As we saw in Section \ref{sec:BRST}, in the geometry formulation of BRST using the principle bundle language, this is a condition added by hand. We will show in the following chapter that this result is equivalent to the Russian formula \eqref{Russian Formula 1}, which now arises naturally from the structure of Lie algebroid (more precisely, from the fact $\rho$ and $j$ are morphisms). 
\par
Later in Subsection \ref{sec:trivial} we will see that the curvature of a Lie algebroid can also be characterized in a trivialization, which also provides equivalent information as the notions of curvature introduced above.

\subsection{The Connection and Curvature Induced by a Representation}
\label{sec:reps}
Once the connection on $A$ specified by the pair of maps $\omega$ and $\sigma$ is introduced, it also induces a connection on the representation algebroid furnished by a vector bundle $E$. More precisely, the representation $\phi_E$ of a Lie algebroid with connection determines a pair of maps $\nabla^E: TM \rightarrow \text{Der}(E)$ and $\omega_E:\Der(E)\to\End(E)$, where $\nabla_E$ can be interpreted as a covariant derivative operator on $E$, and $\omega_E$ is the connection reform on the algebroid $\Der(E)$. To see how this pair of maps comes about, we split $\phi_E(\umX)\in\Der(E)$ by considering $\umX$ as the sum of its horizontal part $\umX_H=\sigma\circ\rho(\un\mX)$ and the vertical part $\umX_V=j\circ\omega(\umX)$:
\begin{align}
\Aconn{E}(\un\mX)&=\Aconn{E}(\sigma\circ\rho(\un\mX)+j\circ\omega(\umX))\nn\\
&=\phi_E\circ\sigma(\rho(\un\mX))+j_E\circ v_E\circ\omega(\umX)\,,
\end{align}
where we used the fact that $\phi_E\circ j=j_E\circ v_E$. Now we define $\nabla^E$ and $\omega_E$ by requiring that
\begin{align}
\label{nablaEomegaE}
\nabla^E_{\rho(\umX)}&=\phi_E\circ\sigma(\rho(\un\mX))=\phi_E(\umX_H)\,,\\
\label{omegaEphiE}
\omega_E\circ\phi_E(\umX)&=v_E\circ\omega(\umX)=v_E\circ\omega(\umX_V)\,.
\end{align}
Then, given any section $\umX$ on $A$, $\phi_E(\umX)\in\Der(E)$ can be split into
\begin{equation} \label{Basic A Connection}
	\Aconn{E}(\un\mX)(\un\psi)= \nabla^E_{\rho(\un\mX)}(\un\psi)-j_E\circ\omega_E\circ\phi_E(\un\mX)(\un\psi) \,,\qquad\forall \un\psi\in\Gamma(E)\,.
\end{equation}
The image of $j_E$ in the second term lives in the vertical sub-bundle of $\Der(E)$, and $\nabla^E_{\rho(\un\mX)}$ defines the horizontal sub-bundle of $\Der(E)$. This also implies that $\text{im}(\nabla^E)=\ker(\omega_E)$. The representation algebroid associated to $A$ and their connections can be expressed diagrammatically as
\begin{equation}
\label{graph2}
\begin{tikzcd}
0
\arrow{r} 
& 
L
\arrow[swap]{dd}{v_E}
\arrow{r}{j} 
& 
A
\arrow{dr}{\rho}
\arrow[swap]{dd}{\phi_E}
\arrow[bend left=20]{l}{\omega} 
& 
&
\\
&&&
TM
\arrow{r} 
\arrow[bend left=20]{ul}{\sigma} 
\arrow[bend left=20]{dl}{}{\nabla^E} &0\,.
\\
0
\arrow{r} 
& 
\End(E)
\arrow{r}{j_E} 
& 
\Der(E)
\arrow{ur}{\rho_E}
\arrow[bend left=20]{l}{\omega_E} 
& 
&
\end{tikzcd}
\end{equation}
The requirements in \eqref{nablaEomegaE} and \eqref{omegaEphiE} ensure that \eqref{graph2} is a commutative diagram in the sense that both the square and triangle parts commute as the arrows go in any directions.
\par
Recall that the representation $\phi_E$ also defines a coboundary operator $\hatd$ through \eqref{phiE}, then for any 0-form $\psi_0\in\Gamma(E)$, the 1-form $\hatd\un\psi_0$ can be obtained from \eqref{Basic A Connection} as
\begin{equation}
\label{dphiE}
(\hatd\un\psi_0)(\umX)= \nabla^E_{\rho(\un\mX)}(\un\psi_0) -\omega_E\circ\phi_E(\un\mX) (\un\psi_0)=\nabla^E_{\rho(\un\mX)}(\un\psi_0)-v_E(\omega(\umX))(\un\psi_0)\,,
\end{equation}
where we omitted $j_E$ since $\End(E)$ is the vertical sub-bundle of $\Der(E)$ and the inclusion $j_E:\End(E)\to\End(E)\subset\Der(E)$ is a trivial map. The two terms on the right-hand side of the above equation separate the action of $\hatd$ into a horizontal part and a vertical part.
\par
To further understand the geometric meaning of $\nabla^E$ as the ``horizontal part" of $\hatd$, we define its curvature as a map ${\cal R}^E:A\times A\times E\to E$:\begin{align}
\label{curvatureRE}
{\cal R}^E(\umX,\umY)(\un\psi_0)\equiv[\nabla^E_{\rho(\umX)},\nabla^E_{\rho(\umY)}]_{\Der(E)}\psi_0-\nabla^E_{\rho([\umX,\umY])}\psi_0\,,
\end{align}
Noticing that $\rho(\umX)=\rho(\umX_H)$, one can readily see that by definition ${\cal R}^E(\umX,\umY_V)=0$, and hence the map is in fact ${\cal R}^E:H\times H\times E\to E$, which is only determined by the horizontal distribution. Furthermore, from the fact that $\phi_E$ is a morphism we can show that
\begin{align}
{\cal R}^E(\umX,\umY)(\un\psi_0)=v_E(\Omega(\umX,\umY))(\un\psi_0)\,.
\end{align}
The detailed derivation will be provided in Appendix \ref{app:ROmega}. This indicates that ${\cal R}^E$ is nothing but another way of representing the curvature of the Lie algebroid, which represents $\Omega$ as an endomorphism on $E$ through $v_E$. Moreover, $\nabla^E_{\rho(\umX)}$ can be considered as a covariant derivative operator on $TM$ (an induced connection) along the $\rho(\umX)$ direction, whose curvature is defined in the familiar way:
\begin{equation}\label{Curvature of Nabla}
R^{E}(\un{X},\un{Y}) \equiv [\nabla^E_{\un{X}},\nabla^E_{\un{Y}}]_{\text{Der}(E)} - \nabla^E_{[\un{X},\un{Y}]_{TM}}\,,\qquad\forall\un{X},\un{Y}\in\Gamma(TM)\,.
\end{equation}
In other words, the curvature of $\nabla^E$ viewed as a connection on $TM$ is determined entirely by the curvature of the horizontal distribution $H$ of $A$.
\par
It is instructive to take a look a special case we encountered before, namely the adjoint representation, where $E$ is the isotropy bundle $L$. In this case $\phi_L$ can be introduced using the Lie bracket defined in \eqref{phiL}. Applying $\omega$ to both sides of \eqref{phiL} yields
\begin{equation}
\phi_L(\umX)(\un\mu)=-\omega([\umX,j(\un\mu)]_A)\,.
\end{equation}
Let us consider $\umX$ as the sum of $\umX_H$ and $\umX_V$, then using \eqref{RomegaV} we have
\begin{align}
\phi_L(\umX_H)(\un\mu)&=-\omega([\umX_H,j(\un\mu)]_A)=-R^\omega(\umX_H,j(\un\mu))\,,\\
\phi_L(\umX_V)(\un\mu)&=-\omega([\umX_V,j(\un\mu)]_A)=\omega([j(\omega(\umX_V)),j(\un\mu)]_A)=\omega(j([\omega(\umX_V),\un\mu]_A))=-[\omega(\umX_V),\un\mu]_L\,,
\end{align}
and thus
\begin{align}
\label{phiL1}
\phi_L(\umX)(\un\mu)=-\omega([\umX_H+\umX_V,j(\un\mu)]_A)=-R^\omega(\umX_H,j(\un\mu))-[\omega(\umX_V),\un\mu]_L\,.
\end{align}
In the adjoint representation, we can take $v_L:L\to\End(L)$ as follows:
\begin{align}
\label{adjrep}
(v_L(\un\mu))(\un\nu)=[\un\mu,\un\nu]_L\,,\qquad\un\mu,\un\nu\in L\,.
\end{align}
Using the above equation and \eqref{phiE}, we can further write \eqref{phiL1} as
\begin{align}
\label{vL}
(\hatd\un\mu)(\umX)=-R^\omega(\umX_H,j(\un\mu))-\omega_L(\phi_L(\umX_V))=-R^\omega(\umX_H,j(\un\mu))-v_L(\omega(\umX_V))(\un\mu)\,.
\end{align}
Comparing this with \eqref{dphiE}, we can recognize that 
\begin{align}
\nabla^L_{\rho(\umX)}\un\mu=-R^\omega(\umX_H,j(\un\mu))\,.
\end{align}
Define the curvature ${\cal R}^L:A\times A\times L\to L$ of $\nabla^L$ as follows:
\begin{align}
{\cal R}^L(\umX,\umY)(\un\mu)\equiv[\nabla^L_{\rho(\umX)},\nabla^L_{\rho(\umY)}]_{\Der(L)}\un\mu-\nabla^L_{\rho([\umX,\umY]_A)}\un\mu\,,
\end{align}
In a more direct way than the case of a general representation, the curvature defined in the above equation can be evaluated to be (see Appendix \ref{app:ROmega} for details)
\begin{align*}
{\cal R}^L(\umX,\umY)(\un\mu)=v_L(\Omega([\umX_H,\umY_H])(\un\mu)\,,
\end{align*}
which means that ${\cal R}^L$ also represents the curvature of the Lie algebroid. Therefore, in the adjoint representation, $\nabla^L$ can be interpreted as the covariant derivative on $TM$ and $\omega_L$ can be represented by the Lie bracket on $L$.

\section{Bases and Lie Brackets}
Before moving on to the discussion of Atiyah Lie algebroids, we finish off this chapter by introducing the maps between bundles in terms of bases, and summarize some useful results by means of index notation to facilitate the discussions later. 
\par
Suppose $\{\un E _\uM\}$ is a basis of $\Gamma(A)$, $\{\un\p_\mu\}$ is a basis of $\Gamma(TM)$, and $\{\un t_A\}$ is a basis of $\Gamma(L)$, where $\uM=1,\cdots,\dim A$, $\mu=1,\cdots,\dim M$, and $A=1,\cdots,\text{rank}\,L$. The maps $\rho$, $\sigma$, $j$, $\omega$ can be expressed as matrices with indices as follows:
\begin{align}
\rho(\un E _\uM)=\rho^\mu{} _\uM\un\p_\mu\,,\qquad \sigma(\un\p_\mu)=\sigma^\uM{}_\mu\un E_\uM\,,\qquad j(\un t_A)=j^\uM{}_A\un E _\uN\,,\qquad\omega(\un E _\uM)=\omega^A{} _\uM\un t_A\,.
\end{align}
Recall the following properties:
\begin{align}
\rho\circ\sigma=Id_{TM}\,,\qquad \omega\circ j=-Id_L\,,\qquad\rho\circ j=0\,,\qquad\omega\circ\sigma=0\,.
\end{align}
Using the index notation these can be written as
\begin{align}
\label{Id}
\rho^\nu{} _\uM\sigma^\uM{}_\mu=\delta^\nu{}_\mu\,,\qquad\omega^A{} _\uM j^\uM{}_B=-\delta^A{}_B\,,\qquad \rho^\mu{} _\uM j^\uM{}_A=0\,,\qquad\omega^A{} _\uM\sigma^\uM{}_\mu=0\,.
\end{align}
Given a section $\umX$ of $A$, its decomposition \eqref{XHXV} can be expressed as
\begin{align}
\umX=\umX^\uM\un E _\uM=\umX^\uM\sigma^\uN{}_\mu\rho^\mu{} _\uM\un E _\uN-\umX^{\uM} j^\uN{}_A\omega^A{} _\uM\un E _\uN\,.
\end{align}
Under a basis transformation, the components of $\umX$ transform correspondingly as
\begin{align}
\label{trans}
\un E _\uM=J^\uN{} _\uM\un E' _\uN\,,\qquad\umX^\uM=(J _\uN{}^\uM)^{-1}\umX'^\uN\,,
\end{align}
so that the vector field $\umX$ is invariant:
\begin{align}
\umX=\umX^\uM\un E _\uM=\umX'^\uN(J _\uN{}^\uM)^{-1}J^{\un P}{} _\uM\un E'_{\un P}=\umX'^\uN\un E' _\uN=\umX'\,.
\end{align}
\par
Now we consider a frame $\{\un E _\uM\}$ where $\un M$ can be separated into $\un M=(\un\alpha,\un A)$ such that $\un E_{\ualpha}$ spans $\Gamma(H)$ ($\ualpha=1,\cdots,\dim M$) and $\un E_{\un A}$ spans $\Gamma(V)$ ($\un A=1,\cdots,\text{rank}\,L$). This kind of frame is called a \emph{split frame}. The transformation matrix in \eqref{trans} between two split frames is block-diagonalized:
\begin{align}
\label{Etrans}
\un E_{\ualpha}=J^{\ubeta}{}_{\ualpha}\un E'_{\ubeta}\,,\qquad\un E_{\un A}=K^{\un B}{}_{\un A}\un E'_{\un B}\,,
\end{align}
where we denoted $J^{\un B}{}_{\un A}$ by $K^{\un B}{}_{\un A}$ for future use. By definition, the image of $\sigma$ is the horizontal sub-bundle $V\subset A$, and the image of $j$ is the vertical sub-bundle $V\subset A$, and hence $\sigma(\un\p_\mu)\in\Gamma(H)$, $j(\un t_A)\in\Gamma(V)$. Also, it follows from \eqref{omegaXH} that $\rho(\un E_{\un A})=\omega(\un E_{\ualpha})=0$. In terms of indices, these indicates that
\begin{align}
\sigma^{\un A}{}_\mu=0\,,\qquad j^{\ualpha}{}_A=0\,,\qquad\rho^\mu{}_{\un A}=0\,,\qquad\omega^A{}_{\ualpha}=0\,.
\end{align}
Then, the non-vanishing components of these maps are $\sigma^{\ualpha}{}_\mu$, $j^{\un A}{}_A$, $\rho^\mu{}_{\ualpha}$ and $\omega^A{}_{\un A}$. Thus, in the split frame we have
\begin{align}
\label{iotat}
j(\un t_A)=j^{\un A}{}_A\un E_{\un A}+j^{\ualpha}{}_A\un E_{\ualpha}=j^{\un A}{}_A\un E_{\un A}\,,\qquad\sigma(\un\p_\mu)=\sigma^{\un A}{}_\mu\un E_{\un A}+\sigma^{\ualpha}{}_\mu\un E_{\ualpha}=\sigma^{\ualpha}{}_\mu\un E_{\ualpha}\,,
\end{align}
and \eqref{Id} becomes
\begin{align}
\rho^\nu{}_{\ualpha}\sigma^{\ualpha}{}_\mu&=\delta^\nu{}_\mu\,,\qquad\omega^A{}_{\un A}j^{\un A}{}_B=-\delta^A{}_B\,.
\end{align}
\par
We can also introduce a dual basis $\{E^\uM\}$, namely a basis of $\Gamma(A^*)$ satisfying $E^\uM(\un E _\uN)=\delta^\uM{} _\uN$. When $\{\un E _\uM\}$ is a split frame $\{E_{\ualpha},E_{\un A}\}$, $\{E^\uM\}$ will be a split dual frame $\{E^{\ualpha},E^{\un A}\}$ with
\begin{align}
\label{EE}
E^{\ualpha}(\un E_{\beta})=\delta^{\ualpha}{}_{\ubeta}\,,\qquad E^{\ualpha}(\un E_{\un A})=0\,,\qquad E^{\un A}(\un E_{\un B})=\delta^{\un A}{}_{\un B}\,,\qquad E^{\un A}(\un E_{\beta})=0\,.
\end{align}
Then the forms on $A$ can be expanded in the dual basis. We also introduce the bases $\{\td x^{\mu}\}$ for $\Gamma(T^*M)$ and $\{t^A\}$ for $\Gamma(L)$, i.e., the dual bases for $\{\un{\partial}_{\mu}\}$ and $\{\un{t}_A\}$, satisfying
\begin{equation}
	\td x^{\mu}(\un{\partial}_{\nu}) = \delta^{\mu}{}_{\nu},\qquad t^A(\un{t}_B) = \delta^A{}_B\,,\qquad\td x^{\mu}(\un{t}_A) = 0\,,\qquad t^A(\un{\partial}_{\mu}) = 0\,.
\end{equation}
These bases will be useful for the discussion of the trivialization of Lie algebroids. In the dual basis on $A$, the connection and curvature reforms can be written as
\begin{align}
\label{basisTML}
\omega=\omega^A{}_{\un A}E^{\un A}\otimes t_A\,,\qquad
\Omega=\Omega^A{}_{\ualpha\ubeta}E^{\ualpha}\wedge E^{\ubeta}\otimes t_A\,,
\end{align}
where we used fact that $\omega$ is vertical ($\omega^A{}_{\ualpha}=0$) and $\Omega$ is horizontal.
\par
Now we look at the vector bundle $E$ and the covariant derivative $\nabla^E$. Suppose $\{\un{e}_a\}$ is a basis of $\Gamma(E)$. Given $\umX\in\Gamma(A)$, $\nabla^E_{\umX}\un e_a$ is a section on $E$, we can expand it using $\{\un e_a\}$: 
\begin{equation}\label{defSpinConn}
\nabla^E_{\rho(\un{\mX})} \un{e}_a = {\cal A}^b{}_a(\un{\mX}_H) \un{e}_b\,,
\end{equation}
where ${\cal A}^b{}_a$ are the connection coefficients of $\nabla^{E}$, which depends linearly on $\umX$. In this way, we can see that the representation $\Aconn{E}$ acts as
\begin{equation}
\label{defSpinConnHV}
	\Aconn{E}(\un\mX)(\un{e}_a) = \Big({\cal A}^b{}_a(\un{\mX}_H)  - (v_E(\omega(\un\mX_V)))^b{}_a\Big) \un{e}_b\,.
\end{equation}
For any $\un\psi\in\Gamma(E)$, we can derive in the basis $\{\un{e}_a\}$ that
\begin{align}
\nabla^E_{\rho(\umX)}\un\psi&=\phi_E(\umX_H)(\psi^a\un e_a)=\psi^a\phi_E(\umX_H)(\un e_a)+(\rho(\umX_H)(\psi^a))\un e_a\nn\\
\label{nablaEpsi}
&=\psi^a\nabla^E_{\rho(\umX)}\un e_a+(\rho(\umX_H)(\psi^a))\un e_a=\psi^a{\cal A}^b{}_a(\umX_H)\un e_b+(\rho(\umX_H)(\psi^a))\un e_a\,,
\end{align}
where we used \eqref{nablaEomegaE} in the first and third equalities and \eqref{phiEf} in the second equality. For the adjoint representation, the action of $\nabla^L_{\umX}\un t_A$ can be represented by:
\begin{align}
\nabla^L_{\umX_H}\un t_A={\cal A}^B{}_A(\umX_H)\un t_B\,.
\end{align}
where ${\cal A}^B{}_A$ are the connection coefficients of $\nabla^L$ in the adjoint representation. Then, for any $\un\mu=\mu^A\un t_A\in\Gamma(L)$,
\begin{align}
\label{nablaLmu}
\nabla^L_{\rho(\umX)}\un\mu=\mu^A{\cal A}^B{}_A(\umX_H)\un t_B+(\rho(\umX_H)(\mu^A))\un t_A\,.
\end{align}
The defining relation \eqref{adjrep} for $v_L$ in the adjoint representation can be written in terms of a basis $\{\un t_A\}$ as
\begin{align}
\label{adjrept}
(v_L(\un t_A))(\un t_B)=f_{AB}{}^C\un t_C\,.
\end{align}
\par
Given a basis $\{\un E _\uM\}$, we can compute the commutators of the basis vectors using the Lie bracket on the Lie algebroid $A$:
\begin{align}
\label{LieEE}
[\un E _\uM,\un E _\uN]_A\equiv C_{\uM\uN}{}^{\un P}\un E_{\un P}\,,
\end{align}
where the commutation coefficients $C_{\uM\uN}{}^{\un P}$ can be considered as encoding the algebraic data of $A$. If $\{\un E _\uM\}$ is a split basis, then \eqref{LieEE} can be decomposed into
\begin{align}
\label{LAEE}
[\un E_{\ualpha},\un E_{\ubeta}]_A&= C_{\ualpha\ubeta}{}^{\un\gamma}\un E_{\un\gamma}+C_{\ualpha\ubeta}{}^{\un A}\un E_{\un A}\,,\\
[\un E_{\ualpha},\un E_{\un A}]_A&= C_{\ualpha\un A}{}^{\un B}\un E_{\un B}\,,\\
\label{LAE}
[\un E_{\un A},\un E_{\un B}]_A&= C_{\un{AB}}{}^{\un C}\un E_{\un C}\,,
\end{align}
where we have used the fact that $[\umX_H,\umY_V]_A\in\Gamma(V)$ and $[\umX_V,\umY_V]_A\in\Gamma(V)$. These commutation coefficients can be found to be
\begin{align}
\label{Cabc}
C_{\ualpha\ubeta}{}^{\un\gamma}&=-\rho^{\mu}{}_{\ualpha}\rho^{\nu}{}_{\ubeta}(\p_{\mu}\sigma^{\un\gamma}{}_{\nu}-\p_{\nu}\sigma^{\un\gamma}{}_{\mu})\,,\\
\label{CabA}
C_{\ualpha\ubeta}{}^{\un A}&=\Omega^A{}_{\ualpha\ubeta}j^{\un A}{}_A\,,\\
\label{CaAB}
C_{\ualpha\un A}{}^{\un B}&={\cal A}_{\ualpha}{}^{B}{}_{A}j^{\un B}{}_B\omega^A{}_{\un A}-(\rho(\un E_{\ualpha})(j^{\un B}{}_A))\omega^A{}_{\un A}\,,\\
\label{CABC}
C_{\un A\un B}{}^{\un C}&=f_{AB}{}^C j^{\un C}{}_C\omega^A{}_{\un A}\omega^B{}_{\un B}\,.
\end{align}
The detailed evaluation of the commutation coefficients will be presented in Appendix \ref{app:commutation}. In a split basis, these coefficients also encode the information of the algebraic structures of the horizontal and vertical sub-bundles. As we can see, $C_{\un A\un B}{}^{\un C}$, which can be regarded as the structure constants of $V$, is directly related to the structure constants $f_{AB}{}^{C}$ of $L$ defined in \eqref{bracketL}. Besides, $C_{\ualpha\un A}{}^{\un B}$ is related to the connection coefficients of $\nabla_L$ in a manner similar to \eqref{nablaLmu}, $C_{\ualpha\ubeta}{}^{\un A}$ corresponds to the curvature of $A$, and $C_{\ualpha\ubeta}{}^{\un\gamma}$ contains the information of the ``exterior derivative'' of $\sigma$.

\chapter{Atiyah Lie Algebroids and the BRST Complex}
\label{chap:trivialized}
The canonical example of a transitive Lie algebroid to which we shall devote our attention in this thesis is the Atiyah Lie algebroid, which is defined through a principal bundle. Since a classical gauge theory already has a description in terms of principal bundles, many observations and intuitions from this framework can be naturally extended to the Atiyah Lie algebroid, which we argue to be a proper geometric formulation of quantum gauge theory. By utilizing the concept of Lie algebroid isomorphism, we can introduce the trivialized algebroid and demonstrate that this geometric framework indeed encompasses the BRST complex.

\section{Atiyah Lie Algebroids}
\label{sec:ALA}
\subsection{From Principal Bundles to Atiyah Lie Algebroids}
\label{sec:principalALA}
\bdefi
Suppose $P(M,G)$ is a principal $G$-bundle over the base manifold $M$ with the structure group $G$. The tangent bundle $TP$ of $P$ is locally described by $(p,\un v_p)$, where $p$ is a point in $P$ and $\un v_p\in T_pP$. The free right action $R_h$ of $h\in G$ on $P$ can also push forward the vector $\un v_p$ at $p$, and thus gives a free right action on $TP$, namely $(p,\un v_p)\mapsto(ph,R_{h*}(\un v_p))$. The vector bundle $TP/G$ over $M$ defined by identifying
\begin{align}
\label{TP/Gclass}
(p,\un v_p)\sim(ph,R_{h*}(\un v_p))\,,\qquad\forall h\in G\,,
\end{align}
is called an \emph{Atiyah Lie algebroid}.
\edefi
In a local trivialization $T_U$ of $P$, we have $p=(x,g)$, where $x=\pi(p)\in U\subset M$, $g\in G$. For convenience's sake, we will assume $T_U$ to be a global trivialization with $U=M$, but the discussion below does not rely on this assumption. Using the projection map $\pi:P\to M$, we can pullback a vector field $\un v$ on $P$ to $M$. Denote $\uX_{\pi(p)}\equiv\pi_*(\un v_p)\in T_{\pi(p)}M$ and $\un\gamma_p=\un v_p-\pi_*^{-1}(\uX_{\pi(p)})$, then $(p,\un v_p)\in TP$ can be expressed as $((x,g),(\uX_{\pi(p)},\un\gamma_p))$, or simply $(x,\uX_{x},\un\gamma_{(x,g)})$ since $\un\gamma_{(x,g)}$ carries the information of $g\in G$. Thus, the equivalence class \eqref{TP/Gclass} is formed by $(x,\uX_{x},\un\gamma_{(x,g)})$ with different $g\in G$, and a point in $TP/G$ corresponds to a representative in this equivalent class. For convenience, we choose $(x,\uX_{x},\un\gamma_{(x,e)})$, with $e$ the identity of $G$. Note that $\un\gamma_{(x,e)}$ can also be identified an element in the Lie algebra $\mathfrak{g}$ of $G$. Hence, a typical fiber of $TP/G$ can be regarded as the combination of $T_xM$ and $\mathfrak{g}$, and so the rank of this vector bundle is $\dim M+\dim G$.
\par
Now we will discuss the Lie algebroid structure of $TP/G$. First, while $TP$ is a bundle over $P$, $TP/G$ is importantly a vector bundle over $M$. Furthermore, $TP/G$ inherits a bracket algebra from $TP$ and possesses an anchor map in the form of the pushforward by the projection, i.e., $\pi_*: TP/G \rightarrow TM$. Moreover, the map $\pi_*$ can easily be seen to be surjective, and hence the algebroid $TP/G$ is automatically transitive. It is also obvious that the map $\pi_*:TP/G\to TM$ has a kernel $(x,0,\un\gamma_{(x,e)})$, and thus at each point $x\in M$ the kernel of $\pi_*$ is identical to the Lie algebra $\mathfrak{g}$. This forms the isotropy bundle $P\times_{\text{Ad}_G}\mathfrak{g}$ (also denoted by $P\times\mathfrak{g}/\sim$), called the \emph{adjoint bundle}, which is an associated bundle of $P$ whose typical fiber is $\mathfrak{g}$. The sections of the adjoint bundle are precisely the local gauge transformations that figured into the analysis of Section \ref{sec:BRST}. Also, there is a natural inclusion map $j:P\times_{\text{Ad}_G}\mathfrak{g}\to TP/G$ as $P\times_{\text{Ad}_G}\mathfrak{g}$ is the vertical sub-bundle of $TP/G$. Therefore, we have the following short exact sequence of vector bundles over $M$:
\begin{equation} \label{Atiyah Lie Algebroid}
\begin{tikzcd}
0
\arrow{r} 
& 
P\times_{\text{Ad}_G} \mathfrak{g}
\arrow{r}{j} 
&
TP/G
\arrow{r}{\pi_*} 
& 
TM
\arrow{r} 
&
0\,.
\end{tikzcd}
\end{equation}
We can see clearly from the above short exact sequence that a section of $TP/G$ can be identified (locally) with the direct sum of a local gauge transformation generated by $\un{\mu} \in \Gamma(L)$ and a diffeomorphism generated by $\un{X} \in \Gamma(TM)$. 

If a connection is defined on $P$, i.e.\ we have a horizontal sub-bundle $H_P$ of $P$, then $H\equiv TH_P/G$ give rise to a horizontal sub-bundle of $TP/G$, and thus we can define a map $\sigma: TM\to TP/G$ whose image is $H$ such that $\pi_*\circ\sigma$ is the identity on $TM$. Therefore, just like a connection on the principal bundle, a connection on an Atiyah Lie algebroid also represents a gauge field in physics, as will we discuss shortly in the next subsection. Having $\sigma$ defined, we can also introduce $\omega:TP/G\to P\times_{\text{Ad}_G} \mathfrak{g}$ whose kernel is $H$, which serves as the connection reform.
\par
For convenience, we will denote the Atiyah Lie algebroid $TP/G$ by $A$, the adjoint bundle $P\times_{\text{Ad}_G} \mathfrak{g}$ by $L$, and the anchor map $\pi_*:A\to TM$ by $\rho$. This will agree with our notation before.

\subsection{Local Trivializations of an Atiyah Lie Algebroid}
\label{sec:trivial}
In Section \ref{sec:principal} we have seen that the local trivialization of a principal bundle is a map $T_{U_i}:P|_{U_i}\to {U_i}\times G$, with $\{U_i\}$ an open cover of the base manifold $M$. The principal connection can be described as a local gauge field in each $U_i\in M$ satisfying the gauge transformation law in the intersection of two open subsets. Similarly, a local trivialization of a Atiyah Lie algebroid $A$ is a map $\tau_i:A^{U_i}\to TU_i\oplus L^{U_i}$, where $A^{U_i}$ and $L^{U_i}$ are the restriction of $A$ and $L$ to their sub-bundles over the local neighborhood $U_i \subset M$; in other words, $A^{U_i}$ and $L^{U_i}$ are vector bundles over $U_i$. Through $\tau_i$, the connection on the algebroid can then be expressed locally as a gauge field. In this subsection we review this notion and set up the stage for discussing the Lie algebroid formulation of BRST complex later in this chapter.

First we need to choose a basis of $\Gamma(A)$ for each coordinate patch $U_i\subset M$, and specify the transformation between two coordinate patches $U_i$ and $U_j$. For a split basis $\{\uE_{\ualpha},\uE_{\uA}\}$ (note that for Atiyah Lie algebroids $\text{rank}\,L=\dim G$), we have according to \eqref{Etrans} that
\begin{align}
\label{transE}
\uE^{U_i}_{\ualpha}=J_{ij}{}^{\ubeta}{}_{\ualpha}\uE^{U_j}_{\ubeta}\,,\qquad\uE^{U_i}_\uA=K_{ij}{}^{\uB}{}_\uA\uE^{U_j}_\uB\,,
\end{align}
where we used the subscript $U_i$ to denote the basis in the patch $U_i$. For a vector bundle $E$ associated to a representation $R$ of the structure group $G$, a basis $\un e^{U_i}_a$ of $\Gamma(E)$ in $U_i$ and the corresponding components of $\un\psi\in\Gamma(E)$ in this basis satisfy
\begin{align}
\label{transpsi}
\un e^{U_i}_a=R(g_{ij})^b{}_a\un e^{U_j}_b\,,\qquad\un\psi_i^a=R(g_{ij}^{-1})^a{}_b\un\psi_j^b\,,
\end{align}
where $g_{ij}$ assigns an element in $G$ pointwisely in $U_i\cap U_j$, which plays the role of the transition function between two local trivialization of the principal bundle $P$. Since $E$ is also the associated bundle of $P$, whose sections are matter fields, we can regard \eqref{transpsi} as the familiar gauge transformation of the matter fields.
\par
Before we discuss the connection on the algebroid directly, let us first look as the covariant derivative $\nabla^E$, namely the induced connection on the representation algebroid. When we split the action of $\hatd$ on $\un\psi\in\Gamma(E)$ into $\hatd\un\psi(\umX)=\nabla^E_{\rho(\umX)}\un\psi-v_E(\omega(\umX))(\un\psi)$, these two terms as a horizontal and a vertical vector field on $A$, respectively, should be invariant under basis transformations. That is,
\begin{align}
\label{nablavE}
(\nabla^E_{\rho(\umX)}\un\psi)_{U_i}=(\nabla^E_{\rho(\umX)}\un\psi)_{U_j}\,,\qquad (v_E\circ\omega)(\umX)(\un\psi))_{U_i}=(v_E\circ\omega)(\umX)(\un\psi))_{U_j}\,.
\end{align}
It follows from \eqref{nablaEpsi} that in two patches $U_i$ and $U_j$, the first equation in \eqref{nablavE} gives
\begin{align*}
\big(\psi_i^b{\cal A}_i^a{}_b(\umX_H)+(\rho(\umX_H)\psi_i^a)\big)\un e_a^{U_i}&=\big(\psi_j^b{\cal A}_j^a{}_b(\umX_H)+(\rho(\umX_H)\psi_j^a)\big)\un e_a^{U_j}\\
&=\big(R(g_{ij})^b{}_c\psi_i^c{\cal A}_j^a{}_b(\umX_H)+\rho(\umX_H)(R(g_{ij})^c{}_d\psi_j^d)\big)R(g_{ij}^{-1})^a{}_c\un e_a^{U_i}\,,
\end{align*}
where we used \eqref{transpsi} in the second equality and relabeled the dummy indices. Taking $\umX_H$ to be a basis vector $\un E_{\ualpha}^{U_i}$ in the above equation, we have
\begin{align*}
\psi_i^b{\cal A}_i^a{}_b(\un E_{\ualpha}^{U_i})+(\rho(\un E_{\ualpha}^{U_i})\psi_i^a)&=R(g_{ij}^{-1})^a{}_c\big(R(g_{ij})^b{}_d\psi_i^d{\cal A}_j^c{}_b(\un E_{\ualpha}^{U_i})+\rho(\un E_{\ualpha}^{U_i})(R(g_{ij})^c{}_d\psi_i^d)\big)\\
\psi_i^b{\cal A}_i^a{}_b(\un E_{\ualpha}^{U_i})+(\rho(\un E_{\ualpha}^{U_i})\psi_i^a)&=R(g_{ij}^{-1})^a{}_cR(g_{ij})^b{}_d\psi_i^d{\cal A}_j^c{}_b(\un E_{\ualpha}^{U_i})\\
&\quad+R(g_{ij}^{-1})^a{}_cR(g_{ij})^c{}_d(\rho(\un E_{\ualpha}^{U_i})\psi_j^d)+R(g_{ij}^{-1})^a{}_c(\rho(\un E_{\ualpha}^{U_i})R(g_{ij})^c{}_d)\psi_i^d\\
\psi_i^b{\cal A}_i^a{}_b(\un E_{\ualpha}^{U_i})&=R(g_{ij}^{-1})^a{}_cR(g_{ij})^d{}_b\psi_i^b{\cal A}_j^c{}_d(J_{ij}{}^{\ubeta}{}_{\ualpha}\uE^{U_j}_{\ubeta})+R(g_{ij}^{-1})^a{}_c(\rho(J_{ij}{}^{\ubeta}{}_{\ualpha}\uE^{U_j}_{\ubeta})R(g_{ij})^c{}_b)\psi_i^b\,,
\end{align*}
where \eqref{transE} is used in the last step. Hence, we obtain that
\begin{align}
\label{AEtransform}
{\cal A}_i^a{}_b(\un E_{\ualpha}^{U_i})&=J_{ij}{}^{\ubeta}{}_{\ualpha}\Big(R(g_{ij}^{-1})^a{}_c{\cal A}_j^c{}_d(\uE^{U_j}_{\ubeta})R(g_{ij})^d{}_b+R(g_{ij}^{-1})^a{}_c(\rho(\uE^{U_j}_{\ubeta})R(g_{ij})^c{}_b)\Big)\,.
\end{align}
This corresponds to the condition (C) in \eqref{condC} for describing a connection on the principal bundle as a local gauge field on the base manifold, which is exactly the familiar transformation for a gauge connection. However, note that unlike the $A_U$ in \eqref{condC}, here ${\cal A}_i^a{}_b$ are not the gauge field components pulled back from the algebroid connection directly but the connection coefficients of $\nabla^E$ on the representation algebroid. On the other hand, the second equation in \eqref{nablavE} gives
\begin{align}
\label{vEtransform}
(v_E\circ\omega)^a{}_b(\un E_{\un A}^{U_i})=K_{ij}{}^B{}_AR(g_{ij}^{-1})^a{}_c(v_E\circ\omega)^c{}_d(\un E_{\un B}^{U_i})R(g_{ij})^d{}_b\,.
\end{align}
This can be recognized as the transformation law for the Maurer-Cartan form on $L$, which is closely related to the notion of ghost as we will see later in this chapter. 
\par
Now we analyze how does a connection on the Atiyah Lie algebroid $A$ itself behave in a trivialization. As we have discussed, the vertical sub-bundle $V$ of $A$ is identical to $L$, while the horizontal sub-bundle $H$ has an ambiguity. On each coordinate patch $U_i$, we introduce the local trivialization as a morphism
\be
\tau_i:A^{U_i}\to TU_i\oplus L^{U_i}.
\ee
For any $\umX\in A$, we can write its image in this trivialization as $\tau_i(\umX)=(\uX_i,\un\mu_i)$, where in the two slots we have $\uX_i\in TU_i$ and $\un\mu_i\in L^{U_i}$. It is natural to require that $\uX_i=\rho(\umX)|_{U_i}$. This is the analogue that for $p\in P$ in a principal bundle we have $T_U(p)=(x,g)$ where $x=\pi(p)$. Then, in this trivialization the split basis vectors are mapped to
\begin{align}
\label{trivialbasis}
\tau_i(\uE^{U_i}_{\ualpha})=\tau_{i}{}^\mu{}_{\ualpha}(\un\p^{U_i}_\mu+b_i^A{}_\mu \un t^{U_i}_A)\equiv\tau_{i}{}^\mu{}_{\ualpha}\un D^{U_i}\,,\qquad\tau_i(\uE^{U_i}_\uA)=\tau_i{}^A{}_\uA \un t^{U_i}_A\,.
\end{align}
where $b_i^A{}_\mu$ is introduced to play the role of the Ehresmann connection, as they represent the ambiguity in the components in $L$ when lifting from $TM$ to $H$, and thus $b_\mu=b_i^A{}_\mu\un t^{U_i}_A$ as an $L$-valued 1-form on $M$ can be viewed as the local gauge field on $M$. In the index notation, the map $\tau_i$ can be decomposed into $\tau_i{}^\mu{}_{\ualpha}=\rho^\mu{}_{\ualpha}$, $\tau_i{}^A{}_{\ualpha}=\rho^\mu{}_{\ualpha}b^A{}_\mu$ and $\tau_i{}^A{}_{\un A}$, while $\tau_i{}^\mu{}_{\un A}=\tau_i{}^A{}_{\ualpha}=0$. One should note that $\tau_i{}^\mu{}_{\ualpha}$ and $\rho^\mu{}_{\ualpha}$ being equal does not mean they are the same map, since $\rho(\un E_{\alpha})=\rho^\mu{}_{\ualpha}\un\p_\mu$ has no component in $L$. We can also define $\tau_i^*:A^*_{U_i}\to T^*U_i\oplus L^*_{U_i}$, the dual map of $\tau_i$,  which preserves the orthogonality condition \eqref{EE}. Then we can write down the dual basis $\{E_{U_i}^{\ualpha},E_{U_i}^\uA\}$ in this trivialization as
\begin{align}
\tau^*_i(E_{U_i}^{\ualpha})=(\tau_{i}^{-1})^{\ualpha}{}_\mu \td x^\mu_i\,,\qquad\tau^*_i(E_{U_i}^\uA)=(\tau_{i}^{-1})^\uA{}_A(t^A_{U_i}-b_i^A{}_\mu \td x^\mu_i)\,,
\end{align}
where $\{\td x_i\}$ and $\{t^A_{U_i}\}$ are the bases of $\Gamma(TU_i)$ and $\Gamma(L^*_{U_i})$ introduced in \eqref{basisTML}. 
\begin{figure}[!htbp]
\center
\includegraphics[width=1.7in]{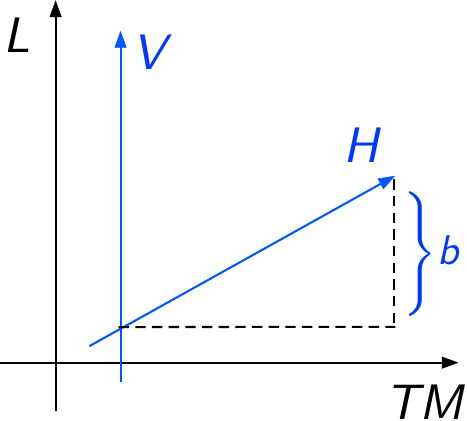}
\caption{A connection on $A$ gives a global split $A=H\oplus V$, which locally can be viewed as determined by a gauge field $b$ defined with respect to ``axes'' corresponding to sub-bundles $TM$ and $L$ \cite{Jia:2023tki}.}
\label{fig:VHcartoon.pdf}
\end{figure}
\par
We will now work in a specific coordinate patch $U_i$ and drop the labels for the patch for brevity. $\tau$ being a morphism means that it satisfies
\begin{align}
\label{taumor}
[\tau(\umX),\tau(\umY)]_{TM\oplus L}=\tau([\umX,\umY]_A)\,.
\end{align}
Evaluating the above condition in different cases gives information on the behavior of the local gauge field and its curvature in a trivialization as we will now demonstrate. For more details of the computations involved in the rest of this subsection, see Appendix \ref{app:trivial}.
\par
First, in the case where $\umX,\umY$ are both vertical, \eqref{taumor} gives
\begin{align}
\label{ttftf}
\tau^A{}_\uA j^\uA{}_D\tau^B{}_\uB j^\uB{}_E f_{AB}{}^C&=\tau^C{}_{\un C}j^{\un C}{}_Ff_{DE}{}^F\,.
\end{align}
Considering that $(\tau\circ j)^A{}_B\equiv\tau^A{}_\uA j^\uA{}_B$ is a local endomorphism on $L$, a convenient choice of $\tau$ is to set $\tau\circ j=Id_L$. In this case, for $\umX_V\in V$ we have $\tau(\umX_V)=(0,-\omega(\umX_V))$, or $\tau^A{}_\uA=-\omega^A{}_\uA$. However, one should note that \eqref{ttftf} does not require that $\tau\circ j=Id_L$ and in general $\tau$ is not related to $\omega$.
\par
Next, we consider $\umX=\umX_H$ to be horizontal and $\umY=\umY_V$ to be vertical. Then \eqref{taumor} together with the fact that $\tau^\mu{}_{\ualpha}=\rho^\mu{}_{\ualpha}$ gives
\begin{align}
\label{bArelation}
{\cal A}_{\ualpha}{}^D{}_C=((\tau\circ j)^{-1})^E{}_C(\rho^\mu{}_{\ualpha}b^A{}_\mu f_{AB}{}^C+\delta^C{}_B\rho^\mu{}_{\ualpha}\p_\mu)(\tau\circ j)^B{}_D\,.
\end{align}
This relates $b_\mu$ with the connection coefficients ${\cal A}_{\ualpha}{}^D{}_C$ of $\nabla^L$. If we make a special choice such that $\tau\circ j=Id_L$, the above equation becomes
\begin{align}
\label{bArelation1}
{\cal A}_{\ualpha}{}^D{}_C=\rho^\mu{}_{\ualpha} b^A{}_\mu f_{AC}{}^D&\,.
\end{align}
which gives a linear correspondence between $b_\mu$ and ${\cal A}_{\ualpha}{}^D{}_C$. Note that unlike the structure group $G$ of a principal bundle, $L$ is a bundle over $M$ and $\tau\circ j$ is defined for each fiber of $L$ pointwisely over $M$. Hence, a general choice of $\tau$ will generate the second term in \eqref{bArelation}, bringing an ambiguity in the relation between ${\cal A}_{\ualpha}{}^D{}_C$ and $b^A{}_\mu$. Nevertheless, if we denote the ${\cal A}_{\ualpha}{}^D{}_C$ in \eqref{bArelation1} as $\tilde{\cal A}_{\ualpha}{}^D{}_C$, then \eqref{bArelation} can be written as
\begin{align}
{\cal A}_{\ualpha}{}^D{}_C=((\tau\circ j)^{-1})^E{}_C(\tilde{\cal A}_{\ualpha}{}^D{}_C+\delta^C{}_B\rho^\mu{}_{\ualpha}\p_\mu)(\tau\circ j)^B{}_D\,,
\end{align}
which is nothing but a gauge transformation of $\tilde{\cal A}_{\ualpha}{}^D{}_C$. This indicates that for a general choice of $\tau$, the deviation of $\tau\circ j$ from the identity map can be viewed as a gauge ambiguity. 
\par
To carry over the above result from $L$ to a general vector bundle $E$, we recall that for the adjoint representation we have $v_L(\un t_A)^C{}_B=f_{AB}{}^C$, and so \eqref{bArelation} can also be expressed as
\begin{align}
\label{bArelationE2}
{\cal A}_{\ualpha}{}^D{}_C=((\tau\circ j)^{-1})^E{}_C(\rho^\mu{}_{\ualpha}b^A{}_\mu v_L(\un t_A)^C{}_B+\delta^C{}_B\rho^\mu{}_{\ualpha}\p_\mu)(\tau\circ j)^B{}_D\,.
\end{align}
And for any vector bundle $E$ we should have the coefficients of $\nabla^E$ as follows:
\begin{align}
\label{bArelationE3}
{\cal A}_{\ualpha}{}^d{}_c=(\lambda_\tau^{-1})^d{}_a(\rho^\mu{}_{\ualpha}b^A{}_\mu v_E(\un t_A)^a{}_b+\delta^a{}_b\rho^\mu{}_{\ualpha}\p_\mu)\lambda_\tau^b{}_d\,.
\end{align}
where now $v_L$ is replaced by $v_E$ and $(\tau\circ j)\in\End(L)$ is replaced an endomorphism $\lambda_\tau\in\End(E)$. Hence, $b$ introduced in a trivialization can be identified with the connection $\nabla^E$ through $\rho\circ{\cal A}=v_E(b)$ up to gauge transformation. Since we have shown that for any vector bundle $E$, the connection coefficients of $\nabla^E$ satisfies the transformation law \eqref{AEtransform}, taking $E$ to be $L$ we can see that $b_\mu$ indeed transforms as local gauge field.
\par
Finally, when $\umX=\umX_H$ and $\umY=\umY_H$ are both horizontal, \eqref{taumor} gives
\begin{align}
\label{FbaseA}
F^A{}_{\mu\nu}\rho^\mu{}_{\ualpha}\rho^\nu{}_{\ubeta}=\Omega^A{}_{\ualpha\ubeta}\,.
\end{align}
where
\begin{align}
F^A{}_{\mu\nu}\equiv\p_\mu b^A{}_\nu-\p_\nu b^A{}_\mu+ b^B{}_\mu b^C{}_\nu f_{BC}{}^A
\end{align}
is the curvature of $b^A{}_\mu$. This indicates that $F_{\mu\nu}\equiv F^A{}_{\mu\nu}\un t_A$ as an $L$-valued 2-form on $M$ also represents the curvature of the Lie algebroid. Physically, $F_{\mu\nu}$ represents the familiar gauge field strength, and \eqref{FbaseA} shows that it can be pulled back from the curvature reform on the algebroid, similar to \eqref{Fbase} for the principal bundle case.

\section{Lie Algebroid Isomorphisms}
\label{sec:morphisms}
In the previous section, we introduced the Atiyah Lie algebroid derived from a principal bundle and discussed its trivialization as an analogy to the trivialization of principal bundles. To further our understanding of Lie algebroid trivialization and to establish a connection with the BRST complex, this section introduces the concept of Lie algebroid isomorphisms for general Lie algebroids. This concept allows us to formulate many results from the previous discussion in a more formal manner.
\par
A Lie algebroid morphism is a map $\varphi: A_1 \rightarrow A_2$ between two Lie algebroids, which preserves the geometric structure of the Lie algebroids as encoded in their brackets. That is, for all $\un\mX,\un\mY \in \Gamma(A_1)$,
\begin{equation} \label{Morphism Condition}
    R^\varphi(\un\mX,\un\mY):=-\varphi([\un\mX , \un\mY]_{A_1}) + [\varphi(\un\mX), \varphi(\un\mY)]_{A_2}=0\,.
\end{equation}
In this section we focus on a subclass of Lie algebroid morphisms which are, in fact, isomorphisms of the underlying vector bundles. Consider a set of Lie algebroids that share the same base manifold and structure group. In general, two such algebroids may be topologically distinct. Our goal is to emphasize that two algebroids in this set, $A_1$ and $A_2$, will be topologically equivalent if there exists an isomorphism between them. To accomplish this goal, we seek to understand the conditions under which the set of structure maps of two Lie algebroids define a commutative diagram of the following form:
\begin{equation}
\label{Transitive Lie Algebroid Morphism}
\begin{tikzcd}
& 
& 
A_1
\arrow[left]{dd}{\varphi}
\arrow[left]{dl}{\omega_1}
\arrow[bend left]{dr}{\rho_1}
& 
&
\\
0
\arrow{r}
&
L
\arrow[bend left]{l} 
\arrow[bend left]{ur}{j_1}
\arrow[bend right, swap]{dr}{ j_2}
&
&
TM
\arrow{r}
\arrow[shift left]{ul}{\sigma_1} 
\arrow[swap]{dl}{\sigma_2} 
&
0\,.
\arrow[bend left]{l} 
\\
& 
&
A_2
\arrow[left, swap]{ul}{\omega_2}
\arrow[bend right,swap]{ur}{\rho_2}
\arrow[shift left]{uu}{\overline{\varphi}}
&
&
\end{tikzcd}
\end{equation}
Notice that with the splitting $A_1=H_1\oplus V_1$ and $A_2=H_2\oplus V_2$, $J\equiv\sigma_2 \circ \rho_1$ is a map from $H_1$ to $H_2$, while $K \equiv j_2 \circ \omega_1$ is a map from $V_1$ to $V_2$. Clearly, we can write $\varphi=J-K$. Our motivation for considering \eqref{Transitive Lie Algebroid Morphism} is that it respects the horizontal and vertical splittings of the two algebroids, and will subsequently provide a useful physical picture for general Lie algebroid isomorphisms.\footnote{Here, we are discussing isomorphisms using an {\it active} language; in the corresponding {\it passive} description, an isomorphism would be understood as a change of basis for the same algebroid.}

By commutativity, the maps $\varphi$ and $\overline{\varphi}$ in \eqref{Transitive Lie Algebroid Morphism} apparently define isomorphisms of the vector bundles $A_1$ and $A_2$. However, it is not immediately clear that these maps respect the algebras defined by the brackets on these bundles. To this end, we will now demonstrate that the map $\varphi$ will be a Lie algebroid morphism if and only if the horizontal distributions of $A_1$ and $A_2$ as defined by their respective connections $\omega_1$ and $\omega_2$ share the same curvature. Recall that the curvature of a connection reform $\omega$ is the horizontal $L$-valued form given by
\beq
\Omega = \hatd\omega + \frac{1}{2}[\omega,\omega]_L\,.
\eeq
(Note that the bracket in the above equation is the graded Lie bracket between $L$-valued forms defined in the first footnote in Subsection \ref{sec:principal}.) Suppose the curvatures of $\omega_1$ and $\omega_2$ are $\Omega_1$ and $\Omega_2$, respectively. We can compute that
\begin{align}
R^{\varphi}(\un\mX_H,\un\mY_H)&= R^{\sigma_2}(\rho_1(\un\mX_H),\rho_1(\un\mY_H))+j_2(R^{-\omega_1}(\un\mX_H,\un\mY_H)\nn\\
&= j_2(\Omega_2(\varphi(\un\mX),\varphi(\un\mY)))-j_2(\Omega_1(\un\mX,\un\mY))\,,
\end{align}
where we used $\varphi=J-K$ and \eqref{curvature}
\beqn\label{Russianformula}
R^\sigma(\rho(\un\mX),\rho(\un\mY))=j(\Omega(\un\mX,\un\mY))=-j(R^{-\omega}(\un\mX_H,\un\mY_H))\,.
\eeqn 
In this way, we see that $\varphi$ will be a morphism of the brackets if and only if
\beq\label{Curvature pullback}
\Omega_1(\un\mX,\un\mY) =\Omega_2(\varphi(\un\mX),\varphi(\un\mY))\,.
\eeq

Provided $\varphi$ is an isomorphism, it will induce a linear transformation on bundles associated to $A_1$ and $A_2$ to preserve Lie algebroid representations. Let $E_1$ and $E_2$ be isomorphic vector bundles over $M$ which are associated, respectively, to $A_1$ and $A_2$ by Lie algebroid representations $\Aconn{E_j}: A_j\to\text{Der}(E_j)$, with $j=1,2$. 
Then, accompanying the Lie algebroid isomorphism $\varphi$, there is a corresponding map on the associated bundles, which can be written as
\begin{equation}
    g_\varphi: E_1 \rightarrow E_2\,.
\end{equation}
By construction, we enforce that this map is compatible with the Lie algebroid representations of $A_1$ and $A_2$ in the sense that
\beq \label{Iso-Rep compatibility}
	\Aconn{E_2} \circ \varphi(\un\mX)(g_{\varphi}(\un{\psi})) = g_{\varphi}(\Aconn{E_1}(\un\mX)(\un{\psi}))\,, \qquad \forall \un\mX \in \Gamma(A_1)\,,\quad \un{\psi} \in \Gamma(E_1)\,.
\eeq
Let $\varphi^*: \Omega(A_2;E_2) \rightarrow \Omega(A_1;E_1)$ denote the Lie algebroid pullback map induced by $\varphi$. Explicitly, given $\eta \in \Omega^r(A_2;E_2)$ and $\un\mX_1, \ldots, \un\mX_r \in \Gamma(A_1)$ we have
\begin{equation}\label{Lie algebroid pullback}
	(\varphi^*\eta)(\un\mX_1, \ldots, \un\mX_r) =g_\varphi^{-1}\big(\eta(\varphi(\un\mX_1), \dots, \varphi(\un\mX_r))\big)\,. 
\end{equation}
Using this notation along with the morphism property \eqref{Morphism Condition} and compatibility condition \eqref{Iso-Rep compatibility}, we can establish that 
\beq\label{Chain Map condition}
\hatd_1\circ \varphi^*=\varphi^*\circ\hatd_2\,,
\eeq
which means that $\varphi$ is a \emph{Lie algebroid chain map} in the exterior algebra sense. To prove this, it is sufficient to show that this condition holds for 0-forms and 1-forms, since $\hatd$ acts as a derivation with respect to the wedge product and the full exterior algebra is generated by the set of 1-forms along with the wedge product. First we look at the 0-form case. Let $\psi \in \Omega^0(A_2;E_2)$, and $\un\mX\in\Gamma(A_1)$. Then,
\begin{align}
&(\varphi^* \hatd_2 \psi)(\un\mX) = g_{\varphi}^{-1}\big(\hatd_2\psi \circ \varphi(\un\mX)\big) = g_{\varphi}^{-1}\big(\Aconn{E_2} \circ \varphi(\un\mX)(\psi)\big) \nn \\
={}& g_{\varphi}^{-1}\big(\Aconn{E_2} \circ \varphi(\un\mX)\big(g_{\varphi}g_{\varphi}^{-1}(\psi)\big)\big) = \Aconn{E_1}(\un\mX)\big(g_{\varphi}^{-1}(\psi)\big) = (\hatd_1 \varphi^*\psi)(\un\mX)\,,
\end{align}
where in the first equality we used \eqref{Lie algebroid pullback}, in the second equality we used the definition of the Lie algebroid differential via the Koszul formula \eqref{dhat}, and in the fourth equality we used \eqref{Iso-Rep compatibility}. Now we move on to the 1-form case. Let $\eta\in\Omega^1(A_2;E_2)$, and take $\un\mX,\un\mY\in\Gamma(A_1)$. We can write
\begin{align}
&(\varphi^*\hat\td_2\eta)(\un\mX,\un\mY)=g_{\varphi}^{-1}[(\hat\td_2\eta)(\varphi(\un\mX),\varphi(\un\mY))]\nn\\
={}& g_{\varphi}^{-1}\Big[\Aconn{E_2}\circ \varphi(\un\mX) (\eta \circ \varphi(\un\mY)) - \Aconn{E_2} \circ \varphi(\un\mY) (\eta \circ \varphi(\un\mX)) - \eta([\varphi(\un\mX),\varphi(\un\mY)]_{A_2})  \Big] \nn \\
={}& \Aconn{E_1}(\un\mX)\big(\varphi^*\eta(\un\mY)\big) - \Aconn{E_1}(\un\mY)\big(\varphi^*\eta(\un\mX)\big) - \varphi^*\eta\big([\un\mX,\un\mY]_{A_1}\big) \nn \\
={}&(\hatd_1\varphi^*\un\eta)(\un\mX,\un\mY)\,,
\end{align}
where again in the first equality we used \eqref{Lie algebroid pullback}, in the second equality we used \eqref{dhat}, and in third equality we applied \eqref{Iso-Rep compatibility} and \eqref{Lie algebroid pullback}. Therefore, a Lie algebroid isomorphism $\varphi: A_1 \rightarrow A_2$ satisfying \eqref{Iso-Rep compatibility} indeed induces a chain map on the exterior algebras of $A_1$ and $A_2$ satisfying \eqref{Chain Map condition}.

Using \eqref{Lie algebroid pullback} we can rewrite \eqref{Curvature pullback} as
\beq \label{Curvature pullback 2}
	\Omega_1 = \varphi^*\Omega_2\,.
\eeq
Eq.~\eqref{Curvature pullback 2} indicates that a Lie algebroid isomorphism of the form \eqref{Transitive Lie Algebroid Morphism} involves a topological consideration about the algebroids in question. In fact, the Chern-Weil homomorphism introduced in Section \ref{sec:Chern} is applicable to Lie algebroid cohomology (see Section \ref{sec:Chernalg}). This will provide a recipe for constructing Atiyah Lie algebroid cohomology classes in terms of characteristic polynomials in curvature. Recall that a characteristic class satisfies the naturality condition \eqref{natural}, which essentially implies that the pullback commutes through the characteristic class; that is, if $\lambda(\Omega)$ is a characteristic class of a curvature $\Omega$, then 
\beq
\lambda(\varphi^*\Omega) = \varphi^*\lambda(\Omega)\,.
\eeq 
Hence, two Lie algebroids whose curvatures are related as \eqref{Curvature pullback} will possess an isomorphism between their cohomologies. Eq.~\eqref{Chain Map condition} similarly implies that isomorphic Lie algebroids possess isomorphic cohomology classes. In light of these observations, we can view the Lie algebroid isomorphism as a device for organizing the set of Atiyah Lie algebroids with connection into topological equivalence classes. Let $(A,\omega)$ denote an Atiyah Lie algebroid $A$ with connection reform $\omega$. Then,
\begin{equation}
[(A,\omega)] := \{(A',\omega') \; | \; \exists \varphi: A \rightarrow A' \text{ s.t. } \Omega = \varphi^*\Omega' \}
\end{equation}
can be regarded as the set of topologically equivalent Atiyah Lie algebroids with connection. 

From a physical perspective Eqs.~\eqref{Curvature pullback} and \eqref{Chain Map condition} establish the fact that the commutative diagram \eqref{Transitive Lie Algebroid Morphism} encodes diffeomorphisms and gauge transformations relating isomorphic Lie algebroids. In particular, it is straightforward to find that the connection coefficients of the horizontal and vertical parts in \eqref{defSpinConnHV} satisfy
\begin{align}
\label{Gauge Transformation}
({\cal A}_1)_{\un\alpha_1}{}^{a_1}{}_{b_1}&=J^{\un\alpha_2}{}_{\un\alpha_1}(g_\varphi^{-1})^{a_1}{}_{a_2}\Big(({\cal A}_2)_{\un\alpha_2}{}^{a_2}{}_{b_2}+\delta^{a_2}{}_{b_2}\rho(\un E_{\un\alpha_2})\Big)g_\varphi^{b_2}{}_{b_1}\,,\\
(v_E(\omega_1))_{\un A_1}{}^{a_1}{}_{b_1}&=K^{\un B_2}{}_{\un A_1}(g_\varphi^{-1})^{a_1}{}_{a_2}(v_E(\omega_2))_{\un B_2}{}^{a_2}{}_{b_2}g_\varphi^{b_2}{}_{b_1}\,.
\end{align}
Immediately, one can observe that the above two equations are reminiscent of the transformations \eqref{AEtransform} and \eqref{vEtransform}. In fact, the latter are indeed a special case of the former, where we consider an isomorphism from $A$ to itself restricted in the overlap of $U_i$ and $U_j$. Therefore, in this formal formulation we can see that the components of ${\cal A}$ and $\omega$ transform like a gauge field and a gauge ghost, respectively. In this respect, we can also identify the Lie algebroid isomorphism \eqref{Transitive Lie Algebroid Morphism} as encoding the data of a gauge transformation. In other words, the equivalence class $[(A,\omega)]$ can be regarded as an orbit of gauge equivalent algebroids. This remark is applied in \cite{klinger2023abc} for constructing the configuration algebroid, which can be regarded as a concise definition of the space of gauge orbits of connections that can be employed in any gauge theory formulated in terms of Atiyah Lie algebroids. Furthermore, as will be discussed in detailed shortly, the trivialization map $\tau$ can be treated as a special kind of Lie algebroid isomorphism from $A$ to the trivialized algebroid, and the results in \eqref{bArelationE2} and \eqref{bArelationE3} are nothing but manifestations of \eqref{Gauge Transformation} for this special isomorphism. 
\par
So far we have shown that there exists a Lie algebroid isomorphism of the form \eqref{Transitive Lie Algebroid Morphism} between Lie algebroids with connection whose horizontal distributions have curvatures related by \eqref{Curvature pullback}. It is worth mentioning that this very same construction was used in constructing a representation of a Lie algebroid $A$ by the Lie algebroid $\Der(E)$, for some associated vector bundle $E$. In fact, this is a slight generalization of what we presented above, in that whereas the isomorphism in question is $\phi_E:A\to \Der(E)$, these two algebroids do not share the same isotropy bundle, but instead there is a further isomorphism $v_E:L\to \End(E)$ between them. Locally this isomorphism can be thought to give a matrix representation (on the fibers of $E$) of the Lie algebra. 

\section{BRST Complex from the Lie Algebroid Trivialization}
\label{sec:BRSTalg}
Given that $\hatd$ is nilpotent on $\Omega(A,E)$, it provides a well-defined notion of cohomology, which we refer to as \emph{Lie algebroid cohomology}. In this section, our intention is to explain how this cohomology is related to the usual notion of BRST cohomology. In the previous section, we showed that two Lie algebroids with connection that are related by an isomorphism are different representatives of an equivalent class, and the cohomology of the respective $\hatd$ agree. In this sense, the $\hatd$ cohomology is invariant under isomorphism. As we have alluded to, the local trivialization can be formalized as a Lie algebroid isomorphism. We will show below that it is in this description that the usual physics notation $\hatd_{\tau} \to \td + s$ is produced, which relates the Lie algebroid cohomology to the usual physics notions of BRST cohomology.

\subsection{Covariant and Consistent Splittings}
\label{sec:local}
Having established the concept of Lie algebroid isomorphisms, now we get back to the discussion of the trivialization of a Lie algebroid. As we mentioned above, a local trivialization of a Lie algebroid can also be thought of as an example of a Lie algebroid isomorphism, with the details presented in terms of the local data in each local subset. Given an open cover $\{U_i\}$ of $M$, we have introduced the $\tau_i:A^{U_i}\to TU_i\oplus L^{U_i}$, and \eqref{trivialbasis} allows us to expresses local sections of $A$ in terms of local bases for $TM$ and $L$:
\beq
\tau_i(\un\mX_H)=\mX_{i,H}^{\ualpha}\tau_i{}^\mu{}_{\ualpha}(\un\p^{U_i}_\mu+b_{i}{}_\mu^A\un t^{U_i}_A)\,,\qquad
\tau_i(\un\mX_V)=\mX_{i,V}^{\un A}\tau_i{}^A{}_{\un A}\un t^{U_i}_A\,.
\label{deftau}
\eeq
For an Atiyah Lie algebroid $A$, we have demonstrated in Subsection \ref{sec:trivial} that the coefficients $b_{i}{}_\mu^A$ are the components of the local gauge field on $M$, which transforms on overlapping open sets as a gauge field by consequence of \eqref{Gauge Transformation}.

Since for each $U_i$ in the open cover of $M$ we realize a Lie algebroid isomorphism $\tau_i: A^{U_i} \rightarrow TU_i \oplus L^{U_i}$,\footnote{Note that here we are using the notion of isomorphism in the active sense, and hence we distinguish $A^{U_i}$ from $TU_i \oplus L^{U_i}$. In what follows, the reader may find it profitable to think from a passive perspective: indeed our use of $A^{U_i}$ versus $TU_i \oplus L^{U_i}$ can be thought of as simply corresponding to a different choice of basis, the first natural from the $H\oplus V$ split, the second natural from the local $TU\oplus L$ split.} we can sew together the aforementioned local charts to obtain a Lie algebroid atlas. Sewing the charts $\tau_i$ together requires that we also specify transition functions $t_{ij}: A^{U_i} \rightarrow A^{U_j}$, which are Lie algebroid isomorphisms with support in the intersection $U_i \cap U_j$ for each pair of $U_i$ and $U_j$. This corresponds to imposing the condition (C) in \eqref{condC}, i.e., overlapping charts in a principal bundle must agree up to a gauge transformation. The presence of non-trivial transition functions in the algebroid context ensures that topological data is preserved under trivialization. Together, the collection $\{U_i, \tau_i, t_{ij}\}$ carries the intuition of the Lie algebroid trivialization into a global context. In the following we will use the abbreviated notation $\tau: A \rightarrow A_{\tau}$ to refer to the local Lie algebroid isomorphism mapping $A$ into the trivialized Lie algebroid $A_{\tau} \simeq TU \oplus L^{U}$ for some $U \subset M$. That is, the notation $A_{\tau}$ serves to remind that $A_{\tau}$ involves restricting $A$ to an open set. We leave the open subset $U$ unspecified with the understanding that the Lie algebroid atlas allows for the algebroid $A$ to be trivialized when restricted to any open neighborhood of the base.

To be precise about details, we will work in explicit bases for the various vector bundles; although we will not indicate so, these should be understood to be valid locally on some open set of $M$. Given the bases for the bundles $TM$ and $L$ introduced in \eqref{basisTML}, we have a choice to make for a basis of sections of the trivialized Lie algebroid $A_\tau$ and we will refer to such choices as  ``splittings". Our analysis will focus on two natural choices of splittings which we refer to as the \emph{consistent splitting} and the \emph{covariant splitting}, respectively. The relevance of this nomenclature will become clear shortly. These two splittings correspond in fact to the two sets of axes shown in Figure~\ref{fig:VHcartoon.pdf}, and they are distinguished precisely because of the non-trivial connection on $(A_\tau,\omega_\tau)$.

By a covariant splitting, we mean to assign a basis on $A_\tau$ by means of a split basis on $A$. Consider an algebroid $(A,\omega)$ for which we take a split basis $\{\un E_{\un\alpha},\un E_{\un A}\}$ with $\un\alpha=1,\ldots,\dim M$, $\un A=1,\ldots,\dim G$. Recall that such a basis has the virtue that $\omega(\un E_{\un\alpha})=\un 0$ and $\rho(\un E_{\un A})=\un 0$, namely they span $\Gamma(H)$ and $\Gamma(V)$, respectively. Given the map $\tau$, it is natural to choose a basis $\{\tau(\un E_{\un\alpha}),\tau(\un E_{\un A})\}$ for $A_\tau$. Since we will now deal directly with $A_\tau$, we will for brevity denote such a basis by $\{\un {\hat E}_{\un\alpha},\un {\hat E}_{\un A}\}$. Thus  a covariant splitting corresponds to a choice of basis sections that are aligned with the global split $A_\tau=H_\tau\oplus V_\tau$. Locally, these sections can be expressed in terms of the bases for $TM$ and $L$ as
\beqn
\label{taubasis}
\un {\hat E}_{\un\alpha}=\rho_\tau^\mu{}_{\un\alpha}(\un\p_\mu+b_\mu^A\un t_A)\,,\qquad 
\un {\hat E}_{\un A}=-\omega_\tau^A{}_{\un{A}}\un t_A\,,
\eeqn
while the dual bases can be written as
\beq\label{taudualbasis}
\hat E^{\un\alpha}=\sigma_\tau^{\un{\alpha}}{}_{\mu} \td x^\mu\,,\qquad
{\hat E}^{\un A}=j_\tau^{\un{A}}{}_{A}(t^A-b^A_\mu \td x^\mu)\,.
\eeq
The coefficients in \eqref{taubasis} and \eqref{taudualbasis} are determined by the choice of $\tau$. The reason that we denote these coefficients in this way is that we can use them to constitute the maps for the trivialized algebroid and get the following diagram:
\begin{equation}
\begin{tikzcd}
& 
& 
A
\arrow[left]{dd}{\tau}
\arrow[left]{dl}{\omega}
\arrow[bend left]{dr}{\rho}
& 
&
\\
0
\arrow{r}
&
L
\arrow[bend left]{l} 
\arrow[bend left]{ur}{j}
\arrow[bend right, swap]{dr}{ j_\tau}
&
&
TM
\arrow{r}
\arrow[shift left]{ul}{\sigma} 
\arrow[swap]{dl}{\sigma_\tau} 
&
0\,.
\arrow[bend left]{l} 
\\
& 
&
A_\tau
\arrow[left, swap]{ul}{\omega_\tau}
\arrow[bend right,swap]{ur}{\rho_\tau}
\arrow[shift left]{uu}{\overline{\tau}}
&
&
\end{tikzcd}
\end{equation}
In this way, $\tau$ gives rise to a well-defined Lie algebroid $A_\tau$ with maps $\rho_\tau$, $\sigma_\tau$, $j_\tau$, $\omega_\tau$. Notice that when we introduce the trivialization map $\tau$ in Section \ref{sec:trivial}, we emphasized that $\tau\circ j:L\to L$ need not to be the identity map, and so $\tau^A{}_\uA=-\omega^A{}_\uA$ is not required. Working in the trivialized algebroid, we now have $\omega_\tau\circ j_\tau=\omega_\tau\circ\tau\circ j=Id$, and so the nontrivality of $\tau\circ j$ is exactly characterized by $\omega_\tau$. This brings us the convenience that \eqref{bArelationE3} on $A_\tau$ can be simply a linear relation:
\begin{align}
\label{bArelationE4}
{\cal A}_{\ualpha}{}^d{}_c=\rho_\tau^\mu{}_{\ualpha}b^A{}_\mu v_E(\un t_A)^a{}_b\,,
\end{align}
since the gauge ambiguity involved in $\tau$ is now put aside. 

Now we are ready to demonstrate the consistent and covariant splittings for $A_\tau$ explicitly. Suppose $\un X=X^\mu\un\p_\mu\in \Gamma(TM)$ and $\un\mu=\mu^A\un t_A\in\Gamma(L)$, then a section $\un\mX$ of $A_\tau$ with $\un X=\rho_\tau(\umX)$ and $\un\mu=\omega_\tau(\umX)$ can be expressed in the covariant splitting as
\beq\label{covariantsection}
\umX=\umX^{\ualpha}\un{\hat E}_{\ualpha}+\umX^{\un A}\un{\hat E}_{\un A}=\umX^{\ualpha}(\rho_\tau^\mu{}_{\ualpha}\un\p_\mu+\rho_\tau^\mu{}_{\ualpha}b^A_\mu\un t_A)+\umX^{\un A}\omega_\tau^A{}_{\un A}\un t_A=X^\mu(\un\p_\mu+b_\mu^A\un t_A)+\mu^A\un t_A\,.
\eeq
On the other hand, by a consistent splitting, we mean a choice of basis for $A_\tau$ that is aligned with the bases for $TM$ and $L$. That is, in the consistent splitting, we can write a section of $A_\tau$ as
\beq
\un\mX=\mX^\mu\un\p_\mu+\mX^A\un t_A\,.
\eeq
By comparing to the covariant split \eqref{covariantsection}, we see that 
\beq
\mX^\mu=\umX^{\ualpha}\rho_\tau^\mu{}_{\ualpha}=X^\mu\,,\qquad \mX^A=\umX^{\un A}\omega_\tau^A{}_{\un A}+\umX^{\ualpha}\rho_\tau^\mu{}_{\ualpha}b^A_\mu=\mu^A+X^\mu b_\mu^A\,,
\eeq
and thus in the consistent splitting, the gauge field is contained in an off-block-diagonal piece of $\sigma_\tau$. 
\par
The next example is a section $\beta$ of $A^*$, i.e., $\beta\in\Omega^1(A_\tau)$. In the covariant splitting we can write
\be
\beta=\beta_{\ualpha}{\hat E}^{\ualpha}+\beta_{\un A}{\hat E}^{\un A}=\beta_{\ualpha}\sigma_\tau^{\un{\alpha}}{}_{\mu} \td x^\mu+\beta_{\un A}j_\tau^{\un{A}}{}_{A}(t^A-b^A_\mu \td x^\mu)\,,
\ee
while in the consistent splitting we have
\be
\beta=\beta_{\mu}\td x^\mu+\beta_{A}t^A\,.
\ee
Comparing the components of $\beta$ in two splittings we can see that
\be
\label{betaspliting}
\beta^a_\mu=\sigma_\tau^{\un\alpha}{}_\mu\beta^a_{\un\alpha}-j_\tau^{\un A}{}_A\beta^a_{\un A}b_\mu^A\,,\qquad\beta^a_A=j_\tau^{\un A}{}_A\beta^a_{\un A}\,.
\ee
This also applies to any $E$-valued 1-form in $\Omega^1(A_\tau;E)$. Furthermore, One can similarly find the conversion between the consistent and covariant splittings for any higher forms in the exterior algebra $\Omega(A_\tau;E)$.

In the current setup, the connection reform $\omega_\tau$ which defines the horizontal distribution through its kernel can be written in the consistent splitting as
\begin{equation} \label{omega = b - c}
	\omega_\tau = \omega_\tau^A{}_{\un{A}} \hat E^{\un{A}} \otimes \un{t}_A = \omega_\tau^A{}_{\un{A}}j_\tau^{\un{A}}{}_B (t^B -  b^B_{\mu} \td x^{\mu}) \otimes \un{t}_A = (b^A_{\mu} \td x^{\mu} - t^A) \otimes \un{t}_A = b - \varpi\,.
\end{equation}
where we defined
\beq \label{MC on L}
\varpi = \varpi^A\otimes\un{t}_A=t^A \otimes \un{t}_A\,,
\eeq
 which can be interpreted as the Maurer-Cartan form on $L$. Recall that $L$ is a bundle of Lie algebras, which means that the $\varpi$ given in \eqref{MC on L} should be interpreted as the Maurer-Cartan form for the group $G$ pointwise on the base manifold $M$. In other words, $\varpi$ is a field of Maurer-Cartan forms, with $\varpi(x)$ being the Maurer-Cartan form for each fiber of $L$ at $x \in M$. The spatial dependence of $\varpi$ will play a significant role in defining the exterior algebra in the consistent splitting.
 
Eq.~\eqref{omega = b - c} explicitly shows that the connection reform can be understood as the sum of two pieces, the first related to the gauge field, and the second related to the Maurer-Cartan form of the gauge algebra, if we interpret it in the consistent splitting (i.e., in terms of the bases for $TM$ and $L$ and their duals). This equation should be compared with the idea of an extended ``connection" $\econn = A + c$ in the BRST complex introduced in Section \ref{sec:BRST}, where $A$ is a local gauge field and $c$ is the ghost field. However, in the algebroid formulation \eqref{omega = b - c} has an advantage over the conventional extended ``connection" defined in the principal bundle context, because now it possesses a manifestly geometric interpretation as $\omega$ is a genuine connection on the Atiyah Lie algebroid. 

\subsection{Trivialized Lie Algebroids and the BRST Complex}
\label{sec:BRSTcohomology}
We now turn our attention to the main focus of this chapter---understanding the BRST complex from the exterior algebra of the trivialized algebroid. 
Similar to the evaluation for the Lie bracket on $A$ in \eqref{LAEE}--\eqref{LAE}, the Lie bracket on $A_{\tau}$ can be written explicitly for the basis sections as
\begin{align} \label{HH}
[\un{\hat E}_{\un{\alpha}}, \un{\hat E}_{\un{\beta}}]_{A_{\tau}}& = \sigma_\tau\left([\rho_\tau(\un{\hat E}_{\un{\alpha}}), \rho_\tau(\un{\hat E}_{\un{\beta}})]_{TM}\right)
+j_\tau( \Omega_{\un{\alpha} \un{\beta}})\,,
\\
\label{HV}
[\un{\hat E}_{\un{\alpha}}, \un{\hat E}_{\un{B}}]_{A_\tau}& = 
-j_\tau\left(R^{-\omega_\tau}(\un{\hat E}_{\un\alpha},\un{\hat E}_{\un B})\right)
=j_\tau\left(\nabla^L_{\un{\hat E}_{\un\alpha}}(\omega_\tau^A{}_{\un B}\un t_A)\right)=j_\tau\left(\Aconn{L}(\un{\hat E}_{\un\alpha})(\omega_\tau^A{}_{\un B}\un t_A)\right)\,,
    \\
\label{VV}
[\un{\hat E}_{\un{A}}, \un{\hat E}_{\un{B}}]_{A_\tau} &= j_\tau\left([\omega_\tau(\un{\hat E}_{\un{A}}),\omega_\tau(\un{\hat E}_{\un{B}})]_L\right)
=-\omega_\tau^A{}_{\un A}\omega_\tau^B{}_{\un B}f_{AB}{}^C\un{\hat E}_{\un C}j_\tau^{\un C}{}_C\,.
\end{align}

The coboundary operator for the complex $\Omega(A_{\tau}; E)$, denoted by $\hatd_{\tau}$, is defined precisely by the Koszul formula \eqref{dhat}. In terms of the isomorphism $\tau:A\to A_\tau$, we have, the chain map condition $\hatd\circ\tau^*=\tau^*\circ\hatd_\tau$. Working in $A_{\tau}$, we now have two different ways of splitting $\Omega(A_{\tau};E)$ into a bi-complex. Firstly, we can use the covariant splitting of $A_{\tau}$ to identify
\begin{equation}
	\Omega^p(A_{\tau};E) = \bigoplus_{r + s = p} \Omega^{(r,s)}(H_{\tau},V_{\tau};E)\,,
\end{equation}
where $\Omega^{(r,s)}(H_{\tau},V_{\tau};E)$ consists of bi-forms of degree $r$ in the algebra of $H_{\tau}$ and degree $s$ in the algebra of $V_{\tau}$. This is certainly the most natural splitting of the exterior algebra, as it is globally defined given a connection. We will show that this is equivalent to, but not the same as, the usual splitting, where $r$ counts the de Rham form degree and $s$ counts ghost number.

Alternatively, using the consistent splitting for $A_{\tau}$ we can identify
\begin{equation}
	\Omega^p(A_{\tau};E) = \bigoplus_{r + s = p} \Omega^{(r,s)}(TM,L;E)\,,
\end{equation}
where $\Omega^p(A_{\tau};E)$ now consists of bi-forms of degree $r$ in the de Rham cohomology of $M$ and degree $s$ in the Chevalley-Eilenberg algebra of $L$.

To understand precisely how this works, we consider the action of $\hatd_{\tau}$ on sections of various bundles. We will show that the action of $\hatd_\tau$ can be interpreted as acting as $\td+\ts$ on the components of sections, which reproduces the coboundary operator $\td_{\rm BRST}$ on the BRST complex. As a first example, we consider an $E$-valued scalar $\un{\psi} = \psi^a \un{e}_a\in\Gamma(E)$. Using the Koszul formula and \eqref{phiEf} and \eqref{defSpinConnHV}, we have
\begin{align}
\label{hatdsalar}
\hatd_{\tau} \un{\psi}
=\hat E^{\un M}\otimes\phi_E(\un {\hat E}_{\un M})(\un\psi)=
 \Big(\td \psi^a
 +  v_E(\un t_A)^a{}_b\varpi^A \psi^b  \Big)\otimes\un{e}_a \,.
\end{align}
Note that the $\phi_E$ and $v_E$ here are associated with the trivialized algebroid $A_\tau$. We can identify the above equation with\footnote{It should be noted that in \cite{Ciambelli:2021ujl} this was written as $\hatd\un\psi=\nabla^E\un\psi+\ts\un\psi$. These results are consistent, given that $\hatd\un\psi=\nabla^E\un\psi+\psi^a \ts\un e_a+ \ts\psi^a\otimes\un e_a=\td\psi^a\otimes \un e_a+ \ts\psi^a\otimes\un e_a$. This is a general feature: by extracting the basis elements, the gauge fields in the covariant derivative are canceled by those coming from $\ts\un e_a$. We will see this pattern repeated in additional examples.}
\beq\label{hatddplusE}
\hatd_{\tau} \un{\psi}=  (\td+\ts)\psi^a\otimes\un{e}_a\,,
\eeq
if we interpret 
\beq 
\ts\psi^a:= v_E(\un t_A)^a{}_b\varpi^A \psi^b\,.
\eeq 
We can recognize that this matches the action of the BRST operator on a scalar shown in \eqref{Dhorizontal} where now $-\varpi$ plays the role of the ghost field $c$.
\par
As a second example, consider a $E$-valued 1-form in $\Omega^1(A_\tau;E)$, namely a section $\beta\in\Gamma(A_\tau^*\times E)$. Employing the Koszul formula (which is most easily employed by translating $\beta$ into the covariant split basis), we find 
\begin{align}
\label{dhat1form}
\hatd_\tau\beta={}&\frac{1}{2}\hat E^{\un M}\wedge \hat E^{\un N}\otimes\Big(
\phi_E(\un {\hat E}_{\un M})(\beta^a_{\un N}\un e_a)
-\phi_E(\un {\hat E}_{\un N})(\beta^a_{\un M}\un e_a)
-\beta([\un {\hat E}_{\un M},\un {\hat E}_{\un N}]_{A_\tau})\Big)\nn\\
={}&
\Big(\td \beta^a_\nu +  v_E(\un t_A)^a{}_bt^A \beta^a_\nu\Big)\wedge \td x^\nu\otimes\un e_a
+\bigg(\td \beta^a_B +  v_E(\un t_A)^a{}_bt^A \beta^b_B -\frac12f_{AB}{}^{C}\beta^a_C t^A
 \bigg)\wedge t^B\otimes\un e_a\,,
\end{align}
and thus we see that
\beq
\hatd_\tau\beta=(\td+\ts)\beta^a_\mu\wedge \td x^\mu\otimes\un e_a + (\td+\ts)\beta^a_A\wedge t^A\otimes\un e_a\,,
\eeq
if we interpret
\beq
\ts\beta^a_\nu=v_E(\un t_A)^a{}_b\varpi^A \beta^a_\nu\,,\qquad
\ts\beta^a_B=v_E(\un t_A)^a{}_b\varpi^A \beta^b_B -\frac12f_{AB}{}^{C}\beta^b_C \varpi^A\,.
\eeq
This is the 1-form version of the scalar example in \eqref{hatddplusE}. The calculation for the scalar and 1-form examples can be carried over to any $E$-valued forms in $\Omega(A_\tau;E)$. For the detailed derivation for \eqref{hatdsalar} and \eqref{dhat1form}, see Appendix \ref{app:hatd}.

As a final example, we consider the connection reform $\omega_\tau$, which we regard as an element of $\Omega^1(A_\tau,L)$. The action of $\hatd_\tau$ gives
\begin{align}
\hatd_\tau\omega_\tau&= \hatd_\tau(b-\varpi)=(\Omega_\tau^A-\frac12 f_{BC}{}^A\omega_\tau^B\wedge \omega_\tau^C)\otimes\un t_A\nn
\\
&= (\td b^A+f_{BC}{}^A\varpi^B\wedge b^C-\frac12 f_{BC}{}^A \varpi^B\wedge\varpi^C)\otimes\un t_A\,,\label{hatdtauomega}
\end{align}
where in the last line we made use of the result \eqref{omega = b - c}, writing $\varpi=\varpi^A\otimes\un t_A$. 
We note that if we identify
\beq \label{s of omega and varpi}
\ts b^A=\td\varpi^A+f_{BC}{}^A\varpi^B\wedge b^C\,,\qquad \ts\varpi^A=\frac12 f_{BC}{}^A \varpi^B\wedge\varpi^C\,,
\eeq
then we obtain
\beq\label{hatdomegads}
\hatd_\tau\omega_\tau=(\td+\ts)\omega_\tau^A\otimes\un t_A \,.
\eeq
Eq.~\eqref{s of omega and varpi} are exactly the action of the BRST operator on the local gauge field and gauge ghost we have seen in \eqref{Russian Formula 2}. To understand \eqref{hatdomegads} one must establish an interpretation for the $\td\varpi^A$ in \eqref{s of omega and varpi}. As we have alluded to below \eqref{MC on L}, $\varpi$ is not spatially constant, and therefore has a nonzero derivative under de Rham differentiation $\td$. Considering the following pair of facts: 
\beq
	\hat{i}_{-j(\un\mu)}\varpi^A=-\mu^A\,, \qquad \hat{\cal L}_{-j(\un\mu)}\varpi^A=0\,,\qquad \forall \un\mu\in\Gamma(L)\,,
\eeq
and noticing that $\hat{\cal L}_{\un\mX}=\hat{i}_{\un\mX}\hatd+\hatd\hat{i}_{\un\mX}$, we have
\beq \label{de Rham of varpi}
	\hat{i}_{-j(\un\mu)} \td \varpi^A = \td \mu^A\,.
\eeq
Then, the first equation in \eqref{s of omega and varpi} is consistent with the standard variation of the gauge field [c.f.~\eqref{Infinitesimal Gauge Transform of A}]:
\beq
	\hat{i}_{-j(\un\mu)}\ts b^A=\td\mu^A+[b,\un\mu]^A=D\un\mu^A\,.
\eeq
Therefore, starting from the formal definition \eqref{dhat} of the nilpotent coboundary operator in the algebroid exterior algebra, we established the relationship  between $\hatd_{\tau}$ and the BRST differentiation $\ts$. Again, we emphasize that this result is a natural consequence of the geometric structure of the algebroid. 

To recapitulate, we have demonstrated how the fundamental features of the BRST complex are geometrically encoded in the Atiyah Lie algebroid. Working in the consistent splitting, the exterior algebra of the trivialized algebroid is a bi-complex consisting of differential forms on the base manifold $M$ and differential forms in the exterior algebra associated to the local gauge group. This is the state of affairs described in the BRST complex but only after making a series of choices \cite{becchi1976renormalization,alvarez1984topological,zumino1985chiral,
zumino1984chiral,kanno1989weil}. We have shown why these choices are reasonable. For example, the counterpart of the extended  ``connection" $\econn = A + c$ is identified with $\omega_\tau = b - \varpi$ in the algebroid context; $b$ corresponds to the gauge field $A$, and $\varpi$ corresponds to the ghost field $c$ (up to a sign difference). Significantly, $\omega_\tau$ is a genuine connection which defines a horizontal distribution on the algebroid. Moreover, the coboundary operator  $\hatd_\tau$ on the trivialized Lie algebroid behaves in the consistent splitting as $\td+\ts$, which reproduces the full BRST complex from the exterior algebra of trivialized algebroid.

As discussed in Subsection \ref{sec:curvature}, the ``Russian formula" central to the BRST analysis is also simply a geometric fact in the algebroid context arising from the observation that the curvature $\Omega$ of a Lie algebroid connection is zero when contracted with a vertical vector field, i.e. $\Omega$ is a horizontal form. Working in the consistent splitting of the trivialized algebroid, this version of the Russian formula can be stated in a more familiar form as
\begin{equation} \label{Algebroid Russian Formula 2}
	\Omega_\tau = \hatd_{\tau} \omega_\tau + \frac{1}{2}[\omega, \omega]_L = (\td + \ts)(b^A - \varpi^A)\otimes \un t_A + \frac{1}{2}[b-\varpi, b-\varpi]_{L} = \td b + \frac{1}{2}[b, b]_L = F\,,
\end{equation}
where $F \equiv \td b + \frac{1}{2}[b, b]_L$ is the gauge field strength of the gauge field $b$. In other words, the curvature $\Omega_\tau$ is now automatically ``ghost free'' without the need to apply any additional requirements.

\chapter{Anomalies from Lie Algebroid Cohomology}
\label{chap:anomaly}
In the BRST context, the Russian formula leads to the descent equations which subsequently characterize anomalies from a topological point of view. This form of the anomaly is referred to as the \emph{consistent anomaly} as it satisfies the Wess-Zumino consistency condition \cite{wess1971consequences}. However, the consistent form of the anomaly is not gauge covariant, and one can separately introduce the corresponding covariantized version, called the \emph{covariant anomaly} \cite{bardeen1984consistent}, as we have reviewed in Subsection \ref{sec:anomalyBRST}. In this final chapter we will demonstrate how this story carries over into the algebroid language. Moreover, we will give an illustration of how the algebroid may afford us with a more complete picture by demonstrating that it is capable of geometrizing the consistent form of the anomaly as well as the covariant form. The conventional analysis of the BRST complex can only cover the former. Here we will be computing anomalies from a purely cohomological perspective which is independent of any specific field theory. In other words, we simply mean that the consistent and covariant anomaly polynomials we derive have the correct topological and algebraic properties to be the anomalous divergences of the consistent and covariant currents that appear in the familiar physical considerations. 

\section{Characteristic Classes and Lie Algebroid Cohomology}
\label{sec:Chernalg}
In Section \ref{sec:BRST} we reviewed the cohomological formulation of anomalies in the BRST language, which begins by considering characteristic classes on a principal bundle and their associated Chern-Simons forms. In this section we will work in the context of an Atiyah Lie algebroid $A$, with connection reform $\omega$ and its curvature reform $\Omega = \hatd\omega + \frac{1}{2}[\omega,\omega]_L$. 

We begin by computing
\begin{equation}\label{Bianchi Identity}
\hatd\Omega = -[\omega,\Omega]_L\,,
\end{equation}
which can be recognized as the Bianchi identity, given $\hat\td^2=0$. The pair of equations
\begin{equation}\label{Algebra of omega and Omega}
\hatd \omega = \Omega - \frac{1}{2}[\omega,\omega]_L, \qquad \hatd \Omega = -[\omega, \Omega]_L
\end{equation}
implies that the ring of polynomials generated by $\omega$ and $\Omega$ form a closed subalgebra of $\Omega(A)$, just as \eqref{AFrelation} for the principal bundle case. This is the basis of the Chern-Weil homomorphism, which states that one can formulate cohomology classes in $\Omega(A)$ using such polynomials. The procedure of this is exactly parallel to what we introduced in Subsection \ref{sec:Chern}. Let $Q^{(l)}: L^{\otimes l} \rightarrow \mathbb{R}$ be a symmetric order-$l$ polynomial function on $L$ which is invariant under Lie algebroid morphisms. Such an object can be represented by a symmetric $l$-linear map in the tensor algebra of $L$. In other words, given the basis $\{t^A\}$ for $\Gamma(L^*)$ with $A = 1, \ldots, \dim G$, we can write 
\be
Q ^{(l)}= Q_{A_1 \ldots A_l} \bigotimes_{j = 1}^l t^{A_j}\,.
\ee
Notice that although this expression looks the same as the $Q^{(l)}$ defined in \eqref{Ql}, now each $t^A$ is a section on $L^*$ which is defined on $M$ pointwisely, while in \eqref{Ql} in the principal bundle case $t_A\in\mg$ does not depend on the point of $M$. In terms of such a symmetric invariant polynomial we can define the characteristic class on $A$ as follows:
\begin{equation}\label{CC of symmetric invariant poly}
\lambda_Q(\Omega) = Q^{(l)}(\underbrace{\Omega, \ldots, \Omega}_l) = Q_{A_1 \ldots A_l} \wedge_{j = 1}^l \Omega^{A_j} \in \Omega^{2l}(A)\,.
\end{equation}
Strictly speaking, the Chern-Weil theorem is proved in the context of principal bundle cohomology. However, the basis of the proof hinges on the fact that the principal connection and curvature satisfy the same algebraic relations as the algebroid connection and curvature given in \eqref{Algebra of omega and Omega}. Hence, the proof carries over to this case as well. (See \cite{fernandes2002lie} for a more rigorous discussion.) Then, the Chern-Weil theorem assures that each $\lambda_Q(\Omega)$ defines an element of the cohomology class of degree $2l$ in the exterior algebra $\Omega(A)$. Specifically, the two statements we introduced in Subsection \ref{sec:Chern} carries over directly to the Lie algebroid version:
\begin{enumerate}
	\item Characteristic classes are closed $2l$-forms in $\Omega(A)$:
\begin{equation}
	\hatd \lambda_Q(\Omega) = l! Q^{(l)}(\hatd \Omega, \underbrace{\Omega, \ldots, \Omega}_{l-1}) = l! Q^{(l)}(\hatd \Omega + [\omega, \Omega]_{L}, \underbrace{\Omega, \ldots, \Omega}_{l-1}) = 0\,,
\end{equation}
which follows from the symmetry of $Q^{(l)}$ and the Bianchi identity. 
	\item Given two different connections $\omega_1$ and $\omega_2$, with respective curvatures $\Omega_1$ and $\Omega_2$, we have that $\lambda_Q(\Omega_2) - \lambda_Q(\Omega_1) \in \Omega^{2l}(A)$ is $\hatd$-exact. The relevant $(2l-1)$-form potential is defined by introducing a one parameter family of connections $\omega_t = \omega_1 + t(\omega_2 - \omega_1)$ which interpolates between $\omega_1$ and $\omega_2$ as $t$ goes from $0$ to $1$. Then,
\begin{equation} \label{Transgression Formula}
	\lambda_Q(\Omega_2) - \lambda_Q(\Omega_1) = \hatd \left[ Q_{A_1 \cdots A_l} \int_{0}^{1} \td t \; (\omega_2 - \omega_1)^{A_1} \wedge_{j = 2}^{l} \left(\hatd \omega_t + \frac{1}{2}[\omega_t , \omega_t]_{L}  \right)^{A_j} \right]\,.
\end{equation}
\end{enumerate}

Once again, the characteristic class $\lambda_Q(\Omega)$ will be globally exact if there exists a one parameter family of connections for which $\omega_2 = \omega$ and $\omega_1$ is any connection that has zero curvature.\footnote{Note that a connection having zero curvature does not imply $\omega = 0$, which would be inconsistent with $\omega \circ j = -Id_L$. Rather, in the consistent splitting one can realize a connection with zero curvature by ensuring that the gauge field vanishes, i.e., $b = 0$. This implies $\omega_{\tau} = -\varpi$, which is consistent with the aforementioned identity. In physical contexts, this corresponds to the case that the connection is ``pure gauge".} Nonetheless, it is always true locally that any characteristic class can be written as $\hatd$ acting on a $(2l-1)$-form:\begin{equation}\label{Chern Simons in algebroid}
	\lambda_Q(\Omega) = \hatd \mathscr{C}_Q(\omega)\,,
\end{equation} 
where
\begin{equation}
\label{transgression}
	\mathscr{C}_Q(\omega):= Q_{A_1 \cdots A_l} \int_{0}^{1} \td t\, \omega^{A_1} \wedge_{j = 2}^{l} \left(t\hatd \omega + \frac{1}{2}t^2[\omega , \omega]_{L}  \right)^{A_j}\,.
\end{equation} 
This transgression formula defines the algebroid Chern-Simons form associated with the symmetric invariant polynomial $Q^{(l)}$. Note that \eqref{Chern Simons in algebroid} indicates that there does not exist $\gamma\in\Omega^{2l-2}(A)$ such that $\scr C_Q=\hat\td\gamma$, and $\scr C_Q$ can only be determined up to a $\hat\td$ closed term. 

\section{Descent Equations and the Consistent Anomaly}
\label{sec:consistent}
Now, let us move into the trivialized algebroid $A_\tau$ and work in the consistent splitting. As we have shown, in the consistent splitting $\omega_\tau = b - \varpi$, and $\hatd_{\tau} \to \td + \ts$. It is therefore natural to organize the Chern-Simons form order by order in the bi-complex $\Omega(TM,L)$ as
\begin{equation} \label{Algebroid CS Expansion}
	\mathscr{C}_Q(b - \varpi) = \sum_{r + s = 2l-1} \alpha^{(r,s)}(b,\varpi)\,,
\end{equation}
where $\alpha^{(r,s)}(b,\varpi) \in \Omega^{(r,s)}(TM,L)$, and $\alpha^{(2l-2,1)}(b,\varpi) = \mathscr{C}_Q(b)$. 

Combining \eqref{Algebroid Russian Formula 2} and \eqref{Chern Simons in algebroid} yields
\begin{equation} \label{Algebroid Descent 1}
	\hatd_\tau \mathscr{C}_Q(b - \varpi) = \lambda_Q(\Omega) = \lambda_Q(F) = \td \mathscr{C}_Q(b)\,.
\end{equation} 
From this point it is straightforward to derive the descent equations simply by plugging (\ref{Algebroid CS Expansion}) into (\ref{Algebroid Descent 1}), and enforcing the equality order by order in the bi-complex $\Omega^{(r,s)}(TM,L)$. The descent equations can be expressed as
\begin{equation}\label{Descent Equations}
	\td \alpha^{(r,s)}(b,\varpi) + \ts\alpha^{(r+1,s-1)}(b,\varpi) = 0\,,\qquad r + s = 2l-1\,,\quad r \neq 2l-1\,,
\end{equation}
In particular, the term with $r = 2l-3$ yields the Wess-Zumino consistency condition:
\begin{equation} \label{WZalgebroid}
	\td\alpha^{(2l-3,2)}(b,\varpi) + \ts\alpha^{(2l-2,1)}(b,\varpi) = 0\,.
\end{equation}
On the other hand, from the fact that $\mathscr{C}_Q(b - \varpi)$ is not $\hat\td_\tau$ exact we also have
\be
\label{WZ not exact}
\alpha^{(2l-2,1)}(b,\varpi)\neq\td \gamma^{(2l-3,1)}(b,\varpi)+\ts\gamma^{(2l-2,0)}(b,\varpi)\,.
\ee
The term $\alpha^{(2l-2,1)}(b,\varpi)$ satisfying \eqref{WZalgebroid} and \eqref{WZ not exact} is a candidate to be the density of the consistent anomaly. Thus, we have now demonstrated that the consistent anomaly arises naturally in the algebroid context:
\begin{equation}
 \label{Aconsistent}
	\mathfrak{a}_{\rm con} = \int_M \alpha^{(2l-2,1)}(b,\varpi)\,.
\end{equation}
This result precisely matches the consistent anomaly \eqref{Aconsistent0} derived from the BRST formalism, with the gauge field $A$ now represented by $b$ and the ghost field $c$ represented by $-\varpi$.

\section{Free Variation and the Covariant Anomaly}
\label{sec:covariant}

Strictly speaking, the results discussed in the previous subsection 
are merely a reformulation of those obtained in the BRST analysis \cite{Bilal:2008qx}, although now they come from a transparent formal and geometric foundation which makes their origin and meaning clear. However, beyond simply improving our interpretation of the BRST analysis, we would now like to demonstrate that the algebroid approach has the potential to produce new results in the study of anomalies. 

As we have stressed, the trivialized algebroid has two relevant splittings. By analyzing the cohomology of the consistent splitting above we found the consistent anomaly. This inspires the question of whether the covariant splitting also has an interpretation related to an anomaly. Following the previous subsection, we can instead organize the Chern-Simons form on $A_{\tau}$ order by order in the bi-complex $\Omega^{(r,s)}(H_{\tau},V_{\tau})$. The most transparent way of doing this is by expanding the Chern-Simons form as a polynomial in the connection $\omega \in \Omega^1(V;L)$ and its curvature $\Omega \in \Omega^2(H;L)$. Here again we see the Russian formula playing a crucial role in dictating that the curvature can generate a sub-algebra of $\Omega(H_{\tau})$. The expansion of the Chern-Simons form can now be written as
\begin{equation}\label{covariantCSexpn}
	\mathscr{C}_Q(\omega) = \sum_{r + s = 2l-1} \beta^{(r,s)}(\omega,\Omega)\,,
\end{equation} 
where $\beta^{(r,s)}(\omega,\Omega) \in \Omega^{(r,s)}(H,V)$ contains $r/2$ factors of the curvature and $s$ factors of the connection.

We will now show that the covariant splitting directly produces the covariant anomaly. As was established in \cite{bardeen1984consistent,stone2012gravitational,Hughes:2012vg} the covariant anomaly is obtained from the free variation of the Chern-Simons form with respect to the connection. Computing this variation in the algebroid context, one arrives at the following formula (see Appendix~\ref{app:covanomaly} for details):
\begin{equation} 
\label{Covariant Anomaly}
	\delta \mathscr{C}_Q(\omega) = l \beta^{(2l-2,1)}(\delta\omega, \Omega) + \hatd\Theta(\omega,\delta\omega)\,,
\end{equation}
where
\begin{equation} 
\label{beta2l-2}
\beta^{(2l-2,1)}(\delta\omega, \Omega)=\frac{1}{l}Q(\underbrace{\Omega,\ldots,\Omega}_{l-1},\delta\omega)\,.
\end{equation}
Hence, the covariant anomaly can be read off from the first term in \eqref{Covariant Anomaly}. We therefore recognize that the covariant anomaly is intimately related to the term of order one in the vertical part of the Lie algebroid exterior algebra appearing in the expansion of the Chern-Simons form. This establishes a pleasant symmetry between the covariant anomaly and the consistent anomaly, since the consistent anomaly was proportional to the ``ghost number" one term in the expansion of the Chern-Simons form when viewed in the consistent splitting. We should note that from this point of view, the consistent and covariant anomalies do not coincide precisely because $V^*$ is not canonical, depending on the connection.

The covariant anomaly does not come with a series of descent equations that leads to a consistency condition. Instead, its defining property is that it is covariant with respect to the gauge transformation. In fact, we can now readily interpret the geometric difference between the consistent and covariant anomalies in the algebroid formulation. The former, being written in the consistent splitting of the algebroid, respects the nilpotency of the coboundary operator $\hatd$ in both factors of its associated bi-complex but spoils the gauge covariance. Conversely, the latter, although it does not admit two nilpotent differential operators, respects the covariant splitting defined by the connection $\omega$ and thus is endowed with gauge covariance. Such a conclusion was not possible from the perspective of the BRST complex, precisely because it lacked a geometry for its connection to define a covariant splitting.

\section{Examples} \label{sec:examples}
After establishing the formalism, now we exhibit the calculation for two illuminating examples: one is the familiar chiral anomaly and the other is the (type A) Lorentz-Weyl anomaly. In both cases the covariant and consistent forms of the anomaly are deduced by analyzing an appropriate characteristic class and its associated Chern-Simons form. The analysis done here can easily be generalized to any arbitrary even dimension.

\subsection{Chiral Anomaly}\label{sec:chiral}
The analysis of the chiral anomaly arises in the context of an Atiyah Lie algebroid $A$ derived from a principal bundle $P(M,G)$, where $G$ is a semisimple Lie group. The characteristic class that is relevant to the chiral anomaly in $2d$ is the second Chern class\footnote{For simplicity, we have taken a basis such that the second Killing form is given by $\delta_{AB}$.}
\begin{equation}
\text{ch}_2(\Omega) = \delta_{AB} \; \Omega^A \wedge \Omega^B\,.
\end{equation}
The Chern-Simons form associated with $\text{ch}_2(\Omega)$ can be deduced by employing the transgression formula \eqref{Transgression Formula}:
\begin{equation}\label{CS form for second Chern class}
\mathscr{C}_2(\omega) = \delta_{AB}\left(\omega^A \wedge \hatd\omega^B + \frac{1}{3}\omega^A \wedge [\omega,\omega]_L^B\right)\,.
\end{equation}

Using \eqref{CS form for second Chern class}, we can easily determine the algebraic form of candidates for the covariant and consistent forms of the anomaly. To begin, still working in the algebroid $A$ we can decompose \eqref{CS form for second Chern class} order by order in the bi-complex $\Omega(H,V)$ by re-expressing it as a polynomial in the curvature and connection; that is, where there is a $\hatd\omega$ we will replace it by $\Omega - \frac{1}{2}[\omega,\omega]_L$. The resulting expression is
\begin{equation}
\mathscr{C}_2(\omega,\Omega) = \delta_{AB} \left(\omega^A \wedge \Omega^B - \frac{1}{6}\omega^A \wedge [\omega,\omega]_L^B\right)\,.
\end{equation}
In other words, the various terms in \eqref{covariantCSexpn} are given by
\begin{equation}
\beta^{(2,1)}(\omega,\Omega) = \delta_{AB} \; \omega^A \wedge \Omega^B\,, \qquad \beta^{(0,3)}(\omega,\Omega) = -\frac{1}{6}\delta_{AB} \; \omega^A \wedge [\omega,\omega]_L^B\,,
\end{equation}
from which we can read off  by applying \eqref{Covariant Anomaly} that the covariant anomaly polynomial is given in terms of the curvature $2\delta_{AB}\Omega^B$, as expected. 

To obtain the consistent anomaly polynomial, we pass to the trivialized Lie algebroid. That is, we specify a map $\tau: A \rightarrow A_{\tau}$ along with its inverse map $\overline{\tau}: A_{\tau} \rightarrow A$. Recall from Subsection \ref{sec:morphisms} that such a morphism implies the following relationships between the connections, curvatures, and coboundary operators of the two algebroids:
\begin{equation}
\overline{\tau}^*\omega = \omega_{\tau} = b - \varpi\,,\qquad 
\overline{\tau}^*\Omega = \Omega_{\tau} = F\,,\qquad 
\overline{\tau}^* \circ \hatd = \hatd_{\tau} \circ \overline{\tau}^* \,.
\end{equation}
Trivializing the Chern-Simons form, it follows from \eqref{hatdtauomega} that
\begin{eqnarray}\label{Trivialized CS Form for Chern Class}
\overline{\tau}^*\mathscr{C}_2(\omega)=\mathscr{C}_2(\omega_\tau)
=\mathscr{C}_2(b) + \delta_{AB} \left(-\varpi^A \wedge \td b^B - \frac{1}{2}b^A \wedge [\varpi,\varpi]_L^B + \frac{1}{6}\varpi^A \wedge [\varpi,\varpi]_L^B\right)\,.
\end{eqnarray}
Then, the expansion \eqref{Algebroid CS Expansion} gives
\begin{equation}
\begin{split}
\alpha^{(3,0)}(b,\varpi) &= \mathscr{C}_2(b)\,,\qquad 
\alpha^{(2,1)}(b,\varpi) = -\delta_{AB}  \varpi^A \wedge \td b^B\,,\\
\alpha^{(1,2)}(b,\varpi) &= -\frac{1}{2}\delta_{AB}  b^A \wedge[\varpi,\varpi]_L^B\,,\qquad
\alpha^{(0,3)}(b,\varpi) = \frac{1}{6}\delta_{AB}  \varpi^A \wedge [\varpi,\varpi]_L^B\,.
\end{split}
\end{equation}
The consistent anomaly polynomial can therefore be read off from the ghost number one contribution to \eqref{Trivialized CS Form for Chern Class}, which is $-\delta_{AB} \varpi^A \wedge \td b^B$. Recall that $-\varpi^A$ corresponds to the ghost field, the consistent anomaly can be recognized $\delta_{AB} \td b^B$, which is again in agreement with the known result. 

As promised, the covariant anomaly, which is written in terms of $\Omega$, is indeed covariant, while the consistent anomaly, which is written in terms of $\td b$, is not. Moreover, it is straightforward to show that the series of terms $\alpha^{(r,s)}(b,\varpi)$ satisfy the descent equations as introduced in \eqref{Descent Equations}.

\subsection{Lorentz-Weyl Anomaly}
\label{sec:LW}
To analyze the Lorentz-Weyl (LW) anomaly, let us begin by introducing the geometric framework and characteristic classes for a Lorentz-Weyl structure in arbitrary even dimension $d = 2l$. Consider an Atiyah Lie algebroid $A$ derived from a principal $G$-structure with $G = SO(1,d-1)\times\RR_+\subset GL(d,\RR)$. Here $SO(1,d-1)$ is the local Lorentz group, while $\RR_+$ corresponds to local Weyl rescaling. The corresponding Lie algebra can be expressed as $\mathfrak g=\mathfrak{so}(1,d-1)\oplus\mathfrak{r}_+$. The adjoint bundle of the group $G$ is given by $L = P \times_{G} \mg = L_{L} \oplus L_{W}$, where $L_L = P \times_{SO(1,d-1)} \mathfrak{so}(1,d-1)$ and $L_{W} = P \times_{\mathbb{R}_+}\mathfrak{r}_+d$ correspond to the Lorentz and Weyl factors, respectively. 
The connection reform on $A$ will therefore split as $\omega =\omega_L + \omega_W$ where $\omega_L$ and $\omega_W$ are the connection reform on the Lorentz and Weyl sub-algebroids, respectively. The curvature of the connection reform $\omega$ will have two pieces
\begin{equation}
\Omega = \hatd\omega + \frac{1}{2}[\omega,\omega]_L = \Omega_{L} + \Omega_{W}\,,
\end{equation}
where $\Omega_{L}\in \Omega^2(H;L_L)$ is related to the Riemann tensor and $\Omega_{W} \in \Omega^2(H;L_{W})$ is the gauge field strength of the Weyl connection. We can see that the curvature $\Omega$ remains horizontal.

There are two natural invariant structures associated with $L$. The Weyl factor $L_{W}$ is an Abelian subalgebra of $L$. Thus, the map $\tr_W: L \rightarrow L_{W}$ which projects an element $\underline{\mu} \in \Gamma(L)$ down to $L_{W}$ will be invariant under the adjoint action of $L$ on itself. In a linear representation of $L$ given by $v_E: L \rightarrow \text{End}(L)$, the generators of $L_L$ are represented by traceless antisymmetric matrices. Hence, as the notation indicates, the map $\tr_W$ can also be understood by selecting a representation and computing the ordinary trace. In other words, for any representation $E$ and given $\tr: \text{End}(E) \rightarrow C^\infty(M)$ we have
\begin{equation} \label{Trace Map}
\tr_W(\underline{\mu}) = \tr \circ v_E(\underline{\mu})\,.
\end{equation} 

Similarly, there is an invariant structure on $L_L$ which will correspond to the Pfaffian. In particular we define
\begin{equation} \label{Pfaffian Map}
\epsilon: L^{\otimes l} \rightarrow C^{\infty}(M)\,.
\end{equation}
One of the defining properties of the map $\epsilon$ is that $\epsilon(\underline{\mu}_1, \ldots , \underline{\mu}_{l}) = 0$ if $\underline{\mu_i} \in \Gamma(L_{W})$ for any $i$. In other words, $\epsilon$ only sees the orthogonal factor of $G$, and is an invariant polynomial on this factor. As was the case with the trace, $\epsilon$ can be computed by passing to a linear representation. To be precise, we should take a $2l$-dimensional representation space $E$ equipped with an inner product $g_E: E \times E \rightarrow C^{\infty}(M)$ of appropriate signature. Then, we can define the map $w_E: L \rightarrow \wedge^2 E^*$ such that given $\underline{\psi}_1, \underline{\psi}_2 \in \Gamma(E)$ we have
\begin{equation}
w_E(\underline{\mu})(\underline{\psi}_1, \underline{\psi}_2) = g_E\left(\underline{\psi}_1, v_E(\underline{\mu})(\underline{\psi}_2)\right)\,.
\end{equation}
Notice that $w_E \circ \tr_W = 0$, since a Weyl rescaling cannot be represented by an antisymmetric matrix. Given an oriented orthonormal basis $\{\underline{e}_a\}$ for $E$ along with its dual basis $\{e^a\}$, with $a=1,\ldots, 2l$, we can define an $SO(1,d-1)$ invariant volume form on $E$\footnote{Note that we are {\it not} specifying a solder form, and so we have no way to pull this volume form back to the base. Similarly the inner product on $E$ is not directly related to a metric on the base. These facts might be thought of as being responsible for the topological nature of the characteristic classes discussed below.}
\begin{equation}
\text{Vol}_E \equiv \epsilon_{a_1 \cdots a_d} e^{a_1} \wedge \cdots \wedge e^{a_d}\,.
\end{equation}
Thus, in this representation we can express:
\begin{equation}\label{Evaluation of Pffafian Map}
\epsilon(\underline{\mu}_1, \ldots, \underline{\mu}_l) = \epsilon_{a_1 b_1 \cdots a_l b_l} w_E(\underline{\mu}_1)^{a_1 b_1} \cdots w_E(\underline{\mu}_l)^{a_l b_l} = \epsilon^{a_1}{}_{b_1} \cdots^{a_l}{}_{b_l} v_E(\underline{\mu}_1)^{b_1}{}_{a_1} \cdots v_E(\underline{\mu}_l)^{b_l}{}_{a_l}\,.
\end{equation}
This construction satisfies the above-mentioned properties since $w_E \circ \tr_W(\underline{\mu}) = 0$ and 
\begin{equation}
\epsilon(\underline{\mu}, \ldots, \underline{\mu}) = \text{Pf}(\underline{\mu})\,.
\end{equation}
Note that this construction requires $d$ to be even, as the $\epsilon^{a_1}{}_{b_1} \cdots^{a_l}{}_{b_l}$ has an equal number of up and down indices (signifying its Weyl invariance).

We are now prepared to introduce the relevant characteristic class for the LW anomaly. If we intend to derive the anomaly for a $d = 2l$ dimensional theory, we must construct a characteristic class of form degree $d+2 = 2(l+1)$. Hence, we must construct a symmetric and invariant linear map $Q^{LW,l+1}: L^{\otimes (l+1)} \rightarrow \mathbb{R}$. As we have discussed, we have at our disposal two invariant objects corresponding to the trace \eqref{Trace Map} and the Pfaffian \eqref{Pfaffian Map}. We therefore obtain an $(l+1)$-order symmetric invariant polynomial by taking the symmetrized product of these two maps:
\begin{equation} \label{symminv}
Q^{LW,l+1}(\underline{\mu}_1, \ldots \underline{\mu}_{l+1}) = \sum_{\pi} \epsilon(\underline{\mu}_{\pi(1)}, \ldots, \underline{\mu}_{\pi(l)})\,\tr_{W}(\underline{\mu}_{\pi(l+1)})\,,
\end{equation}
where $\pi$ denotes the permutations of $(1,\ldots,l+1)$. The characteristic class associated with $Q^{LW,l+1}$ is therefore given by $\lambda_{Q^{LW,l+1}}(\Omega)$ as dictated in \eqref{CC of symmetric invariant poly}.  While $\lambda_{Q^{LW,l+1}}$ is the appropriate characteristic class in the LW context, in other situations (such as a simple or semi-simple group) one finds an Euler class.\footnote{Indeed in the literature \cite{Boulanger:2007st,Boulanger:2018rxo,Francois:2015oca,Francois:2015pg} there is an analysis of Cartan geometry, in which the symmetry is enhanced to $SO(2,d)$, and the type A conformal anomaly comes from the Euler class. Descending to the subgroup $SO(1,d-1)\times\RR_+$ considered here, one obtains \eqref{symminv}.}

Let us now specialize to the case $d = 2$ and show that $\lambda_{Q^{LW,2}}$ gives rise to the LW anomaly. The characteristic class of interest takes the following form:
\begin{equation}\label{Gen 2 Euler}
\lambda_{Q^{LW,2}}(\Omega) = \frac{1}{2}\left(\epsilon(\Omega) \wedge \tr_W(\Omega) + \tr_W(\Omega) \wedge \epsilon(\Omega)\right)\,.
\end{equation}
In the $2d$ case, since the structure group $G = SO(1,1) \times \mathbb{R}_{+}$ is Abelian, we can write $\Omega = \hatd\omega$. Hence, the Chern-Simons form can be obtained as
\begin{equation}\label{CS form for second generalized euler class}
\mathscr{C}_{LW,2}(\omega,\Omega) = \frac{1}{2}\left(\epsilon(\omega) \wedge \tr_W(\Omega) + \tr_W(\omega) \wedge \epsilon(\Omega)\right)\,.
\end{equation}
To read off the covariant form of the anomaly polynomial let us pass to a representation on $E$. Then using \eqref{Trace Map} and \eqref{Evaluation of Pffafian Map} we can write the covariant anomaly as (ignoring the constant factor)
\begin{equation}
\label{LW cov anom}
\Omega_{W}\epsilon^{a}{}_b + \text{Pf}(\Omega_{L})\delta^{a}{}_b\,.
\end{equation}
Noticing that $\epsilon(\omega)$ and $\tr_W(\omega)$ picks out the Lorentz and Weyl part of the connection, respectively, the first term in the above result should be interpreted as the Lorentz anomaly, which vanishes when the Weyl connection is turned off; the second term is the Weyl anomaly in $2d$, which is proportional to the Ricci scalar of the spacetime. Therefore, the LW anomaly is the mixed anomaly between the Lorentz and Weyl symmetry. In fact, it is easy to see that by adding a total derivative term, one can remove the Lorentz anomaly or Weyl anomaly but cannot remove both simultaneously.

To obtain the consistent form, we must employ a Lie algebroid trivialization. Under the trivialization we find that
\begin{equation}
\overline{\tau}^*\omega = b - \varpi_{L} + a - \varpi_{W}\,,\qquad 
\overline{\tau}^*\Omega = R + f\,,\qquad 
\overline{\tau}^* \circ \hatd = (\td + \ts_{L} + \ts_{W}) \circ \overline{\tau}^*\,,
\end{equation}
where $b$ and $a$ are the spin connection and Weyl connection on $M$, and $R$ and $f$ are their curvature 2-forms, respectively. The pairs $(\varpi_{L},\ts_{L})$ and $(\varpi_{W},\ts_{W})$ are the Maurer-Cartan forms and BRST operators for the $SO(1,1)$ and $\mathbb{R}_+$ factors of $L$. Let $B = b + a$ and $\varpi = \varpi_{L} + \varpi_{W}$ denote the combined gauge field and Maurer-Cartan forms. We subsequently identify the consistent LW anomaly from $Q^{LW,2}(\varpi,\td B)$. Since in the index notation of the representation we have 
\be
(\td B)^a{}_b = R\epsilon^a{}_b + f\delta^a{}_b\,,
\ee
the consistent form of the LW anomaly is merely the pullback of the covariant form by the trivialization $\overline{\tau}$, which reads
\begin{equation}
f\epsilon^{a}{}_b + \text{Pf}(R)\delta^{a}{}_b\,,
\end{equation}
which has the same form as \eqref{LW cov anom}. This follows in this particular case from the fact that $G$ is an Abelian group when $d=2$. A simplified account of the LW anomaly in two dimensions appeared also in Appendix A of \cite{Campoleoni:2022wmf}.

Note that here we have focused on the type A Weyl anomaly, and the type B Weyl anomaly remains an open question in general dimension. In Part I we have seen that the building blocks of the holographic Weyl anomaly are the Schouten tensor and obstruction tensors, and conjectured that it is true for the Weyl anomaly of a general theory. Since obstruction tensors, which prevents the type B Weyl anomaly to be topological [in the sense of \eqref{varX}], are expected to make an appearance in the type B Weyl anomaly, more consideration may be necessary in addition to the standard characteristic class construction.

\section{Discussion}
\label{sec:disc2}
\subsection{Summary and Outlook}
In Chapter \ref{chap:intro2} we raised a series of questions about the BRST formalism. We have provided answers to each of these questions in Part II of this thesis by geometrically formalizing the BRST complex in terms of the Atiyah Lie algebroid. As we promised in the introduction, each answer follows immediately from the geometry of the Atiyah Lie algebroid.

\textbf{Q:} Why should the Grassmann-valued fields $c^A(x)$, which started their life in the BRST quantization procedure have an interpretation as the generators of local gauge transformations? And why is it reasonable to combine the de Rham complex and the ghost algebra into a single exterior bi-algebra?

\textbf{A:} In the algebroid context the Maurer-Cartan form $\varpi \in \Omega^1(L;L)$ plays the role of the gauge ghost, and is also a generator of local gauge transformations. Working in the consistent splitting the exterior algebra of the trivialized algebroid $A_{\tau}$ subsequently takes the form of a bi-complex $\Omega^{(p,q)}(TM,L;E)$, where $p$ is the form degree with respect to the de Rham cohomology of $M$, and $q$ is the ``ghost number". The coboundary operator $\hat{\td}_{\tau}$ takes explicitly the form $\td + \ts$  on this exterior algebra, where $\td$ is the de Rham differential and $\ts$ is the BRST operator. 

\textbf{Q:} Why is it reasonable to consider $\econn = A + c$ as a ``connection'', and moreover what horizontal distribution does it define?

\textbf{A:} Still in the context of the trivialized Lie algebroid, one can introduce a connection reform, $\omega_\tau: A_{\tau} \rightarrow L$, defining the horizontal distribution $H_{\tau} = \text{ker}(\omega_\tau)$ for which $A_{\tau} = H_{\tau} \oplus V_{\tau}$. In the consistent splitting $\omega_\tau = b - \varpi$, where $b: TM \rightarrow L$ is a local gauge field, and $\varpi: L \rightarrow L$ is the Maurer-Cartan form on $L$. Hence, $\omega$ reproduces the ``connection" $\econn$ defined in the BRST complex, where again we see the role of the gauge ghost being played by the Maurer-Cartan form. 

\textbf{Q:} Why should the ``curvature" $\ecurv$ be taken to have ghost number zero? And why does enforcing this requirement turn the BRST operator $\ts$ into the Chevalley-Eilenberg operator for the Lie algebra of the structure group?

\textbf{A:} $\ecurv$ in the context of the trivialized Lie algebroid is represented by the curvature associated with $\omega_\tau$, $\Omega_\tau = \hat{\td}_\tau\omega_\tau + \frac{1}{2}[\omega_\tau, \omega_\tau]_{L}$, which is fully horizontal as a built-in geometric property of the algebroid. In the consistent splitting, this reproduces the Russian formula and the BRST transformation as presented in \eqref{Algebroid Russian Formula 2}. 

The culmination of all of these facts gives rise to the descent equations \eqref{Descent Equations} and the Wess-Zumino consistency condition \eqref{WZalgebroid}. Given a characteristic class $\lambda_{Q}(\Omega)$ with associated Chern-Simons form $\mathscr{C}_{Q}(\omega)$ we have
\begin{equation}
    \hat{\td}_{\tau} \mathscr{C}_{Q}(\omega) = (\td + \ts) \mathscr{C}_{Q}(b - \varpi) = \td \mathscr{C}_{Q}(b)\,.
\end{equation}
From the above equation, one can immediately compute the \emph{consistent} anomaly polynomial, which corresponds to the ghost number one contribution to $\mathscr{C}_Q(b - \varpi)$, and can be shown to be an element of the first cohomology of the BRST operator $\ts$ once integrated over a space of appropriate dimension. Furthermore, one can also obtain the \emph{covariant} form of the anomaly by viewing the Chern-Simons form in the covariant splitting and extracting the terms contributing with one exterior power in the vertical sub-bundle of the associated exterior algebra (multiplied by the order $l$ of $Q$). Although the formulas for finding the consistent and covariant anomalies have been known \cite{bardeen1984consistent}, our approach to these anomalies provides a meaningful explanation as to why the consistent anomaly is consistent and the covariant anomaly is covariant. From the algebroid perspective, they just correspond to different choices of splitting.

To understand the complete picture of the consistent and covariant anomalies, we will have to further exploit the structure of the configuration space of Lie algebroid connections. In this thesis we established a powerful approach for studying Lie algebroid isomorphisms in terms of commutative diagrams, which found a physical interpretation as a unified tool for implementing diffeomorphisms and gauge transformations. The authors of \cite{klinger2023abc} have made use of this construction to define a new geometric formalism, called the configuration algebroid, for understanding the extended configuration space of arbitrary gauge theories. From the point of view of the configuration algebroid, the presence of anomalies is associated with the question of whether the charge algebra is centrally extended. 

As mentioned in Chapter \ref{chap:intro2}, our analysis of anomalies so far applies to the perturbative anomalies for continuous symmetries. One possible direction is to investigate how to extend this geometric setup to discuss perturbative anomalies of large gauge transformations or discrete symmetries, which may involve studying the corresponding groupoid structure. Furthermore, it is also natural to consider how this formalism can be carried over to study anomalies of generalized symmetries.

Having a geometric understanding of the BRST formalism in the algebroid language, we also hope to further understand other interesting physical aspects of quantum gauge theory. One example is the Gribov problem \cite{Gribov:1977wm,Vandersickel:2012tz}, which states that when one restricts the space of gauge fields to the so-called Gribov region, some features related to confinement become manifest but the BRST symmetry is broken. The remedy for this issue requires analyzing the global topology of the Lie algebroid, which has been touched upon in \cite{AliAhmad:2024wja} in the context of the $G$-framed algebra. It would be valuable to explore this further and find applications of the geometric formulation presented in this thesis to understanding topics such as QCD and confinement.

\par
\subsection{Comments on the Weyl Anomaly}
At the end of this thesis, we would like to comment on some new insights into the Weyl anomaly, combining the understanding from Part I and Part II. In Part I we focused on the holographic Weyl anomaly and utilized the WFG gauge which provides a Weyl geometry background for the boundary theory. In Part II we studied the WL structure and identified the mixed anomaly nature of the type A Weyl anomaly. Although neither case addresses the most general form of the Weyl anomaly, these observations reveal some previously underemphasized features of it.

First, the Weyl connection plays a crucial role in identifying the mixed anomaly. In Subsection \ref{sec:LW}, we we explicitly demonstrated that the connection of the WL structure, split as $\omega=\omega_L+\omega_W$, gives rise to the mixed Lorentz-Weyl mixed anomaly, where the Weyl anomaly depends on the curvature of the Lorentz connection $\omega_L$. Now if we look back at the holographic Weyl anomaly derived in Chapter \ref{chap:HWA}, since we turned on two backgrounds fields $g$ and $a$ on the Weyl geometry background, the Weyl anomaly can be interpreted as a Weyl-diffeomorphism mixed anomaly. In fact, by turning on $g$ one also turns on the unique affine connection $\nabla$ satisfying $\nabla g=2ag$, and the Weyl-LC connection satisfying $\hat\nabla g=\nabla g-2ag$ [or equivalently \eqref{WLCsplit}] is precisely the counterpart of $\omega=\omega_L+\omega_W$. Similar discussion can also be applied to theories with other gauge groups. For example, with the Weyl connection, the famous trace anomaly of $4d$ QED or QCD can be recognized as the mixed anomaly between the Weyl and $U(1)$ or $SU(N)$ symmetries. 

Second, although the Weyl anomaly is sometimes considered to have no anomaly inflow, the holographic picture in the WFG gauge provides a natural anomaly inflow for it. Recall that in the anomaly inflow picture, the boundary anomaly matches the variation of the bulk theory induced on the boundary, and the boundary connection and symmetry transformation should also be induced from those in the bulk. In the WFG gauge, not only can the boundary anomaly be obtained from the bulk variation, but the boundary Weyl-LC connection and Weyl symmetry are indeed also induced from the bulk LC connection and the Weyl diffeomorphism. However, this anomaly inflow is unconventional in the following senses: (1) the boundary is not a finite boundary but an asymptotic boundary; (2) the bulk effective theory is not a topological field theory. The first property may be related to the fact that the Weyl anomaly is a real factor in the path integral transformation rather than a phase. The second is related to another distinctive property of the Weyl anomaly, namely it is not robust (not a 't Hooft anomaly) but monotonically decreases under the RG flow \cite{Zamolodchikov:1986gt,Komargodski:2011vj}.\footnote{See, however, \cite{Schwimmer:2010za,Schwimmer:2023nzk} for the discussions on the anomaly matching of the Weyl anomaly between the unbroken and spontaneously broken phases.} Therefore, holography not only potentially offers an inflow picture but may also unravel the peculiarities of the Weyl anomaly compared with other anomalies. It is appealing to unify the holographic and finite boundary pictures of anomaly inflow and find the relationship of this picture with the recently developed symmetry topological field theory (SymTFT) \cite{Bah:2019rgq,Apruzzi:2021nmk}.\footnote{Also note that in holography there is a duality between the bulk an boundary theories instead of having a theory coupled to the boundary of the bulk.}

Finally, we have seen that the holographic Weyl anomaly can be cast into the compact form \eqref{X1D}--\eqref{X4D} using the Schouten tensor and obstruction tensors, which provides clues for a general expression in arbitrary even dimensions. In Subsection \ref{sec:LW}, we also found that the type A Weyl anomaly can be derived from a characteristic class constructed from cuvature. Based on these results, for a general non-holographic theory, we expect that the building blocks of the Weyl anomaly are the Riemann curvature, Schouten tensor, and obstruction tensors. In this way, the Weyl tensor can be expressed as \eqref{Weyltensor} in terms of the Riemann tensor and Schouten tensor, and the derivatives of the Weyl tensor, which appear in the type B Weyl anomaly in $d > 4$, should be organized into the obstruction tensors. However, as we have previously remarked, the general geometric structure may require techniques beyond cohomology. Note that the holographic Weyl anomaly (also recognized as the Q-curvature) is constrained by the Einstein theory in the bulk; for example, for a $4d$ boundary, we have $a = c$ in \eqref{WA2d4d}. In the general case, to realize the holographic anomaly inflow, the bulk effective theory may need to be deformed to other theories, such as higher curvature theories.

The Weyl anomaly sits at the intersection of three topics explored in this thesis: Weyl geometry, holography and cohomology. We hope that our investigation from these three perspectives can shed light on the fundamental understanding of Nature.

\begin{figure}[!htbp]
\center
\includegraphics[width=2.6in]{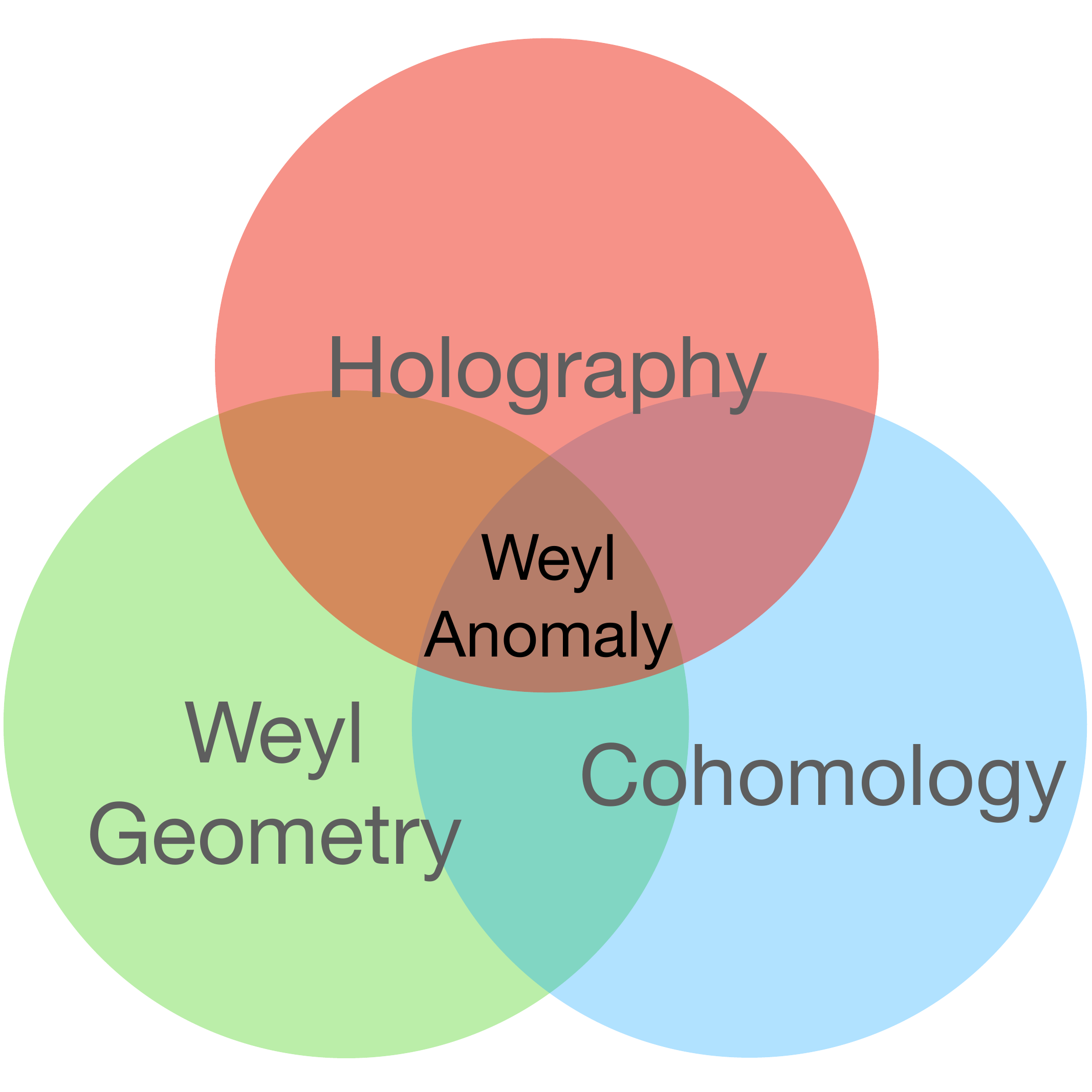}
\caption{The three-legged stool of the Weyl anomaly.}
\label{fig:three-legged}
\end{figure}



\backmatter

\clearpage
\setcounter{counterforappendices}{\value{page}}
\mainmatter
\setcounter{page}{\value{counterforappendices}}

\appendix
\addcontentsline{toc}{part}{Appendices}
\part*{Appendices}

\chapter{Supplement to Part I}
\section{Coordinate Systems of the Flat Ambient Space}\label{app:coords}
In this appendix section we demonstrate the transformation between the flat ambient metric in different coordinate systems introduced in Section \ref{sec:ambient}.
\par
Start with Minkowski spacetime $\mathbb{R}^{1,d+1}$ in Lorentzian coordinates $\{X^{0},X^{i}\}$ with $i=1,\dots , d+1$:
\begin{equation}\label{app:1}
\eta= - (\td X^{0})^{2}+\sum_{i=1}^{d+1}(\td X^{i})^{2}\,.
\end{equation}
First, we can define a stereographic coordinate system $\{\ell,r,x^{i}\}$ as follows:
\begin{equation}\label{Lorentzian_to_stereo}
X^{0}= \ell \frac{L^2+ r^2}{L^2-r^2}\,,\qquad X^{i}= \ell \frac{2L}{L^2-r^2} x^{i},\qquad i=1,\dots, d+1\,,
\end{equation}
where $r^2 = \sum\limits_{i=1}^{d+1} (x^{i})^2 $ and $L$ is a positive constant. In this system, the Minkowski metric \eqref{app:1} becomes
\begin{equation}\label{app:2}
\eta= -\td \ell^2 + \frac{\ell^2}{L^2} \frac{4}{(1- (r/L)^2)^2}\sum_{i=1}^{d+1}(\td x^i)^2 = -\td \ell^2 + \frac{\ell^2}{L^2} \frac{4}{(1- (r/L)^2)^2}\left(\td r^2 + r^2 \td\Omega^{2}_{d}\right)\,,
\end{equation}
where in the second equality we expressed $\{x^{i}\}$ in the spherical coordinates. The coordinate patch is $\ell>0$, $0 \leqslant r<L$, which covers the interior of the future light cone. Notice that in  these coordinates the metric has a ``cone'' form \eqref{Flat_Ambient_4}, with $g^+$ given in \eqref{eq:adsglob}, which is the $(d+1)$-dimensional Euclidean AdS metric $g^+_G$ in global coordinates. This AdS metric can be converted into the FG from by transforming the coordinate $r$ to a coordinate $z$
\begin{equation}\label{stereo_to_FG}
r= L\left(\frac{2L-z}{2L+z}\right)\,.
\end{equation}
Then, the metric \eqref{app:2} takes the form
\begin{equation}
\label{eq:etaFG}
\eta = - \td \ell^2 + \frac{\ell^2}{z^2} \left(\td z^2 + L^2(1- \frac{1}{4}(z/L)^2)^2\td\Omega_{d}^2\right)\,,
\end{equation}
and the interior of the future light cone is now covered by  $\ell>0$, $ 0<z< 2L$. We can further convert \eqref{eq:etaFG} into the ambient form \eqref{ambient_metric} by setting
\begin{equation}\label{FG_to_ambient}
\ell=z t\,,\qquad z^2 =-2\rho\,,
\end{equation}
and the metric turns into the form shown in \eqref{Flat_Ambient_2}:
\begin{equation}\label{app:3}
\eta = 2\rho \td t^2 + 2t \td t \td\rho + t^2 (1+ \frac{\rho}{2 L^2})^2 L^2 \td\Omega_{d}^2 \,.
\end{equation} 
Plugging \eqref{FG_to_ambient} and \eqref{stereo_to_FG} into \eqref{Lorentzian_to_stereo} we find that
\beq
X^{0}+ R= 2L t\,,\qquad\tan\alpha\equiv\frac{R}{X^{0}}= \frac{1+ \frac{\rho}{2L^2}}{1- \frac{\rho}{2L^2}}\,,
\eeq
where $R^2 = \sum\limits_{i=1}^{d+1} (X^{i})^2 $. From the above equation one can see that the constant-$t$ and constant-$\rho$ surfaces are indeed the cones depicted in Figure~\ref{fig:cones}, with  $m$ the angle of the constant-$\rho$ cone with respect to the $X^{0}$-axis.
\par
The Minkowski metric \eqref{app:1} can also be written in the cone form with $g^+=g^+_P$ the Euclidean AdS metric in Poincar\'e coordinates given in \eqref{eq:adsPoin}. Introduce another coordinate system $\{\ell,x^{i},z\}$ as follows:
\begin{equation}
X^{0}= \frac{\ell}{2L z}\left(L^{2}+ \sum_{i=1}^{d}(x^{i})^{2} + z^{2}\right)\,,\quad X^{d+1}= \frac{\ell}{2Lz}\left(L^{2}- \sum_{i=1}^{d}(x^{i})^{2}- z^{2}\right)\,,\quad X^{i}= \frac{\ell x^{i}}{z}\,.
\end{equation}
The metric \eqref{app:1} becomes 
\begin{equation}\label{app:4}
\eta= -\td \ell^2 + \frac{\ell^2}{z^2}\left(\td z^2+ \delta_{ij} \td x^{i}\td x^{j}\right),\qquad i=1,\cdots,d\,,\qquad z>0\,.
\end{equation}
Define the ambient coordinate system $\{t,x^{i},\rho\}$ as
\begin{equation}
\ell=zt\,,\qquad z^2 =-2\rho\,,
\end{equation}
then the metric \eqref{app:4} will have the form shown in \eqref{Flat_Ambient_3}
\begin{equation}\label{app:5}
\eta = 2\rho \td t^2 + 2t\td t \td\rho + t^2\delta_{ij} \td x^{i}\td x^{j}\,,\qquad i=1,\cdots,d\,.
\end{equation}

\section{Details of Null Frame Calculations}
\label{app:Null}
In Section \ref{sec:topdown} we introduced the following frame:
\begin{align}
\label{e+-}
\bm e^+&=\td t+ta_i\td x^i\,,\qquad\bm e^-=t\td\rho+\rho \td t-t\rho a_i\td x^i\,,\qquad \bm e^i=\td x^i\,,\\
\un D_+&=\un\p_t-\frac{\rho}{t}\un\p_\rho\,,\qquad\un D_-=\frac{1}{t}\un\p_\rho\,,\qquad\un D_i=\un\p_i-ta_i\un\p_t+2\rho a_i\un\p_\rho\,.
\end{align}
The metric \eqref{Weyl_ambient} can be written in this frame as
\begin{align*}
\tilde g=\bm e^+\otimes\bm e^-+\bm e^-\otimes\bm e^++t^2\gamma_{ij}\bm e^i\otimes\bm e^j\,,
\end{align*}
and the metric components read
\begin{align*}
\tilde g_{+-}&=\tilde g_{-+}=1\,,\qquad \tilde g_{ij}=t^2\gamma_{ij}\,,\qquad \tilde g^{+-}=\tilde g^{-+}=1\,,\qquad \tilde g^{ij}=\frac{1}{t^2}\gamma^{ij}\,.
\end{align*}
The commutation relations of the frame are as follows:
\begin{align}
\begin{split}
[\un D_+,\un D_i]&=-(a_i-\rho \varphi_i)\un D_+-\rho^2 \varphi_i\un D_-\,,\qquad[\un D_+,\un D_-]=0\,,\\
[\un D_-,\un D_i]&=(a_i+\rho\varphi_i)\un D_--\varphi_i\un D_+\,,\qquad[\un D_i,\un D_j]=-tf_{ij}\un D_++t\rho f_{ij}\un D_-\,,
\end{split}
\end{align}
where $\varphi=\p_\rho a_i$, and $f_{ij}=D_ia_j-D_ja_i$. From the above commutators we can read off the commutation coefficients:
\begin{align}
\begin{split}
C_{+i}{}^+&=-a_i+\rho \varphi_i\,,\qquad C_{+i}{}^-=-\rho^2 \varphi_i\,,\qquad C_{-i}{}^+=-\varphi_i\,,\\
C_{-i}{}^-&=a_i+\rho \varphi_i\,,\qquad C_{ij}{}^+=-tf_{ij}\,,\qquad C_{ij}{}^-=t\rho f_{ij}\,.
\end{split}
\end{align}
Then, we can compute the connection coefficients $\tilde\Gamma^{P}{}_{MN}$ of the ambient LC connection:
\begin{align}
\begin{split}
\tilde\Gamma^{P}{}_{MN}=&\frac{1}{2}\tilde g^{PQ}(D_M\tilde g_{NQ}+D_N\tilde g_{QM}-D_Q\tilde g_{MN})\\
&-\frac{1}{2} \tilde g^{PQ}(C_{MQ}{}^{R}\tilde g_{RN}+C_{NM}{}^{R}\tilde g_{RQ}-C_{QN}{}^{R}\tilde g_{RM})\,.
\end{split}
\end{align}
The nonvanishing components are
\begin{align}
\tilde\Gamma^{+}{}_{i+}&=a_i\,,\qquad
\tilde\Gamma^{+}{}_{ij}=-\frac{t}{2}(\p_\rho\gamma_{ij}+f_{ij})\,,\qquad
\tilde\Gamma^{-}{}_{ij}=-t\gamma_{ij}+\frac{\rho t}{2}(\p_\rho \gamma_{ij}+f_{ij})\,,\nn\\
\tilde\Gamma^{-}{}_{i-}&=-a_i\,,\qquad\tilde\Gamma^{i}{}_{j-}=\frac{1}{2t}\gamma^{ik}(\p_\rho\gamma_{jk}+f_{jk})\,,\qquad\tilde\Gamma^{i}{}_{j+}=\frac{1}{t}\delta^i{}_j-\frac{\rho}{2t}\gamma^{ik}(\p_\rho\gamma_{jk}+f_{jk})\,,\nn\\
\tilde\Gamma^{i}{}_{jk}&=\frac{1}{2}\gamma^{il}(\p_j\gamma_{lk}+\p_k\gamma_{jl}-\p_l\gamma_{jk})-(a_j\delta^i{}_k+a_k\delta^i{}_j-a^i\gamma_{jk})+\rho\gamma^{il}(a_j\p_\rho \gamma_{lk}+a_k\p_\rho \gamma_{jl}-a_l\p_\rho \gamma_{jk})\,,\nn\\
\tilde\Gamma^{+}{}_{+i}&=\rho \varphi_i\,,\qquad\tilde\Gamma^{i}{}_{++}=\frac{\rho^2}{t^2}\gamma^{ij}\varphi_j\,,\qquad\tilde\Gamma^{-}{}_{+i}=-\rho^2\varphi_i\,,\qquad\tilde\Gamma^{i}{}_{+-}=-\frac{\rho}{t^2}\gamma^{ij}\varphi_j\,,\nn\\
\tilde\Gamma^{+}{}_{-i}&=-\varphi_i\,,\qquad\tilde\Gamma^{i}{}_{-+}=-\frac{\rho}{t^2}\gamma^{ij}\varphi_j\,,\qquad\tilde\Gamma^{-}{}_{-i}=\rho \varphi_i\,,\qquad\tilde\Gamma^{i}{}_{--}=\frac{1}{t^2}\gamma^{ij}\varphi_j\,,\nn\\
\tilde\Gamma^{i}{}_{+j}&=\frac{1}{t}\delta^i{}_j-\frac{\rho}{2t}\gamma^{ik}(\p_\rho\gamma_{jk}+f_{jk})\,,\qquad\tilde\Gamma^{i}{}_{-j}=\frac{1}{2t}\gamma^{ik}(\p_\rho\gamma_{jk}+f_{jk})\,,
\end{align}
which constitute the connection 1-form $\tilde{\bm\omega}^{M}{}_{N}$ presented in \eqref{eq:conn1form}. Then, using Cartan's second structure equation
\begin{align}
\tilde{\bm R}^{M}{}_{N}=\td\tilde{\bm\omega}^{M}{}_{N}+\tilde{\bm\omega}^{M}{}_{P}\wedge\tilde{\bm\omega}^{P}{}_{N}\,,
\end{align}
we can find the ambient curvature 2-form, the nonvanishing components are
\begin{align}
\tilde{\bm R}^+{}_i={}&-t(\hat\nabla_j\psi_{ki}-\rho \varphi_if_{jk})\bm e^j\wedge \bm e^k+(\p_\rho\psi_{ji}-\psi_{jk}\psi_i{}^k-\hat\nabla_j\varphi_i-2\rho\varphi_i\varphi_j)\bm e^j\wedge(\bm e^--\rho\bm e^+)\nn\,,\\
\tilde{\bm R}^-{}_i={}&\rho t(\hat\nabla_j\psi_{ki}-\rho \varphi_if_{jk})\bm e^j\wedge \bm e^k-\rho(\p_\rho\psi_{ji}-\psi_{jk}\psi_i{}^k-\hat\nabla_j\varphi_i-2\rho\varphi_i\varphi_j)\bm e^j\wedge(\bm e^--\rho\bm e^+)\nn\,,\\
\tilde{\bm R}^i{}_+={}&-\frac{\rho}{t}(\hat\nabla_j\psi_k{}^i-\rho\varphi^if_{jk})\bm e^j\wedge \bm e^k+\frac{\rho}{t^2}(\p_\rho\psi_j{}^i+\psi_k{}^i\psi_j{}^k-\hat\nabla_j\varphi^i-2\rho\varphi^i\varphi_j) \bm e^j\wedge(\bm e^--\rho\bm e^+)\nn\,,\\
\tilde{\bm R}^i{}_-={}&\frac{1}{t}(\hat\nabla_j\psi_k{}^i-\rho\varphi^if_{jk})\bm e^j\wedge \bm e^k\nn-\frac{1}{t^2}(\p_\rho\psi_j{}^i+\psi_k{}^i\psi_j{}^k-\hat\nabla_j\varphi^i-2\rho\varphi^i\varphi_j)\bm e^j\wedge(\bm e^--\rho\bm e^+)\nn\,,\\
\tilde{\bm R}^i{}_j={}&\frac{1}{2}(\bar R^i{}_{jkl}+\delta^i{}_jf_{kl})\bm e^k\wedge \bm e^l-(\delta_k{}^i\psi_{lj}+\psi_k{}^i\gamma_{lj}-2\rho\psi_k{}^i\psi_{lj}+\rho\psi_j{}^if_{kl})\bm e^k\wedge \bm e^l\nn\\
\label{eq:tildeRcomp}
&+\frac{1}{t}\gamma^{il}(\hat\nabla_l\psi_{jk}-\hat\nabla_j\psi_{lk}+2\rho f_{jl}\varphi_k) \bm e^k\wedge(\bm e^--\rho\bm e^+)\,,
\end{align}
where $\hat\nabla$ is introduced in \eqref{eq:hatnabla}, $\psi_{ij}\equiv\frac{1}{2}(\p_\rho\gamma_{ij}+f_{ij})$, and
\begin{align}
\bar R^{i}{}_{jkl}\equiv D_{k}\ti \Gamma^{i}{}_{lj}-D_{l}\ti \Gamma^{i}{}_{kj}+\ti \Gamma^{i}{}_{km}\ti \Gamma^{m}{}_{lj}-\ti \Gamma^{i}{}_{lm}\ti \Gamma^{m}{}_{kj}\,.
\end{align}
The components in \eqref{eq:tildeRcomp} constitute the curvature 2-form $\tilde{\bm R}^{M}{}_{N}$ presented in \eqref{eq:curv2form}. 
\par
Now one can derive the extended Weyl-obstruction tensors according to Definition \ref{def1}. For example, $\hat\Omega^{(1)}_{ij}$ and $\hat\Omega^{(2)}_{ij}$ can be computed as follows:
\begin{align*}
\tilde R_{-ij-}={}&\p_\rho\gamma_{ij}-\psi_{ik}\psi_j{}^k-\hat\nabla_{(i}\varphi_{j)}-2\rho\varphi_i\varphi_j\,,\\
\nabla_-\tilde R_{-ij-}={}&\frac{1}{t}\Big[\p^2_\rho\gamma_{ij}-2\psi_j{}^k{\cal B}_{ki}-2\psi_{i}{}^{k}{\cal B}_{kj}-\hat\nabla_{(i}(\p_\rho\varphi_{j)})-6\varphi_i\varphi_j+\varphi^k\varphi_k\gamma_{ji}-\psi_i{}^k\hat\nabla_j\varphi_k-\psi_j{}^k\hat\nabla_i\varphi_k\\
&\quad+\varphi^k(\hat\nabla_i\psi_{jk}+2\hat\nabla_j\psi_{ki}-2\hat\nabla_k\psi_{ji}+\hat\nabla_i\psi_{kj}-\hat\nabla_k\psi_{ij})\\
&\quad+2\rho\big(\varphi^k(\varphi_j\psi_{ik}+\varphi_i\psi_{kj}-\varphi_k\psi_{ij})-2\varphi^k\varphi_{(i}\psi_{j)k}-3\p_\rho\varphi_{(i}\varphi_{j)}-2\varphi^k\varphi_{(i}f_{j)k}\big)\Big]\,.
\end{align*}
Plugging the on-shell solution \eqref{Pgf}--\eqref{g6} in to the above expressions, one obtains the extended Weyl-obstruction tensors $\hat\Omega^{(1)}_{ij}$ and $\hat\Omega^{(2)}_{ij}$ given in \eqref{eq:Omega1} and \eqref{eq:Omega2}.
\par
From the components of the ambient Riemann curvature, we can also find the Ricci components in this frame:
\begin{align*}
\tilde R_{++}&=-\rho\tilde  R_{+-}=-\rho\tilde  R_{-+}=\rho^2\tilde R_{--}=-\frac{\rho^2}{t^2}(\gamma^{ij}\p_\rho\psi_{ji}+\psi_k{}^i\psi_i{}^k-\hat\nabla_i\varphi^i-2\rho\varphi^i\varphi_i)\,, \\
\tilde R_{i+}&=\tilde R_{+i}=-\rho \tilde R_{i-}=-\rho \tilde R_{-i}=-\frac{\rho}{t}(\hat\nabla_j\psi_i{}^j-\hat\nabla_i\theta-2\rho\varphi^jf_{ji})\,,\\
\tilde R_{ij}&=\bar R_{ij}+f_{ij}-(d-2)\psi_{ji}-\theta\gamma_{ji}+2\rho({\cal B}_{ij}+\theta\psi_{ji}-\psi_j{}^k\psi_{ki}-\psi_i{}^kf_{kj})\,,
\end{align*}
where ${\cal B}_{ij}$ is defined in \eqref{eq:Bij}. The Ricci-flatness condition gives the following three equations:
\begin{align}
\label{eq:R++0}
0&=\gamma^{ij}\p_\rho\psi_{ji}+\psi_k{}^i\psi_i{}^k-\hat\nabla_i\varphi^i-2\rho\varphi^i\varphi_i \,,\\
\label{eq:R+i0}
0&=\hat\nabla_j\psi_i{}^j-\hat\nabla_i\theta-2\rho\varphi^jf_{ji}\,,\\
\label{eq:Rij0}
0&=\bar R_{ij}+f_{ij}-(d-2)\psi_{ji}-\theta\gamma_{ji}+2\rho({\cal B}_{ij}+\theta\psi_{ji}-\psi_j{}^k\psi_{ki}-\psi_i{}^kf_{kj})\,.
\end{align}
In the leading order when $\rho=0$, the condition \eqref{eq:R++0} leads to the fact that $\hat\Omega^{(1)}_{ij}$ is traceless, and \eqref{eq:Rij0} gives the Bianchi identity $\hat\nabla^{(0)}_i\hat P^i{}_j=\hat\nabla^{(0)}\hat P$, where $\hat P$ is the trace of $\hat P_{ij}$.
\par
Differentiating $\bar R_{ij}$ with respect to $\rho$ yields
\begin{align}
\p_\rho\bar R_{ij}={}&\hat\nabla_k\hat\nabla_j\psi^{k}{}_i+\hat\nabla_k\hat\nabla_i\psi_j{}^k-\hat\nabla_k\hat\nabla^k\psi_{ji}-\hat\nabla_j\hat\nabla_i\theta-\hat\nabla_i\varphi_j+(d-1)\hat\nabla_i\varphi_j+\gamma_{ij}\hat\nabla_k\varphi^k\nn\\
&+4\rho a_k(\varphi_j\psi^k{}_i+\varphi_i\psi^k{}_j-\varphi^k\psi_{ij})-4\rho a_j\varphi_i\theta+2\rho\hat\nabla(\varphi_j\psi^k{}_i+\varphi_i\psi^k{}_j-\varphi^k\psi_{ij})-2\rho\hat\nabla_i\theta\nn\\
&+2\rho\varphi_k(\nabla_j\psi^k{}_i+\nabla_i\psi_j{}^k-\nabla^k\psi_{ji})-2\rho\varphi_j\hat\nabla_i\theta-2\rho\big((d+2)\varphi_i\varphi_j-\varphi_k\varphi^k\gamma_{ij}\big)\nn\\
\label{eq:Rijprho}
&+2\rho \varphi_k(\varphi_j\psi^k{}_i+\varphi_i\psi^k{}_j-\varphi^k\psi_{ij})-2\rho \varphi_j\varphi_i\theta\,,
\end{align}
which leads to \eqref{eq:psipole2} when differentiating \eqref{eq:Rij0}.

\section{Solving the Bulk Einstein Equations}
\label{app:B0}
To solve for $\gamma^{(2k)}_{ij}$ in the expansion \eqref{hex} in the WFG gauge from the Einstein equations, we first introduced the following notations:
\begin{align}
\varphi_i&\equiv D_za_i\,,\qquad f_{ij}\equiv D_i a_j-D_j a_i\,,\qquad\rho_{ij}\equiv\frac{1}{2}D_zh_{ij}\,,\qquad\theta\equiv\tr\rho\,,\nn\\
\label{quantities}
\psi_{ij}&\equiv\rho_{ij}+\frac{L}{2}f_{ij}\,,\qquad
\gamma^k{}_{ij}\equiv\Gamma^k{}_{ij}=\frac{1}{2}h^{kl}( D_i h_{lj} +D_j h_{il} - D_l h_{ji})\,.
\end{align}
Since the integral curves of $\un D_z$ form a congruence, some of these quantities can be interpreted as the properties of this congruence: $\varphi^i$ is the acceleration, $f_{ij}$ is the twist, $\theta$ is the expansion and $\sigma_{ij}\equiv\rho_{ij}-\frac{1}{d}\theta h_{ij}$ is the shear. By plugging in the expansions \eqref{hex} and \eqref{aex}, one can obtain the expansions of the quantities above. A list of these expansions enough for capturing the first two leading orders of the Einstein equations can be found in the Appendix of \cite{Ciambelli:2019bzz}.
\par
Using the connection coefficients $\Gamma^k{}_{ij}$ in the bulk, one can compute the curvature tensors and the Einstein tensor. Then, the vacuum Einstein equations can be written as
\begin{align}
\label{eomzz}
0&=G_{zz}+g_{zz}\Lambda=-\frac{1}{2}\tr(\rho\rho)-\frac{3L^2}{8}\tr(ff)-\frac{1}{2}\bar{R}+\frac{1}{2}\theta^2+\Lambda\,\\
\label{eomzm}
0&=G_{zi}+g_{zi}\Lambda=\nabla_j\psi^j{}_i-D_i\theta+L^2f_{ji}\varphi^j\,,\\
\label{eommn}
0&=G_{ij}+g_{ij}\Lambda=\bar{G}_{ij}-(D_z+\theta)\psi_{ij}-L\nabla_j\varphi_i+2\rho_{jk}\rho^k{}_i+\frac{L^2}{2}f_{jk}f^k{}_i-L^2\varphi_i\varphi_j\nn\\
&\qquad\qquad\qquad\qquad+h_{ij}\left(L\nabla_i\varphi^i+D_z\theta+\frac{1}{2}\tr(\rho\rho)-\frac{L^2}{8}\tr(ff)+L^2\varphi^2+\frac{1}{2}\theta^2+\Lambda\right)\,.
\end{align}
where $\Lambda=-\frac{d(d-1)}{2L^2}$ is the cosmological constant, and $\bar{R}=h^{ij}\bar R_{ij}$ with
\begin{align}
\bar{R}_{ij}=D_k\gamma^k
{}_{ji}-D_j\gamma^k{}_{ki}+\gamma^{k}{}_{kl}\gamma^{l}{}_{ji}-\gamma^{k}{}_{jl}\gamma^{l}{}_{ki}\,.
\end{align}
Denote $m_{(2k)j}^{i}\equiv\gamma_{(0)}^{ik}\gamma^{(2k)}_{kj}$ and $n_{(2k)j}^{i}\equiv\gamma_{(0)}^{ik}\pi^{(2k)}_{kj}$. Expanding \eqref{eomzz}--\eqref{eommn} using \eqref{hex} and \eqref{aex}, one can solve the Einstein equations order by order. First, the $zz$-component of the Einstein equations gives
\begin{align}
0={}&\bigg[\frac{d(d-1)}{2L^2}+\Lambda\bigg]-\frac{z^2}{L^2}\bigg[\frac{R^{(0)}}{2}+\frac{d-1}{L^2}X^{(1)}\bigg]+\frac{z^4}{L^4}\bigg[\frac{d}{2L^2}(X^{(1)})^2-\frac{2(d-1)}{L^2}X^{(2)}-\frac{1}{2L^2}\tr(m_{(2)}^2)\nn\\
&-\frac{3L^2}{8}\tr(f_{(0)}^2)-\frac{1}{2}\Big(\gamma_{(0)}^{kj}\hat\nabla^{(0)}_k\hat\nabla_i \big(m_{(2)}{}^{i}{}_{j}-\tr (m_{(2)})\delta^i{}_j\big)
+2(d-1)\hat\nabla\cdot a^{(2)}
-\tr\big(m_{(2)}\gamma_{(0)}^{-1}R^{(0)}\big)\Big)\bigg]\nn\\
&+\cdots-\frac{z^d}{L^d}(d-1)\bigg[\frac{d}{2L^2}Y^{(1)}+\hat\nabla\cdot p_{(0)}\bigg]+\cdots\,,
\end{align}
where $X^{(1)}$, $X^{(2)}$ and $Y^{(1)}$ are given in expansion \eqref{sqrth}, which can be expressed in terms of the expansion of $h_{ij}$ as
\begin{align}
\label{XY}
X^{(1)}&=\tr(m_{(2)})\,,\quad X^{(2)}=\tr(m_{(4)})-\frac{1}{2}\tr(m_{(2)}^2)+\frac{1}{4}\left(\tr(m_{(2)})\right)^2\,,\,\cdots\,,Y^{(1)}=\tr(n_{(0)})\,,\,\cdots\,.
\end{align}
At the $O(1)$-order, the $zz$-equation is trivially satisfied, and at the $O(z^2)$-order, we can find that
\begin{align}
X^{(1)}=-\frac{L^2}{2(d-1)} R^{(0)}=-L^2\hat P\,.
\end{align}
Then, using the above result we can obtain from the $O(z^4)$-order that
\begin{align}
X^{(2)}&=-\frac{1}{4}\tr(m_{(2)}^2)+\frac{1}{4}(X^{(1)})^2-\frac{L^2}{2}\hat\nabla\cdot a^{(2)}-\frac{L^4}{16}\tr(f^{(0)}f^{(0)})\nn\\
&=-\frac{L^4}{4}\tr(\hat P^2)+\frac{L^4}{4}\hat P^2-\frac{L^2}{2}\hat\nabla\cdot a^{(2)}\,,
\end{align}
where we used \eqref{Pgf0}. Also notice that the $O(z^d)$-order gives the Weyl-Ward identity
\begin{align}
0=\frac{d}{2L^2}Y^{(1)}+\hat\nabla\cdot p_{(0)}\,.
\end{align}
\par
Now we look at the $ij$-components of the Einstein equations:
\begin{align}
0={}&\bigg[G^{(0)}_{ij}+\frac{d}{2}f^{(0)}_{ij}-\frac{d-2}{L^2}X^{(1)}\gamma^{(0)}_{ij}+\frac{d-2}{L^2}\gamma^{(2)}_{ij}\bigg]+\frac{z^2}{L^2}\bigg[\frac{1}{2}\hat\nabla_k\Big(\gamma_{(0)}^{kl}\Big(\hat\nabla_j \gamma^{(2)}_{li}+\hat\nabla_i \gamma^{(2)}_{lj}-\hat\nabla_l \gamma^{(2)}_{ij}\Big)\Big)\nn\\
&-\frac{1}{2}\gamma^{(0)}_{ij}\hat\nabla_i\hat\nabla_j \Big(\gamma_{(2)}^{ij}-X^{(1)}\gamma_{(0)}^{ij}\Big)-\frac{1}{2}\hat\nabla_{(i}\hat\nabla_{j)} X^{(1)}+(d-4)\big(\hat\nabla^{(0)}_{(i} a_{j)}^{(2)}-\gamma_{ij}^{(0)}\hat\nabla\cdot a^{(2)})\nn\\
\label{emn}
&+\frac{2(d-4)}{L^{2}}\gamma^{(4)}_{ij}+\frac{2}{L^{2}}m_{(2)}^k{}_i\gamma^{(2)}_{kj}+\frac{L^2}{2}f^{(0)}_{jk}f^{(0)}_{li}\gamma_{(0)}^{lk}+\Big(\frac{1}{2}\tr(m_{(2)}\gamma_{(0)}^{-1}{R}^{(0)})-\frac{L^2}{8}\tr(f^{(0)}f^{(0)})\nn\\
&-\frac{2(d-4)}{L^2}X^{(2)}+\frac{d-3}{2L^2}(X^{(1)})^2+\frac{1}{2L^2}\tr(m_ {(2)}^2)\Big)\gamma^{(0)}_{ij}\bigg]+\cdots\,.
\end{align}
Note that $\gamma_{(2)}^{ij}\equiv(\gamma_{(0)}^{-1}\gamma^{(2)}\gamma_{(0)}^{-1})^{ij}$ is not the inverse of $\gamma^{(2)}_{ij}$. Plugging in the results we got from the $zz$-equation, we obtain from the first two leading orders of \eqref{emn} that 
\begin{align}
\gamma^{(2)}_{ij}={}&-\frac{L^2}{d-2}\bigg(R^{(0)}_{(ij)}-\frac{1}{2(d-1)}R^{(0)}\gamma^{(0)}_{ij}\bigg)\,,\\
\gamma^{(4)}_{ij}={}&-\frac{L^2}{4(d-4)}\bigg(2\hat\nabla_k\hat\nabla_{(i} m_{(2)}{}^{k}{}_{j)}
-\hat\nabla\cdot\hat\nabla \gamma^{(2)}_{ij}-\hat\nabla_{(i}\hat\nabla_{j)} X^{(1)}-\frac{1}{L^2}\gamma^{(0)}_{ij}\tr(m_{(2)}^2)+\frac{4}{L^2}m_{(2)}^k{}_i\gamma^{(2)}_{kj}\nn\\
&\qquad\qquad\qquad+L^2f^{(0)}_{jk}f^{(0)}_{li}\gamma_{(0)}^{lk}-\frac{L^2}{4}\tr(f^{(0)}f^{(0)})\gamma^{(0)}_{ij}\bigg)-\frac{L^2}{2}\hat\nabla^{(0)}_{(i} a_{j)}^{(2)}\,.
\end{align}
\par
Furthermore, expanding \eqref{emn} to the $O(z^4)$-order one obtains
\begin{align}
\gamma^{(6)}_{ij}=&-\frac{L^2}{3(d-6)}\bigg[\hat\nabla_k\hat\gamma^k_{(4)ij}-\frac{1}{2}\hat\nabla_{(i}\hat\nabla_{j)} \tr(m_{(4)})-\hat\nabla_k(\hat\gamma^l_{(2)ij}{m}_{(2)}^k{}_l)+\hat\nabla_{(j}(\hat\gamma^l_{(2)i)k}{m}_{(2)}^k{}_l)\\
&+\frac{1}{2}\hat\nabla_l X^{(1)}\hat\gamma^l_{(2)ij}-\hat\gamma^l_{(2)ik}\hat\gamma^k_{(2)lj}-\frac{2}{L^2}(m_{(2)}^3)^k{}_j\gamma^{(0)}_{ik}+\frac{8}{L^2}\gamma^{(4)}_{k(i}m_{(2)}^k{}_{j)}-\frac{1}{L^2}\gamma^{(4)}_{ij}X^{(1)}\nn\\
&-\frac{L^2}{2}f^{(0)}_{li}f^{(0)}_{jk}\gamma_{(2)}^{kl}+L^2f^{(2)}_{l(i}f^{(0)}_{j)k}\gamma_{(0)}^{kl}-\frac{1}{L^2}\gamma^{(0)}_{ij}\Big(\tr(m_{(4)}m_{(2)})-\frac{1}{2}\tr(m_{(2)}^3)-\frac{L^4}{8}\tr(m_{(2)} f_{(0)}^2)\nn\\
&-\frac{L^4}{4}\hat\nabla_k a^{(2)}_l f_{(0)}^{kl}-\frac{L^2}{4}\hat\nabla_l X^{(1)} a_{(2)}^l+\frac{L^2}{2}\hat\nabla_k(\gamma_{(2)}^{kl} a^{(2)}_l)\Big)+2\hat\nabla_k({m}_{(2)}^k{}_{(j} a^{(2)}_{i)})-2\gamma^{(0)}_{l(j}\hat\gamma^l_{(2)i)k}a_{(2)}^k\nn\\
&-a^{(2)}_{(j}\hat\nabla_{i)} X^{(1)}-\hat\nabla_{(j}(X^{(1)} a^{(2)}_{i)})\bigg]-\frac{L^2}{3}\hat\nabla_{(i}a^{(4)}_{j)}-L^2a^{(2)}_{i}a^{(2)}_{j}+\frac{L^2}{6}a^{(2)}\cdot a^{(2)}\gamma^{(0)}_{ij}+\frac{L^2}{3}\hat\gamma^k_{(2)ij}a^{(2)}_{k}\nn\,,
\end{align}
where $f^{(2)}_{ij}\equiv\hat\nabla_i a^{(2)}_j-\hat\nabla_j a^{(2)}_i$, and 
\begin{align}
\hat\gamma^k_{(2)ij}&=\frac{1}{2}\gamma_{(0)}^{kl}(\hat\nabla^{(0)}_i\gamma^{(2)}_{jl}+\hat\nabla^{(0)}_j\gamma^{(2)}_{il}-\hat\nabla^{(0)}_l\gamma^{(2)}_{ij})=-\frac{L^2}{2}(\hat\nabla^{(0)}_i \hat P^k{}_{j}+\hat\nabla^{(0)}_j\hat P_i{}^k-\hat\nabla_{(0)}^k\hat P_{ij})\,.
\end{align} 
(In the second step we used $\hat\nabla^{(0)}_i f^{(0)}_{jk}+\hat\nabla^{(0)}_j f^{(0)}_{ki}+\hat\nabla^{(0)}_k f^{(0)}_{ij}=0$.) The $\gamma^{(4)}_{ij}$ and $\gamma^{(6)}_{ij}$ above can be organized in to \eqref{g40} and \eqref{g60}, respectively.
\par
Finally, the $zi$-component of the Einstein equations gives
\begin{align}
0=&-\frac{L}{d-2}\frac{z^2}{L^2}\gamma_{(0)}^{mn}\hat\nabla_m^{(0)}\hat G^{(0)}_{ni}+L^{-1}\frac{z^4}{L^4}\bigg[\hat\nabla_m\big (2m_{(4)i}^m-(m_{(2)}^2)^m{}_i\big)+\frac{1}{2}m_{(2)i}^m\hat\nabla_m X^{(1)}\nn\\
&+\frac{L^2}{2}\bigg(\hat\nabla\cdot\hat\nabla a_i^{(2)}-\hat\nabla_i\hat\nabla\cdot a^{(2)}+(R^{(0)}_{ni}+4f^{(0)}_{ni})\gamma_{(0)}^{mn}a^{(2)}_m-\hat\nabla_m\big(f^{(0)}_{ni}m_{(2)k}^m\gamma_{(0)}^{kn}\big)\nn\\
&-f^{(0)}_{jk}\gamma_{(0)}^{mj}\hat\nabla_m m_{(2)}^k{}_i+\frac{1}{2}f^{(0)}_{ni}\gamma_{(0)}^{mn}\hat\nabla_m X^{(1)}\bigg)-2\hat\nabla_i X^{(2)}+\frac{1}{2}\hat\nabla_i (X^{(1)})^2-\frac{1}{4}\hat\nabla_i\tr(m_{(2)}^2)
\bigg]+\cdots\nn\\
\label{emz}
&+\frac{z^d}{L^d}\bigg[\frac{d}{2L}\hat\nabla_m n_{(0)i}^m+\frac{L}{2}(\hat\nabla\cdot\hat\nabla p_i^{(0)}+\hat\nabla_m\hat\nabla_i p_{(0)}^m)\bigg]+\cdots\,.
\end{align}
One can observe that the $O(z^2)$-order of the above equation is exactly the contraction of the Weyl-Bianchi identity as shown in \eqref{BI}. By plugging in the results we got from the $zz$-equation, the $O(z^4)$-order can be organized into the identity \eqref{divWB}, which demonstrates the divergence of the Bach tensor. Also, the $O(z^d)$-order gives the conservation law of the improved energy-momentum tensor defined in \eqref{improvedT}.

\section{Expansions of the Raychaudhuri Equation and $\sqrt{-\det h}$}
\label{app:expansion}
Using the components of the Einstein equations \eqref{eomzz}--\eqref{eommn}, one can construct the following equation \cite{Ciambelli:2019bzz}:
\begin{align}
\label{eqa}
0&=\frac{g^{MN}(G_{MN}+\Lambda g_{MN})}{d-1}+(G_{zz}+\Lambda g_{zz})\nn\\
&=D_z\theta+L\nabla_j\varphi^j+L^2\varphi^2+\tr(\rho\rho)+\frac{L^2}{4}\tr(ff)-\frac{d}{L^2}\,,
\end{align}
where the indices $M,N$ represent the bulk components as $M=(z,i)$. This equation can be recognized as the Raychaudhuri equation of the congruence generated by $\un D_z$. Expanding each term in the above equation, we can write down a general expansion of this equation to any order. This combination of the components of the Einstein equations contains all the information we need for deriving $X^{(k)}$. We here provide some details of deriving $X^{(3)}$ and $X^{(4)}$ by means of the Raychaudhuri equation. 
\par
Recall that we have the expansion \eqref{hinv} of the inverse of $h_{ij}$:
\begin{align}
h^{ij}(z;x)&=\frac{z^2}{L^2}\left[\gamma_{(0)}^{ij}(x)+\frac{z^2}{L^2}\gamma_{(2)}^{ij}(x)+...\right]+\frac{z^{d+2}}{L^{d+2}}\left[\pi_{(0)}^{ij}(x)+\frac{z^2}{L^2}\pi_{(2)}^{ij}(x)+...\right]\\
&=\frac{z^2}{L^2}\left[\gamma_{(0)}^{ij}(x)-\frac{z^2}{L^2}\tilde{m}_{(2)k}^{i}\gamma_{(0)}^{kj}(x)-\frac{z^4}{L^4}\tilde{m}_{(4)k}^{i}\gamma_{(0)}^{kj}(x)+\cdots\right]\nn+\frac{z^{d+2}}{L^{d+2}}\left[\tilde{n}_{(2)k}^{i}\gamma_{(0)}^{kj}(x)+\cdots\right]\,,
\end{align}
where $\tilde{m}_{(2k)j}^{i}\equiv-\gamma_{(2k)}^{ik}\gamma^{(0)}_{kj}$, $\tilde{n}_{(2k)j}^{i}\equiv-\pi_{(2k)}^{ik}\gamma^{(0)}_{kj}$. By taking the inverse of the metric, one finds the following relation:
\begin{align}
m_{(2p)}-\tilde{m}_{(2p)}=\sum^{p-1}_{k=1}\tilde{m}_{(2k)}m_{(2p-2k)}\,.
\end{align}
Specifically, we have
\begin{align}
m_{(2)}-\tilde{m}_{(2)}=0\,,\qquad m_{(4)}-\tilde{m}_{(4)}=m^2_{(2)}\,,\qquad m_{(6)}-\tilde{m}_{(6)}=m_{(2)}m_{(4)}+\tilde m_{(4)}m_{(2)}\,.
\end{align}
Now we expand the quantities defined in \eqref{quantities} to an arbitrary order by plugging the expansions \eqref{hex}, \eqref{aex} and \eqref{hinv} into their definitions. For the purpose of finding the Weyl anomaly, here we only keep the $m_{(2p)}$ and $a_{(2p)}$ terms in the first series of $h_{ij}$ and $a_i$ and neglect the $n_{(2p)}$ and $p_{(2p)}$ terms. The expansions of these quantities are

\begin{align}
\label{exp1}
\rho^i{}_j=&-\delta^i{}_j+\frac{1}{2}\sum_{p=1}^\infty\left(\frac{z}{L}\right)^{2p}\Big[p(m_{(2p)}+\tilde{m}_{(2p)})+\sum_{k=1}^{p-1}(2k-p)\tilde{m}_{(2k)}m_{(2p-2k)}\Big]^i{}_j+O(z^d)\,,\\
\theta=&-\frac{d}{L}+\frac{1}{2L}\sum^{\infty}_{p=1}\bigg(\frac{z}{L}\bigg)^{2p}\bigg[p\tr(m_{(2p)}+\tilde{m}_{(2p)})+\sum^{p-1}_{k=1}(2k-p)\tr\tilde{m}_{(2k)}m_{(2p -2k)}\bigg]+O(z^d)\,,
\end{align}
\begin{align}
\varphi_i={}&\frac{1}{L}\sum^\infty_{p=0}\left(\frac{z}{L}\right)^{2p}2pa_i^{(2p)}+O(z^{d-2})\,,\\
f_{ij}=&\sum_{p=0}^{\infty}\left(\frac{z}{L}\right)^{2p}\big[f^{(2p)}_{ij}+\sum_{q=1}^{p-1} 2q(a_i^{(2p-2q)} p_j^{(2q)}-a_j^{(2p-2q)} p_i^{(2q)})\big]+O(z^{d-2})\,,\\
\gamma^k{}_{ij}={}&\gamma^k_{(0)ij}-\sum_{p=1}^\infty\left(\frac{z}{L}\right)^{2p}\bigg(\sum_{q=0}^{p-1}\tilde{m}_{(2q)}^k{}_l\hat\gamma^l_{(2p-2q)ij}+\frac{1}{2}\sum_{q=0}^{p-1}[\tilde{m}_{(2q)}\gamma_{(0)}^{-1}]^{kl}\sum_{k=0}^{p-q-1}(2k-2)\qquad\nn\\
\label{exp5}
&\times(a^{(2p-2q-2k)}_i\gamma^{(2k)}_{jl}+a^{(2p-2q-2k)}_j\gamma^{(2k)}_{il}-a^{(2p-2q-2k)}_l\gamma^{(2k)}_{ij})\bigg)+O(z^{d-2})\,,
\end{align}
where 
\begin{align*}
f^{(0)}_{ij}&=\p_i a_j^{(0)}-\p_j a_i^{(0)}\,,\qquad f^{(2k)}_{ij}=\hat\nabla^{(0)}_i a_j^{(2k)}-\hat\nabla^{(0)}_j a_i^{(2k)}\quad(k>0)\,,\\
\gamma^k_{(0)ij}&=\frac{1}{2}\gamma_{(0)}^{kl}\big(
 \p_i \gamma^{(0)}_{jl}
 +\p_j \gamma^{(0)}_{il}-\p_l \gamma^{(0)}_{ij}\big)-\big(a^{(0)}_i\delta^k{}_j+a^{(0)}_j\delta^k{}_i-a^{(0)}_l\gamma_{(0)}^{kl}\gamma^{(0)}_{ij}\big)\,,\\
\hat\gamma^k_{(2k)ij}&=\frac{1}{2}\gamma_{(0)}^{kl}(\hat\nabla^{(0)}_i\gamma^{(2k)}_{jl}+\hat\nabla^{(0)}_j\gamma^{(2k)}_{il}-\hat\nabla^{(0)}_l\gamma^{(2k)}_{ij})\quad(k>0)\,.
\end{align*} 
Expanding everything in \eqref{eqa} using \eqref{exp1}--\eqref{exp5}, we obtain the following equation:
\begin{align}
0={}&\frac{1}{L^2}p(p-1)\tr(m_{(2p)}+\tilde{m}_{(2p)})+\frac{1}{L^2}\sum^{p-1}_{q=1}(p-1)(2q-p)\tr\tilde{m}_{(2q)}m_{(2p-2q)}\nn\\
&-\sum^{p-1}_{q=1}2q\hat\nabla_i a_j^{(2q)}\big[\tilde{m}_{(2p-2q-2)}\gamma_{(0)}^{-1}\big]^{ij}-\sum^{p-1}_{q=1}\sum^{q-1}_{k=0}(2p-2q+2k)2ka_i^{(2p-2q)}a_j^{(2k)}\big[\tilde{m}_{(2q-2k-2))}\gamma_{(0)}^{-1}\big]^{ij}\nn\\
&-\sum^{p-1}_{q=1}\sum_{k=0}^{q-1}\sum_{n=0}^{p-q-1}na_k^{(2n)}[\tilde{m}_{(2p-2q-2n-2)}\gamma^{-1}_{(0)}]^{ij}\bigg(\tilde{m}_{(2k)}^k{}_l\hat\gamma^l_{(2q-2k)ij}\nn\\
&\quad-[\tilde{m}_{(2k)}\gamma_{(0)}^{-1}]^{kl}\sum_{m=0}^{q-k-1}(2-2m)(a^{(2q-2k-2m)}_i\gamma^{(2m)}_{jl}+a^{(2q-2k-2m)}_j\gamma^{(2m)}_{il}-a^{(2q-2k-2m)}_l\gamma^{(2m)}_{ij})\bigg)\nn\\
&+\frac{1}{4L^2}\sum_{q=1}^{p-1}(p-q)\tr\bigg[(m_{(2p-2q)}+\tilde{m}_{(2p-2q)})\Big[q(m_{(2q)}+\tilde{m}_{(2q)})+\sum_{k=1}^{q-1}2(2k-q)\tilde{m}_{(2k)}m_{(2q-2k)}\Big]\bigg]\nn\\
&+\frac{1}{4L^2}\sum_{q=1}^{p-1}\sum_{k=1}^{q-1}\sum_{m=1}^{p-q-1}(2k-q)(2m-p+q)\tr\big[\tilde{m}_{(2k)}m_{(2q-2k)}\tilde{m}_{(2m)}m_{(2p-2q-2m)}\big]\nn\\
&+\frac{L^2}{4}\sum_{q=1}^{p-1}\sum_{k=0}^{q-1}\big[f^{(2k)}_{il}+\sum_{m=1}^{k-1} 2m(a_i^{(2k-2m)} a_l^{(2m)}-a_l^{(2k-2m)} a_i^{(2m)})\big][\tilde{m}_{(2q-2k-2)}\gamma^{-1}_{(0)}]^{lj}\nn\\
\label{Rayex}
&\quad\times\sum_{n=0}^{p-q-1}\big[f^{(2n)}_{jl}+\sum_{s=1}^{n-1} 2s(a_j^{(2n-2s)} a_l^{(2s)}-a_l^{(2n-2s)} a_j^{(2s)})\big][\tilde{m}_{(2p-2q-2n-2)}\gamma^{-1}_{(0)}]^{li}\,.
\end{align}
From this equation, one can find $\tr(m_{(2p)}+\tilde{m}_{(2p)})$ in terms of $m_{(2q)}$ and $\tilde{m}_{(2q)}$ for all $q<p$. 
\par
Taking $p=3$ we get the Raychaudhuri equation at the $O(z^6)$-order:
\begin{align}
0={}&\frac{6}{L^2}\tr(m_{(6)}+\tilde{m}_{(6)})+\frac{4}{L^2}\tr(m_{(4)}m_{(2)})-\frac{4}{L^2}\tr(m_{(2)}^3)-\frac{L^2}{2}m_{(2)}^i{}_m f^m_{(0)n} f^n_{(0)i}\nn\\
\label{Ray6d}
&+4\hat\nabla\cdot a^{(4)}-2m_{(2)}^i{}_k\gamma_{(0)}^{kj}\hat\nabla_j a_i^{(2)}-2\gamma_{(0)}^{ij}\hat\gamma^{k}_{(2)ij}a^{(2)}_k-2(d-6)a^2_{(2)}+\frac{L^2}{2}f^{(2)}_{ij}f_{(0)}^{ji}\,.
\end{align}
And for $p=4$, we have the Raychaudhuri equation at the $O(z^8)$-order: 
\begin{align}
0={}&\frac{12}{L^2}\tr(m_{(8)}+\tilde{m}_{(8)})+\frac{6}{L^2}\tr(m_{(6)}m_{(2)})-\frac{22}{L^2}\tr(m_{(4)}m_{(2)}^2)+\frac{9}{L^2}\tr(m_{(2)}^4)+\frac{4}{L^2}\tr(m_{(4)}^2)\nn\\
&+\frac{L^2}{4}f^{(0)}_{ik}f_{(0)}^{jl}m_{(2)}^k{}_j m_{(2)}^i{}_l+\frac{L^2}{2}f^{(0)}_{ik}f_{(0)}^{kl}(m_{(2)}^2)^i{}_l-\frac{L^2}{2}f^{(0)}_{ik}f_{(0)}^{kl}(m_{(4)})^i{}_l+6\hat\nabla\cdot a^{(6)}\nn\\
&-4\hat\nabla_i a_j^{(4)}\gamma_{(2)}^{ij}+L^2\hat\nabla_{[i}a^{(4)}_{k]}f_{(0)}^{ki}-4a_l^{(4)}\gamma_{(0)}^{ij}\hat\gamma_{(2)}^l{}_{ij}-6(d-8)a^{(4)}\cdot a^{(2)}-2\hat\nabla_i a_j^{(2)}\gamma_{(4)}^{ij}\nn\\
&-2a_l^{(2)}\gamma_{(0)}^{ij}\hat\gamma_{(4)}^l{}_{ij}+2\hat\nabla_i a_j^{(2)}(m_{(2)}^2)^i{}_k\gamma_{(0)}^{kj}+L^2\hat\nabla_{[i}a^{(2)}_{k]}\hat\nabla^{[k}a_{(2)}^{i]}-2L^2\hat\nabla_{[i}a^{(2)}_{k]}f_{(0)}^{kl}m_{(2)}^i{}_l\nn\\
\label{Ray8d}
&+2a_l^{(2)}\gamma_{(2)}^{ij}\hat\gamma_{(2)}^l{}_{ij}+2a_k^{(2)}\gamma_{(0)}^{ij}m_{(2)}^k{}_l\hat\gamma_{(2)}^l{}_{ij}+2(d-8)a^{(2)}_i a^{(2)}_j\gamma_{(2)}^{ij}+2X^{(1)}a^{(2)}\cdot a^{(2)}\,.
\end{align}
\par
Now let us look at the expansion of $\sqrt{-\det h}$. Using the fact that $\theta=D_z(\ln \sqrt{-\det h})$, we can write down the expansion of $\sqrt{-\det h}$ to any order as
\begin{align}
\label{deth}
\sqrt{-\det h}={}&\sqrt{-\det\gamma_{(0)}}\bigg(\frac{z}{L}\bigg)^{-d}\sum_0^\infty\frac{1}{n!}\\
&\times\left[\frac{1}{2}\sum^{\infty}_{m=1}\bigg(\frac{z}{L}\bigg)^{2m}\bigg[\frac{1}{2}\tr(m_{(2m)}+\tilde{m}_{(2m)})+\sum^{m-1}_{k=1}\bigg(\frac{k}{m}-\frac{1}{2}\bigg)\tr(\tilde{m}_{(2k)}m_{(2m-2k)})\bigg]\right]^n\nn\,.
\end{align}
Comparing with \eqref{sqrth}, at the $O(z^6)$-order and the $O(z^8)$-order, the above equation gives respectively
\begin{align}
\label{X3}
X^{(3)}={}&\frac{1}{2}\tr(m_{(6)}+\tilde{m}_{(6)})-\frac{1}{6}\tr (m_{(2)}^3)+\frac{1}{2}X^{(1)}X^{(2)}-\frac{1}{12}(X^{(1)})^3\,,\\
X^{(4)}={}&\frac{1}{2}\tr(m_{(8)}+\tilde{m}_{(8)})-\frac{1}{2}\tr(m_{(4)}m_{(2)}^2)+\frac{1}{4}\tr(m_{(2)}^4)\nn\\
\label{X4}
&+\frac{1}{2}X^{(3)}X^{(1)}-\frac{1}{4}X^{(2)}(X^{(1)})^2+\frac{1}{4}(X^{(2)})^2+\frac{1}{32}(X^{(1)})^{4}\,.
\end{align}
Now solving for $\tr(m_{(6)}+\tilde{m}_{(6)})$ from \eqref{Ray6d} and plugging \eqref{g2}, \eqref{g40} and \eqref{X1X2} into \eqref{X3}, we can organize all the $m_{(2)}$ and $f_{(0)}$ terms in $X^{(3)}$ and get \eqref{X3P}. Similarly, plugging $\tr(m_{(8)}+\tilde{m}_{(8)})$ obtained from \eqref{Ray8d} into \eqref{X4}, the expression for $X^{(4)}$ can be organized in terms of the Weyl-Schouten tensor and extended Weyl-obstruction tensors as
\begin{align}
\frac{24}{L^2}X^{(4)}={}&L^6\bigg(\frac{1}{8}\hat P^4-\frac{3}{4}\tr(\hat P^2)\hat P^2+\frac{3}{8}[\tr(\hat P^2)]^2+\tr(\hat P^3)\hat P-\frac{3}{4}\tr(\hat P^4)-\tr(\hat\Omega_{(1)}\hat P)\hat P+\tr(\hat\Omega_{(1)}\hat P^2)\nn\\
&-\frac{1}{4}\tr(\hat\Omega_{(1)}^2)-\frac{1}{4}\tr(\hat\Omega_{(2)}\hat P)\bigg)+2(d-8)\big[3a^{(4)}\cdot a^{(2)}+a^{(2)}_i a^{(2)}_j(\hat P^{ij}-\hat P\gamma_{(0)}^{ij})\big]-6\hat\nabla\cdot a^{(6)}\nn\\
&-L^2\hat\nabla_i \big[a_j^{(4)}(4\hat P^{ij}+2\hat P^{ji}-4\hat P\gamma_{(0)}^{ij})\big]-\frac{L^2}{2}\hat\nabla_i \big[a^{(2)}_{j}(3\hat\nabla^{j}a_{(2)}^{i}+\hat\nabla^{i} a^{j}_{(2)}-3\hat\nabla\cdot a_{(2)}\gamma_{(0)}^{ij})\big]\nn\\
&+L^4\hat\nabla_i\big[a^{(2)}_j(3\hat P^{ij}\hat P+\hat P^{ji}\hat P)\big]+\frac{3L^4}{2}\hat\nabla^i\big[a^{(2)}_i(\tr(\hat P^2)-\hat P^2) \big]-\frac{3L^4}{2}\hat\nabla_i (a_j^{(2)}\hat\Omega_{(1)}^{ij})\nn\\
\label{X8t}
&-\frac{L^4}{4}\hat\nabla_{i}\big[a_{j}^{(2)}(3\hat P^{ki}\hat P^j{}_k-5\hat P^{ki}\hat P_k{}^j+7\hat P^{ik}\hat P_k{}^j-9\hat P^{ik}\hat P^j{}_k)\big]\,,
\end{align}
which leads to \eqref{X4P}.

\section{Proof of Lemma \ref{lem:Riem}}
\label{app:Lemproof}
\begin{proof}[Proof of Lemma \ref{lem:Riem}]
We will prove this identity by induction. First, noticing that $\tilde R_{-+MN}=0$, when $n=0$ we have
\begin{align*}
\tilde\nabla_i\tilde R_{-+MN}&=-\tilde\Gamma^j{}_{i-}\tilde R_{j+MN}-\tilde\Gamma^{j}{}_{i+}\tilde R_{-jMN}=\frac{1}{t}\psi_i{}^j \tilde R_{+jMN}-\frac{1}{t}(\delta^j{}_i-\rho\psi^j{}_i )\tilde R_{-jMN}\\
&=-\frac{\rho}{t}\psi_i{}^j \tilde R_{-jMN}-\frac{1}{t}(\delta^j{}_i-\rho\psi_i{}^j)\tilde R_{-jMN}=-\frac{1}{t}\tilde R_{-iMN}\,,\\
\tilde\nabla_-\tilde R_{-+MN}&=-\tilde\Gamma^j{}_{--}\tilde R_{j+MN}-\tilde\Gamma^{j}{}_{-+}\tilde R_{-jMN}=0\,,\\
\tilde\nabla_+\tilde R_{-+MN}&=-\tilde\Gamma^j{}_{+-}\tilde R_{j+MN}-\tilde\Gamma^{j}{}_{++}\tilde R_{-jMN}=0\,,
\end{align*}
where we used the fact that $\tilde\Gamma^i{}_{M+}=-\rho\tilde\Gamma^i{}_{M-}$ and $\tilde R_{+jMN}=-\rho\tilde R_{-jMN}$, which can be seen from
\eqref{eq:conn1form} and \eqref{eq:curv2form}, respectively. Thus, for $n=0$ we have $\nabla_{P}\tilde R_{-+MN}=-\frac{1}{t}\delta^{i}{}_{P}\tilde R_{-iMN}$. Assuming that this lemma holds for all $n\leqslant k-1$, now we show that it will hold for $n=k>0$:
\begin{align*}
&\tilde\nabla_i\underbrace{\tilde\nabla_-\cdots\tilde\nabla_-}_{k}\tilde R_{-+MN}\\
={}&D_i\underbrace{\tilde\nabla_-\cdots\tilde\nabla_-}_{k-1}\tilde R_{-+MN}-\tilde\Gamma^j{}_{i-}\tilde\nabla_j\underbrace{\tilde\nabla_-\cdots\tilde\nabla_-}_{k-1}\tilde R_{-+MN}-\cdots-\tilde\Gamma^j{}_{i-}\underbrace{\tilde\nabla_-\cdots\tilde\nabla_-}_{k-1}\tilde\nabla_j\tilde R_{-+MN}
\\
&-\tilde\Gamma^+{}_{i-}\tilde\nabla_+\underbrace{\tilde\nabla_-\cdots\tilde\nabla_-}_{k-1}\tilde R_{-+MN}-\cdots-\tilde\Gamma^+{}_{i-}\underbrace{\tilde\nabla_-\cdots\tilde\nabla_-}_{k-1}\tilde\nabla_+\tilde R_{-+MN}
\\
&-\tilde\Gamma^j{}_{i-}\underbrace{\tilde\nabla_-\cdots\tilde\nabla_-}_{k}\tilde R_{j+MN}-\tilde\Gamma^j{}_{i+}\underbrace{\tilde\nabla_-\cdots\tilde\nabla_-}_{k}\tilde R_{-jMN}\\
&-\tilde\Gamma^{P}{}_{i M}\underbrace{\tilde\nabla_-\cdots\tilde\nabla_-}_{k}\tilde R_{-+ P N}-\tilde\Gamma^{P}{}_{i N}\underbrace{\tilde\nabla_-\cdots\tilde\nabla_-}_{k}\tilde R_{-+MP}\\
={}&\frac{k}{t^2}\psi_i{}^j \underbrace{\tilde\nabla_-\cdots\tilde\nabla_-}_{k-1}\tilde R_{-jMN}- \frac{1}{t}\psi_i{}^j\underbrace{\tilde\nabla_-\cdots\tilde\nabla_-}_{k}(\rho \tilde R_{-jMN})-\frac{1}{t}(\delta^j{}_i-\rho\psi_i{}^j )\underbrace{\tilde\nabla_-\cdots\tilde\nabla_-}_{k}\tilde R_{-jMN}\\
={}&-\frac{1}{t}\underbrace{\tilde\nabla_-\cdots\tilde\nabla_-}_{k}\tilde R_{-jMN}\,,\\
&\tilde\nabla_-\underbrace{\tilde\nabla_-\cdots\tilde\nabla_-}_{k}\tilde R_{-+MN}\\
={}&D_-\underbrace{\tilde\nabla_-\cdots\tilde\nabla_-}_{k-1}\tilde R_{-+MN}-\tilde\Gamma^j{}_{--}\tilde\nabla_j\underbrace{\tilde\nabla_-\cdots\tilde\nabla_-}_{k-1}\tilde R_{-+MN}-\cdots-\tilde\Gamma^j{}_{--}\underbrace{\tilde\nabla_-\cdots\tilde\nabla_-}_{k-1}\tilde\nabla_j\tilde R_{-+MN}
\\
&-\tilde\Gamma^j{}_{--}\underbrace{\tilde\nabla_-\cdots\tilde\nabla_-}_{k}\tilde R_{j+MN}-\tilde\Gamma^j{}_{-+}\underbrace{\tilde\nabla_-\cdots\tilde\nabla_-}_{k}\tilde R_{-jMN}\\
&-\tilde\Gamma^{P}{}_{- M}\underbrace{\tilde\nabla_-\cdots\tilde\nabla_-}_{k}\tilde R_{-+PN}-\tilde\Gamma^{P}{}_{-N}\underbrace{\tilde\nabla_-\cdots\tilde\nabla_-}_{k}\tilde R_{-+MP}\\
={}&\frac{k}{t^2}\varphi^j\underbrace{\tilde\nabla_-\cdots\tilde\nabla_-}_{k-1}\tilde R_{-jMN}-\frac{1}{t}\varphi^j\underbrace{\tilde\nabla_-\cdots\tilde\nabla_-}_{k}(\rho \tilde R_{-jMN})+\frac{\rho}{t}\varphi^j\underbrace{\tilde\nabla_-\cdots\tilde\nabla_-}_{k}\tilde R_{-jMN}=0\,,\\
&\tilde\nabla_+\underbrace{\tilde\nabla_-\cdots\tilde\nabla_-}_{k}\tilde R_{-+MN}\\
={}&D_+\underbrace{\tilde\nabla_-\cdots\tilde\nabla_-}_{k-1}\tilde R_{-+MN}-\tilde\Gamma^j{}_{+-}\tilde\nabla_j\underbrace{\tilde\nabla_-\cdots\tilde\nabla_-}_{k-1}\tilde R_{-+MN}-\cdots-\tilde\Gamma^j{}_{+-}\underbrace{\tilde\nabla_-\cdots\tilde\nabla_-}_{k-1}\tilde\nabla_j\tilde R_{-+MN}
\\
&-\tilde\Gamma^j{}_{+-}\underbrace{\tilde\nabla_-\cdots\tilde\nabla_-}_{k}\tilde R_{j+MN}-\tilde\Gamma^j{}_{++}\underbrace{\tilde\nabla_-\cdots\tilde\nabla_-}_{k}\tilde R_{-jMN}\\
&-\tilde\Gamma^{P}{}_{+M}\underbrace{\tilde\nabla_-\cdots\tilde\nabla_-}_{k}\tilde R_{-+PN}-\tilde\Gamma^{P}{}_{+N}\underbrace{\tilde\nabla_-\cdots\tilde\nabla_-}_{k}\tilde R_{-+MP}\\
={}&-\frac{k\rho}{t^2}\varphi^j\underbrace{\tilde\nabla_-\cdots\tilde\nabla_-}_{k-1}\tilde R_{-jMN}+\frac{\rho}{t}\varphi^j\underbrace{\tilde\nabla_-\cdots\tilde\nabla_-}_{k}(\rho \tilde R_{-jMN})-\frac{\rho^2}{t}\varphi^j\underbrace{\tilde\nabla_-\cdots\tilde\nabla_-}_{k}\tilde R_{-jMN}=0\,.
\end{align*}
Therefore, $\tilde\nabla_{P}\underbrace{\tilde\nabla_-\cdots\tilde\nabla_-}_{n}\tilde R_{-+MN}=-\frac{1}{t}\delta^{i}{}_{P}\underbrace{\tilde\nabla_-\cdots\tilde\nabla_-}_{n}\tilde R_{-iMN}$ holds for $n=k$ if it is valid for all $n\leqslant k-1$, which completes the proof.
\end{proof}

\chapter{Supplement to Part II}
\section{Nilpotency and Linearity of $\hatd$}
\label{app:dhat}
In this appendix section, we will show that the coboundary operator $\hat\td:\Omega^p(A;E)\to\Omega^{p+1}(A,E)$ in Definition \ref{def:hatd} is nilpotent and linear.
First we verify the nilpotency of $\hatd$ when it acts on $E$-valued 0-forms and 1-forms on $A$. The action of $\hatd$ on $\un\psi_1\in\Omega^{1}(A;E)$ reads
\begin{equation}
\label{d2b}
(\hatd\un\psi_1)(\umX_1,\umX_2)=\phi_E(\umX_1)\un\psi_1(\umX_2)-\phi_E(\umX_2)\un\psi_1(\umX_1)-\un\psi_1([\umX_1,\umX_2]_A)\,.
\end{equation}
Taking $\un\psi_1=\hatd\un\psi_0$, we have
\begin{align*}
(\hatd\hatd\un\psi_0)(\umX_1,\umX_2)&=\phi_E(\umX_1)\hatd\un\psi_0(\umX_2)-\phi_E(\umX_2)\hatd\un\psi_0(\umX_1)-\hatd\un\psi_0([\umX_1,\umX_2]_A)\\
&=[\phi_E(\umX_1),\phi_E(\umX_2)]_{\Der(E)}(\un\psi_0)-\phi_E([\umX_1,\umX_2]_A)(\un\psi_0)\,.
\end{align*}
Thus, $\hatd$ is nilpotent when acting twice on a 0-form provided that $\phi_E$ is a morphism.
\par
The action of $\hatd$ on $\un\psi_2\in\Omega^{2}(A,E)$ reads
\begin{align}
(\hatd\un\psi_2)(\umX_1,\umX_2,\umX_3)&=\phi_E(\umX_1)\un\psi_2(\umX_2,\umX_3)-\phi_E(\umX_2)\un\psi_2(\umX_1,\umX_3)+\phi_E(\umX_3)\un\psi_2(\umX_1,\umX_2)\nn\\
&\quad-\un\psi_2([\umX_1,\umX_2]_A,\umX_3)+\un\psi_2([\umX_1,\umX_3]_A,\umX_2)-\un\psi_2([\umX_2,\umX_3]_A,\umX_1)\,.
\end{align}
Taking $\un\psi_2=\hatd\un\psi_1$, we have
\begin{align*}
(\hatd\hatd\un\psi_1)(\umX_1,\umX_2,\umX_3)&=\phi_E(\umX_1)\hatd\un\psi_1(\umX_2,\umX_3)-\phi_E(\umX_2)\hatd\un\psi_1(\umX_1,\umX_3)+\phi_E(\umX_3)\hatd\un\psi_1(\umX_1,\umX_2)\\
&\quad-\hatd\un\psi_1([\umX_1,\umX_2]_A,\umX_3)+\hatd\un\psi_1([\umX_1,\umX_3]_A,\umX_2)-\hatd\un\psi_1([\umX_2,\umX_3]_A,\umX_1)\\
&=\phi_E(\umX_1)\phi_E(\umX_2)\un\psi_1(\umX_3)-\phi_E(\umX_1)\phi_E(\umX_3)\un\psi_1(\umX_2)-\phi_E(\umX_1)\un\psi_1([\umX_2,\umX_3]_A)\\
&\quad-\phi_E(\umX_2)\phi_E(\umX_1)\un\psi_1(\umX_3)+\phi_E(\umX_2)\phi_E(\umX_3)\un\psi_1(\umX_1)+\phi_E(\umX_2)\un\psi_1([\umX_1,\umX_3]_A)\\
&\quad+\phi_E(\umX_3)\phi_E(\umX_1)\un\psi_1(\umX_2)-\phi_E(\umX_3)\phi_E(\umX_2)\un\psi_1(\umX_1)-\phi_E(\umX_3)\un\psi_1([\umX_1,\umX_2]_A)\\
&\quad-\phi_E([\umX_1,\umX_2]_A)\un\psi_1(\umX_3)+\phi_E(\umX_3)\un\psi_1([\umX_1,\umX_2]_A)+\un\psi_1([[\umX_1,\umX_2]_A,\umX_3]_A)\\
&\quad+\phi_E([\umX_1,\umX_3]_A)\un\psi_1(\umX_2)-\phi_E(\umX_2)\un\psi_1([\umX_1,\umX_3]_A)-\un\psi_1([[\umX_1,\umX_3]_A,\umX_2]_A)\\
&\quad-\phi_E([\umX_2,\umX_3]_A)\un\psi_1(\umX_1)+\phi_E(\umX_1)\un\psi_1([\umX_2,\umX_3]_A)+\un\psi_1([[\umX_2,\umX_3]_A,\umX_1]_A)\\
&=\un\psi_1([[\umX_1,\umX_2]_A,\umX_3]_A)-\un\psi_1([[\umX_1,\umX_3]_A,\umX_2]_A)+\un\psi_1([[\umX_2,\umX_3]_A,\umX_1]_A)\,,
\end{align*}
where in the third equality we treated $\phi_E$ as a morphism. This indicates that $\hatd$ is nilpotent when acting twice on 1-forms if the Lie bracket on $A$ satisfies the Jacobi identity. Having these observations, we can carry this over to any higher forms. 
\bthm
The operator $\hatd$ is nilpotent, i.e.\ $\hatd\hatd\un\psi_n=0$ $\forall\un\psi_n\in\Omega^{n}(A;E)$, if
\par
(a) $\phi_E([\umX,\umY]_A)=[\phi_E(\umX),\phi_E(\umY)]_{\Der(E)}\,,\qquad\forall\umX,\umY\in\Gamma(A)$;
\par
(b) $[[\umX,\umY]_A,\umZ]_A+[[\umX,\umY]_A,\umZ]_A+[[\umX,\umY]_A,\umZ]_A=0\,,\qquad\forall\umX,\umY,\umZ\in\Gamma(A)$.
\ethm
\bpf
Suppose $\un\psi_n=\hatd\un\psi_{n-1}$, then
\begin{align*}
&\quad(\hatd\hatd\un\psi_{n-1})(\umX_1,\cdots,\umX_{n+1})\\
&=\sum_{r=1}^{n+1}(-1)^{r+1}\phi_E(\umX_r)(\hatd\un\psi_{n-1}(\umX_1,\cdots,\widehat{\umX_r},\cdots,\umX_{n+1}))\\
&\quad+\sum_{r<s}^{n+1}(-1)^{r+s}\hatd\un\psi_{n-1}([\umX_r,\umX_s]_A,\umX_1,\cdots,\widehat{\umX_r},\cdots,\widehat{\umX_s},\cdots,\umX_{n+1})\\
&=\sum_{r>s}(-1)^{r+s}\phi_E(\umX_r)\phi_E(\umX_s)(\un\psi_{n-1}(\umX_1,\cdots,\widehat{\umX_s},\cdots,\widehat{\umX_r},\cdots,\umX_{n+1}))\\
&\quad-\sum_{r<s}(-1)^{r+s}\phi_E(\umX_r)\phi_E(\umX_s)(\un\psi_{n-1}(\umX_1,\cdots,\widehat{\umX_r},\cdots,\widehat{\umX_s},\cdots,\umX_{n+1}))\\
&\quad+\sum_{s<t<r}(-1)^{r+s+t+1}\phi_E(\umX_r)\un\psi_{n-1}([\umX_s,\umX_t]_A,\umX_1,\cdots,\widehat{\umX_s},\cdots,\widehat{\umX_t},\cdots,\widehat{\umX_r},\cdots,\umX_{n+1})\\
&\quad+\sum_{s<r<t}(-1)^{r+s+t}\phi_E(\umX_r)\un\psi_{n-1}([\umX_s,\umX_t]_A,\umX_1,\cdots,\widehat{\umX_s},\cdots,\widehat{\umX_r},\cdots,\widehat{\umX_t},\cdots,\umX_{n+1})\\
&\quad+\sum_{r<s<t}(-1)^{r+s+t+1}\phi_E(\umX_r)\un\psi_{n-1}([\umX_s,\umX_t]_A,\umX_1,\cdots,\widehat{\umX_r},\cdots,\widehat{\umX_s},\cdots,\widehat{\umX_t},\cdots,\umX_{n+1})\\
&\quad+\sum_{r<s}^{n+1}(-1)^{r+s}\phi_E([\umX_r,\umX_s]_A)(\un\psi_{n-1}(\umX_1,\cdots,\widehat{\umX_r},\cdots,\widehat{\umX_s},\cdots,\umX_{n+1}))\\
&\quad+\sum_{t<r<s}(-1)^{r+s+t}\phi_E(\umX_t)(\un\psi_{n-1}([\umX_r,\umX_s]_A,\umX_1,\cdots,\widehat{\umX_t},\cdots,\widehat{\umX_r},\cdots,\widehat{\umX_s},\cdots,\umX_{n+1}))\\
&\quad+\sum_{r<t<s}(-1)^{r+s+t+1}\phi_E(\umX_t)(\un\psi_{n-1}([\umX_r,\umX_s]_A,\umX_1,\cdots,\widehat{\umX_r},\cdots,\widehat{\umX_t},\cdots,\widehat{\umX_s},\cdots,\umX_{n+1}))\\
&\quad+\sum_{r<s<t}(-1)^{r+s+t}\phi_E(\umX_t)(\un\psi_{n-1}([\umX_r,\umX_s]_A,\umX_1,\cdots,\widehat{\umX_r},\cdots,\widehat{\umX_s},\cdots,\widehat{\umX_t},\cdots,\umX_{n+1}))\\
&\quad+\sum_{t<r<s}^{n+1}(-1)^{r+s+t}\un\psi_{n-1}([[\umX_r,\umX_s]_A,\umX_t]_A,\umX_1,\cdots,\widehat{\umX_t},\cdots,\widehat{\umX_r},\cdots,\widehat{\umX_s},\cdots,\umX_{n+1})\\
&\quad+\sum_{r<t<s}^{n+1}(-1)^{r+s+t+1}\un\psi_{n-1}([[\umX_r,\umX_s]_A,\umX_t]_A,\umX_1,\cdots,\widehat{\umX_r},\cdots,\widehat{\umX_t},\cdots,\widehat{\umX_s},\cdots,\umX_{n+1})\\
&\quad+\sum_{r<s<t}^{n+1}(-1)^{r+s+t}\un\psi_{n-1}([[\umX_r,\umX_s]_A,\umX_t]_A,\umX_1,\cdots,\widehat{\umX_r},\cdots,\widehat{\umX_s},\cdots,\widehat{\umX_t},\cdots,\umX_{n+1})\\
&\quad+\sum_{t<u<r<s}^{n+1}(-1)^{r+s+t+u}\un\psi_{n-1}([\umX_t,\umX_u]_A,[\umX_r,\umX_s]_A,\umX_1,\cdots,\widehat{\umX_t},\cdots,\widehat{\umX_u},\cdots,\widehat{\umX_r},\cdots,\widehat{\umX_s},\cdots,\umX_{n+1})\\
&\quad+\sum_{t<r<u<s}^{n+1}(-1)^{r+s+t+u+1}\un\psi_{n-1}([\umX_t,\umX_u]_A,[\umX_r,\umX_s]_A,\umX_1,\cdots,\widehat{\umX_t},\cdots,\widehat{\umX_r},\cdots,\widehat{\umX_u},\cdots,\widehat{\umX_s},\cdots,\umX_{n+1})\\
&\quad+\sum_{t<r<s<u}^{n+1}(-1)^{r+s+t+u}\un\psi_{n-1}([\umX_t,\umX_u]_A,[\umX_r,\umX_s]_A,\umX_1,\cdots,\widehat{\umX_t},\cdots,\widehat{\umX_r},\cdots,\widehat{\umX_s},\cdots,\widehat{\umX_u},\cdots,\umX_{n+1})\\
&\quad+\sum_{r<t<u<s}^{n+1}(-1)^{r+s+t+u}\un\psi_{n-1}([\umX_t,\umX_u]_A,[\umX_r,\umX_s]_A,\umX_1,\cdots,\widehat{\umX_r},\cdots,\widehat{\umX_t},\cdots,\widehat{\umX_u},\cdots,\widehat{\umX_s},\cdots,\umX_{n+1})\\
&\quad+\sum_{r<t<s<u}^{n+1}(-1)^{r+s+t+u+1}\un\psi_{n-1}([\umX_t,\umX_u]_A,[\umX_r,\umX_s]_A,\umX_1,\cdots,\widehat{\umX_r},\cdots,\widehat{\umX_t},\cdots,\widehat{\umX_s},\cdots,\widehat{\umX_u},\cdots,\umX_{n+1})\\
&\quad+\sum_{r<s<t<u}^{n+1}(-1)^{r+s+t+u}\un\psi_{n-1}([\umX_t,\umX_u]_A,[\umX_r,\umX_s]_A,\umX_1,\cdots,\widehat{\umX_r},\cdots,\widehat{\umX_s},\cdots,\widehat{\umX_t},\cdots,\widehat{\umX_u},\cdots,\umX_{n+1})\\
&=\sum_{t<r<s}^{n+1}(-1)^{r+s+t}\un\psi_{n-1}([[\umX_r,\umX_s]_A,\umX_t]_A,\umX_1,\cdots,\widehat{\umX_t},\cdots,\widehat{\umX_r},\cdots,\widehat{\umX_s},\cdots,\umX_{n+1})\\
&\quad+\sum_{t<r<s}^{n+1}(-1)^{r+s+t}\un\psi_{n-1}([[\umX_s,\umX_t]_A,\umX_r]_A,\umX_1,\cdots,\widehat{\umX_t},\cdots,\widehat{\umX_r},\cdots,\widehat{\umX_s},\cdots,\umX_{n+1})\\
&\quad+\sum_{t<r<s}^{n+1}(-1)^{r+s+t}\un\psi_{n-1}([[\umX_t,\umX_r]_A,\umX_s]_A,\umX_1,\cdots,\widehat{\umX_t},\cdots,\widehat{\umX_r},\cdots,\widehat{\umX_s},\cdots,\umX_{n+1})\\
&=\sum_{t<r<s}^{n+1}(-1)^{r+s+t}\\
&\quad\un\psi_{n-1}([[\umX_r,\umX_s]_A,\umX_t]_A+[[\umX_s,\umX_t]_A,\umX_r]_A+[[\umX_t,\umX_r]_A,\umX_s]_A,\umX_1,\cdots,\widehat{\umX_t},\cdots,\widehat{\umX_r},\cdots,\widehat{\umX_s},\cdots,\umX_{n+1})\\
&=0\,.
\end{align*}
Thus, $\hatd\hatd\un\psi_n=0$ as long as $\phi_E$ is a morphism and the Lie bracket on $A$ satisfies the Jacobi identity. 
\epf
The next thing we want to verify is that the Koszul formula is linear in the sections $\umX_1,\cdots,\umX_{n+1}$. Let $f\in C^\infty(M)$, then for any $p=1,\cdots,n+1$ we can derive that
\begin{align*}
&\quad(\hatd\un\psi_n)(\umX_1,\cdots,f\umX_p,\cdots,\umX_{n+1})\\
&=\sum_{r=1}^{p-1}(-1)^{r+1}\phi_E(\umX_r)(\un\psi_n(\umX_1,\cdots,\widehat{\umX_r},\cdots,f\umX_p,\cdots,\umX_{n+1}))\\
&\quad+(-1)^{p+1}\phi_E(f\umX_p)(\un\psi_n(\umX_1,\cdots,\widehat{\umX_p},\cdots,\umX_{n+1}))\\
&\quad+\sum_{r=p+1}^{n+1}(-1)^{r+1}\phi_E(\umX_r)(\un\psi_n(\umX_1,\cdots,f\umX_p,\cdots,\widehat{\umX_r},\cdots,\umX_{n+1}))\\
&\quad+\sum_{s=2}^{p-1}\sum_{r=1}^{s-1}(-1)^{r+s}\un\psi_n([\umX_r,\umX_s]_A,\umX_1,\cdots,\widehat{\umX_r},\cdots,\widehat{\umX_s},\cdots,f\umX_p,\cdots,\umX_{n+1})\\
&\quad+\sum_{s=p+1}^{n+1}\sum_{r=p}^{s-1}(-1)^{r+s}\un\psi_n([\umX_r,\umX_s]_A,\umX_1,\cdots,f\umX_p,\cdots,\widehat{\umX_r},\cdots,\widehat{\umX_s},\cdots,\umX_{n+1})\\
&\quad+\sum_{s=p+1}^{n+1}\sum_{r=1}^{p-1}(-1)^{r+s}\un\psi_n([\umX_r,\umX_s]_A,\umX_1,\cdots,\widehat{\umX_r},\cdots,f\umX_p,\cdots,\widehat{\umX_s},\cdots,\umX_{n+1})\\
&\quad+\sum_{r=1}^{p-1}(-1)^{r+p}\un\psi_n([\umX_r,f\umX_p]_A,\umX_1,\cdots,\widehat{\umX_r},\cdots,\widehat{\umX_p},\cdots,\umX_{n+1})\\
&\quad+\sum_{r=p+1}^{n+1}(-1)^{p+s}\un\psi_n([f\umX_p,\umX_s]_A,\umX_1,\cdots,\widehat{\umX_p},\cdots,\widehat{\umX_s},\cdots,\umX_{n+1})\\
&=\sum_{r\neq p}(-1)^{r+1}\phi_E(\umX_r)(f\un\psi_n(\umX_1,\cdots,\widehat{\umX_r},\cdots,\umX_{n+1}))\\
&\quad+(-1)^{p+1}\phi_E(f\umX_p)(\un\psi_n(\umX_1,\cdots,\widehat{\umX_p},\cdots,\umX_{n+1}))\\
&\quad+\sum_{p\neq r<s\neq p}(-1)^{r+s}f\un\psi_n([\umX_r,\umX_s]_A,\umX_1,\cdots,\widehat{\umX_r},\cdots,\widehat{\umX_s},\cdots,\umX_{n+1})\\
&\quad+\sum_{r=1}^{p-1}(-1)^{r+p}\un\psi_n([\umX_r,f\umX_p]_A,\umX_1,\cdots,\widehat{\umX_r},\cdots,\widehat{\umX_p},\cdots,\umX_{n+1})\\
&\quad+\sum_{s=p+1}^{n+1}(-1)^{p+s}\un\psi_n([f\umX_p,\umX_s]_A,\umX_1,\cdots,\widehat{\umX_p},\cdots,\widehat{\umX_s},\cdots,\umX_{n+1})\\
&=\sum_{r\neq p}(-1)^{r+1}f\phi_E(\umX_r)(\un\psi_n(\umX_1,\cdots,\widehat{\umX_r},\cdots,\umX_{n+1}))\\
&\quad+\sum_{r\neq p}(-1)^{r+1}\rho(\umX_r)(f)(\un\psi_n(\umX_1,\cdots,\widehat{\umX_r},\cdots,\umX_{n+1}))\\
&\quad+(-1)^{p+1}f\phi_E(\umX_p)(\un\psi_n(\umX_1,\cdots,\widehat{\umX_p},\cdots,\umX_{n+1}))\\
&\quad+\sum_{p\neq r<s\neq p}(-1)^{r+s}f\un\psi_n([\umX_r,\umX_s]_A,\umX_1,\cdots,\widehat{\umX_r},\cdots,\widehat{\umX_s},\cdots,\umX_{n+1})\\
&\quad+\sum_{r=1}^{p-1}(-1)^{r+p}f\un\psi_n([\umX_r,\umX_p]_A,\umX_1,\cdots,\widehat{\umX_r},\cdots,\widehat{\umX_p},\cdots,\umX_{n+1})\\
&\quad+\sum_{r=1}^{p-1}(-1)^{r}\rho(\umX_r)(f)\un\psi_n(\umX_1,\cdots,\widehat{\umX_r},\cdots,\umX_p,\cdots,\umX_{n+1})\\
&\quad+\sum_{s=p+1}^{n+1}(-1)^{p+s}f\un\psi_n([\umX_p,\umX_s]_A,\umX_1,\cdots,\widehat{\umX_p},\cdots,\widehat{\umX_s},\cdots,\umX_{n+1})\\
&\quad+\sum_{s=p+1}^{n+1}(-1)^{s}\rho(\umX_s)(f)\un\psi_n(\umX_1,\cdots,\umX_p,\cdots,\widehat{\umX_s},\cdots,\umX_{n+1}))\\
&=\sum_{r=1}^{n+1}(-1)^{r+1}f\phi_E(\umX_r)(\un\psi_n(\umX_1,\cdots,\widehat{\umX_r},\cdots,\umX_{n+1})\\
&\quad+\sum_{r<s}^{n+1}(-1)^{r+s}f\un\psi_n([\umX_r,\umX_s]_A,\umX_1,\cdots,\widehat{\umX_r},\cdots,\widehat{\umX_s},\cdots,\umX_{n+1})\\
&=f(\hatd\un\psi_n)(\umX_1,\cdots,\umX_p,\cdots,\umX_{n+1})\,.
\end{align*}
Therefore, the operator $\hatd$ defined through the Koszul formula is linear.

\section{Relation Between Curvatures ${\cal R}^E$ and $\Omega$}
\label{app:ROmega}
Given a representation $\phi_E$ of a Lie algebroid, the curvature of the induced connection $\nabla^E$ on a representation algebroid introduced in Subsection \ref{sec:reps} is defined as
\begin{align}
{\cal R}^E(\umX,\umY)(\un\psi_0)\equiv\nabla^E_{\rho(\umX_H)}(\nabla^E_{\rho(\umY_H)}\psi_0)-\nabla^E_{\rho(\umY_H)}(\nabla^E_{\rho(\umX_H)}\psi_0)-\nabla^E_{\rho([\umX_H,\umY_H]_H)}\psi_0\,,
\end{align}
where $[\umX_H,\umY_H]_H$ represents the horizontal part of $[\umX_H,\umY_H]_A$. Using the condition that $\phi_E$ is a morphism, we have
\begin{align*}
0&=[\phi_E(\umX),\phi_E(\umY)]_{\Der(E)}(\un\psi_0)-\phi_E([\umX,\umY]_A)(\un\psi_0)\\
&=\phi_E(\umX)(\phi_E(\umY)(\un\psi_0))-\phi_E(\umY)(\phi_E(\umX)(\un\psi_0))-\phi_E([\umX,\umY]_A)(\un\psi_0)\\
&=\phi_E(\umX)(\nabla^E_{\rho(\umY)}\psi_0-v_E(\omega(\umY))(\un\psi_0))-\phi_E(\umY)(\nabla^E_{\rho(\umX)}\psi_0-v_E(\omega(\umX))(\un\psi_0))\\
&\quad-\nabla^E_{\rho([\umX,\umY]_H)}\psi_0+v_E(\omega([\umX,\umY]_V))(\un\psi_0)\\
&=\nabla^E_{\rho(\umX)}(\nabla^E_{\rho(\umY)}\psi_0-v_E(\omega(\umY))(\un\psi_0))-v_E(\omega(\umX))(\nabla^E_{\rho(\umY)}\psi_0-v_E(\omega(\umY))(\un\psi_0))\\
&\quad-\nabla^E_{\rho(\umY)}(\nabla^E_{\rho(\umX)}\psi_0-v_E(\omega(\umX))(\un\psi_0))+v_E(\omega(\umY))(\nabla^E_{\rho(\umX)}\psi_0-v_E(\omega(\umX))(\un\psi_0))\\
&\quad-\nabla^E_{\rho([\umX,\umY]_H)}\psi_0+v_E(\omega([\umX,\umY]_V))(\un\psi_0)\\
&={\cal R}^E(\umX,\umY)(\un\psi_0)-\nabla^E_{\rho(\umX)}(v_E(\omega(\umY))(\un\psi_0))+\nabla^E_{\rho(\umY)}(v_E(\omega(\umX))(\un\psi_0))-v_E(\omega(\umX))(\nabla^E_{\rho(\umY)}\psi_0)\\
&\quad+v_E(\omega(\umY))(\nabla^E_{\rho(\umX)}\psi_0)+v_E(\omega(\umX_V))v_E(\omega(\umY_V))(\un\psi_0)-v_E(\omega(\umY_V))v_E(\omega(\umX_V))(\un\psi_0)+v_E(\omega([\umX,\umY]_V))(\un\psi_0)\\
&={\cal R}^E(\umX,\umY)(\un\psi_0)-\nabla^E_{\rho(\umX)}(v_E(\omega(\umY))(\un\psi_0))+\nabla^E_{\rho(\umY)}(v_E(\omega(\umX))(\un\psi_0))-v_E(\omega(\umX))(\nabla^E_{\rho(\umY)}\psi_0)\\
&\quad+v_E(\omega(\umY))(\nabla^E_{\rho(\umX)}\psi_0)+v_E([\omega(\umX_V),\omega(\umY_V)]_L)(\un\psi_0)+v_E(\omega([\umX,\umY]_V))(\un\psi_0)\\
&={\cal R}^E(\umX,\umY)(\un\psi_0)+\nabla^E_{\rho(\umX)}(v_E(\omega(\umY))(\un\psi_0))-\nabla^E_{\rho(\umY)}(v_E(\omega(\umX))(\un\psi_0))-v_E(\omega(\umX))(\nabla^E_{\rho(\umY)}\psi_0)\\
&\quad+v_E(\omega(\umY))(\nabla^E_{\rho(\umX)}\psi_0)+v_E(R^\omega(\umX,\umY))(\un\psi_0)\,,
\end{align*}
where we used the fact that $v_E$ is a morphism in the sixth equality. Since ${\cal R}^E(\umX_V,\umY_V)=0$ and $R^\omega(\umX_V,\umY_V)=0$, we can see that when $\umX$ and $\umY$ are purely vertical, this expression identically vanishes. For the case $\umX$ being horizontal and $\umY$ being vertical, we have
\begin{align*}
0&=[\phi_E(\umX_H),\phi_E(\umY_V)]_{\Der(E)}(\un\psi_0)-\phi_E([\umX_H,\umY_V]_A)(\un\psi_0)\\
&=\nabla^E_{\rho(\umX)}(v_E(\omega(\umY))(\un\psi_0))+v_E(\omega(\umY))(\nabla^E_{\rho(\umX)}\psi_0)+v_E(R^\omega(\umX_H,\umY_V))(\un\psi_0)\\
&=\nabla^E_{\rho(\umX)}(v_E(\omega(Y_V))\un\psi)-v_E(\omega(\umY_V))(\nabla^E_{\rho(\umX)}\un\psi_0)-v_E(\nabla^L_{X_H}\omega(Y_V))(\un\psi_0)\,.
\end{align*}
This can be regarded as a Leibniz rule relating $\nabla^E$ to the induced connection $\nabla^L$ in the adjoint representation. 
Finally we look at the case where $\umX$ and $\umY$ are both horizontal,
\begin{align*}
0&=[\phi_E(\umX_H),\phi_E(\umY_H)]_{\Der(E)}(\un\psi_0)-\phi_E([\umX_H,\umY_H]_A)(\un\psi_0)\\
&={\cal R}^E(\umX_H,\umY_H)(\un\psi_0)+v_E(R^\omega(\umX_H,\umY_H))(\un\psi_0)\\
&={\cal R}^E(\umX_H,\umY_H)(\un\psi_0)-v_E(\Omega(\umX,\umY))(\un\psi_0)\,.
\end{align*}
Thus,
\begin{align}
{\cal R}^E(\umX_H,\umY_H)(\un\psi_0)=v_E(\Omega(\umX,\umY))(\un\psi_0)\,,
\end{align}
which relates ${\cal R}^E$ to the curvature reform $\Omega$ of the Lie algebroid. 
\par
In the special case of the adjoint representation, the morphisms $\phi_E$ and $v_E$ can be expressed in terms of the Lie brackets:
\begin{align}
\label{phiLm}
\phi_L(\umX)(\un\mu)&=-\omega([\umX,j(\un\mu)]_A)\,,\qquad\forall\umX\in A, \un\in L\\
(v_L(\un\mu))(\un\nu)&=[\un\mu,\un\nu]_L\,,\qquad\forall\un\mu,\un\nu\in L\,,
\end{align}
and we have seen that the induced connection $\nabla^L$ behaves as
\begin{align}
\label{nablaL}
\nabla^L_{\rho(\umX)}\un\mu=\nabla^L_{\rho(\umX_H)}\un\mu=-R^\omega(\umX_H,j(\un\mu))\,.
\end{align}
Define the curvature ${\cal R}^L:A\times A\times L\to L$ of $\nabla^L$ as follows (which is in fact ${\cal R}^L:H\times H\times L\to L$): 
\begin{align}
{\cal R}^L(\umX,\umY)(\un\mu)\equiv\nabla^L_{\rho(\umX_H)}(\nabla^L_{\rho(\umY_H)}\un\mu)-\nabla^L_{\rho(\umY_H)}(\nabla^L_{\rho(\umX_H)}\un\mu)-\nabla^L_{\rho([\umX_H,\umY_H]_H)}\un\mu\,.
\end{align}
Using \eqref{nablaL}, the equation above can be evaluated directly as follows
\begin{align}
{\cal R}^L(\umX,\umY)(\un\mu)&=-\nabla^L_{\rho(\umX_H)}(\omega([\umY_H,j(\un\mu)]_A))+\nabla^L_{\rho(\umY_H)}(\omega([\umY_H,j(\un\mu)]_A))+\omega([[\umX_H,\umY_H]_H,j(\un\mu)]_A)\,,\nn\\
&=\omega([\umX_H,j(\omega([\umY_H,j(\un\mu)]_A))]_A)-\omega([\umY_H,j(\omega([\umX_H,j(\un\mu)]_A))]_A)+\omega([[\umX_H,\umY_H]_H,j(\un\mu)]_A)\,,\nn\\
&=-\omega([\umX_H,[\umY_H,j(\un\mu)]_A]_A)+\omega([\umY_H,[\umX_H,j(\un\mu)]_A]_A)+\omega([[\umX_H,\umY_H]_H,j(\un\mu)]_A)\,,\nn\\
&=\omega([\umX_H,[j(\un\mu)]_A,\umY_H]_A)+\omega([\umY_H,[\umX_H,j(\un\mu)]_A]_A)+\omega([j(\un\mu),[\umY_H,\umX_H]_A)\nn\\
&\quad-\omega([[\umX_H,\umY_H]_A,j(\un\mu)]_A)+\omega([[\umX_H,\umY_H]_H,j(\un\mu)]_A)\,,\nn\\
&=-\omega([[\umX_H,\umY_H]_V,j(\un\mu)]_A)\,,\nn\\
&=-R^\omega([\umX_H,\umY_H]_V,j(\un\mu))+[\omega([\umX_H,\umY_H]_V),\omega(j(\un\mu))]_L\,,\nn\\
&=-[\omega([\umX_H,\umY_H]_V),\un\mu]_L\,,\nn\\
&=-v_L(\omega([\umX_H,\umY_H]_V)(\un\mu)\,,\nn\\
&=v_L(\Omega([\umX_H,\umY_H])(\un\mu)\,,
\end{align}
where we used the Jacobi identity in the fifth equality, the fact that $R^\omega(\umX_V,\umY_V)=0$ is used in the seventh equality, \eqref{vL} is used in the eighth equality, and \eqref{Omega2d} is used in the last equality. Thus, ${\cal R}^L$ also represents the curvature of the Lie algebroid. 

\section{Commutation Coefficients of the Algebroid Lie Bracket}
\label{app:commutation}
Given a split basis $\{\un E_{\ualpha},\un E_{\un A}\}$, the Lie bracket on $A$ gives
\begin{align}
[\un E_{\ualpha},\un E_{\ubeta}]_A&= C_{\ualpha\ubeta}{}^{\un\gamma}\un E_{\un\gamma}+C_{\ualpha\ubeta}{}^{\un A}\un E_{\un A}\,,\\
[\un E_{\ualpha},\un E_{\un A}]_A&= C_{\ualpha\un A}{}^{\un B}\un E_{\un B}\,,\\
\label{LAE1}
[\un E_{\un A},\un E_{\un B}]_A&= C_{\un{AB}}{}^{\un C}\un E_{\un C}\,,
\end{align}
First we evaluate $C_{\un{AB}}{}^{\un C}$ in \eqref{LAE1}. Recall that in a basis $\{\un t_A\}$ of $\Gamma(L)$ we have
\begin{align}
\label{Liett}
[\un t_A,\un t_B]_L= f_{AB}{}^C\un t_C\,.
\end{align}
Applying $j$ to both sides of \eqref{Liett} yields
\begin{align}
j([\un t_A,\un t_B]_L)&= f_{AB}{}^Cj(\un t_C)\\
[j(\un t_A),j(\un t_B)]_A&= f_{AB}{}^Cj(\un t_C)\\
[j^{\un A}{}_AE_{\un A},j^{\un B}{}_BE_{\un B}]_A&= f_{AB}{}^Cj^{\un C}{}_CE_{\un C}\,,
\end{align}
where we used \eqref{iotat} in the last step. Comparing this with \eqref{LAE1} yields
\begin{align}
\label{Cf}
C_{\un{AB}}{}^{\un C}j^{\un A}{}_Aj^{\un B}{}_B= f_{AB}{}^Cj^{\un C}{}_C\,,
\end{align}
which leads to \eqref{CABC}.
\par
For a horizontal section $\umX_H\in\Gamma(H)$ and a vertical section $j(\un\mu)\in\Gamma(V)$ with $\un\mu\in\Gamma(L)$, the Lie bracket gives
\begin{align}
[\umX_H,j(\un\mu)]_A&=[\mX_H^{\ualpha}\un E_{\ualpha},j(\mu^A\un t_A)]_A\nn\\
&=[\mX_H^{\ualpha}\un E_{\ualpha},\mu^Aj^{\un A}{}_A\un E_{\un A}]_A\nn\\
&=\mX_H^{\ualpha}\mu^Aj^{\un A}{}_A[\un E_{\ualpha},\un E_{\un A}]_A+\mX_H^{\ualpha}\rho(\un E_{\ualpha})(j^{\un A}{}_A\mu^A)\un E_{\un A}-\mu^A\rho(\un E_{\un A})(j^{\un A}{}_A\mX_H^{\ualpha})\un E_{\ualpha}\nn\\
&=\mX_H^{\ualpha}\mu^Aj^{\un A}{}_AC_{\ualpha\un A}{}^{\un B}\un E_{\un B}+\mX_H^{\ualpha}\rho(\un E_{\ualpha})(j^{\un A}{}_A\mu^A)\un E_{\un A}\nn\\
&=\mX_H^{\ualpha}\big(\mu^Aj^{\un A}{}_AC_{\ualpha\un A}{}^{\un B}+\rho(\un E_{\ualpha})(j^{\un B}{}_A\mu^A)\big)\un E_{\un B}\nn\\
\label{XHiota0}
&=\mX_H^{\ualpha}\big(\mu^Aj^{\un A}{}_AC_{\ualpha\un A}{}^{\un B}+\rho(\un E_{\ualpha})(\mu^B)j^{\un B}{}_B+\rho(\un E_{\ualpha})(j^{\un B}{}_A)\mu^A\big)\un E_{\un B}\,.
\end{align}
On the other hand, it follows from \eqref{nablaL} that
\begin{align}
[\umX_H,j(\un\mu)]_A=j(\nabla^L_{\umX_H}\un\mu)\,,
\end{align}
and it follows from \eqref{phiLm} and \eqref{phiEf} that
\begin{align}
[\umX_H,j(\un\mu)]_A&=j(\phi_L(\umX_H)(\mu^A\un t_A))\nn\\
&=j(\mu^A\phi_L(\umX_H)(\un t_A)+(\rho(\umX_H)\mu^A)\un t_A)\nn\\
\label{XHiota1}
&=j(\mu^A\nabla^L_{\umX_H}\un t_A+(\rho(\umX_H)\mu^A)\un t_A)\,.
\end{align}
Since $\nabla^L_{\umX}\un t_A$ is a section on $L$, we can expand it using $\{\un t_A\}$:
\begin{align}
\label{nablaLt}
\nabla^L_{\umX_H}\un t_A={\cal A}^B{}_A(\umX_H)\un t_B\,.
\end{align}
where ${\cal A}^B{}_A(\umX)$ are the connection coefficients, which depends linearly on $\umX$. Thus, now \eqref{XHiota1} becomes
\begin{align}
[\umX_H,j(\un\mu)]_A&=j(\mu^A{\cal A}^B{}_A(\umX_H)\un t_B+(\rho(\umX_H)\mu^A)\un t_A)\nn\\
&=\mu^A{\cal A}^B{}_A(\umX_H)j^{\un B}{}_B\un E_{\un B}+(\rho(\umX_H)\mu^A)j^{\un A}{}_A\un E_{\un A}\nn\\
\label{XHiota2}
&=\umX_H^{\ualpha}\big(\mu^A{\cal A}^B{}_A(\un E_{\ualpha})+\rho(\un E_{\ualpha})\mu^B\big)j^{\un B}{}_B\un E_{\un B}\,.
\end{align}
Comparing \eqref{XHiota0} and \eqref{XHiota2} yields
\begin{align}
j^{\un A}{}_AC_{\ualpha\un A}{}^{\un B}+\rho(\un E_{\ualpha})(j^{\un B}{}_A)={\cal A}_{\ualpha}{}^B{}_Aj^{\un B}{}_B\,.
\end{align}
where ${\cal A}_{\ualpha}{}^B{}_A\equiv {\cal A}^B{}_A(\un E_{\ualpha})$. This equation gives rise to \eqref{CaAB}. 
\par
Plugging $\un E_{\ualpha},\un E_{\ubeta}$ into \eqref{Omega1}, we have
\begin{align*}
j(\Omega(\un E_{\ualpha},\un E_{\ubeta}))&=[\un E_{\ualpha},\un E_{\ubeta}]_V\\
j(\Omega^A(\un E_{\ualpha},\un E_{\ubeta})\un t_A)&=C_{\ualpha\ubeta}{}^{\un A}\un E_{\un A}\\
\Omega^A(\un E_{\ualpha},\un E_{\ubeta})j^{\un A}{}_A\un E_{\un A}&=C_{\ualpha\ubeta}{}^{\un A}\un E_{\un A}\,.
\end{align*}
Thus,
\begin{align}
C_{\ualpha\ubeta}{}^{\un A}=\Omega^A{}_{\ualpha\ubeta}j^{\un A}{}_A\,,
\end{align}
where $\Omega^A{}_{\ualpha\ubeta}\equiv\Omega^A(\un E_{\ualpha},\un E_{\ubeta})$. Now we consider two horizontal sections $\sigma(\uX)$ and $\sigma(\uY)$ of $A$ with $\uX,\uY\in \Gamma(TM)$. The commutator gives
\begin{align*}
[\sigma(\uX),\sigma(\uY)]_A&=[X^\mu\sigma^{\ualpha}{}_\mu\un E_{\ualpha},Y^\nu\sigma^{\ubeta}{}_\nu\un E_{\ubeta}]_A\\
&=X^\mu\sigma^{\ualpha}{}_\mu Y^\nu\sigma^{\ubeta}{}_\nu[\un E_{\ualpha},\un E_{\ubeta}]_A+X^\mu\sigma^{\ualpha}{}_\mu\rho(\un E_{\ualpha})(Y^\nu\sigma^{\ubeta}{}_\nu)\un E_{\ubeta}-Y^\nu\sigma^{\ubeta}{}_\nu\rho(\un E_{\ubeta})(X^\mu\sigma^{\ualpha}{}_\mu)\un E_{\ualpha}\\
&=X^\mu\sigma^{\ualpha}{}_\mu Y^\nu\sigma^{\ubeta}{}_\nu(C_{\ualpha\ubeta}{}^{\un\gamma}\un E_{\un\gamma}+C_{\ualpha\ubeta}{}^{\un A}\un E_{\un A})+X^\mu\sigma^{\ualpha}{}_\mu\rho^\rho{}_{\ualpha}\un\p_\rho(Y^\nu\sigma^{\ubeta}{}_\nu)\un E_{\ubeta}-Y^\nu\sigma^{\ubeta}{}_\nu\rho^\rho{}_{\ubeta}\un\p_\rho(X^\mu\sigma^{\ualpha}{}_\mu)\un E_{\ualpha}\\
&=X^\mu\sigma^{\ualpha}{}_\mu Y^\nu\sigma^{\ubeta}{}_\nu(C_{\ualpha\ubeta}{}^{\un\gamma}\un E_{\un\gamma}+C_{\ualpha\ubeta}{}^{\un A}\un E_{\un A})+X^\mu\un\p_\mu(Y^\nu\sigma^{\ubeta}{}_\nu)\un E_{\ubeta}-Y^\nu\un\p_\nu(X^\mu\sigma^{\ualpha}{}_\mu)\un E_{\ualpha}\\
&=X^\mu\sigma^{\ualpha}{}_\mu Y^\nu\sigma^{\ubeta}{}_\nu(C_{\ualpha\ubeta}{}^{\un\gamma}\un E_{\un\gamma}+C_{\ualpha\ubeta}{}^{\un A}\un E_{\un A})+X^\mu(\un\p_\mu Y^\nu)\sigma^{\ubeta}{}_\nu\un E_{\ubeta}+X^\mu Y^\nu(\un\p_\mu\sigma^{\ubeta}{}_\nu)\un E_{\ubeta}\\
&\quad-Y^\nu(\un\p_\nu X^\mu)\sigma^{\ualpha}{}_\mu\un E_{\ualpha}-Y^\nu X^\mu(\un\p_\nu\sigma^{\ualpha}{}_\mu)\un E_{\ualpha}\\
&=X^\mu\sigma^{\ualpha}{}_\mu Y^\nu\sigma^{\ubeta}{}_\nu(C_{\ualpha\ubeta}{}^{\un\gamma}\un E_{\un\gamma}+C_{\ualpha\ubeta}{}^{\un A}\un E_{\un A})+[\uX,\uY]^\mu\sigma^{\un\gamma}{}_\mu\un E_{\un\gamma}+X^\mu Y^\nu(\un\p_\mu\sigma^{\un\gamma}{}_\nu-\un\p_\nu\sigma^{\un\gamma}{}_\mu)\un E_{\un\gamma}\\
&=X^\mu\sigma^{\ualpha}{}_\mu Y^\nu\sigma^{\ubeta}{}_\nu(C_{\ualpha\ubeta}{}^{\un\gamma}\un E_{\un\gamma}+C_{\ualpha\ubeta}{}^{\un A}\un E_{\un A})+\sigma([\uX,\uY]_{TM})+X^\mu Y^\nu(\un\p_\mu\sigma^{\un\gamma}{}_\nu-\un\p_\nu\sigma^{\un\gamma}{}_\mu)\un E_{\un\gamma}\,,
\end{align*}
and hence
\begin{align*}
R^\sigma(\uX,\uY)&=[\sigma(\uX),\sigma(\uY)]_A-\sigma([\uX,\uY]_{TM})\\
&=X^\mu\sigma^{\ualpha}{}_\mu Y^\nu\sigma^{\ubeta}{}_\nu(C_{\ualpha\ubeta}{}^{\un\gamma}\un E_{\un\gamma}+C_{\ualpha\ubeta}{}^{\un A}\un E_{\un A})+X^\mu Y^\nu(\un\p_\mu\sigma^{\un\gamma}{}_\nu-\un\p_\nu\sigma^{\un\gamma}{}_\mu)\un E_{\un\gamma}\,.
\end{align*}
Since it follows from \eqref{Rsigma} that $R^\sigma(\uX,\uY)$ is purely vertical, it only has components in the $\un E_{\un A}$-direction. Thus, we can read off from the above equation that 
\begin{align}
C_{\ualpha\ubeta}{}^{\un\gamma}\sigma^{\ualpha}{}_\mu \sigma^{\ubeta}{}_\nu =-\un\p_\mu\sigma^{\un\gamma}{}_\nu+\un\p_\nu\sigma^{\un\gamma}{}_\mu\,,
\end{align}
which is equivalent to \eqref{Cabc}.

\section{Calculations for Lie Algebroid Trivializations}
\label{app:LAT}
\subsection{Connection and Curvature in a Local Trivialization}
\label{app:trivial}
Starting from the morphism condition of $\tau$, i.e.,
\begin{align}
\label{taumor1}
[\tau(\umX),\tau(\umY)]_{TM\oplus L}=\tau([\umX,\umY]_A)\,,
\end{align}
we now derive explicitly the results in \eqref{ttftf}, \eqref{bArelation} and \eqref{FbaseA}. Note that in this appendix section we work in a specific open set $U\subset M$ without specifying in the notation. 
\par
First we should define the Lie bracket on $TM\oplus L$. Given a basis $\{\un\p_\mu,\ut_A\}$, we can define the Lie bracket following condition (b) in Definition \ref{LA}:
\begin{align}
[\un\p_\mu,\un\p_\nu]_{TM\oplus L}&=0\,,\qquad[f\un\p_\mu,g\un\p_\nu]_{TM\oplus L}=f(\p_\mu g)\un\p_\nu-g(\p_\nu f)\un\p_\mu\,,\\
[\un\p_\mu,\ut_A]_{TM\oplus L}&=0\,,\qquad[f\un\p_\mu,g\ut_A]_{TM\oplus L}=f(\p_\mu g)\ut_A\,,\\
[\ut_A,\ut_B]_{TM\oplus L}&=f_{AB}{}^C\ut_C\,,\qquad[f\ut_A,g\ut_B]_{TM\oplus L}=fgf_{AB}{}^C\ut_C\,,\qquad f,g\in C^\infty(M)\,.
\end{align}
In the case where $\umX,\umY$ are both vertical, the condition \eqref{taumor1} gives
\begin{align*}
[\tau(\umX_V),\tau(\umY_V)]_{TM\oplus L}&=\tau([\umX_V,\umY_V]_A)\\
[\umX_V^\uA\tau(\uE_\uA),Y_V^\uB\tau(\uE_\uB)]_{TM\oplus L}&=\tau([\umX_V^\uA\uE_\uA,Y_V^\uB\uE_\uB]_A)\\
[\tau^A{}_\uA\ut_A,\tau^B{}_\uB\ut_B]_L&=\tau(C_{\un{AB}}{}^{\un C}\uE_{\un C})\\
\tau^A{}_\uA\tau^B{}_\uB f_{AB}{}^C\ut_C&=C_{\un{AB}}{}^{\un C}\tau^C{}_{\un C}\ut_C\,.
\end{align*}
Thus,
\begin{align}
\tau^A{}_\uA\tau^B{}_\uB f_{AB}{}^C&=C_{\un{AB}}{}^{\un C}\tau^C{}_{\un C}\,.
\end{align}
Applying $j^\uA{}_Dj^\uB{}_E$ to both sides of the above equation and considering \eqref{Cf} we get
\begin{align}
\label{ttftf1}
\tau^A{}_\uA j^\uA{}_D\tau^B{}_\uB j^\uB{}_E f_{AB}{}^C&=\tau^C{}_{\un C}j^{\un C}{}_Ff_{DE}{}^F\,.
\end{align}
\par
Now we take $\umX=\umX_H$ to be horizontal and $\umY=j(\un\mu)$ to be vertical. Then \eqref{taumor1} gives
\begin{align*}
[\tau(\umX_H),\tau(j(\un\mu))]_{TM\oplus L}&=\tau([\umX_H,j(\un\mu)]_A)\\
[\umX_H^{\ualpha}\tau(\un E_{\alpha}),\mu^C(\tau\circ j)(\un t_C)]_{TM\oplus L}&=\tau(j(\mu^A{\cal A}^B{}_A(\umX_H)\un t_B+(\rho(\umX_H)\mu^A)\un t_A))\\
[\umX_H^{\ualpha}\tau^\mu{}_{\ualpha}(\un\p_\mu+b^A{}_\mu \un t_A),\mu^C(\tau\circ j)^B{}_C\un t_B]_{TM\oplus L}&=\umX_H^{\ualpha}(\mu^A{\cal A}^B{}_A(\uE_{\ualpha})(\tau\circ j)^C{}_B\un t_C+(\rho(\uE_{\ualpha})\mu^A)(\tau\circ j)^B{}_A\un t_B)\\
\umX_H^{\ualpha}\tau^\mu{}_{\ualpha}\mu^D(\p_\mu(\tau\circ j)^C{}_D+b^A{}_\mu(\tau\circ j)^B{}_D f_{AB}{}^C)\un t_C&=\umX_H^{\ualpha}\mu^A{\cal A}_{\ualpha}{}^B{}_A(\tau\circ j)^C{}_B\un t_C\,,
\end{align*}
where we used \eqref{XHiota2} in the second step and the fact that $\tau^\mu{}_{\ualpha}=\rho^\mu{}_{\ualpha}$ in the last step. Then, we obtain that
\begin{align}
{\cal A}_{\ualpha}{}^D{}_C=((\tau\circ j)^{-1})^E{}_C(\rho^\mu{}_{\ualpha}b^A{}_\mu f_{AB}{}^C+\delta^C{}_B\rho^\mu{}_{\ualpha}\p_\mu)(\tau\circ j)^B{}_D\,.
\end{align}
Using \eqref{ttftf1}, this can be written alternatively as
\begin{align}
{\cal A}_{\ualpha}{}^D{}_C=\rho^\mu{}_{\ualpha}(b^A{}_\mu((\tau\circ j)^{-1}){}^B{}_A f_{BC}{}^D+((\tau\circ j)^{-1}\p_\mu(\tau\circ j))^D{}_C)\,.
\end{align}
\par
When $\umX=\umX_H$ and $\umY=\umY_H$ are both horizontal, the condition \eqref{taumor1} gives
\begin{align}
\label{taumorH}
[\tau(\umX_H),\tau(\umY_H)]_{TM\oplus L}&=\tau([\umX_H,\umY_H]_A)
\end{align}
The left-hand side of this equation can be evaluated as follows:
\begin{align*}
[\tau(\umX_H),\tau(\umY_H)]_{TM\oplus L}&=[\umX_H^{\ualpha}\tau^\mu{}_{\ualpha}(\un\p_\mu+b^A{}_\mu \un t_A),\umY_H^{\ubeta}\tau^\nu{}_{\ubeta}(\un\p_\nu+b^B{}_\nu \un t_B)]_{TM\oplus L}\nn\\
&=[\umX_H^{\ualpha}\tau^\mu{}_{\ualpha}\un\p_\mu,\umY_H^{\ubeta}\tau^\nu{}_{\ubeta}\un\p_\nu]_{TM\oplus L}+[\umX_H^{\ualpha}\tau^\mu{}_{\ualpha}\un\p_\mu,\umY_H^{\ubeta}\tau^\nu{}_{\ubeta} b^B{}_\nu \un t_B]_{TM\oplus L}\nn\\
&\quad+[\umX_H^{\ualpha}\tau^\mu{}_{\ualpha} b^A{}_\mu \un t_A,\umY_H^{\ubeta}\tau^\nu{}_{\ubeta}\un\p_\nu]_{TM\oplus L}+[\umX_H^{\ualpha}\tau^\mu{}_{\ualpha} b^A{}_\mu \un t_A,\umY_H^{\ubeta}\tau^\nu{}_{\ubeta} b^B{}_\nu \un t_B]_{TM\oplus L}\nn\\
&=\umX_H^{\ualpha}\tau^\mu{}_{\ualpha}\p_\mu(\umY_H^{\ubeta}\tau^\nu{}_{\ubeta})\un\p_\nu-\umY_H^{\ubeta}\tau^\nu{}_{\ubeta}\p_\nu(\umX_H^{\ualpha}\tau^\mu{}_{\ualpha})\un\p_\mu\nn\\
&\quad+\umX_H^{\ualpha}\tau^\mu{}_{\ualpha}\p_\mu(\umY_H^{\ubeta}\tau^\nu{}_{\ubeta} b^B{}_\nu)\un t_B-\umY_H^{\ubeta}\tau^\nu{}_{\ubeta}\p_\nu(\umX_H^{\ualpha}\tau^\mu{}_{\ualpha} b^A{}_\mu) \un t_A\nn\\
&\quad+\umX_H^{\ualpha}\tau^\mu{}_{\ualpha}\umY_H^{\ubeta}\tau^\nu{}_{\ubeta} b^A{}_\mu b^B{}_\nu f_{AB}{}^C\ut_C\,.
\end{align*}
Since we have $\tau^\mu{}_{\ualpha}=\rho^\mu{}_{\ualpha}$, we can notice that
\begin{align}
\umX_H^{\ualpha}\tau^\mu{}_{\ualpha}\un\p_\mu=\umX_H^{\ualpha}\rho^\mu{}_{\ualpha}\un\p_\mu=\rho(\umX_H)=\uX=X^\mu\un\p_\mu\,,
\end{align}
and thus $\umX_H^{\ualpha}\tau^\mu{}_{\ualpha}=X^\mu$. Hence, we have
\begin{align}
&\quad[\tau(\umX_H),\tau(\umY_H)]_{TM\oplus L}\nn\\
&= X^\mu\p_\mu Y^\nu\un\p_\nu-Y^\nu\p_\nu X^\mu\un\p_\mu+X^\mu\p_\mu(Y^\nu b^B{}_\nu)\ut_B-Y^\nu\p_\nu(X^\mu b^A{}_\mu) \ut_A+X^\mu Y^\nu b^A{}_\mu b^B{}_\nu f_{AB}{}^C\ut_C\nn\\
&=X^\mu\p_\mu Y^\nu(\un\p_\nu+b^B{}_\nu\ut_B)-Y^\nu\p_\nu X^\mu(\un\p_\mu+b^A{}_\mu\ut_A)+X^\mu Y^\nu(\p_\mu b^A{}_\nu-\p_\nu b^A{}_\mu+ b^B{}_\mu b^C{}_\nu f_{BC}{}^A)\ut_A\nn\\
&=[\uX,\uY]^\mu(\un\p_\mu+b^A{}_\mu\ut_A)+X^\mu Y^\nu(\p_\mu b^A{}_\nu-\p_\nu b^A{}_\mu+ b^B{}_\mu b^C{}_\nu f_{BC}{}^A)\ut_A\nn\\
\label{tauHtauH}
&=[\uX,\uY]^\mu \un D_\mu+X^\mu Y^\nu F^A{}_{\mu\nu}\ut_A\,,
\end{align}
where we defined $\un D_\mu\equiv\un\p_\mu+b^A{}_\mu\ut_A$ and the curvature of $b^A{}_\mu$:
\begin{align}
F^A{}_{\mu\nu}\equiv\p_\mu b^A{}_\nu-\p_\nu b^A{}_\mu+ b^B{}_\mu b^C{}_\nu f_{BC}{}^A
\end{align}
On the other hand, the right-hand side of \eqref{taumorH} is
\begin{align}
\tau([\umX_H,\umY_H]_A)&=\tau([\sigma(\uX),\sigma(\uY)]_A)\nn\\
&=\tau(\sigma([\uX,\uY]_{TM}))+\tau(R^\sigma(\uX,\uY))\nn\\
&=\tau(\sigma([\uX,\uY]^\nu\un\p_\nu))+\tau((R^\sigma(\uX,\uY))^\uA\uE_A)\nn\\
&=\tau(\sigma^{\ualpha}{}_\nu[\uX,\uY]^\nu\uE_{\ualpha})+\tau^A{}_\uA (R^\sigma(\uX,\uY))^\uA\ut_A\nn\\
&=\tau^\mu{}_{\ualpha}\sigma^{\ualpha}{}_\nu[\uX,\uY]^\nu\un D_\mu-\omega^A{}_\uA (R^\sigma(\uX,\uY))^\uA\ut_A\nn\\
&=\rho^\mu{}_{\ualpha}\sigma^{\ualpha}{}_\nu[\uX,\uY]^\nu\un D_\mu-\omega(R^\sigma(\uX,\uY))\nn\\
&=[\uX,\uY]^\mu\un D_\mu+\Omega(\umX_H,\umY_H)\nn\\
\label{tauHH}
&=[\uX,\uY]^\mu\un D_\mu+\Omega^A{}_{\ualpha\ubeta}\umX_H^{\ualpha}\umY_H^{\ubeta}\ut_A\,.
\end{align}
Comparing \eqref{tauHtauH} and \eqref{tauHH} and noticing that $\umX_H^{\ualpha}\rho^\mu{}_{\ualpha}=X^\mu$, we obtain
\begin{align}
F^A{}_{\mu\nu}\rho^\mu{}_{\ualpha}\rho^\nu{}_{\ubeta}=\Omega^A{}_{\ualpha\ubeta}\,,
\end{align}
which indicates that $F^A{}_{\mu\nu}$ also represents the curvature of the Lie algebroid.

\subsection{The Decomposition of $\hatd_\tau$ on a Trivialized Algebroid}
\label{app:hatd}
In this part of the appendix we present the calculation details of \eqref{hatdsalar} and \eqref{dhat1form}.
First, for an $E$-valued scalar $\un{\psi} = \psi^a \un{e}_a\in\Gamma(E)$. Using the Koszul formula \eqref{dhat}, we have
\begin{align}
\hatd_{\tau} \un{\psi}
&=\hat E^{\un M}\otimes\phi_E(\un {\hat E}_{\un M})(\psi^a \un{e}_a)\nn\\
&=\rho_\tau(\un {\hat E}_{\un M})\psi^a\hat E^{\un M}\otimes\un e_a+\psi^a\hat E^{\un M}\otimes\phi_E(\un {\hat E}_{\un M})(\un e_a)\nn\\
&=\rho_\tau(\un {\hat E}_{\ualpha})\psi^a\hat E^{\ualpha}\otimes\un e_a+\psi^a{\cal A}_{\ualpha}{}^b{}_a\hat E^{\ualpha}\otimes\un e_b - \psi^a\hat E^{\un A}\otimes v_E(\omega_\tau(\un {\hat E}_{\un A}))^b{}_a \un{e}_b\nn\\
&=\rho_\tau^{\mu}{}_{\un{\alpha}}\left(\partial_{\mu} \psi^a + b^A{}_{\mu}v_E(\un t_A)^a{}_b \psi^b\right)  \sigma_\tau^{\un{\alpha}}{}_{\nu} \td x^\nu\otimes \un{e}_a- \omega_\tau^A{}_{\un A}v_E(\un t_A)^a{}_b \psi^b j_\tau^{\un{A}}{}_{B}(t^B-b^B_\mu \td x^\mu)\otimes\un{e}_{a}\nn\\
\label{hatddsscalar}
&=
 \Big(\td \psi^a
 +  v_E(\un t_A)^a{}_b\varpi^A \psi^b  \Big)\otimes\un{e}_a \,,
\end{align}
where in the second equality we used \eqref{phiEf}, in the third equality we used \eqref{defSpinConnHV}, in the fourth equality we plugged in \eqref{taubasis}, \eqref{taudualbasis} and \eqref{bArelationE4},\footnote{This derivation can also be done in the trivialization introduced in \eqref{trivialbasis} without introducing the basis \eqref{taubasis}, \eqref{taudualbasis} for the trivialized algebroid. In this case the linear relation \eqref{bArelationE4} does not hold and one should use \eqref{bArelationE3}. However, the inhomogeneous term therein can be absorbed by redefining the Maurer-Cartan form $\varpi$ and so the algebra proceeds similarly.} and in the last step we used the fact that $\rho_\tau^{\mu}{}_{\un{\alpha}} \sigma_\tau^{\un{\alpha}}{}_{\nu}=\delta^\mu{}_\nu$ and $\omega_\tau^A{}_{\un A}j_\tau^{\un{A}}{}_{B}=\delta^A{}_B$.

Next, we consider $\beta\in\Gamma(A_\tau^*\times E)$. Employing the Koszul formula (which is most easily employed by translating $\alpha$ into the covariant split basis), we find
\beqn
\hatd_\tau\beta&=&\hatd_\tau(\beta^a_{\un M}\hat E^{\un M}\otimes\un e_a)\nn\\
&=&\frac12\hat E^{\un M}\wedge \hat E^{\un N}\otimes\Big(
\phi_E(\un {\hat E}_{\un M})(\beta^a_{\un N}\un e_a)
-\phi_E(\un {\hat E}_{\un N})(\beta^a_{\un M}\un e_a)
-\beta([\un {\hat E}_{\un M},\un {\hat E}_{\un N}]_{A_\tau})\Big)\nn\\
&=&-\Big(\hat E^{\un N}\wedge\Big(\td \beta^a_{\un N} +  v_E(\un t_A)^a{}_bt^A \beta^b_{\un N}  \Big)
+\frac12C_{\un M\un N}{}^{\un P}\beta^a_{\un P} \hat E^{\un M}\wedge \hat E^{\un N}\Big)\otimes\un e_a
\nn\\
&=&\Big(\td \beta^a_{\un\alpha} +  v_E(\un t_A)^a{}_bt^A \beta^b_{\un\alpha}  \Big)\wedge\hat E^{\un\alpha}\otimes\un e_a
+\Big(\td \beta^a_{\un B} +  v_E(\un t_A)^a{}_bt^A \beta^b_{\un B}  \Big)\wedge\hat E^{\un B}\otimes\un e_a
-\frac12C_{\un\alpha\un\beta}{}^{\un\gamma}\beta^a_{\un\gamma} \hat E^{\un\alpha}\wedge \hat E^{\un\beta}\otimes\un e_a
\nn\\&&
-\frac12C_{\un\alpha\un\beta}{}^{\un C}\beta^a_{\un C} \hat E^{\un\alpha}\wedge\hat E^{\un\beta}\otimes\un e_a
-C_{\un\alpha\un A}{}^{\un B}\beta^a_{\un B}\hat E^{\un\alpha}\wedge \hat E^{\un A}\otimes\un e_a
-\frac12C_{\un A\un B}{}^{\un C}\beta^a_{\un C} \hat E^{\un A}\wedge \hat E^{\un B}\otimes\un e_a
\nn\\
&=&\bigg[\Big(\td (\sigma_\tau^{\un\alpha}{}_\nu\beta^a_{\un\alpha}) +  v_E(\un t_A)^a{}_bt^A (\sigma_\tau^{\un\alpha}{}_\nu\beta^b_{\un\alpha})  \Big)\wedge \td x^\nu
+\Big(\td (j_\tau^{\un B}{}_B\beta^a_{\un B}) +  v_E(\un t_A)^a{}_bt^A (j_\tau^{\un B}{}_B\beta^b_{\un B})  \Big)\wedge (t^B-b^B_\nu \td x^\nu)
\nn\\&&
-\frac12F^B{}_{\mu\nu}\beta^a_{B}  \td x^\mu\wedge \td x^\nu
-
\sigma_\tau^{\un\alpha}{}_\mu {\cal A}_{\un\alpha}{}^A{}_B\beta^a_{A}   \td x^\mu\wedge  (t^B-b^B_\nu \td x^\nu)\nn\\
&&
-\frac12f_{AB}{}^{C}\beta^a_{C}  (t^A-b^A_\mu \td x^\mu)\wedge  (t^B-b^B_\nu \td x^\nu)\bigg]\otimes\un e_a
\nn\\
&=&\bigg[\Big(\td (\sigma_\tau^{\un\alpha}{}_\nu\beta^a_{\un\alpha}-j_\tau^{\un B}{}_B\beta^a_{\un B}b_\nu^B) +  v_E(\un t_A)^a{}_bt^A (\sigma_\tau^{\un\alpha}{}_\nu\beta^a_{\un\alpha}-j_\tau^{\un B}{}_B\beta^a_{\un B}b_\nu^B)  \Big)\wedge \td x^\nu
\nn\\&&
+\Big(\td (j_\tau^{\un B}{}_B\beta^a_{\un B}) +  v_E(\un t_A)^a{}_bt^A (j_\tau^{\un B}{}_B\beta^b_{\un B}) -\frac12f_{AB}{}^{C}(j_\tau^{\un B}{}_C\beta^b_{\un B}) t^A
 \Big)\wedge t^B
\nn\\&&+(\sigma_\tau^{\un\alpha}{}_\nu {\cal A}_{\un\alpha}{}^C{}_A+f_{AB}{}^{C}b^B_\nu) (j_\tau^{\un B}{}_C\beta^b_{\un B}) t^A\wedge \td x^\nu+(\sigma_\tau^{\un\alpha}{}_\mu {\cal A}_{\un\alpha}{}^C{}_B-f_{AB}{}^Cb_\mu^A) b_\nu^B(j_\tau^{\un B}{}_C\beta^a_{\un B})  \td x^\mu\wedge \td x^\nu\Big]\otimes\un e_a\nn\\
&=&\Big(\td (\sigma_\tau^{\un\alpha}{}_\nu\beta^a_{\un\alpha}-j_\tau^{\un B}{}_B\beta^a_{\un B}b_\nu^B) +  v_E(\un t_A)^a{}_bt^A (\sigma_\tau^{\un\alpha}{}_\nu\beta^a_{\un\alpha}-j_\tau^{\un B}{}_B\beta^a_{\un B}b_\nu^B)  \Big)\wedge \td x^\nu\otimes\un e_a
\nn\\&&
+\Big(\td (j_\tau^{\un B}{}_B\beta^a_{\un B}) +  v_E(\un t_A)^a{}_bt^A (j_\tau^{\un B}{}_B\beta^b_{\un B}) -\frac12f_{AB}{}^{C}(j_\tau^{\un B}{}_C\beta^b_{\un B}) t^A
 \Big)\wedge t^B\otimes\un e_a\,,
\eeqn
where in the third equality we applied the result from \eqref{hatddsscalar}, in the fifth equality we plugged in the commutation coefficients \eqref{Cabc}--\eqref{CABC}, and in the last equality the terms are canceled by means of \eqref{bArelationE4}. Recognizing from \eqref{betaspliting} that $\beta^a_\nu=\sigma_\tau^{\un\alpha}{}_\nu\beta^a_{\un\alpha}-j_\tau^{\un B}{}_B\beta^a_{\un B}b_\nu^B$ and $\beta^a_A=j_\tau^{\un B}{}_A\beta^a_{\un B}$, we obtain the result in \eqref{dhat1form}:
\beq
\hatd_\tau\beta=
\Big(\td \beta^a_\nu +  v_E(\un t_A)^a{}_bt^A \beta^a_\nu\Big)\wedge \td x^\nu\otimes\un e_a
+\Big(\td \beta^a_B +  v_E(\un t_A)^a{}_bt^A \beta^b_B -\frac12f_{AB}{}^{C}\beta^a_C t^A
 \Big)\wedge t^B\otimes\un e_a\,.
\eeq

\section{The Free Variation of the Chern-Simons Form}
\label{app:covanomaly}
In Subsection \ref{sec:covariant}, we introduced that the covariant anomaly can be derived by taking the free variation of the Chern-Simons form $\scr C_Q(\omega)$ in the covariant splitting, as shown in equation \eqref{Covariant Anomaly}. We will now provide an explicit demonstration of this derivation. Following the approach presented in \cite{bardeen1984consistent}, we introduce a nilpotent operator $K:\Omega^{p}(A;L)\to\Omega^{p-1}(A;L)$ that acts as follows:
\begin{align}
K\omega=0\,,\qquad K\Omega=\delta\omega\,,\qquad K\delta\omega=0\,.
\end{align}
Then, the variation operator on $\omega$ and $\Omega$ can be written as
\begin{align}
\delta=K\hat\td+\hat\td K\,.
\end{align}
When performing the variation of the Chern-Simons form:
\begin{align}
\label{KddKC}
\delta\mathscr{C}_Q=K\hat\td\mathscr{C}_Q+\hat\td K\mathscr{C}_Q\,,
\end{align}
the second term is a total derivative, and thus all we have to show is that the first term in \eqref{KddKC} gives rise to the first term in \eqref{Covariant Anomaly}, namely $\beta^{(2l-2,1)}(\delta\omega, \Omega)$. Using the transgression formula \eqref{transgression}, one finds
\begin{align}
K\hat\td\mathscr{C}_Q(\omega)={}& Q_{A_1 \cdots A_l} \int_{0}^{1} \td t\,\delta\omega^{A_1} \wedge_{j = 2}^{l} \left(t\Omega + \frac{1}{2}(t^2-t)[\omega , \omega]_{L}  \right)^{A_j}\nn\\
&+(l-1)Q_{A_1 \cdots A_l} \int_{0}^{1} \td t\,\td\omega^{A_1} t\delta\omega^{A_2} \wedge_{j = 3}^{l} \left(t\Omega + \frac{1}{2}(t^2-t)[\omega , \omega]_{L}  \right)^{A_j}\nn\\
&+(l-1) Q_{A_1 \cdots A_l} \int_{0}^{1} \td t\, \omega^{A_1} t^2[\delta\omega , \omega]_{L}^{A_2} \wedge_{j = 3}^{l} \left(t\Omega + \frac{1}{2}(t^2-t)[\omega , \omega]_{L}  \right)^{A_j}\nn\\
&+(l-1)(l-2) Q_{A_1 \cdots A_l} \int_{0}^{1} \td t\, \omega^{A_1} t^3[\Omega , \omega]_{L}^{A_2} \delta\omega^{A_3}\wedge_{j = 4}^{l} \left(t\Omega + \frac{1}{2}(t^2-t)[\omega , \omega]_{L}  \right)^{A_j}\nn\\
={}& Q_{A_1 \cdots A_l} \int_{0}^{1} \td t\,\delta\omega^{A_1} \wedge_{j = 2}^{l} \left(t\Omega + \frac{1}{2}(t^2-t)[\omega , \omega]_{L}  \right)^{A_j}\nn\\
&+(l-1)Q_{A_1 \cdots A_l} \int_{0}^{1} \td t\,\td\omega^{A_1} t\delta\omega^{A_2} \wedge_{j = 3}^{l} \left(t\Omega + \frac{1}{2}(t^2-t)[\omega , \omega]_{L}  \right)^{A_j}\nn\\
&+(l-1) Q_{A_1 \cdots A_l} \int_{0}^{1} \td t\,t^2 \delta\omega^{A_1} [ \omega,\omega]^{A_2} \wedge_{j = 3}^{l} \left(t\Omega + \frac{1}{2}(t^2-t)[\omega , \omega]_{L}  \right)^{A_j}\nn\\
\label{KdC1}
={}&l Q_{A_1 \cdots A_l} \int_{0}^{1} \td t\,\delta\omega^{A_1} \wedge_{j = 2}^{l} \left(t\Omega + \frac{1}{2}(t^2-t)[\omega , \omega]_{L}  \right)^{A_j}\nn\\
&+\frac{l-1}{2}Q_{A_1 \cdots A_l} \int_{0}^{1} \td t\,t^2\delta\omega^{A_1} [\omega , \omega]_{L}^{A_2}\wedge_{j = 3}^{l} \left(t\Omega + \frac{1}{2}(t^2-t)[\omega , \omega]_{L}  \right)^{A_j}\,,
\end{align}
To further evaluate this, we first perform the integral of the following form:
\begin{align}
&\int_{0}^{1} \td t \left[ l(t A+\frac{t^2-t}{2}B)^{l-1}+\frac{l-1}{2}t^2 B(t A+\frac{t^2-t}{2}B)^{l-2} \right] \nn\\
=&\int_{0}^{1} \td t \left[ l\sum_{n=0}^{l-1}C^n_{l-1}t^n (\frac{t^2-t}{2})^{l-1-n}A^nB^{l-1-n}+\frac{l-1}{2}\sum_{n=0}^{l-2}C^n_{l-2} t^{n+2}(\frac{t^2-t}{2})^{l-2-n}A^nB^{l-1-n} \right] \nn\\
=&\int_{0}^{1} \td t \left[lt^{l-1}A^{l-1}+ \sum_{n=0}^{l-2}\left(lC^n_{l-1}t^n (\frac{t^2-t}{2})^{l-1-n}+\frac{l-1}{2}C^n_{l-2} t^{n+2}(\frac{t^2-t}{2})^{l-2-n}\right)A^nB^{l-1-n} \right] \nn\\
=&\int_{0}^{1} \td t \left[lt^{l-1}A^{l-1}+ \sum_{n=0}^{l-2}\left(\frac{l(l-1)!}{n!(l-1-n)!}\frac{t^2-t}{2}+\frac{(l-1)(l-2)!}{n!(l-2-n)!}\frac{t^{2}}{2}\right)t^n(\frac{t^2-t}{2})^{l-2-n}A^nB^{l-1-n} \right] \nn\\
=&\int_{0}^{1} \td t \left[lt^{l-1}A^{l-1}+ \sum_{n=0}^{l-2}\frac{(l-1)!}{n!(l-1-n)!}\left(l\frac{t-1}{2}+(l-1-n)\frac{t}{2}\right)t^{n+1}(\frac{t^2-t}{2})^{l-2-n}A^nB^{l-1-n} \right] \nn\\
=&\int_{0}^{1} \td t \left[lt^{l-1}A^{l-1}+ \sum_{n=0}^{l-2}\frac{(l-1)!}{n!(l-1-n)!2^{l-1-n}}[l(t-1)+(l-1-n)t]t^{l-1}(t-1)^{l-2-n}A^nB^{l-1-n} \right] \nn\\
=&\int_{0}^{1} \td t lt^{l-1}A^{l-1}+ \sum_{n=0}^{l-2}\frac{(l-1)!A^nB^{l-1-n}}{n!(l-1-n)!2^{l-1-n}}t^{l}(t-1)^{l-1-n}\Big|_0^1 \nn\\
=&A^{l-1}\,.
\end{align}
Then, taking $A$ as $\Omega$ and $B$ as $[\omega,\omega]_L$, the integral in \eqref{KdC1} yields 
\begin{align}
K\hat\td\mathscr{C}_Q(\omega)=Q(\underbrace{\Omega,\ldots,\Omega}_{l-1},\delta\omega)\,.
\end{align}
Now we can compare this with $\beta^{(2l-2,1)}(\delta\omega, \Omega)$. From \eqref{transgression}, one can pick up the term with a single $\omega$ and find
\begin{align}
 \beta^{(2l-2,1)}(\omega, \Omega)=Q_{A_1 \cdots A_l} \int_{0}^{1} \td t\,\omega^{A_1} t^{l-1}\wedge_{j = 2}^{l} \Omega^{A_j}=\frac{1}{l}Q(\underbrace{\Omega,\ldots,\Omega}_{l-1},\omega)\,,
\end{align}
and hence
\begin{equation} 
\label{beta2l-2a}
\beta^{(2l-2,1)}(\delta\omega, \Omega)= \frac{1}{l}Q(\underbrace{\Omega,\ldots,\Omega}_{l-1},\delta\omega)\,.
\end{equation}
Therefore, we can see that \eqref{KddKC} can be written as
\begin{equation} 
\label{deltaC}
\delta \mathscr{C}_Q(\omega) = l \beta^{(2l-2,1)}(\delta\omega, \Omega) + \hat\td\Theta(\omega,\delta\omega)\,,
\end{equation}
where $\Theta\equiv K\mathscr{C}_Q$. The covariant anomaly can be read off from the first term, while the $\Theta$ in the second term serves as the Bardeen-Zumino polynomial which covariantizes the consistent anomaly when added to the anomalous current \cite{bardeen1984consistent}.


\printbibliography[heading=bibintoc,title={References}]

@article{Freidel:2021cjp,
	archiveprefix = {arXiv},
	author = {Freidel, Laurent and Oliveri, Roberto and Pranzetti, Daniele and Speziale, Simone},
	doi = {10.1007/JHEP09(2021)083},
	eprint = {2104.12881},
	journal = {JHEP},
	pages = {083},
	primaryclass = {hep-th},
	title = {{Extended corner symmetry, charge bracket and Einstein's equations}},
	volume = {09},
	year = {2021}}

@article{Freidel:2021fxf,
	archiveprefix = {arXiv},
	author = {Freidel, Laurent and Oliveri, Roberto and Pranzetti, Daniele and Speziale, Simone},
	doi = {10.1007/JHEP07(2021)170},
	eprint = {2104.05793},
	journal = {JHEP},
	pages = {170},
	primaryclass = {hep-th},
	title = {{The Weyl BMS group and Einstein's equations}},
	volume = {07},
	year = {2021}}

@article{Weyl:1918pdp,
	author = {Weyl, Hermann},
	doi = {10.1007/BF01199420},
	journal = {Math. Z.},
	number = {3-4},
	pages = {384--411},
	title = {{Reine Infinitesimalgeometrie}},
	volume = {2},
	year = {1918}}

@article{CartanLesE,
	author = {Cartan, {\'E}lie},
	journal = {Ann. Soc. pol. Math.},
	pages = {171--221},
	title = {Les espaces {\`a} connexion conforme},
	volume = {2},
	year = {1923}}

@article{thomas1925invariants,
	author = {Thomas, Tracy Yerkes},
	journal = {Proc. Natl. Acad. Sci.},
	number = {12},
	pages = {722--725},
	publisher = {National Acad Sciences},
	title = {Invariants of relative quadratic differential forms},
	volume = {11},
	year = {1925}}

@article{Penrose:1962ij,
	author = {Penrose, Roger},
	doi = {10.1103/PhysRevLett.10.66},
	journal = {Phys. Rev. Lett.},
	pages = {66--68},
	title = {{Asymptotic properties of fields and space-times}},
	volume = {10},
	year = {1963}}

@article{Mannheim:2011ds,
	archiveprefix = {arXiv},
	author = {Mannheim, Philip D.},
	doi = {10.1007/s10701-011-9608-6},
	eprint = {1101.2186},
	journal = {Found. Phys.},
	pages = {388--420},
	primaryclass = {hep-th},
	title = {Making the Case for Conformal Gravity},
	volume = {42},
	year = {2012}}

@article{Maldacena:1997re,
	archiveprefix = {arXiv},
	author = {Maldacena, Juan Martin},
	doi = {10.1023/A:1026654312961},
	eprint = {hep-th/9711200},
	journal = {Adv. Theor. Math. Phys.},
	pages = {231--252},
	reportnumber = {HUTP-97-A097, HUTP-98-A097},
	title = {{The Large N limit of superconformal field theories and supergravity}},
	volume = {2},
	year = {1998}}

@article{Witten:1998qj,
	archiveprefix = {arXiv},
	author = {Witten, Edward},
	doi = {10.4310/ATMP.1998.v2.n2.a2},
	eprint = {hep-th/9802150},
	journal = {Adv. Theor. Math. Phys.},
	pages = {253--291},
	reportnumber = {IASSNS-HEP-98-15},
	title = {{Anti-de Sitter space and holography}},
	volume = {2},
	year = {1998}}

@article{Capper:1974ic,
	author = {Capper, D. M. and Duff, M. J.},
	doi = {10.1007/BF02748300},
	journal = {Nuovo Cim. A},
	pages = {173--183},
	title = {{Trace anomalies in dimensional regularization}},
	volume = {23},
	year = {1974}}

@article{Deser:1993yx,
	archiveprefix = {arXiv},
	author = {Deser, Stanley and Schwimmer, A.},
	doi = {10.1016/0370-2693(93)90934-A},
	eprint = {hep-th/9302047},
	journal = {Phys. Lett. B},
	pages = {279--284},
	reportnumber = {BRX-343, SISSA-14-93-EP},
	title = {{Geometric classification of conformal anomalies in arbitrary dimensions}},
	volume = {309},
	year = {1993}}

@book{alexakis2012decomposition,
	author = {Alexakis, Spyros},
	doi = {10.23943/princeton/9780691153476.001.0001},
	publisher = {Princeton University Press},
	title = {{The Decomposition of Global Conformal Invariants (AM-182)}},
	year = {2012}}

@article{cotton1899varietes,
	address = {Toulouse},
	author = {Cotton, \'E.},
	journal = {Annales de la Facult{\'e} des Sciences de Toulouse},
	number = {4},
	pages = {385--438},
	publisher = {\'Ed. Privat, Imprimeur-Libraire},
	series = {2},
	title = {Sur les vari\'et\'es \`a trois dimensions},
	volume = {1},
	year = {1899}}

@article{bach1921weylschen,
	author = {Bach, Rudolf},
	doi = {10.1007/BF01378338},
	journal = {Math. Zeit. 9 (1921), 110.},
	number = {1},
	pages = {110--135},
	publisher = {Springer},
	title = {Zur weylschen relativit{\"a}tstheorie und der weylschen erweiterung des kr{\"u}mmungstensorbegriffs},
	volume = {9},
	year = {1921}}

@article{fefferman1979parabolic,
	author = {Fefferman, Charles},
	journal = {Adv. Math.},
	number = {2},
	pages = {131--262},
	publisher = {Academic Press},
	title = {Parabolic invariant theory in complex analysis},
	volume = {31},
	year = {1979}}

@article{graham2005ambient,
	archiveprefix = {arXiv},
	author = {Graham, C Robin and Hirachi, Kengo},
	doi = {10.4171/013-1/3},
	eprint = {math/0405068},
	journal = {IRMA Lect. Math. Theor. Phys.},
	pages = {59--71},
	primaryclass = {math.DG},
	publisher = {Eur. Math. Soc. Z{\"u}rich},
	title = {{The ambient obstruction tensor and Q-curvature}},
	volume = {8},
	year = {2005}}

@article{branson1991explicit,
	author = {Branson, Thomas P and {\O}rsted, Bent},
	journal = {Proc. Amer. Math. Soc.},
	number = {3},
	pages = {669--682},
	title = {Explicit functional determinants in four dimensions},
	volume = {113},
	year = {1991}}

@article {Branson1995,
    AUTHOR = {Branson, Thomas P.},
     TITLE = {Sharp inequalities, the functional determinant, and the
              complementary series},
   JOURNAL = {Trans. Amer. Math. Soc.},
  FJOURNAL = {Transactions of the American Mathematical Society},
    VOLUME = {347},
      YEAR = {1995},
    NUMBER = {10},
     PAGES = {3671--3742},
      ISSN = {0002-9947,1088-6850},
   MRCLASS = {58G26 (22E46 53A30)},
  MRNUMBER = {1316845},
MRREVIEWER = {Friedbert\ Pr\"{u}fer},
       DOI = {10.2307/2155203}
}

@article{Fefferman2003,
   title={Ambient metric construction of $Q$-curvature in conformal and CR geometries},
   volume={10},
   ISSN={1945-001X},
   url={http://dx.doi.org/10.4310/MRL.2003.v10.n6.a9},
   DOI={10.4310/mrl.2003.v10.n6.a9},
   number={6},
   journal={Math. Res. Lett.},
   publisher={International Press of Boston},
   author={Fefferman, Charles and Hirachi, Kengo},
   year={2003},
   pages={819–831} }

@article{Chang2008,
	author = {Chang, S.-Y. Alice and Eastwood, Michael and {\O}rsted, Bent and Yang, Paul C.},
	date = {2008},
	doi = {10.1007/s10440-008-9229-z},
	id = {Chang2008},
	journal = {Acta. Appl. Math.},
	number = {2},
	pages = {119--125},
	title = {What is {Q}-Curvature?},
	volume = {102},
	year = {2008}}

@article{10.4310/jdg/1214429379,
	author = {Gerald B. Folland},
	doi = {10.4310/jdg/1214429379},
	journal = {J. Diff. Geom.},
	number = {2},
	pages = {145 -- 153},
	publisher = {Lehigh University},
	title = {{Weyl manifolds}},
	volume = {4},
	year = {1970}}

@article{doi:10.1063/1.529582,
	author = {Hall,G. S.},
	doi = {10.1063/1.529582},
	journal = {J. Math. Phys.},
	number = {7},
	pages = {2633-2638},
	title = {Weyl manifolds and connections},
	volume = {33},
	year = {1992}}

@article{scholz2018unexpected,
	archiveprefix = {arXiv},
	author = {Scholz, Erhard},
	doi = {10.1007/978-1-4939-7708-6_11},
	eprint = {1703.03187},
	journal = {Einstein Stud.},
	pages = {261--360},
	primaryclass = {math.HO},
	title = {{The unexpected resurgence of Weyl geometry in late 20-th century physics}},
	volume = {14},
	year = {2018}}

@article{Jia:2021hgy,
	archiveprefix = {arXiv},
	author = {Jia, Weizhen and Karydas, Manthos},
	doi = {10.1103/PhysRevD.104.126031},
	eprint = {2109.14014},
	journal = {Phys. Rev. D},
	number = {12},
	pages = {126031},
	primaryclass = {hep-th},
	title = {{Obstruction tensors in Weyl geometry and holographic Weyl anomaly}},
	volume = {104},
	year = {2021}}

@article{Fefferman:2007rka,
	archiveprefix = {arXiv},
	author = {Fefferman, Charles and Graham, C. Robin},
	eprint = {0710.0919},
	journal = {Ann. Math. Stud.},
	pages = {1--128},
	primaryclass = {math.DG},
	title = {{The ambient metric}},
	volume = {178},
	year = {2011}}

@incollection{AST_1985__S131__95_0,
	address = {Paris},
	author = {Fefferman, Charles and Graham, C. Robin},
	booktitle = {\'Elie Cartan et les math\'ematiques d'aujourd'hui (Lyon, 1984)},
	language = {en},
	mrnumber = {837196},
	number = {S131},
	publisher = {Soci\'et\'e math\'ematique de France},
	series = {Ast\'erisque},
	title = {Conformal invariants},
	year = {1985},
	zbl = {0602.53007}}

@article{Graham:1991jqw,
	author = {Graham, C. Robin and Lee, John M.},
	doi = {10.1016/0001-8708(91)90071-E},
	journal = {Adv. Math.},
	number = {2},
	pages = {186--225},
	title = {{Einstein metrics with prescribed conformal infinity on the ball}},
	volume = {87},
	year = {1991}}

@incollection{graham2008inhomogeneous,
	archiveprefix = {arXiv},
	author = {Graham, C Robin and Hirachi, Kengo},
	booktitle = {Symmetries and Overdetermined Systems of Partial Differential Equations},
	doi = {10.1007/978-0-387-73831-4_20},
	eprint = {math/0611931},
	pages = {403--420},
	primaryclass = {math.DG},
	publisher = {Springer},
	title = {Inhomogeneous ambient metrics},
	year = {2008}}

@article{2001math.....10271F,
	archiveprefix = {arXiv},
	author = {Charles Fefferman and Graham, {C. Robin}},
	doi = {10.4310/MRL.2002.v9.n2.a2},
	eprint = {math/0110271},
	issn = {1073-2780},
	journal = {Math. Res. Lett.},
	language = {English (US)},
	number = {2-3},
	pages = {139--151},
	primaryclass = {math.DG},
	publisher = {International Press of Boston, Inc.},
	title = {{Q-curvature and Poincar{\'e} metrics}},
	volume = {9},
	year = {2002}}

@article{10.2996/kmj/1138845392,
	author = {Koichi Ogiue},
	doi = {10.2996/kmj/1138845392},
	journal = {Kodai Mathematical Seminar Reports},
	number = {2},
	pages = {193 -- 224},
	publisher = {Tokyo Institute of Technology, Department of Mathematics},
	title = {{Theory of conformal connections}},
	volume = {19},
	year = {1967}}

@article{Graham2007jet,
	archiveprefix = {arXiv},
	author = {Graham, Robin C.},
	eprint = {0710.1671},
	journal = {Arch. Math.},
	language = {eng},
	number = {5},
	pages = {389-415},
	primaryclass = {math.DG},
	publisher = {Department of Mathematics, Faculty of Science of Masaryk University, Brno},
	title = {Jet isomorphism for conformal geometry},
	volume = {043},
	year = {2007}}

@article{Ciambelli:2019lap,
	archiveprefix = {arXiv},
	author = {Ciambelli, Luca and Leigh, Robert G. and Marteau, Charles and Petropoulos, P. Marios},
	doi = {10.1103/PhysRevD.100.046010},
	eprint = {1905.02221},
	journal = {Phys. Rev. D},
	number = {4},
	pages = {046010},
	primaryclass = {hep-th},
	reportnumber = {CPHT-RR025.052019, CPHT-RR010.022019},
	title = {Carroll Structures, Null Geometry and Conformal Isometries},
	volume = {100},
	year = {2019}}

@article{Ciambelli:2018xat,
	archiveprefix = {arXiv},
	author = {Ciambelli, Luca and Marteau, Charles and Petkou, Anastasios C. and Petropoulos, P. Marios and Siampos, Konstantinos},
	doi = {10.1088/1361-6382/aacf1a},
	eprint = {1802.05286},
	journal = {Class. Quant. Grav.},
	number = {16},
	pages = {165001},
	primaryclass = {hep-th},
	reportnumber = {CPHT-RR048.082017, CERN-TH-2017-228},
	title = {{Covariant Galilean versus Carrollian hydrodynamics from relativistic fluids}},
	volume = {35},
	year = {2018}}

@article{Pasterski:2016qvg,
	archiveprefix = {arXiv},
	author = {Pasterski, Sabrina and Shao, Shu-Heng and Strominger, Andrew},
	doi = {10.1103/PhysRevD.96.065026},
	eprint = {1701.00049},
	journal = {Phys. Rev. D},
	number = {6},
	pages = {065026},
	primaryclass = {hep-th},
	title = {Flat Space Amplitudes and Conformal Symmetry of the Celestial Sphere},
	volume = {96},
	year = {2017}}

@article{Raclariu:2021zjz,
	archiveprefix = {arXiv},
	author = {Raclariu, Ana-Maria},
	eprint = {2107.02075},
	primaryclass = {hep-th},
	title = {Lectures on Celestial Holography},
	year = {2021}}

@inproceedings{Pasterski:2021raf,
	archiveprefix = {arXiv},
	author = {Pasterski, Sabrina and Pate, Monica and Raclariu, Ana-Maria},
	booktitle = {{2022 Snowmass Summer Study}},
	eprint = {2111.11392},
	primaryclass = {hep-th},
	title = {Celestial Holography},
	year = {2021}}

@article{Akal:2020wfl,
	archiveprefix = {arXiv},
	author = {Akal, Ibrahim and Kusuki, Yuya and Takayanagi, Tadashi and Wei, Zixia},
	doi = {10.1103/PhysRevD.102.126007},
	eprint = {2007.06800},
	journal = {Phys. Rev. D},
	number = {12},
	pages = {126007},
	primaryclass = {hep-th},
	reportnumber = {YITP-20-91, IPMU20-0079},
	title = {{Codimension two holography for wedges}},
	volume = {102},
	year = {2020}}

@article{Ogawa:2022fhy,
    author = "Ogawa, Naoki and Takayanagi, Tadashi and Tsuda, Takashi and Waki, Takahiro",
    title = "{Wedge holography in flat space and celestial holography}",
    eprint = "2207.06735",
    archivePrefix = "arXiv",
    primaryClass = "hep-th",
    reportNumber = "YITP-22-71, IPMU22-0036",
    doi = "10.1103/PhysRevD.107.026001",
    journal = "Phys. Rev. D",
    volume = "107",
    number = "2",
    pages = "026001",
    year = "2023"
}

@article{Parisini:2022wkb,
    author = "Parisini, Enrico and Skenderis, Kostas and Withers, Benjamin",
    title = "{Embedding formalism for CFTs in general states on curved backgrounds}",
    eprint = "2209.09250",
    archivePrefix = "arXiv",
    primaryClass = "hep-th",
    doi = "10.1103/PhysRevD.107.066022",
    journal = "Phys. Rev. D",
    volume = "107",
    number = "6",
    pages = "066022",
    year = "2023"
}

@inbook{Stora1977,
	address = {Boston, MA},
	author = {Stora, R.},
	booktitle = {New Developments in Quantum Field Theory and Statistical Mechanics Carg{\`e}se 1976},
	doi = {10.1007/978-1-4615-8918-1_8},
	editor = {L{\'e}vy, Maurice and Mitter, Pronob},
	isbn = {978-1-4615-8918-1},
	pages = {201--224},
	publisher = {Springer US},
	title = {Continuum Gauge Theories},
	year = {1977}}

@article{Sorella:1992dr,
	archiveprefix = {arXiv},
	author = {Sorella, Silvio P.},
	doi = {10.1007/BF02099759},
	eprint = {hep-th/9302136},
	journal = {Commun. Math. Phys.},
	pages = {231--243},
	reportnumber = {UGVA-DPT-1992-08-781},
	title = {{Algebraic characterization of the Wess-Zumino consistency conditions in gauge theories}},
	volume = {157},
	year = {1993}}

@article{Boulanger:2007ab,
	archiveprefix = {arXiv},
	author = {Boulanger, Nicolas},
	doi = {10.1103/PhysRevLett.98.261302},
	eprint = {0706.0340},
	journal = {Phys. Rev. Lett.},
	pages = {261302},
	primaryclass = {hep-th},
	title = {Algebraic Classification of {Weyl} Anomalies in Arbitrary Dimensions},
	volume = {98},
	year = {2007}}

@article{ThierryMieg:1987um,
	author = {Thierry-Mieg, Jean},
	doi = {10.1016/0370-2693(87)90402-3},
	journal = {Phys. Lett. B},
	pages = {368},
	reportnumber = {DAMTP-87-14},
	title = {{BRS} Analysis of {Zamolodchikov's} Spin Two and Three Current Algebra},
	volume = {197},
	year = {1987}}

@article{Dragon:1996md,
	archiveprefix = {arXiv},
	author = {Dragon, Norbert},
	eprint = {hep-th/9602163},
	reportnumber = {ITP-UH-3-96},
	title = {{BRS} Symmetry and Cohomology},
	year = {1996}}

@article{Attard:2019pvw,
	archiveprefix = {arXiv},
	author = {Attard, Jeremy and Fran\c{c}ois, Jordan and Lazzarini, Serge and Masson, Thierry},
	doi = {10.1016/j.geomphys.2019.103541},
	eprint = {1904.04915},
	journal = {J. Geom. Phys.},
	pages = {103541},
	primaryclass = {math-ph},
	title = {{Cartan connections and Atiyah Lie algebroids}},
	volume = {148},
	year = {2020}}

@article{Zumino:1984ws,
	author = {Zumino, Bruno},
	doi = {10.1016/0550-3213(85)90543-7},
	journal = {Nucl. Phys. B},
	pages = {477--493},
	reportnumber = {NSF-ITP-84-150},
	title = {Cohomology of Gauge Groups: Cocycles and {Schwinger} Terms},
	volume = {253},
	year = {1985}}

@article{Atiyah:1957,
	author = {M. F. Atiyah},
	doi = {10.2307/1992969},
	journal = {Trans. Amer. Math. Soc.},
	number = {1},
	pages = {181--207},
	publisher = {American Mathematical Society},
	title = {Complex Analytic Connections in Fibre Bundles},
	volume = {85},
	year = {1957}}

@article{ThierryMieg:1982un,
	author = {Thierry-Mieg, Jean and Baulieu, Laurent},
	doi = {10.1016/0550-3213(83)90324-3},
	journal = {Nucl. Phys. B},
	pages = {259--284},
	reportnumber = {PAR LPTHE 82/14},
	title = {Covariant Quantization of non-{Abelian} Antisymmetric Tensor Gauge Theories},
	volume = {228},
	year = {1983}}

@article{marle2008differential,
	archiveprefix = {arXiv},
	author = {Charles-Michel Marle},
	eprint = {0804.2451},
	primaryclass = {math.DG},
	title = {{Differential calculus on a Lie algebroid and Poisson manifolds}},
	year = {2008}}

@article{kosmannschwarzbach2002differential,
	archiveprefix = {arXiv},
	author = {Y. Kosmann-Schwarzbach and K.C.H. Mackenzie},
	eprint = {math/0209337},
	primaryclass = {math.DG},
	title = {Differential Operators and Actions of {Lie} Algebroids},
	year = {2002}}

@article{atiyah1984dirac,
	author = {Atiyah, M. F. and Singer, I. M.},
	doi = {10.1073/pnas.81.8.2597},
	journal = {Proc. Nat. Acad. Sci.},
	pages = {2597--2600},
	title = {Dirac Operators Coupled to Vector Potentials},
	volume = {81},
	year = {1984}}

@article{Henneaux:1989rq,
	author = {Henneaux, Marc},
	doi = {10.1016/0003-4916(89)90274-1},
	journal = {Annals Phys.},
	pages = {281},
	reportnumber = {ULB-TH2-89-01},
	title = {Duality Theorems in {BRST} Cohomology},
	volume = {194},
	year = {1989}}

@article{Bonora:1981rw,
	author = {Bonora, L. and Pasti, P. and Tonin, M.},
	doi = {10.1007/BF02812376},
	journal = {Nuovo Cim. A},
	pages = {307},
	reportnumber = {IFPD 37/81},
	title = {{Extended BRS symmetry in non-Abelian gauge theories}},
	volume = {64},
	year = {1981}}

@article{Fournel:2012uv,
	archiveprefix = {arXiv},
	author = {Fournel, Cedric and Lazzarini, Serge and Masson, Thierry},
	doi = {10.1016/j.geomphys.2012.11.005},
	eprint = {1205.6725},
	journal = {J. Geom. Phys.},
	pages = {174--191},
	primaryclass = {math-ph},
	title = {{Formulation of gauge theories on transitive Lie algebroids}},
	volume = {64},
	year = {2013}}

@inbook{Jordan:2014uza,
	archiveprefix = {arXiv},
	author = {Fran\c{c}ois, Jordan and Serge, Lazzarini and Thierry, Masson},
	booktitle = {{Mathematical Structures of the Universe}},
	editor = {Heller, Michael and Eckstein, Michael and Fast, Sebastian},
	eprint = {1404.4604},
	pages = {177--226},
	primaryclass = {math-ph},
	title = {{Gauge field theories: various mathematical approaches}},
	year = {2014}}

@article{Tyutin:1975qk,
	archiveprefix = {arXiv},
	author = {Tyutin, I.V.},
	eprint = {0812.0580},
	primaryclass = {hep-th},
	reportnumber = {LEBEDEV-75-39},
	title = {Gauge Invariance in Field Theory and Statistical Physics in Operator Formalism},
	year = {1975}}

@inbook{Atiyah:1978,
	address = {Berlin, Heidelberg},
	author = {Atiyah, F.},
	booktitle = {Mathematical Problems in Theoretical Physics: International Conference Held in Rome, June 6--15, 1977},
	doi = {10.1007/3-540-08853-9_18},
	editor = {Dell'Antonio, G. and Doplicher, S. and Jona-Lasinio, G.},
	isbn = {978-3-540-35811-4},
	pages = {216--221},
	publisher = {Springer Berlin Heidelberg},
	title = {{Geometry of Yang-Mills fields}},
	year = {1978}}

@inproceedings{Neeman:1979cvl,
	author = {Ne'eman, Y. and Regge, T. and Thierry-Mieg, J.},
	booktitle = {{19th International Conference on High-Energy Physics}},
	pages = {552--554},
	reportnumber = {TAUP-700-78},
	title = {Ghost-Fields, {BRS} and Extended Supergravity as Applications of Gauge Geometry},
	year = {1979}}

@article{Brandt:1989gv,
	author = {Brandt, Friedemann and Dragon, Norbert and Kreuzer, Miximilian},
	doi = {10.1016/0550-3213(90)90038-F},
	journal = {Nucl. Phys. B},
	pages = {250--260},
	reportnumber = {DESY-89-088, ITP-UH-6-89},
	title = {{Lie algebra cohomology}},
	volume = {332},
	year = {1990}}

@article{fern2007lie,
	archiveprefix = {arXiv},
	author = {Rui Loja Fernandes and Ivan Struchiner},
	eprint = {0712.3198},
	primaryclass = {math.DG},
	title = {Lie Algebroids and Classification Problems in Geometry},
	year = {2007}}

@article{Kotov:2016lpx,
	archiveprefix = {arXiv},
	author = {Kotov, Alexei and Strobl, Thomas},
	doi = {10.1142/S0129055X19500156},
	eprint = {1603.04490},
	journal = {Rev. Math. Phys.},
	number = {04},
	pages = {1950015},
	primaryclass = {math.DG},
	title = {{Lie algebroids, gauge theories, and compatible geometrical structures}},
	volume = {31},
	year = {2018}}

@book{mackenzie_1987,
	author = {Kirill C. H. Mackenzie},
	collection = {London Mathematical Society Lecture Note Series},
	doi = {10.1017/CBO9780511661839},
	place = {Cambridge},
	publisher = {Cambridge University Press},
	series = {London Mathematical Society Lecture Note Series},
	title = {{Lie Groupoids and Lie Algebroids in Differential Geometry}},
	year = {1987}}

@article{Brandt:1996mh,
	archiveprefix = {arXiv},
	author = {Brandt, Friedemann},
	doi = {10.1007/s002200050248},
	eprint = {hep-th/9604025},
	journal = {Commun. Math. Phys.},
	pages = {459--489},
	reportnumber = {KUL-TF-96-8},
	title = {Local {BRST} Cohomology and Covariance},
	volume = {190},
	year = {1997}}

@article{Barnich:2000zw,
	archiveprefix = {arXiv},
	author = {Barnich, Glenn and Brandt, Friedemann and Henneaux, Marc},
	doi = {10.1016/S0370-1573(00)00049-1},
	eprint = {hep-th/0002245},
	journal = {Phys. Rept.},
	pages = {439--569},
	reportnumber = {ITP-UH-03-00, ULB-TH-00-05},
	title = {{Local BRST cohomology in gauge theories}},
	volume = {338},
	year = {2000}}

@article{LC,
	author = {Levi-Civita, T.},
	da = {1916/12/01},
	doi = {10.1007/BF03014898},
	id = {Levi-Civita1916},
	isbn = {0009-725X},
	journal = {Rendiconti del Circolo Matematico di Palermo (1884-1940)},
	number = {1},
	pages = {173--204},
	title = {{Nozione di parallelismo in una variet{\`a} qualunque e conseguente specificazione geometrica della curvatura riemanniana}},
	ty = {JOUR},
	volume = {42},
	year = {1916}}

@article{Carow-Watamura:2016lob,
	archiveprefix = {arXiv},
	author = {Carow-Watamura, Ursula and Heller, Marc Andre and Ikeda, Noriaki and Kaneko, Tomokazu and Watamura, Satoshi},
	doi = {10.1093/ptep/ptx100},
	eprint = {1612.02612},
	journal = {PTEP},
	number = {8},
	pages = {083B01},
	primaryclass = {hep-th},
	reportnumber = {TU-1035},
	title = {Off-Shell Covariantization of Algebroid Gauge Theories},
	volume = {2017},
	year = {2017}}

@article{Henneaux:1995ex,
	archiveprefix = {arXiv},
	author = {Henneaux, Marc},
	doi = {10.1016/0370-2693(95)01449-7},
	eprint = {hep-th/9510116},
	journal = {Phys. Lett. B},
	pages = {163--169},
	title = {{On the gauge fixed BRST cohomology}},
	volume = {367},
	year = {1996}}

@article{Baulieu:1985md,
	author = {Baulieu, L. and Bellon, Marc P.},
	doi = {10.1016/0550-3213(86)90178-1},
	journal = {Nucl. Phys. B},
	pages = {75--124},
	reportnumber = {LPTENS 85/7},
	title = {$p$ Forms and Supergravity: Gauge Symmetries in Curved Space},
	volume = {266},
	year = {1986}}

@article{zbMATH01186367,
	author = {Alan Weinstein},
	doi = {10.1016/S0926-2245(98)00022-9},
	issn = {0926-2245},
	journal = {Differ. Geom. Appl.},
	number = {1},
	pages = {213 - 238},
	title = {Poisson Geometry},
	volume = {9},
	year = {1998}}

@article{Hirshfeld:1980te,
	author = {Hirshfeld, Allen C. and Leschke, Hajo},
	doi = {10.1016/0370-2693(81)90486-X},
	journal = {Phys. Lett. B},
	pages = {48--50},
	reportnumber = {DO-TH 80/24},
	title = {Quantization of {Yang-Mills} Fields in Superspace},
	volume = {101},
	year = {1981}}

@article{DuboisViolette:1991is,
	author = {Dubois-Violette, Michel and Henneaux, Marc and Talon, Michel and Viallet, Claude-Michel},
	doi = {10.1016/0370-2693(91)90527-W},
	journal = {Phys. Lett. B},
	pages = {81--87},
	reportnumber = {PAR-LPTHE-91-19, LPTHE-ORSAY-91-32},
	title = {{Some results on local cohomologies in field theory}},
	volume = {267},
	year = {1991}}

@article{Hoyos:1981pb,
	author = {Hoyos, J. and Quiros, M. and Ramirez Mittelbrunn, J. and de Urries, F.J.},
	doi = {10.1063/1.525523},
	journal = {J. Math. Phys.},
	pages = {1504},
	reportnumber = {Print-81-0641 (MADRID)},
	title = {Superfiber Bundle Structure of Quantized Gauge Theories},
	volume = {23},
	year = {1982}}

@article{RevModPhys.52.175,
	author = {Daniel, M. and Viallet, C. M.},
	doi = {10.1103/RevModPhys.52.175},
	issue = {1},
	journal = {Rev. Mod. Phys.},
	numpages = {0},
	pages = {175--197},
	publisher = {American Physical Society},
	title = {The geometrical setting of gauge theories of the {Yang-Mills} type},
	volume = {52},
	year = {1980}}

@article{Hull:1990qg,
	author = {Hull, C.M. and Spence, Bill J. and Vazquez-Bello, J.L.},
	doi = {10.1016/0550-3213(91)90224-L},
	journal = {Nucl. Phys. B},
	pages = {108--124},
	reportnumber = {QMW-PH-90-11},
	title = {The Geometry of Quantum Gauge Theories: A Superspace formulation of {BRST} symmetry},
	volume = {348},
	year = {1991}}

@article{Baulieu:1981sb,
	author = {Baulieu, L. and Thierry-Mieg, J.},
	doi = {10.1016/0550-3213(82)90454-0},
	journal = {Nucl. Phys. B},
	pages = {477--508},
	reportnumber = {CU-TP-196},
	title = {The Principle of {BRS} Symmetry: An Alternative Approach to {Yang-Mills} Theories},
	volume = {197},
	year = {1982}}

@article{MR0216409,
	author = {Pradines, Jean},
	journal = {C. R. Acad. Sci. Paris S\'er. A-B},
	mrclass = {53.42},
	mrnumber = {0216409 (35 \#7242)},
	mrreviewer = {K.-T. Chen},
	pages = {A245--A248},
	title = {{Th\'eorie de {L}ie pour les groupo\"i des diff\'erentiables. {C}alcul diff\'erenetiel dans la cat\'egorie des groupo\"i des infinit\'esimaux}},
	volume = {264},
	year = {1967}}

@article{MR0214103,
	author = {Pradines, Jean},
	journal = {C. R. Acad. Sci. Paris S\'er. A-B},
	mrclass = {57.70},
	mrnumber = {0214103 (35 \#4954)},
	mrreviewer = {H. H. Johnson},
	pages = {A907--A910},
	title = {{Th\'eorie de {L}ie pour les groupo\"i des diff\'erentiables. {R}elations entre propri\'et\'es locales et globales}},
	volume = {263},
	year = {1966}}

@article{Bandelloni:1986wz,
	author = {Bandelloni, G.},
	doi = {10.1063/1.527323},
	journal = {J. Math. Phys.},
	pages = {2551--2557},
	title = {{Yang-Mills} Cohomology in Four-dimensions},
	volume = {27},
	year = {1986}}

@article{Wu:1975es,
	author = {Wu, Tai Tsun and Yang, Chen Ning},
	doi = {10.1103/PhysRevD.12.3845},
	editor = {Hsu, Jong-Ping and Fine, D.},
	journal = {Phys. Rev. D},
	pages = {3845--3857},
	reportnumber = {ITP-SB-75-31},
	title = {Concept of Nonintegrable Phase Factors and Global Formulation of Gauge Fields},
	volume = {12},
	year = {1975}}

@article{Yang:1954ek,
	author = {Yang, Chen-Ning and Mills, Robert L.},
	doi = {10.1103/PhysRev.96.191},
	editor = {Hsu, Jong-Ping and Fine, D.},
	journal = {Phys. Rev.},
	pages = {191--195},
	title = {Conservation of Isotopic Spin and Isotopic Gauge Invariance},
	volume = {96},
	year = {1954}}

@article{Nelson:1984gu,
	author = {Nelson, Philip C. and Alvarez-Gaume, Luis},
	doi = {10.1007/BF01466595},
	journal = {Commun. Math. Phys.},
	pages = {103},
	reportnumber = {HUTP-84/A067, NSF-ITP-84-149},
	title = {Hamiltonian Interpretation of Anomalies},
	volume = {99},
	year = {1985}}

@article{Bonora:2021cxn,
	archiveprefix = {arXiv},
	author = {Bonora, Loriano and Malik, Rudra Prakash},
	doi = {10.3390/universe7080280},
	eprint = {2107.12796},
	journal = {Universe},
	number = {8},
	pages = {280},
	primaryclass = {hep-th},
	reportnumber = {SISSA 17/2021/FISI},
	title = {{BRST} and Superfield Formalism\textemdash{}A Review},
	volume = {7},
	year = {2021}}

@inproceedings{Thierry-Mieg:1985zmg,
	author = {Thierry-Mieg, Jean},
	booktitle = {{Symposium on Anomalies, Geometry, Topology}},
	reportnumber = {LBL-19671},
	title = {{Classical geometrical interpretation of ghost fields and anomalies in Yang-Mills theory and quantum gravity}},
	year = {1985}}

@article{Boulanger:2007st,
	archiveprefix = {arXiv},
	author = {Boulanger, Nicolas},
	doi = {10.1088/1126-6708/2007/07/069},
	eprint = {0704.2472},
	journal = {JHEP},
	pages = {069},
	primaryclass = {hep-th},
	title = {{General solutions of the Wess-Zumino consistency condition for the Weyl anomalies}},
	volume = {07},
	year = {2007}}

@article{Boulanger:2018rxo,
	archiveprefix = {arXiv},
	author = {Boulanger, Nicolas and Fran\c{c}ois, Jordan and Lazzarini, Serge},
	doi = {10.1088/1751-8121/ab01af},
	eprint = {1809.05445},
	journal = {J. Phys. A},
	number = {11},
	pages = {115201},
	primaryclass = {math-ph},
	title = {{A classification of global conformal invariants}},
	volume = {52},
	year = {2019}}

@article{Francois:2015oca,
	archiveprefix = {arXiv},
	author = {Fran\c{c}ois, Jordan and Lazzarini, Serge and Masson, Thierry},
	doi = {10.1007/JHEP09(2015)195},
	eprint = {1504.08297},
	journal = {JHEP},
	pages = {195},
	primaryclass = {math-ph},
	title = {{Residual Weyl symmetry out of conformal geometry and its BRST structure}},
	volume = {09},
	year = {2015}}

@article{Francois:2015pg,
	archiveprefix = {arXiv},
	author = {Fran\c{c}ois, Jordan and Lazzarini, Serge and Masson, Thierry},
	doi = {10.1063/1.4943595},
	eprint = {1508.07666},
	journal = {J. Math. Phys.},
	number = {3},
	pages = {033504},
	primaryclass = {math-ph},
	title = {{Becchi-Rouet-Stora-Tyutin structure for the mixed Weyl-diffeomorphism residual symmetry}},
	volume = {57},
	year = {2016}}

@article{klinger2023abc,
    author = "Klinger, Marc S. and Leigh, Robert G. and Pai, Pin-Chun",
    title = "{Extended phase space in general gauge theories}",
    eprint = "2303.06786",
    archivePrefix = "arXiv",
    primaryClass = "hep-th",
    doi = "10.1016/j.nuclphysb.2023.116404",
    journal = "Nucl. Phys. B",
    volume = "998",
    pages = "116404",
    year = "2024"
}

@article{Ciambelli:2021ujl,
	archiveprefix = {arXiv},
	author = {Ciambelli, Luca and Leigh, Robert G.},
	doi = {10.1016/j.nuclphysb.2021.115553},
	eprint = {2101.03974},
	journal = {Nucl. Phys. B},
	pages = {115553},
	primaryclass = {hep-th},
	title = {Lie algebroids and the geometry of off-shell {BRST}},
	volume = {972},
	year = {2021}}

@article{Ciambelli:2019bzz,
	archiveprefix = {arXiv},
	author = {Ciambelli, Luca and Leigh, Robert G.},
	doi = {10.1103/PhysRevD.101.086020},
	eprint = {1905.04339},
	journal = {Phys. Rev. D},
	number = {8},
	pages = {086020},
	primaryclass = {hep-th},
	reportnumber = {CPHT-RR020.052019},
	title = {{Weyl connections and their role in holography}},
	volume = {101},
	year = {2020}}

@article{Ciambelli:2021vnn,
	archiveprefix = {arXiv},
	author = {Ciambelli, Luca and Leigh, Robert G.},
	doi = {10.1103/PhysRevD.104.046005},
	eprint = {2104.07643},
	journal = {Phys. Rev. D},
	number = {4},
	pages = {046005},
	primaryclass = {hep-th},
	title = {{Isolated surfaces and symmetries of gravity}},
	volume = {104},
	year = {2021}}

@article{Ciambelli:2022cfr,
	archiveprefix = {arXiv},
	author = {Ciambelli, Luca and Leigh, Robert G.},
	doi = {10.1016/j.nuclphysb.2022.116053},
	eprint = {2207.06441},
	journal = {Nucl. Phys. B},
	pages = {116053},
	primaryclass = {hep-th},
	title = {{Universal corner symmetry and the orbit method for gravity}},
	volume = {986},
	year = {2023}}

@article{Bilal:2008qx,
	archiveprefix = {arXiv},
	author = {Bilal, Adel},
	eprint = {0802.0634},
	primaryclass = {hep-th},
	reportnumber = {LPTENS-08-05},
	title = {{Lectures on anomalies}},
	year = {2008}}

@article{bardeen1984consistent,
	author = {Bardeen, William A. and Zumino, Bruno},
	doi = {10.1016/0550-3213(84)90322-5},
	journal = {Nucl. Phys. B},
	pages = {421--453},
	reportnumber = {FERMILAB-PUB-84-038-T, LBL-17639, UCB-PTH-84-12, FERMILAB-PUB-84-038-T-REV, LBL-17639-REV, UCB-PTH-84-12-REV},
	title = {Consistent and Covariant Anomalies in Gauge and Gravitational Theories},
	volume = {244},
	year = {1984}}

@article{bonora1983some,
	author = {Bonora, L. and Cotta-Ramusino, P.},
	doi = {10.1007/BF01208267},
	journal = {Commun. Math. Phys.},
	pages = {589},
	reportnumber = {HUTMP 81/B112},
	title = {{Some remarks on BRS transformations, anomalies and the cohomology of the Lie algebra of the group of gauge transformations}},
	volume = {87},
	year = {1983}}

@article{kanno1989weil,
	author = {Kanno, Hiroaki},
	doi = {10.1007/BF01506544},
	journal = {Z. Phys. C},
	pages = {477},
	reportnumber = {RIMS-647},
	title = {Weil Algebra Structure and Geometrical Meaning of {BRST} Transformation in Topological Quantum Field Theory},
	volume = {43},
	year = {1989}}

@article{thierry1980explicit,
	author = {Thierry-Mieg, Jean},
	doi = {10.1007/BF02732091},
	journal = {Nuovo Cim. A},
	pages = {396},
	reportnumber = {CALT-68-733},
	title = {Explicit Classical Construction of the {Faddeev-Popov} Ghost Field},
	volume = {56},
	year = {1980}}

@article{thierry1980geometrical,
	author = {Thierry-Mieg, Jean},
	doi = {10.1063/1.524385},
	journal = {J. Math. Phys.},
	pages = {2834--2838},
	title = {{Geometrical reinterpretation of Faddeev-Popov ghost particles and BRS transformations}},
	volume = {21},
	year = {1980}}

@article{quiros1981geometrical,
	author = {Quiros, M. and de Urries, F. J. and Hoyos, J. and Mazon, M. L. and Rodriguez, E.},
	doi = {10.1063/1.525123},
	journal = {J. Math. Phys.},
	pages = {1767},
	reportnumber = {Print-80-0650 (MADRID)},
	title = {Geometrical Structure of {Faddeev-Popov} Fields and Invariance Properties of Gauge Theories},
	volume = {22},
	year = {1981}}

@article{witten1983global,
	author = {Witten, Edward},
	doi = {10.1016/0550-3213(83)90063-9},
	journal = {Nucl. Phys. B},
	pages = {422--432},
	reportnumber = {PRINT-83-0262 (PRINCETON)},
	title = {Global Aspects of Current Algebra},
	volume = {223},
	year = {1983}}

@article{callan1985anomalies,
	author = {Callan, Jr., Curtis G. and Harvey, Jeffrey A.},
	doi = {10.1016/0550-3213(85)90489-4},
	journal = {Nucl. Phys. B},
	pages = {427--436},
	reportnumber = {Print-84-0860 (PRINCETON)},
	title = {Anomalies and Fermion Zero Modes on Strings and Domain Walls},
	volume = {250},
	year = {1985}}

@article{fernandes2002lie,
	archiveprefix = {arXiv},
	author = {Fernandes, Rui Loja},
	doi = {https://doi.org/10.1006/aima.2001.2070},
	eprint = {math/0007132},
	journal = {Adv. Math.},
	number = {1},
	pages = {119--179},
	primaryclass = {math.DG},
	publisher = {Elsevier},
	title = {Lie Algebroids, Holonomy and Characteristic Classes},
	volume = {170},
	year = {2002}}

@book{nakahara2018geometry,
	address = {Boca Raton},
	author = {Nakahara, Mikio},
	doi = {10.1201/9781315275826},
	publisher = {CRC press},
	title = {Geometry, Topology and Physics},
	year = {2018}}

@article{alvarez1984topological,
	author = {Alvarez-Gaume, Luis and Ginsparg, Paul H.},
	doi = {10.1016/0550-3213(84)90487-5},
	journal = {Nucl. Phys. B},
	pages = {449--474},
	reportnumber = {HUTP-83/A081},
	title = {The Topological Meaning of Non-{Abelian} Anomalies},
	volume = {243},
	year = {1984}}

@article{fujikawa1979path,
	author = {Fujikawa, Kazuo},
	doi = {10.1103/PhysRevLett.42.1195},
	journal = {Phys. Rev. Lett.},
	pages = {1195--1198},
	reportnumber = {INS-328},
	title = {Path Integral Measure for Gauge Invariant Fermion Theories},
	volume = {42},
	year = {1979}}

@article{becchi1976renormalization,
	author = {Becchi, C. and Rouet, A. and Stora, R.},
	doi = {10.1016/0003-4916(76)90156-1},
	journal = {Ann. Phys.},
	pages = {287--321},
	reportnumber = {CPT-75-P.723-MARSEILLE},
	title = {Renormalization of Gauge Theories},
	volume = {98},
	year = {1976}}

@article{wess1971consequences,
	author = {Wess, J. and Zumino, B.},
	doi = {10.1016/0370-2693(71)90582-X},
	journal = {Phys. Lett. B},
	pages = {95--97},
	title = {Consequences of Anomalous {Ward} Identities},
	volume = {37},
	year = {1971}}

@book{mackenzie2005general,
	address = {Cambridge},
	author = {Mackenzie, Kirill CH},
	doi = {10.1017/CBO9781107325883},
	publisher = {Cambridge University Press},
	title = {{General Theory of Lie Groupoids and Lie Algebroids}},
	year = {2005}}

@article{stone2012gravitational,
	archiveprefix = {arXiv},
	author = {Stone, Michael},
	doi = {10.1103/PhysRevB.85.184503},
	eprint = {1201.4095},
	journal = {Phys. Rev. B},
	pages = {184503},
	primaryclass = {cond-mat.mes-hall},
	reportnumber = {NSF-KITP-12-005},
	title = {Gravitational Anomalies and Thermal {Hall} Effect in Topological Insulators},
	volume = {85},
	year = {2012}}

@article{Hughes:2012vg,
	archiveprefix = {arXiv},
	author = {Hughes, Taylor L. and Leigh, Robert G. and Parrikar, Onkar},
	doi = {10.1103/PhysRevD.88.025040},
	eprint = {1211.6442},
	journal = {Phys. Rev. D},
	number = {2},
	pages = {025040},
	primaryclass = {hep-th},
	title = {Torsional Anomalies, {Hall} Viscosity, and Bulk-boundary Correspondence in Topological States},
	volume = {88},
	year = {2013}}

@article{Parrikar:2014usa,
    author = "Parrikar, Onkar and Hughes, Taylor L. and Leigh, Robert G.",
    title = {Torsion, Parity-odd Response and Anomalies in Topological States},
    eprint = "1407.7043",
    archivePrefix = "arXiv",
    primaryClass = "cond-mat.mes-hall",
    doi = "10.1103/PhysRevD.90.105004",
    journal = "Phys. Rev. D",
    volume = "90",
    number = "10",
    pages = "105004",
    year = "2014"
}

@inproceedings{Witten:2019bou,
	archiveprefix = {arXiv},
	author = {Witten, Edward and Yonekura, Kazuya},
	booktitle = {{The Shoucheng Zhang Memorial Workshop}},
	eprint = {1909.08775},
	primaryclass = {hep-th},
	title = {Anomaly Inflow and the $\eta$-Invariant},
	year = {2019}}

@article{zumino1984chiral,
	author = {Zumino, Bruno and Wu, Yong-Shi and Zee, A.},
	doi = {10.1016/0550-3213(84)90259-1},
	journal = {Nucl. Phys. B},
	pages = {477--507},
	reportnumber = {LBL-16443, DOE-ER-40048-18-P3},
	title = {Chiral Anomalies, Higher Dimensions, and Differential Geometry},
	volume = {239},
	year = {1984}}

@article{adler1969axial,
	author = {Adler, Stephen L.},
	doi = {10.1103/PhysRev.177.2426},
	journal = {Phys. Rev.},
	pages = {2426--2438},
	title = {Axial Vector Vertex in Spinor Electrodynamics},
	volume = {177},
	year = {1969}}

@article{t1976symmetry,
	author = {'t Hooft, Gerard},
	doi = {10.1103/PhysRevLett.37.8},
	editor = {Shifman, Mikhail A.},
	journal = {Phys. Rev. Lett.},
	pages = {8--11},
	reportnumber = {PRINT-76-0254 (HARVARD)},
	title = {Symmetry Breaking through {Bell-Jackiw} Anomalies},
	volume = {37},
	year = {1976}}

@article{bell1969pcac,
	author = {Bell, J. S. and Jackiw, R.},
	doi = {10.1007/BF02823296},
	journal = {Nuovo Cim. A},
	pages = {47--61},
	title = {A {PCAC} Puzzle: $\pi^0 \to \gamma \gamma$ in the $\sigma$-model},
	volume = {60},
	year = {1969}}

@incollection{zumino1985chiral,
	author = {Bruno Zumino},
	booktitle = {Current Algebra and Anomalies},
	doi = {10.2307/j.ctt7ztmc8.7},
	pages = {361--392},
	publisher = {Princeton University Press},
	title = {Chiral Anomalies and Differential Geometry},
	year = {1985}}

@article{Stora:1983ct,
    author = "Stora, Raymond",
    editor = "'t Hooft, Gerard and Jaffe, A. and Lehmann, H. and Mitter, P. K. and Singer, I. M. and Stora, R.",
    title = "ALGEBRAIC STRUCTURE AND TOPOLOGICAL ORIGIN OF ANOMALIES",
    reportNumber = "LAPP-TH-94",
    year = "1983"
}

@article{chern1946characteristic,
	author = {Chern, Shiing-Shen},
	doi = {10.2307/1969037},
	journal = {Ann. Math.},
	pages = {85--121},
	publisher = {JSTOR},
	title = {Characteristic Classes of {Hermitian} Manifolds},
	year = {1946}}

@article{zbMATH03070474,
	author = {Cartan, H.},
	journal = {S\'eminaire Henri Cartan},
	language = {fr},
	publisher = {Secr\'etariat math\'ematique},
	title = {Cohomologie r\'eelle d'un espace fibr\'e principal diff\'erentiable. {I} : notions d'alg\`ebre diff\'erentielle, alg\`ebre de {Weil} d'un groupe de {Lie}},
	volume = {2},
	year = {1949-1950}}

@incollection{WeilLetter,
	address = {Berlin},
	author = {Weil, Andre},
	booktitle = {Andr\'e Weil Oeuvres Scientifiques/Collected Papers, vol. 1 (1926-1951)},
	language = {French},
	page = {422-436},
	publisher = {Springer},
	title = {{G{\'e}om{\'e}trie diff{\'e}rentielle des espaces fibr{\'e}s (letters to Chevalley and Koszul), 1949}},
	year = {1979}}

@article{becchi1974abelian,
    author = "Becchi, C. and Rouet, A. and Stora, R.",
    title = "{The Abelian Higgs-Kibble model, unitarity of the S-operator}",
    reportNumber = "CPT-74/P.625-MARSEILLE",
    doi = "10.1016/0370-2693(74)90058-6",
    journal = "Phys. Lett. B",
    volume = "52",
    pages = "344--346",
    year = "1974"
}

@article{chevalley1948cohomology,
	author = {Chevalley, Claude and Eilenberg, Samuel},
	doi = {10.1090/S0002-9947-1948-0024908-8},
	journal = {Trans. Am. Math. Soc.},
	pages = {85--124},
	title = {{Cohomology theory of Lie groups and Lie algebras}},
	volume = {63},
	year = {1948}}

@article{lazzarini2012connections,
	author = {Lazzarini, Serge and Masson, Thierry},
	doi = {10.1016/j.geomphys.2011.11.002},
	journal = {J. Geom. Phys.},
	number = {2},
	pages = {387--402},
	publisher = {Elsevier},
	title = {Connections on Lie algebroids and on derivation-based noncommutative geometry},
	volume = {62},
	year = {2012}}

@article{chern1966geometry,
	author = {Chern, Shiing-Shen},
	doi = {bams/1183527777},
	journal = {Bull. Amer. Math. Soc.},
	number = {2},
	pages = {167--219},
	title = {{The geometry of G-structures}},
	volume = {72},
	year = {1966}}

@book{Gockeler:1987an,
	author = {{G\"ockeler}, M. and {Sch\"ucker}, Thomas},
	doi = {10.1017/CBO9780511628818},
	publisher = {Cambridge University Press},
	series = {Cambridge Monographs on Mathematical Physics},
	title = {{Differential Geometry, Gauge Theories, and Gravity}},
	year = {1989}}

@article{Jia:2023gmk,
    author = "Jia, Weizhen and Karydas, Manthos and Leigh, Robert G.",
    title = "{Weyl-ambient geometries}",
    eprint = "2301.06628",
    archivePrefix = "arXiv",
    primaryClass = "hep-th",
    doi = "10.1016/j.nuclphysb.2023.116224",
    journal = "Nucl. Phys. B",
    volume = "991",
    pages = "116224",
    year = "2023"
}

@article{graham2009extended,
	archiveprefix = {arXiv},
	author = {Graham, C Robin},
	eprint = {0810.4203},
	journal = {Adv. Math.},
	number = {6},
	pages = {1956--1985},
	primaryclass = {math.DG},
	publisher = {Elsevier},
	title = {Extended obstruction tensors and renormalized volume coefficients},
	volume = {220},
	doi ={10.1016/j.aim.2008.11.015},
	year = {2009}}

@article{Parisini:2023nbd,
    author = "Parisini, Enrico and Skenderis, Kostas and Withers, Benjamin",
    title = "The Ambient Space Formalism",
    eprint = "2312.03820",
    archivePrefix = "arXiv",
    primaryClass = "hep-th",
    year = "2023"
}

@article{Ciambelli:2022vot,
    author = "Ciambelli, Luca",
    title = "From Asymptotic Symmetries to the Corner Proposal",
    eprint = "2212.13644",
    archivePrefix = "arXiv",
    primaryClass = "hep-th",
    doi = "10.22323/1.435.0002",
    journal = "PoS",
    volume = "Modave2022",
    pages = "002",
    year = "2023"
}

@article{Ciambelli:2023ott,
    author = "Ciambelli, Luca and Delfante, Arnaud and Ruzziconi, Romain and Zwikel, C\'eline",
    title = "{Symmetries and charges in Weyl-Fefferman-Graham gauge}",
    eprint = "2308.15480",
    archivePrefix = "arXiv",
    primaryClass = "hep-th",
    doi = "10.1103/PhysRevD.108.126003",
    journal = "Phys. Rev. D",
    volume = "108",
    number = "12",
    pages = "126003",
    year = "2023"
}

@article{Ghilencea:2023sti,
    author = "Ghilencea, D. M.",
    title = "{Weyl conformal geometry vs Weyl anomaly}",
    eprint = "2309.11372",
    archivePrefix = "arXiv",
    primaryClass = "hep-th",
    doi = "10.1007/JHEP10(2023)113",
    journal = "JHEP",
    volume = "10",
    pages = "113",
    year = "2023"
}

@article{Zanusso:2023vkn,
    author = "Zanusso, Omar",
    title = "{Consequences of gauging the Weyl symmetry and the two-dimensional conformal anomaly}",
    eprint = "2309.09598",
    archivePrefix = "arXiv",
    primaryClass = "hep-th",
    doi = "10.1103/PhysRevD.108.125018",
    journal = "Phys. Rev. D",
    volume = "108",
    number = "12",
    pages = "125018",
    year = "2023"
}

@article{Condeescu:2023izl,
    author = "Condeescu, C. and Ghilencea, D. M. and Micu, A.",
    title = "{Weyl quadratic gravity as a gauge theory and non-metricity vs torsion duality}",
    eprint = "2312.13384",
    archivePrefix = "arXiv",
    primaryClass = "hep-th",
    doi = "10.1140/epjc/s10052-024-12644-6",
    journal = "Eur. Phys. J. C",
    volume = "84",
    number = "3",
    pages = "292",
    year = "2024"
}

@article{Jia:2023tki,
    author = "Jia, Weizhen and Klinger, Marc S. and Leigh, Robert G.",
    title = "{BRST cohomology is Lie algebroid cohomology}",
    eprint = "2303.05540",
    archivePrefix = "arXiv",
    primaryClass = "hep-th",
    doi = "10.1016/j.nuclphysb.2023.116317",
    journal = "Nucl. Phys. B",
    volume = "994",
    pages = "116317",
    year = "2023"
}

@book{Leigh,
	author = {Leigh, Robert G.},
	title = {{Strings, Holography and Applications}},
	note = {Unpublished Lecture Notes},
	}

@article{Imbimbo_2000,
    author = "Imbimbo, C. and Schwimmer, A. and Theisen, S. and Yankielowicz, S.",
    editor = "Lechtenfeld, O. and Louis, J. and Lust, Dieter and Nicolai, H.",
    title = "{Diffeomorphisms and holographic anomalies}",
    eprint = "hep-th/9910267",
    archivePrefix = "arXiv",
    reportNumber = "TAUP-2590-99",
    doi = "10.1088/0264-9381/17/5/322",
    journal = "Class. Quant. Grav.",
    volume = "17",
    pages = "1129--1138",
    year = "2000"
}

@article{Bautier:2000mz,
    author = "Bautier, K. and Englert, F. and Rooman, M. and Spindel, P.",
    title = "{The Fefferman-Graham ambiguity and AdS black holes}",
    eprint = "hep-th/0002156",
    archivePrefix = "arXiv",
    reportNumber = "ULB-TH-00-03, UMH-MG-00-01",
    doi = "10.1016/S0370-2693(00)00339-7",
    journal = "Phys. Lett. B",
    volume = "479",
    pages = "291--298",
    year = "2000"
}

@article{Rooman:2000ei,
    author = "Rooman, M. and Spindel, Ph.",
    title = "{Uniqueness of the asymptotic AdS(3) geometry}",
    eprint = "gr-qc/0011005",
    archivePrefix = "arXiv",
    reportNumber = "ULB-TH-00-25, UMH-MG-00-05",
    doi = "10.1088/0264-9381/18/11/309",
    journal = "Class. Quant. Grav.",
    volume = "18",
    pages = "2117--2124",
    year = "2001"
}

@article{Henningson:1998gx,
    author = "Henningson, M. and Skenderis, K.",
    title = "{The holographic Weyl anomaly}",
    eprint = "hep-th/9806087",
    archivePrefix = "arXiv",
    reportNumber = "CERN-TH-98-188, KUL-TF-98-21",
    doi = "10.1088/1126-6708/1998/07/023",
    journal = "JHEP",
    volume = "07",
    pages = "023",
    year = "1998"
}

@article{Capper1974,
    author = "Capper, D. M. and Duff, M. J.",
    title = "{Trace anomalies in dimensional regularization}",
    doi = "10.1007/BF02748300",
    journal = "Nuovo Cim. A",
    volume = "23",
    pages = "173--183",
    year = "1974"
}

@article{Deser1976,
    author = "Deser, Stanley and Duff, M. J. and Isham, C. J.",
    title = "Nonlocal Conformal Anomalies",
    reportNumber = "Print-76-0374 (KING S COLL.)",
    doi = "10.1016/0550-3213(76)90480-6",
    journal = "Nucl. Phys. B",
    volume = "111",
    pages = "45--55",
    year = "1976"
}

@article{Duff1977,
    author = "Duff, M. J.",
    title = "Observations on Conformal Anomalies",
    reportNumber = "Print-77-0361 (QUEEN MARY COLL.)",
    doi = "10.1016/0550-3213(77)90410-2",
    journal = "Nucl. Phys. B",
    volume = "125",
    pages = "334--348",
    year = "1977"
}

@article{Polyakov1981,
    author = "Polyakov, Alexander M.",
    editor = "Khalatnikov, I. M. and Mineev, V. P.",
    title = "Quantum Geometry of Bosonic Strings",
    reportNumber = "Print-81-0351 (LANDAU INST)",
    doi = "10.1016/0370-2693(81)90743-7",
    journal = "Phys. Lett. B",
    volume = "103",
    pages = "207--210",
    year = "1981"
}

@article{Fradkin:1983tg,
    author = "Fradkin, E. S. and Tseytlin, Arkady A.",
    title = "Conformal Anomaly in {Weyl} Theory and Anomaly Free Superconformal Theories",
    reportNumber = "LEBEDEV-83-180, LEBEDEV-83-185",
    doi = "10.1016/0370-2693(84)90668-3",
    journal = "Phys. Lett. B",
    volume = "134",
    pages = "187",
    year = "1984"
}

@article{Bonora:1985cq,
    author = "Bonora, L. and Pasti, P. and Bregola, M.",
    title = "WEYL COCYCLES",
    reportNumber = "DFPD-28/85",
    doi = "10.1088/0264-9381/3/4/018",
    journal = "Class. Quant. Grav.",
    volume = "3",
    pages = "635",
    year = "1986"
}

@article{Osborn1991,
    author = "Osborn, H.",
    title = "{Weyl consistency conditions and a local renormalization group equation for general renormalizable field theories}",
    reportNumber = "DAMTP-91-1",
    doi = "10.1016/0550-3213(91)80030-P",
    journal = "Nucl. Phys. B",
    volume = "363",
    pages = "486--526",
    year = "1991"
}

@article{Deser1993,
    author = "Deser, Stanley and Schwimmer, A.",
    title = "{Geometric classification of conformal anomalies in arbitrary dimensions}",
    eprint = "hep-th/9302047",
    archivePrefix = "arXiv",
    reportNumber = "BRX-343, SISSA-14-93-EP",
    doi = "10.1016/0370-2693(93)90934-A",
    journal = "Phys. Lett. B",
    volume = "309",
    pages = "279--284",
    year = "1993"
}

@article{Liu:1998bu,
    author = "Liu, Hong and Tseytlin, Arkady A.",
    title = "{D = 4 super Yang-Mills, D = 5 gauged supergravity, and D = 4 conformal supergravity}",
    eprint = "hep-th/9804083",
    archivePrefix = "arXiv",
    reportNumber = "IMPERIAL-TP-97-98-39, NSF-ITP-98-049",
    doi = "10.1016/S0550-3213(98)00443-X",
    journal = "Nucl. Phys. B",
    volume = "533",
    pages = "88--108",
    year = "1998"
}

@article{deHaro:2000vlm,
    author = "de Haro, Sebastian and Solodukhin, Sergey N. and Skenderis, Kostas",
    title = "{Holographic reconstruction of space-time and renormalization in the AdS / CFT correspondence}",
    eprint = "hep-th/0002230",
    archivePrefix = "arXiv",
    reportNumber = "SPIN-2000-05, ITP-UU-00-03, PUTP-1921",
    doi = "10.1007/s002200100381",
    journal = "Commun. Math. Phys.",
    volume = "217",
    pages = "595--622",
    year = "2001"
}

@article{Skenderis2002,
    author = "Skenderis, Kostas",
    editor = "de Wit, B. and Vandoren, S.",
    title = "{Lecture notes on holographic renormalization}",
    eprint = "hep-th/0209067",
    archivePrefix = "arXiv",
    reportNumber = "PUTP-2047",
    doi = "10.1088/0264-9381/19/22/306",
    journal = "Class. Quant. Grav.",
    volume = "19",
    pages = "5849--5876",
    year = "2002"
}

@article{Manvelyan:2001pv,
    author = "Manvelyan, R. and Mkrtchian, R. and Muller-Kirsten, H. J. W.",
    title = "{Holographic trace anomaly and cocycle of Weyl group}",
    eprint = "hep-th/0103082",
    archivePrefix = "arXiv",
    reportNumber = "KL-TH-01-02",
    doi = "10.1016/S0370-2693(01)00550-0",
    journal = "Phys. Lett. B",
    volume = "509",
    pages = "143--150",
    year = "2001"
}

@article{BONORA1983305,
    author = "Bonora, L. and Cotta-Ramusino, P. and Reina, C.",
    title = "Conformal Anomaly and Cohomology",
    reportNumber = "CERN-TH-3542",
    doi = "10.1016/0370-2693(83)90169-7",
    journal = "Phys. Lett. B",
    volume = "126",
    pages = "305--308",
    year = "1983"
}

@article{Callan1970,
    author = "Callan, Jr., Curtis G. and Coleman, Sidney R. and Jackiw, Roman",
    title = "{A new improved energy-momentum tensor}",
    doi = "10.1016/0003-4916(70)90394-5",
    journal = "Annals Phys.",
    volume = "59",
    pages = "42--73",
    year = "1970"
}

@article{BELINFANTE1940449,
title = {On the current and the density of the electric charge, the energy, the linear momentum and the angular momentum of arbitrary fields},
journal = {Physica},
volume = {7},
number = {5},
pages = {449-474},
year = {1940},
issn = {0031-8914},
doi = {10.1016/S0031-8914(40)90091-X},
author = {F.J. Belinfante},
}

@article{Mazur2001,
    author = "Mazur, Pawel O. and Mottola, Emil",
    title = "{Weyl cohomology and the effective action for conformal anomalies}",
    eprint = "hep-th/0106151",
    archivePrefix = "arXiv",
    reportNumber = "LA-UR-01-3121",
    doi = "10.1103/PhysRevD.64.104022",
    journal = "Phys. Rev. D",
    volume = "64",
    pages = "104022",
    year = "2001"
}

@article{Boulanger:2004zf,
    author = "Boulanger, Nicolas and Erdmenger, Johanna",
    title = "{A classification of local Weyl invariants in $D$ = 8}",
    eprint = "hep-th/0405228",
    archivePrefix = "arXiv",
    reportNumber = "HU-EP-04-28, DAMTP-2004-48",
    doi = "10.1088/0264-9381/21/18/003",
    journal = "Class. Quant. Grav.",
    volume = "21",
    number = "18",
    pages = "4305--4316",
    year = "2004",
    note = "[Erratum: Class.Quant.Grav. 39, 039501 (2022)]"
}

@article{Frampton:1983nr,
    author = "Frampton, Paul H. and Kephart, Thomas W.",
    title = "The Analysis of Anomalies in Higher Space-time Dimensions",
    reportNumber = "IFP-193-UNC",
    doi = "10.1103/PhysRevD.28.1010",
    journal = "Phys. Rev. D",
    volume = "28",
    pages = "1010",
    year = "1983"
}

@article{Zumino:1983rz,
    author = "Zumino, Bruno and Wu, Yong-Shi and Zee, A.",
    title = "Chiral Anomalies, Higher Dimensions, and Differential Geometry",
    reportNumber = "LBL-16443, DOE-ER-40048-18-P3",
    doi = "10.1016/0550-3213(84)90259-1",
    journal = "Nucl. Phys. B",
    volume = "239",
    pages = "477--507",
    year = "1984"
}

@book{Bertlmann:1996xk,
    author = "Bertlmann, R. A.",
    title = "{Anomalies in quantum field theory}",
    year = "1996",
    publisher = {Oxford University Press},
    address = {Oxford},
    doi={10.1093/acprof:oso/9780198507628.001.0001}
}

@article{Anastasiou:2020zwc,
    author = "Anastasiou, Giorgos and Miskovic, Olivera and Olea, Rodrigo and Papadimitriou, Ioannis",
    title = "{Counterterms, Kounterterms, and the variational problem in AdS gravity}",
    eprint = "2003.06425",
    archivePrefix = "arXiv",
    primaryClass = "hep-th",
    reportNumber = "KIAS-P20010",
    doi = "10.1007/JHEP08(2020)061",
    journal = "JHEP",
    volume = "08",
    pages = "061",
    year = "2020"
}

@article{Campoleoni:2022wmf,
    author = "Campoleoni, Andrea and Ciambelli, Luca and Delfante, Arnaud and Marteau, Charles and Petropoulos, P. Marios and Ruzziconi, Romain",
    title = "{Holographic Lorentz and Carroll frames}",
    eprint = "2208.07575",
    archivePrefix = "arXiv",
    primaryClass = "hep-th",
    doi = "10.1007/JHEP12(2022)007",
    journal = "JHEP",
    volume = "12",
    pages = "007",
    year = "2022"
}

@article{Noether:1918zz,
    author = "Noether, Emmy",
    title = "Invariant Variation Problems",
    eprint = "physics/0503066",
    archivePrefix = "arXiv",
    doi = "10.1080/00411457108231446",
    journal = "Gott. Nachr.",
    volume = "1918",
    pages = "235--257",
    year = "1918"
}

@book{Wigner1931,
    author = "Eugen Wigner",
    title = "{Gruppentheorie und ihre Anwendung auf die Quantenmechanik der Atomspektren}",
    year = "1931",
    publisher = {Vieweg},
    address = {Wiesbaden},
    doi={10.1007/978-3-663-02555-9}
}

@book{Weyl1928,
    author = "Hermann Weyl",
    title = "{Gruppentheorie und Quantenmechanik}",
    year = "1928",
    publisher = {Hirzel},
    address = {Leipzig}
}

@book{Liang1,
    author = "Canbin Liang and Bin Zhou",
    title = "{Differential Geometry and General Relativity}",
    volume = "1",
    year = "2023",
    publisher = {Springer},
    address = {Singapore},
    translator = {Weizhen Jia and Bin Zhou},
    doi = {10.1007/978-981-99-0022-0},
    url-1 = {https://doi.org/10.1007/978-981-99-0022-0}}

@book{Liang3,
    author = "Canbin Liang and Bin Zhou",
    title = "{Differential Geometry and General Relativity}",
    volume = "3",
    year = "2009",
    publisher = {Science Press},
    address = {Beijing},
    note = {in Chinese}
}

@book{Hatcher,
    author = "Allen Hatcher",
    title = "{Algebraic Topology}",
    year = "2002",
    publisher = {Cambridge University Press},
    address = {Cambridge}
}

@article{tHooft:1993dmi,
    author = "'t Hooft, Gerard",
    title = "{Dimensional reduction in quantum gravity}",
    eprint = "gr-qc/9310026",
    archivePrefix = "arXiv",
    reportNumber = "THU-93-26",
    journal = "Conf. Proc. C",
    volume = "930308",
    pages = "284--296",
    year = "1993"
}

@article{Susskind:1994vu,
    author = "Susskind, Leonard",
    title = "{The world as a hologram}",
    eprint = "hep-th/9409089",
    archivePrefix = "arXiv",
    reportNumber = "SU-ITP-94-33",
    doi = "10.1063/1.531249",
    journal = "J. Math. Phys.",
    volume = "36",
    pages = "6377--6396",
    year = "1995"
}

@article{Herzog:2007ij,
    author = "Herzog, Christopher P. and Kovtun, Pavel and Sachdev, Subir and Son, Dam Thanh",
    title = "{Quantum critical transport, duality, and M-theory}",
    eprint = "hep-th/0701036",
    archivePrefix = "arXiv",
    reportNumber = "NSF-KITP-06-129, INT-PUB-07-01",
    doi = "10.1103/PhysRevD.75.085020",
    journal = "Phys. Rev. D",
    volume = "75",
    pages = "085020",
    year = "2007"
}

@article{Hartnoll:2009sz,
    author = "Hartnoll, Sean A.",
    editor = "Uranga, A. M.",
    title = "{Lectures on holographic methods for condensed matter physics}",
    eprint = "0903.3246",
    archivePrefix = "arXiv",
    primaryClass = "hep-th",
    doi = "10.1088/0264-9381/26/22/224002",
    journal = "Class. Quant. Grav.",
    volume = "26",
    pages = "224002",
    year = "2009"
}

@article{Sachdev:2010ch,
    author = "Sachdev, Subir",
    title = "{Condensed matter and AdS/CFT}",
    eprint = "1002.2947",
    archivePrefix = "arXiv",
    primaryClass = "hep-th",
    doi = "10.1007/978-3-642-04864-7_9",
    journal = "Lect. Notes Phys.",
    volume = "828",
    pages = "273--311",
    year = "2011"
}

@article{Kim:2011ey,
    author = "Kim, Youngman and Yi, Deokhyun",
    title = "{Holography at work for nuclear and hadron physics}",
    eprint = "1107.0155",
    archivePrefix = "arXiv",
    primaryClass = "hep-ph",
    doi = "10.1155/2011/259025",
    journal = "Adv. High Energy Phys.",
    volume = "2011",
    pages = "259025",
    year = "2011"
}

@article{Pahlavani:2014dma,
    author = "Pahlavani, M. R. and Morad, R.",
    title = "{Application of AdS/CFT in nuclear physics}",
    eprint = "1403.2501",
    archivePrefix = "arXiv",
    primaryClass = "hep-th",
    doi = "10.1155/2014/863268",
    journal = "Adv. High Energy Phys.",
    volume = "2014",
    pages = "863268",
    year = "2014"
}

@article{Aoki:2012th,
    author = "Aoki, Sinya and Hashimoto, Koji and Iizuka, Norihiro",
    title = "Matrix Theory for Baryons: An Overview of Holographic {QCD} for Nuclear Physics",
    eprint = "1203.5386",
    archivePrefix = "arXiv",
    primaryClass = "hep-th",
    reportNumber = "CERN-PH-TH-2012-072, RIKEN-MP-42",
    doi = "10.1088/0034-4885/76/10/104301",
    journal = "Rept. Prog. Phys.",
    volume = "76",
    pages = "104301",
    year = "2013"
}

@article{Policastro:2002se,
    author = "Policastro, Giuseppe and Son, Dam T. and Starinets, Andrei O.",
    title = "{From AdS/CFT correspondence to hydrodynamics}",
    eprint = "hep-th/0205052",
    archivePrefix = "arXiv",
    reportNumber = "INT-PUB-02-32",
    doi = "10.1088/1126-6708/2002/09/043",
    journal = "JHEP",
    volume = "09",
    pages = "043",
    year = "2002"
}

@article{Bhattacharyya:2007vjd,
    author = "Bhattacharyya, Sayantani and Hubeny, Veronika E and Minwalla, Shiraz and Rangamani, Mukund",
    title = "Nonlinear Fluid Dynamics from Gravity",
    eprint = "0712.2456",
    archivePrefix = "arXiv",
    primaryClass = "hep-th",
    reportNumber = "TIFR-TH-07-44, DCPT-07-73, NI07097",
    doi = "10.1088/1126-6708/2008/02/045",
    journal = "JHEP",
    volume = "02",
    pages = "045",
    year = "2008"
}

@article{Haack:2008cp,
    author = "Haack, Michael and Yarom, Amos",
    title = "{Nonlinear viscous hydrodynamics in various dimensions using AdS/CFT}",
    eprint = "0806.4602",
    archivePrefix = "arXiv",
    primaryClass = "hep-th",
    reportNumber = "LMU-ASC-39-08",
    doi = "10.1088/1126-6708/2008/10/063",
    journal = "JHEP",
    volume = "10",
    pages = "063",
    year = "2008"
}

@inproceedings{Hubeny:2011hd,
    author = "Hubeny, Veronika E. and Minwalla, Shiraz and Rangamani, Mukund",
    title = "{The fluid/gravity correspondence}",
    booktitle = "{Theoretical Advanced Study Institute in Elementary Particle Physics}: {String theory and its Applications: From meV to the Planck Scale}",
    eprint = "1107.5780",
    archivePrefix = "arXiv",
    primaryClass = "hep-th",
    pages = "348--383",
    year = "2012"
}

@article{Banks:2020dus,
    author = "Banks, Tom",
    title = "Holographic Space-time and Quantum Information",
    eprint = "2001.08205",
    archivePrefix = "arXiv",
    primaryClass = "hep-th",
    reportNumber = "RUNHETC-2020-04",
    doi = "10.3389/fphy.2020.00111",
    journal = "Front. in Phys.",
    volume = "8",
    pages = "111",
    year = "2020"
}

@article{Almheiri:2014lwa,
    author = "Almheiri, Ahmed and Dong, Xi and Harlow, Daniel",
    title = "Bulk Locality and Quantum Error Correction in {AdS/CFT}",
    eprint = "1411.7041",
    archivePrefix = "arXiv",
    primaryClass = "hep-th",
    reportNumber = "SU-ITP-14-30",
    doi = "10.1007/JHEP04(2015)163",
    journal = "JHEP",
    volume = "04",
    pages = "163",
    year = "2015"
}

@article{Dong:2016eik,
    author = "Dong, Xi and Harlow, Daniel and Wall, Aron C.",
    title = "Reconstruction of Bulk Operators within the Entanglement Wedge in Gauge-Gravity Duality",
    eprint = "1601.05416",
    archivePrefix = "arXiv",
    primaryClass = "hep-th",
    reportNumber = "NSF-KITP-16-005",
    doi = "10.1103/PhysRevLett.117.021601",
    journal = "Phys. Rev. Lett.",
    volume = "117",
    number = "2",
    pages = "021601",
    year = "2016"
}

@article{Pastawski:2015qua,
    author = "Pastawski, Fernando and Yoshida, Beni and Harlow, Daniel and Preskill, John",
    title = "{Holographic quantum error-correcting codes: Toy models for the bulk/boundary correspondence}",
    eprint = "1503.06237",
    archivePrefix = "arXiv",
    primaryClass = "hep-th",
    doi = "10.1007/JHEP06(2015)149",
    journal = "JHEP",
    volume = "06",
    pages = "149",
    year = "2015"
}

@article{Hayden:2016cfa,
    author = "Hayden, Patrick and Nezami, Sepehr and Qi, Xiao-Liang and Thomas, Nathaniel and Walter, Michael and Yang, Zhao",
    title = "{Holographic duality from random tensor networks}",
    eprint = "1601.01694",
    archivePrefix = "arXiv",
    primaryClass = "hep-th",
    doi = "10.1007/JHEP11(2016)009",
    journal = "JHEP",
    volume = "11",
    pages = "009",
    year = "2016"
}

@article{Ji:1994av,
    author = "Ji, Xiang-Dong",
    title = "{A QCD analysis of the mass structure of the nucleon}",
    eprint = "hep-ph/9410274",
    archivePrefix = "arXiv",
    reportNumber = "MIT-CTP-2368",
    doi = "10.1103/PhysRevLett.74.1071",
    journal = "Phys. Rev. Lett.",
    volume = "74",
    pages = "1071--1074",
    year = "1995"
}

@article{Ji:1995sv,
    author = "Ji, Xiang-Dong",
    title = "{Breakup of hadron masses and energy - momentum tensor of QCD}",
    eprint = "hep-ph/9502213",
    archivePrefix = "arXiv",
    reportNumber = "MIT-CTP-2407",
    doi = "10.1103/PhysRevD.52.271",
    journal = "Phys. Rev. D",
    volume = "52",
    pages = "271--281",
    year = "1995"
}

@book{Polchinski:1998rq,
    author = "Polchinski, J.",
    title = "{String Theory. Vol. 1: An Introduction to the Bosonic String}",
    doi = "10.1017/CBO9780511816079",
    publisher = "Cambridge University Press",
    series = "Cambridge Monographs on Mathematical Physics",
    year = "2007"
}

@book{Polchinski:1998rr,
    author = "Polchinski, J.",
    title = "{String Theory. Vol. 2: Superstring Theory and Beyond}",
    doi = "10.1017/CBO9780511618123",
    publisher = "Cambridge University Press",
    series = "Cambridge Monographs on Mathematical Physics",
    year = "2007"
}

@article{Zamolodchikov:1986gt,
    author = "Zamolodchikov, A. B.",
    title = "Irreversibility of the Flux of the Renormalization Group in a 2D Field Theory",
    journal = "JETP Lett.",
    volume = "43",
    pages = "730--732",
    year = "1986"
}

@article{Komargodski:2011vj,
    author = "Komargodski, Zohar and Schwimmer, Adam",
    title = "On Renormalization Group Flows in Four Dimensions",
    eprint = "1107.3987",
    archivePrefix = "arXiv",
    primaryClass = "hep-th",
    doi = "10.1007/JHEP12(2011)099",
    journal = "JHEP",
    volume = "12",
    pages = "099",
    year = "2011"
}

@article{Gaiotto:2014kfa,
    author = "Gaiotto, Davide and Kapustin, Anton and Seiberg, Nathan and Willett, Brian",
    title = "Generalized Global Symmetries",
    eprint = "1412.5148",
    archivePrefix = "arXiv",
    primaryClass = "hep-th",
    doi = "10.1007/JHEP02(2015)172",
    journal = "JHEP",
    volume = "02",
    pages = "172",
    year = "2015"
}

@article{Lawler_2004,
   title={{Quantum Hall smectics, sliding symmetry, and the renormalization group}},
   volume={70},
   DOI={10.1103/physrevb.70.165310},
   number={16},
   journal={Phys. Rev. B},
   publisher={American Physical Society (APS)},
   author={Lawler, Michael J. and Fradkin, Eduardo},
   year={2004} }

@article{Seiberg:2020bhn,
    author = "Seiberg, Nathan and Shao, Shu-Heng",
    title = "Exotic Symmetries, Duality, and Fractons in 2+1-Dimensional Quantum Field Theory",
    eprint = "2003.10466",
    archivePrefix = "arXiv",
    primaryClass = "cond-mat.str-el",
    doi = "10.21468/SciPostPhys.10.2.027",
    journal = "SciPost Phys.",
    volume = "10",
    number = "2",
    pages = "027",
    year = "2021"
}

@article{Bhardwaj:2017xup,
    author = "Bhardwaj, Lakshya and Tachikawa, Yuji",
    title = "{On finite symmetries and their gauging in two dimensions}",
    eprint = "1704.02330",
    archivePrefix = "arXiv",
    primaryClass = "hep-th",
    reportNumber = "IPMU-17-0049",
    doi = "10.1007/JHEP03(2018)189",
    journal = "JHEP",
    volume = "03",
    pages = "189",
    year = "2018"
}

@article{Chang:2018iay,
    author = "Chang, Chi-Ming and Lin, Ying-Hsuan and Shao, Shu-Heng and Wang, Yifan and Yin, Xi",
    title = "Topological Defect Lines and Renormalization Group Flows in Two Dimensions",
    eprint = "1802.04445",
    archivePrefix = "arXiv",
    primaryClass = "hep-th",
    reportNumber = "CALT-TH-2017-067, CALT-TH 2017-067, PUPT-2546",
    doi = "10.1007/JHEP01(2019)026",
    journal = "JHEP",
    volume = "01",
    pages = "026",
    year = "2019"
}

@article{Shao:2023gho,
    author = "Shao, Shu-Heng",
    title = "What's Done Cannot Be Undone: {TASI} Lectures on Non-Invertible Symmetries",
    eprint = "2308.00747",
    archivePrefix = "arXiv",
    primaryClass = "hep-th",
    reportNumber = "YITP-SB-2023-19",
    year = "2023"
}

@article{Vijay:2016phm,
    author = "Vijay, Sagar and Haah, Jeongwan and Fu, Liang",
    title = "Fracton Topological Order, Generalized Lattice Gauge Theory and Duality",
    eprint = "1603.04442",
    archivePrefix = "arXiv",
    primaryClass = "cond-mat.str-el",
    doi = "10.1103/PhysRevB.94.235157",
    journal = "Phys. Rev. B",
    volume = "94",
    number = "23",
    pages = "235157",
    year = "2016"
}

@article{Schwinger:1951nm,
    author = "Schwinger, Julian S.",
    editor = "Milton, K. A.",
    title = "{On gauge invariance and vacuum polarization}",
    doi = "10.1103/PhysRev.82.664",
    journal = "Phys. Rev.",
    volume = "82",
    pages = "664--679",
    year = "1951"
}

@article{Johnson:1963vz,
    author = "Johnson, K.",
    title = "{gamma(5) invariance}",
    doi = "10.1016/S0375-9601(63)95573-7",
    journal = "Phys. Lett.",
    volume = "5",
    pages = "253--255",
    year = "1963"
}

@article{Cheng:2022sgb,
    author = "Cheng, Meng and Seiberg, Nathan",
    title = "{Lieb-Schultz-Mattis, Luttinger, and 't Hooft - anomaly matching in lattice systems}",
    eprint = "2211.12543",
    archivePrefix = "arXiv",
    primaryClass = "cond-mat.str-el",
    doi = "10.21468/SciPostPhys.15.2.051",
    journal = "SciPost Phys.",
    volume = "15",
    number = "2",
    pages = "051",
    year = "2023"
}

@article{tHooft:1979rat,
    author = "'t Hooft, Gerard",
    editor = "'t Hooft, Gerard and Itzykson, C. and Jaffe, A. and Lehmann, H. and Mitter, P. K. and Singer, I. M. and Stora, R.",
    title = "{Naturalness, chiral symmetry, and spontaneous chiral symmetry breaking}",
    reportNumber = "PRINT-80-0083 (UTRECHT)",
    doi = "10.1007/978-1-4684-7571-5_9",
    journal = "NATO Sci. Ser. B",
    volume = "59",
    pages = "135--157",
    year = "1980"
}

@article{Cordova:2019bsd,
    author = "C\'ordova, Clay and Ohmori, Kantaro",
    title = "Anomaly Obstructions to Symmetry Preserving Gapped Phases",
    eprint = "1910.04962",
    archivePrefix = "arXiv",
    primaryClass = "hep-th",
    year = "2019"
}

@article{Gaiotto:2017yup,
    author = "Gaiotto, Davide and Kapustin, Anton and Komargodski, Zohar and Seiberg, Nathan",
    title = "Theta, Time Reversal, and Temperature",
    eprint = "1703.00501",
    archivePrefix = "arXiv",
    primaryClass = "hep-th",
    doi = "10.1007/JHEP05(2017)091",
    journal = "JHEP",
    volume = "05",
    pages = "091",
    year = "2017"
}

@article{Lieb:1961fr,
    author = "Lieb, Elliott H. and Schultz, Theodore and Mattis, Daniel",
    title = "{Two soluble models of an antiferromagnetic chain}",
    doi = "10.1016/0003-4916(61)90115-4",
    journal = "Annals Phys.",
    volume = "16",
    pages = "407--466",
    year = "1961"
}

@article{Cheng:2015kce,
    author = "Cheng, Meng and Zaletel, Michael and Barkeshli, Maissam and Vishwanath, Ashvin and Bonderson, Parsa",
    title = "Translational Symmetry and Microscopic Constraints on Symmetry-Enriched Topological Phases: A View from the Surface",
    eprint = "1511.02263",
    archivePrefix = "arXiv",
    primaryClass = "cond-mat.str-el",
    doi = "10.1103/PhysRevX.6.041068",
    journal = "Phys. Rev. X",
    volume = "6",
    number = "4",
    pages = "041068",
    year = "2016"
}

@article{Cho:2017fgz,
    author = "Cho, Gil Young and Ryu, Shinsei and Hsieh, Chang-Tse",
    title = "Anomaly Manifestation of {Lieb-Schultz-Mattis} Theorem and Topological Phases",
    eprint = "1705.03892",
    archivePrefix = "arXiv",
    primaryClass = "cond-mat.str-el",
    doi = "10.1103/PhysRevB.96.195105",
    journal = "Phys. Rev. B",
    volume = "96",
    number = "19",
    pages = "195105",
    year = "2017"
}

@book{Weinberg:1996kr,
    author = "Weinberg, Steven",
    title = "{The Quantum Theory of Fields. Vol. 2: Modern Applications}",
    doi = "10.1017/CBO9781139644174",
    publisher = "Cambridge University Press",
    year = "2013"
}

@article{Green:1984sg,
    author = "Green, Michael B. and Schwarz, John H.",
    title = "Anomaly Cancellation in Supersymmetric D=10 Gauge Theory and Superstring Theory",
    reportNumber = "CALT-68-1182",
    doi = "10.1016/0370-2693(84)91565-X",
    journal = "Phys. Lett. B",
    volume = "149",
    pages = "117--122",
    year = "1984"
}

@article{Fradkin:1986pq,
    author = "Fradkin, Eduardo H. and Dagotto, Elbio and Boyanovsky, Daniel",
    title = "Physical Realization of the Parity Anomaly in Condensed Matter Physics",
    reportNumber = "ILL-TH-86-10, ILL-337-86-10",
    doi = "10.1103/PhysRevLett.57.2967",
    journal = "Phys. Rev. Lett.",
    volume = "57",
    pages = "2967",
    year = "1986",
    note = "[Erratum: Phys.Rev.Lett. 58, 961 (1987)]"
}

@article{Stone:1990iw,
    author = "Stone, Michael",
    title = "Edge Waves in the Quantum {Hall} Effect",
    reportNumber = "ILL-TH-90-8",
    doi = "10.1016/0003-4916(91)90177-A",
    journal = "Annals Phys.",
    volume = "207",
    pages = "38--52",
    year = "1991"
}

@article{Freed:2014iua,
    author = "Freed, Daniel S.",
    editor = "Donagi, Ron and Douglas, Michael R. and Kamenova, Ljudmila and Rocek, Martin",
    title = "Anomalies and Invertible Field Theories",
    eprint = "1404.7224",
    archivePrefix = "arXiv",
    primaryClass = "hep-th",
    doi = "10.1090/pspum/088/01462",
    journal = "Proc. Symp. Pure Math.",
    volume = "88",
    pages = "25--46",
    year = "2014"
}

@article{Freed:2016rqq,
    author = "Freed, Daniel S. and Hopkins, Michael J.",
    title = "{Reflection positivity and invertible topological phases}",
    eprint = "1604.06527",
    archivePrefix = "arXiv",
    primaryClass = "hep-th",
    doi = "10.2140/gt.2021.25.1165",
    journal = "Geom. Topol.",
    volume = "25",
    pages = "1165--1330",
    year = "2021"
}

@article{Monnier:2019ytc,
    author = "Monnier, Samuel",
    title = "A Modern Point of View on Anomalies",
    eprint = "1903.02828",
    archivePrefix = "arXiv",
    primaryClass = "hep-th",
    doi = "10.1002/prop.201910012",
    journal = "Fortsch. Phys.",
    volume = "67",
    number = "8-9",
    pages = "1910012",
    year = "2019"
}

@article{Chen:2011pg,
    author = "Chen, Xie and Gu, Zheng-Cheng and Liu, Zheng-Xin and Wen, Xiao-Gang",
    title = "{Symmetry protected topological orders and the group cohomology of their symmetry group}",
    eprint = "1106.4772",
    archivePrefix = "arXiv",
    primaryClass = "cond-mat.str-el",
    doi = "10.1103/PhysRevB.87.155114",
    journal = "Phys. Rev. B",
    volume = "87",
    number = "15",
    pages = "155114",
    year = "2013"
}

@article{Gu:2009dr,
    author = "Gu, Zheng-Cheng and Wen, Xiao-Gang",
    title = "Tensor-Entanglement-Filtering Renormalization Approach and Symmetry Protected Topological Order",
    eprint = "0903.1069",
    archivePrefix = "arXiv",
    primaryClass = "cond-mat.str-el",
    doi = "10.1103/PhysRevB.80.155131",
    journal = "Phys. Rev. B",
    volume = "80",
    pages = "155131",
    year = "2009"
}

@article{Senthil:2014ooa,
    author = "Senthil, T.",
    title = "Symmetry Protected Topological phases of Quantum Matter",
    eprint = "1405.4015",
    archivePrefix = "arXiv",
    primaryClass = "cond-mat.str-el",
    doi = "10.1146/annurev-conmatphys-031214-014740",
    journal = "Ann. Rev. Condens. Matter Phys.",
    volume = "6",
    pages = "299",
    year = "2015"
}

@article{Witten:1982fp,
    author = "Witten, Edward",
    editor = "Shifman, Mikhail A.",
    title = "An {SU(2)} Anomaly",
    doi = "10.1016/0370-2693(82)90728-6",
    journal = "Phys. Lett. B",
    volume = "117",
    pages = "324--328",
    year = "1982"
}

@article{Wang:2018qoy,
    author = "Wang, Juven and Wen, Xiao-Gang and Witten, Edward",
    title = "A New {SU(2)} Anomaly",
    eprint = "1810.00844",
    archivePrefix = "arXiv",
    primaryClass = "hep-th",
    doi = "10.1063/1.5082852",
    journal = "J. Math. Phys.",
    volume = "60",
    number = "5",
    pages = "052301",
    year = "2019"
}

@article{Alvarez-Gaume:1984zst,
    author = "Alvarez-Gaume, Luis and Della Pietra, S. and Moore, Gregory W.",
    title = "Anomalies and Odd Dimensions",
    reportNumber = "HUTP-84-A028",
    doi = "10.1016/0003-4916(85)90383-5",
    journal = "Annals Phys.",
    volume = "163",
    pages = "288",
    year = "1985"
}

@article{Niemi:1983rq,
    author = "Niemi, A. J. and Semenoff, G. W.",
    title = "Axial Anomaly Induced Fermion Fractionization and Effective Gauge Theory Actions in Odd Dimensional Space-Times",
    reportNumber = "Print-83-0988 (IAS,PRINCETON)",
    doi = "10.1103/PhysRevLett.51.2077",
    journal = "Phys. Rev. Lett.",
    volume = "51",
    pages = "2077",
    year = "1983"
}

@article{Redlich:1983kn,
    author = "Redlich, A. N.",
    editor = "Taylor, J. C.",
    title = "Gauge Noninvariance and Parity Violation of Three-Dimensional Fermions",
    reportNumber = "MIT-CTP-1107",
    doi = "10.1103/PhysRevLett.52.18",
    journal = "Phys. Rev. Lett.",
    volume = "52",
    pages = "18",
    year = "1984"
}

@article{Redlich:1983dv,
    author = "Redlich, A. N.",
    editor = "Wilczek, Frank",
    title = "Parity Violation and Gauge Noninvariance of the Effective Gauge Field Action in Three-Dimensions",
    reportNumber = "MIT-CTP-1128",
    doi = "10.1103/PhysRevD.29.2366",
    journal = "Phys. Rev. D",
    volume = "29",
    pages = "2366--2374",
    year = "1984"
}

@article{Witten:2015aba,
    author = "Witten, Edward",
    title = "Fermion Path Integrals And Topological Phases",
    eprint = "1508.04715",
    archivePrefix = "arXiv",
    primaryClass = "cond-mat.mes-hall",
    doi = "10.1103/RevModPhys.88.035001",
    journal = "Rev. Mod. Phys.",
    volume = "88",
    number = "3",
    pages = "035001",
    year = "2016"
}

@article{Atiyah:1968mp,
    author = "Atiyah, M. F. and Singer, I. M.",
    title = "{The Index of elliptic operators. I}",
    doi = "10.2307/1970715",
    journal = "Annals Math.",
    volume = "87",
    pages = "484--530",
    year = "1968"
}

@article{Atiyah:1975jf,
    author = "Atiyah, M. F. and Patodi, V. K. and Singer, I. M.",
    title = "{Spectral asymmetry and Riemannian geometry. I}",
    doi = "10.1017/S0305004100049410",
    journal = "Math. Proc. Cambridge Phil. Soc.",
    volume = "77",
    pages = "43",
    year = "1975"
}

@article{Kapustin:2014tfa,
    author = "Kapustin, Anton",
    title = "Symmetry Protected Topological Phases, Anomalies, and Cobordisms: Beyond Group Cohomology",
    eprint = "1403.1467",
    archivePrefix = "arXiv",
    primaryClass = "cond-mat.str-el",
    year = "2014"
}

@article{Kapustin:2014dxa,
    author = "Kapustin, Anton and Thorngren, Ryan and Turzillo, Alex and Wang, Zitao",
    title = "Fermionic Symmetry Protected Topological Phases and Cobordisms",
    eprint = "1406.7329",
    archivePrefix = "arXiv",
    primaryClass = "cond-mat.str-el",
    doi = "10.1007/JHEP12(2015)052",
    journal = "JHEP",
    volume = "12",
    pages = "052",
    year = "2015"
}

@article{Yonekura:2018ufj,
    author = "Yonekura, Kazuya",
    title = "{On the cobordism classification of symmetry protected topological phases}",
    eprint = "1803.10796",
    archivePrefix = "arXiv",
    primaryClass = "hep-th",
    reportNumber = "IPMU-18-0040",
    doi = "10.1007/s00220-019-03439-y",
    journal = "Commun. Math. Phys.",
    volume = "368",
    number = "3",
    pages = "1121--1173",
    year = "2019"
}

@article{Faddeev:1967fc,
    author = "Faddeev, L. D. and Popov, V. N.",
    editor = "Hsu, Jong-Ping and Fine, D.",
    title = "{Feynman diagrams for the Yang-Mills field}",
    doi = "10.1016/0370-2693(67)90067-6",
    journal = "Phys. Lett. B",
    volume = "25",
    pages = "29--30",
    year = "1967"
}

@article{Ramandi:2014qoa,
    author = "Ramandi, Gh. Fasihi and Boroojerdian, N.",
    title = "Forces Unification in the Framework of Transitive {Lie} Algebroids",
    doi = "10.1007/s10773-014-2357-5",
    journal = "Int. J. Theor. Phys.",
    volume = "54",
    number = "5",
    pages = "1581--1593",
    year = "2015"
}

@inproceedings{Roytenberg:2002nu,
    author = "Roytenberg, Dmitry",
    title = "{On the structure of graded symplectic supermanifolds and Courant algebroids}",
    booktitle = "{Workshop on Quantization, Deformations, and New Homological and Categorical Methods in Mathematical Physics}",
    eprint = "math/0203110",
    archivePrefix = "arXiv",
    year = "2002"
}

@article{Kosmann-Schwarzbach:2003skv,
    author = "Kosmann-Schwarzbach, Yvette",
    title = "{Derived brackets}",
    eprint = "math/0312524",
    archivePrefix = "arXiv",
    doi = "10.1007/s11005-004-0608-8",
    journal = "Lett. Math. Phys.",
    volume = "69",
    pages = "61--87",
    year = "2004"
}

@article{Blohmann:2010jd,
    author = "Blohmann, Christian and Fernandes, Marco Cezar Barbosa and Weinstein, Alan",
    title = "{Groupoid symmetry and constraints in general relativity}",
    eprint = "1003.2857",
    archivePrefix = "arXiv",
    primaryClass = "math.DG",
    doi = "10.1142/S0219199712500617",
    journal = "Commun. Contemp. Math.",
    volume = "15",
    number = "01",
    pages = "1250061",
    year = "2013"
}

@article{Blumenhagen:2012nt,
    author = "Blumenhagen, Ralph and Deser, Andreas and Plauschinn, Erik and Rennecke, Felix",
    title = "{Non-geometric strings, symplectic gravity and differential geometry of Lie algebroids}",
    eprint = "1211.0030",
    archivePrefix = "arXiv",
    primaryClass = "hep-th",
    reportNumber = "MPP-2012-147, DFPD-2012-TH-14",
    doi = "10.1007/JHEP02(2013)122",
    journal = "JHEP",
    volume = "02",
    pages = "122",
    year = "2013"
}

@article{Dai:1994kq,
    author = "Dai, Xian-zhe and Freed, Daniel S.",
    title = "{Eta-invariants and determinant lines}",
    eprint = "hep-th/9405012",
    archivePrefix = "arXiv",
    doi = "10.1063/1.530747",
    journal = "J. Math. Phys.",
    volume = "35",
    pages = "5155--5194",
    year = "1994",
    note = "[Erratum: J.Math.Phys. 42, 2343--2344 (2001)]"
}

@book{kobayashi2012transformation,
	address = {Heidelberg},
	author = {Kobayashi, Shoshichi},
	doi = {10.1007/978-3-642-61981-6},
	publisher = {Springer Berlin},
	title = {Transformation Groups in Differential Geometry},
	year = {1995}}

@book{Wald,
    author = "Wald, Robert M.",
    title = "{General Relativity}",
    doi = "10.7208/chicago/9780226870373.001.0001",
    publisher = "Chicago University Press",
    address = "Chicago",
    year = "1984"
}

@article{Northe:2022tjr,
    author = "Northe, Christian and Zhang, Chunxu and Wawrzy\'nczak, Rafa\l{} and Gooth, Johannes and Galeski, Stanislaw and Hankiewicz, Ewelina M.",
    title = "{Conformal anomaly in magnetic finite temperature response of strongly interacting one-dimensional spin systems}",
    eprint = "2210.07972",
    archivePrefix = "arXiv",
    primaryClass = "cond-mat.str-el",
    year = "2022"
}

@article{Klinger:2023auu,
    author = "Klinger, Marc S. and Leigh, Robert G.",
    title = "Crossed Products, Conditional Expectations and Constraint Quantization",
    eprint = "2312.16678",
    archivePrefix = "arXiv",
    primaryClass = "hep-th",
    year = "2023"
}

@article{Gomis:2017ixy,
    author = "Gomis, Jaume and Komargodski, Zohar and Seiberg, Nathan",
    title = "Phases Of Adjoint {QCD}$_3$ And Dualities",
    eprint = "1710.03258",
    archivePrefix = "arXiv",
    primaryClass = "hep-th",
    doi = "10.21468/SciPostPhys.5.1.007",
    journal = "SciPost Phys.",
    volume = "5",
    number = "1",
    pages = "007",
    year = "2018"
}

@article{Gribov:1977wm,
    author = "Gribov, V. N.",
    editor = "Nyiri, J.",
    title = "Quantization of Nonabelian Gauge Theories",
    reportNumber = "LENINGRAD-77-367",
    doi = "10.1016/0550-3213(78)90175-X",
    journal = "Nucl. Phys. B",
    volume = "139",
    pages = "1",
    year = "1978"
}

@article{Vandersickel:2012tz,
    author = "Vandersickel, N. and Zwanziger, Daniel",
    title = "{The Gribov problem and QCD dynamics}",
    eprint = "1202.1491",
    archivePrefix = "arXiv",
    primaryClass = "hep-th",
    doi = "10.1016/j.physrep.2012.07.003",
    journal = "Phys. Rept.",
    volume = "520",
    pages = "175--251",
    year = "2012"
}

@article{AliAhmad:2024wja,
    author = "Ali Ahmad, Shadi and Chemissany, Wissam and Klinger, Marc S. and Leigh, Robert G.",
    title = "Quantum Reference Frames from Top-Down Crossed Products",
    eprint = "2405.13884",
    archivePrefix = "arXiv",
    primaryClass = "hep-th",
    year = "2024"
}

@article{Schwimmer:2010za,
    author = "Schwimmer, A. and Theisen, S.",
    title = "Spontaneous Breaking of Conformal Invariance and Trace Anomaly Matching",
    eprint = "1011.0696",
    archivePrefix = "arXiv",
    primaryClass = "hep-th",
    doi = "10.1016/j.nuclphysb.2011.02.003",
    journal = "Nucl. Phys. B",
    volume = "847",
    pages = "590--611",
    year = "2011"
}

@article{Schwimmer:2023nzk,
    author = "Schwimmer, Adam and Theisen, Stefan",
    title = "{Comments on trace anomaly matching}",
    eprint = "2307.14957",
    archivePrefix = "arXiv",
    primaryClass = "hep-th",
    doi = "10.1088/1751-8121/ad0012",
    journal = "J. Phys. A",
    volume = "56",
    number = "46",
    pages = "465402",
    year = "2023"
}

@article{Bah:2019rgq,
    author = "Bah, Ibrahima and Bonetti, Federico and Minasian, Ruben and Nardoni, Emily",
    title = "{Anomalies of QFTs from M-theory and holography}",
    eprint = "1910.04166",
    archivePrefix = "arXiv",
    primaryClass = "hep-th",
    doi = "10.1007/JHEP01(2020)125",
    journal = "JHEP",
    volume = "01",
    pages = "125",
    year = "2020"
}

@article{Apruzzi:2021nmk,
    author = "Apruzzi, Fabio and Bonetti, Federico and Garc\'\i{}a Etxebarria, I\~naki and Hosseini, Saghar S. and Schafer-Nameki, Sakura",
    title = "{Symmetry TFTs from string theory}",
    eprint = "2112.02092",
    archivePrefix = "arXiv",
    primaryClass = "hep-th",
    doi = "10.1007/s00220-023-04737-2",
    journal = "Commun. Math. Phys.",
    volume = "402",
    number = "1",
    pages = "895--949",
    year = "2023"
}

@article{Gubser:2010ze,
    author = "Gubser, Steven S.",
    title = "{Symmetry constraints on generalizations of Bjorken flow}",
    eprint = "1006.0006",
    archivePrefix = "arXiv",
    primaryClass = "hep-th",
    reportNumber = "PUPT-2340",
    doi = "10.1103/PhysRevD.82.085027",
    journal = "Phys. Rev. D",
    volume = "82",
    pages = "085027",
    year = "2010"
}

@article{Gubser:2010ui,
    author = "Gubser, Steven S. and Yarom, Amos",
    title = "{Conformal hydrodynamics in Minkowski and de Sitter spacetimes}",
    eprint = "1012.1314",
    archivePrefix = "arXiv",
    primaryClass = "hep-th",
    reportNumber = "PUPT-2358",
    doi = "10.1016/j.nuclphysb.2011.01.012",
    journal = "Nucl. Phys. B",
    volume = "846",
    pages = "469--511",
    year = "2011"
}

@article{JLN,
    author = "Jia, Weizhen and Leigh, Robert G. and Noronha, Jorge",
    title = "Conformal Hydrodynamics from Ambient Space",
    note = "To Appear"
}

\end{document}